\documentclass[11pt]{article}
\usepackage{graphicx}
\usepackage[utf8]{inputenc}
\usepackage[margin=1in]{geometry}
\usepackage{amsmath,amssymb,amsfonts,amsthm,stmaryrd}
\usepackage{thmtools,thm-restate}
\usepackage[colorlinks=true,linkcolor=blue,allcolors=blue]{hyperref}
\usepackage{todonotes}
\usepackage{physics}
\usepackage[ruled,vlined,linesnumbered]{algorithm2e}
\usepackage[capitalize,nameinlink]{cleveref}
\usepackage{dsfont}
\usepackage{quantikz}

\newcommand{\opn}[1]{\operatorname{#1}}
\newcommand{\im}{\opn{im}}
\newcommand{\spn}{\opn{span}}

\newcommand{\Enc}{\opn{Enc}}

\newcommand{\supp}{\opn{supp}}

\newcommand{\poly}{\opn{poly}}

\newcommand{\1}{\mathds{1}}
\newcommand{\Cl}[1]{{#1}^{\mathsf{C}}}
\newcommand{\Qu}[1]{{#1}^{\mathsf{Q}}}
\newcommand{\Rf}[1]{{#1}^{\mathsf{R}}}
\newcommand{\Ad}[1]{{#1}^{\mathsf{A}}}
\newcommand{\gI}{\mathsf{I}}
\newcommand{\gX}{\mathsf{X}}
\newcommand{\gY}{\mathsf{Y}}
\newcommand{\gZ}{\mathsf{Z}}
\newcommand{\gH}{\mathsf{H}}
\newcommand{\gCNOT}{\mathsf{CNOT}}
\newcommand{\gInitX}{\mathsf{Init}_{\gX}}
\newcommand{\gInitZ}{\mathsf{Init}_{\gZ}}
\newcommand{\gTerm}{\mathsf{Term}}
\newcommand{\gMX}{\mathsf{M}_{\gX}}
\newcommand{\gMZ}{\mathsf{M}_{\gZ}}
\newcommand{\gCt}[1]{\mathsf{C-}{#1}}
\newcommand{\gCF}{\mathsf{CF}}
\newcommand{\Rd}{\opn{Rd}}
\newcommand{\Rdn}{\opn{\overline{Rd}}}
\newcommand{\Lev}{\opn{Lev}}
\newcommand{\Shad}{\opn{Shad}}
\newcommand{\id}{\opn{id}}
\newcommand{\comp}[1]{{#1}^{\mathsf{c}}}
\newcommand{\ps}{\mathsf{PS}}

\newcommand{\cA}{\mathcal{A}}\newcommand{\cB}{\mathcal{B}}
\newcommand{\cC}{\mathcal{C}}\newcommand{\cD}{\mathcal{D}}
\newcommand{\cE}{\mathcal{E}}\newcommand{\cF}{\mathcal{F}}
\newcommand{\cG}{\mathcal{G}}\newcommand{\cH}{\mathcal{H}}

\newcommand{\cM}{\mathcal{M}}
\newcommand{\cO}{\mathcal{O}}
\newcommand{\cR}{\mathcal{R}}

\newcommand{\bC}{\mathbb{C}}
\newcommand{\bE}{\mathbb{E}}\newcommand{\bF}{\mathbb{F}}

\newcommand{\bN}{\mathbb{N}}

\newcommand{\bR}{\mathbb{R}}

\newcommand{\bZ}{\mathbb{Z}}

\newcommand{\bfG}{\mathbf{G}}
\newcommand{\bfN}{\mathbf{N}}
\newcommand{\bfn}{\mathbf{n}}

\newtheorem{theorem}{Theorem}[section]
\newtheorem{lemma}[theorem]{Lemma}
\newtheorem{claim}[theorem]{Claim}
\newtheorem{definition}[theorem]{Definition}
\newtheorem{corollary}[theorem]{Corollary}
\newtheorem{proposition}[theorem]{Proposition}

\theoremstyle{remark} % Define remark style
\newtheorem{remark}[theorem]{Remark}

\title{Constant-Overhead Injection into Quantum Codes}
\author{
  Louis Golowich\thanks{Department of EECS, UC Berkeley. Email: 
  \href{mailto:lgolowich@berkeley.edu}{\texttt{lgolowich@berkeley.edu}}. Supported by a Google PhD Fellowship, a ONR grant N00014-24-1-2491, and a National Science Foundation Graduate Research Fellowship under Grant No.~DGE 2146752.}
  \and
  Venkatesan Guruswami\thanks{Simons Institute for the Theory of Computing, and Departments of EECS \& Mathematics, UC Berkeley. Email: 
  \href{mailto:venkatg@berkeley.edu}{\texttt{venkatg@berkeley.edu}}. Research supported in part by a Simons Investigator award, ONR grant N00014-24-1-2491, and a UC Noyce initiative award.}
}
\begin{document}

\maketitle
\thispagestyle{empty}
\begin{abstract}
  We construct the first known fault-tolerant scheme for injecting states into quantum error-correcting codes with constant space and time overhead. That is, we construct a family of constant-rate quantum error-correcting codes for which a set of bare physical qubits can be injected, i.e.~fault-tolerantly encoded, into a code block. Similarly, a code state can be ejected, i.e.~fault-tolerantly decoded, back into bare physical qubits. We show that these injection and ejection procedures succeed under circuit-level locally stochastic noise, while incurring just a small constant probability of corrupting each qubit, which is unavoidable for bare physical qubits. We also show how to perform fault-tolerant error correction and code-state preparation under locally stochastic noise. All of our gadgets can be implemented with constant quantum circuit depth (i.e.~are single-shot), and with a number of physical qubits growing linearly with the number of logical qubits, assuming the ability to run polynomial-sized noiseless classical circuits on the side.

  We construct our quantum codes by taking a high-dimensional hypergraph product of classical LDPC codes, which in turn are a simplified version of Spielman's linear-time encodable codes (STOC'95). As our resulting product codes are only resilient to physical errors occurring with non-uniform probabilities across qubits, we then show how to concatenate with inner codes of various sizes to obtain fault-tolerance against uniform noise.
\end{abstract}

\newpage

\tableofcontents
\newpage
\pagenumbering{arabic}

\section{Introduction}
\label{sec:intro}
Reducing the space-time overhead of fault-tolerant quantum computation is a fundamental problem in both theory and practice. Schemes for quantum fault-tolerance typically perform computations on quantum data that is encoded in quantum error-correcting codes, in order to continuously correct any
errors that arise. A particularly fundamental and important task therefore is to interface between data encoded in quantum codes and data stored as bare physical qubits. This task can be accomplished with two related primitives:
\begin{enumerate}
\item \emph{Injection}, in which a bare physical state is encoded into a quantum code, and
\item \emph{Ejection}, in which a quantum code state is unencoded into bare physical form.
\end{enumerate}
As we typically make the physically natural assumption that bare physical qubits are continuously exposed to some constant noise rate, injection and ejection cannot be entirely error-free. We therefore instead define injection and ejection to be fault-tolerant if each data qubit is corrupted with only some small probability, which decays with the physical noise rate. Our main result is a construction of the first known scheme for such fault-tolerant injection and ejection that uses both constant space overhead (i.e.~maintains a constant encoding rate) and constant time overhead (i.e.~is single-shot).

\subsection{Motivation}
\label{sec:motivation}
Fault-tolerant injection and ejection are fundamental primitives, with numerous applications to various aspects of quantum computation.
At a basic level, injection and ejection facilitate quantum processing of real-world data obtained from physical quantum processes. For instance, various quantum learning and quantum sensing algorithms have been developed to efficiently process coherent quantum data from physical processes, while avoiding inefficiencies that would arise if the data were first converted to classical measurement outcomes (see e.g.~\cite{huang_quantum_2022,aharonov_quantum_2022,allen_quantum_2025,chen_exponential_2022}). To be effective, many of these quantum algorithms must be executed fault-tolerantly. Hence the physical data often must be injected into a quantum code. Indeed, \cite{kannan_fault-tolerant_2026} provides a comprehensive study of this injection-based approach, and shows how it can lead to large improvements over alternative approaches that avoid injection.

Injection and ejection also underly many techniques used in fault-tolerant quantum computation, even for applications with entirely classical inputs and outputs, so that quantum data is never stored in bare physical qubits. For example, to perform certain (e.g.~non-Clifford) gates on encoded data, many quantum fault-tolerance schemes inject appropriate (e.g.~``magic'') resource states into quantum codes. These resource states are then distilled to reduce the noise introduced from injection, and subsequently consumed to fault-tolerantly implement the respective gates (see e.g.~\cite{bravyi_universal_2005,bravyi_magic-state_2012,wills_constant-overhead_2024}).

Alternatively, ejection can be used to fault-tolerantly prepare resource states. For instance, \cite{aharonov_fault-tolerant_1997,gottesman_fault-tolerant_2014} show how to prepare a resource state encoded in a desired quantum code $Q_1$ by first preparing the state in the concatenation of $Q_1$ with a second code $Q_2$. Assuming $Q_2$ supports fault-tolerant computation and ejection, we can prepare this concatenated state and then eject out of $Q_2$ leaving the desired state encoded in just $Q_1$.

More generally, injection and ejection provide powerful primitives for switching between different quantum codes, as described in more detail below in \Cref{fig:codeswitch}. In short, given a state encoded in $Q_1$, we can inject its physical qubits in $Q_2$ to obtain a concatenated code state. We can then modify the underlying $Q_1$ state, such as by unencoding out of $Q_1$ and reencoding into some different $Q_1'$. Here we simply require that $Q_2$ supports a fault-tolerant implementation of the (un)encoding circuits of $Q_1,Q_1'$. Finally, we can eject out of $Q_2$ to obtain our state encoded in just~$Q_1'$.

\subsection{Main Result}
\label{sec:maininf}
The fundamental nature of injection and ejection in various applications raises a natural question:
\begin{center}
  \emph{How efficiently can we perform fault-tolerant injection and ejection?}
\end{center}
More precisely, for a given number $k$ of data qubits, what is the minimal quantum space-time overhead needed to inject our bare physical qubits into a quantum code, or to eject the code state back into bare physical form? While we may design the quantum code to facilitate injection and ejection, we require it to also support fault-tolerant error-correction, so that it can maintain an injected state in memory without further logical errors. We also ideally want this code to support a sufficient set of fault-tolerant Clifford gates to implement the resource state preparation and code switching applications described above.

For concreteness, we ask these fault-tolerant protocols (apart from injection/ejection) to succeed except with exponentially small $2^{-\poly(k)}$ logical error probability, assuming locally stochastic noise occurs at every timestep. Because bare physical qubits are always exposed to noise, we must allow the logical error probability for injection and ejection to instead be bounded by a small constant, which decays with the physical noise rate.

Our main result in this paper is the first known scheme for fault-tolerant injection and ejection satisfying the above properties with \emph{constant quantum space-time overhead}, given the ability to run polynomial-sized noiseless classical side-computations:

\begin{theorem}[Main result; informal statement of results in \Cref{sec:ftnonu} and Appendix~\ref{sec:concat}]\footnote{In particular, we instantiate our gadgets in \Cref{sec:ftnonu} with $r_X=r_Z=2$ and with $\bar{\ell}$ growing arbitrarily large. Here $r_X=r_Z=2$ means that we place our code qubits at level $2$ of a $4$-dimensional cochain complex, as defined in \Cref{def:prodcode}. The code's block length grows with $\bar{\ell}$.}
  \label{thm:maininf}
  There exists a family of $[[n=\Theta(k),\; k,\; d=\poly(k)]]$ quantum codes that support the following operations under locally stochastic noise with constant error rate $p$. Specifically, each operation below can be implemented by a quantum circuit using space (i.e.~width) $O(k)$ and time (i.e.~depth) $O(1)$, which is allowed to run an arbitrary $\poly(k)$-sized noiseless classical circuit in each timestep.
  \begin{enumerate}
  \item Injection of $k$ bare physical qubits into a code state. Each input qubit is corrupted with a small constant probability that decays with $p$.
  \item Ejection of a logical code state into $k$ bare physical qubits. Each output qubit is corrupted with a small constant probability that decays with $p$.
  \item Preparation of logical $\ket{0}^{\otimes k}$ and $\ket{+}^{\otimes k}$ code states.
  \item Logical $\gCNOT$ gates on all pairs of respective qubits across two code blocks.
  \item\label{it:miec} Error correction on a code state.
  \end{enumerate}
  All operations above except injection and ejection fail (i.e.~have uncorrectable output errors) with exponentially small probability $2^{-\poly(k)}$.
\end{theorem}

Our statement of \Cref{thm:maininf} is intentionally informal, as formalizing the relevant fault-tolerance notions is itself a difficult task that has been the study of recent works including \cite{nguyen_quantum_2025,he_composable_2025,christandl_fault-tolerant_2026-1}. Indeed, the formal fault-tolerance framework we use, which is heavily inspired by \cite{nguyen_quantum_2025,he_composable_2025,breuckmann_fault-tolerant_2026}, is developed throughout \Cref{sec:prelim,sec:codeconstruct,sec:ftnonu}. Nevertheless, we now expand upon the claims in \Cref{thm:maininf} to provide more context for understanding our results.

While we assume corruptions from faults are locally stochastic, we characterize correctability of a given error on a code state by a deterministic condition (see \Cref{def:badfams}). We show that this condition is violated with exponentially small probability $2^{-\poly(k)}$ under locally stochastic noise, assuming we repeatedly run our (single-shot) error-correction gadget in \Cref{it:miec} in \Cref{thm:maininf}. Intuitively, this error-correction gadget (which is formalized in \Cref{lem:errcorr}) is able to reduce a large constant noise rate to a smaller constant noise rate, so that a low error rate is maintained under repeated rounds of noise.

Due to the structure of our construction as described in \Cref{sec:tech}, we first analyze our gadgets under non-uniform locally stocastic noise. That is, we assign each physical qubit a ``level'' $\ell\in\bZ_{\geq 0}$, and assume that each level-$\ell$ qubit is corrupted with probability $p^{2^{\Omega(\ell)}}$ for a small constant $p>0$. In Appendix~\ref{sec:concat}, we show how to reduce uniform locally stochastic noise, in which all qubits have the same constant error probability, to this notion of non-uniform noise. The main idea in this reduction is to concatenate our main gadgets (which expect non-uniform noise) with non-uniformly-sized inner codes. That is, we encode each level-$\ell$ qubit into a $2^{\Theta(\ell)}$-sized inner code block. We specifically concatente with the construction of \cite{golowich_constant-overhead_2025}, to ensure that all necessary operations on the inner code blocks can be performed fault-tolerantly with constant space-time overhead. Because the fraction of level-$\ell$ qubits drops exponentially in $\ell$, the concatenated scheme maintains a constant encoding rate. Hence the concatenation allows us to simulate the desired non-uniform noise under uniform physical noise rates, with just constant space-time overhead.

One drawback of our main result stated in \Cref{thm:maininf} is that the errors on data qubits from injection and ejection may not be independent. Rather, we show that except with a small constant failure probability, a large constant fraction of data qubits remain uncorrupted during injection and ejection. Therefore once we fix a probability distribution from which the fault is sampled, we obtain a well-defined probability distribution on the set of corrupted data qubits in our injection and ejection gadgets. In particular, the fault-tolerance of our gadgets ensures that there exists a large constant fraction of the data qubits that are ``good,'' meaning they have low marginal probabilities of lying within the set of corrupted qubits. Hence we may always encode information into these ``good'' positions in our codes. As we describe in \Cref{sec:apps} below, our low marginal error probability guarantee (for the ``good'' qubits) is sufficient for many applications of injection and ejection described in \Cref{sec:motivation}.

% One drawback of our main result stated in \Cref{thm:maininf} is that the errors on data qubits from injection and ejection may not be independent. Rather, we show that except with a small constant failure probability, a large constant fraction of data qubits remain uncorrupted during injection and ejection. Therefore once the fault probability distribution is fixed, we have a well-defined data qubit error distribution, so we may identify a large constant fraction of data qubits that all have low marginal error probabilities. Nevertheless, as we describe in \Cref{sec:apps} below, this low marginal error probability guarantee is sufficient for many applications of injection and ejection described in \Cref{sec:motivation}.

\subsection{Relation to Prior Constructions}
To the best of our knowledge, all prior schemes providing fault-tolerant injection and ejection as defined in \Cref{sec:maininf} required either the space or time overhead to grow as the target logical error probability of a round of error correction decayed towards $0$. Indeed, various schemes based on concatenated codes (e.g.~\cite{aharonov_fault-tolerant_1997}), surface codes (e.g.~\cite{dennis_topological_2002,mazurek_long-distance_2014,li_magic_2015}), and more general LDPC codes (e.g.~\cite{lodyga_simple_2015,zhang_constant-overhead_2025,bhardwaj_high-rate_2026}) have been proposed, which obtained logical error probability $2^{-\poly(n)}$ for length-$n$ code blocks. However, to perform injection/ejection, such schemes required a $\poly(n)$ multiplicative space-time overhead. That is, letting $k$ denote the number of logical qubits, then such schemes required either at least $k\cdot\poly(n)$ physical qubits, or $\poly(n)$ timesteps. In contrast, our scheme in \Cref{thm:maininf} uses just constant space-time overhead, meaning that we use $O(k)$ physical qubits and $O(1)$ timesteps for injection/ejection, as well as for gadgets such as $\ket{0},\ket{+}$ state preparation and error correction.

Our codes underlying \Cref{thm:maininf} are tensor (i.e.~hypergraph) products of carefully constructed classical codes, which in turn are based in the linear-time-encodable codes of \cite{spielman_linear-time_1996}. As we describe in more detail in \Cref{sec:tech}, these codes have a sort of hierarchical structure that \cite{spielman_linear-time_1996} uses to ensure linear-time encodability. We instead leverage this structure to ensure that the resulting tensor product codes support ejection via Pauli measurements on appropriate physical code qubits. To perform injection, we show how to prepare logical Bell pairs across two code blocks, and then eject out of one code block. The resulting resource state can be used to teleport bare physical qubits into the remaining code block.

Our approach to ejection via appropriate physical Pauli measurements shares some properties with prior works such as \cite{mazurek_long-distance_2014,li_magic_2015} mentioned above. However, our construction is distinct in its crucial use of ideas from \cite{spielman_linear-time_1996} to obtain constant-rate codes for which such measurements yield fault-tolerant ejection.

In this respect, our construction is perhaps more related to the quantum codes of \cite{wills_linear-time_2026}, which also leverage ideas from \cite{spielman_linear-time_1996}. However, \cite{wills_linear-time_2026} obtain linear-time encodability and decodability against a linear number of adversarial errors \emph{without} fault-tolerance properties. That is, \cite{wills_linear-time_2026} allows an adversary to corrupt a linear number of qubits, but only after encoding and prior to decoding. In contrast, we focus on fault-tolerant injection/ejection and error correction against locally stochastic noise. That is, we require corruptions to be stochastic, which are generally easier to correct than adversarial errors, but we allow the errors to occur on any qubit at any time. Hence our results are incomparable to, and complement, those of \cite{wills_linear-time_2026}.

Our single-shot error-correction and state-preparation protocols in \Cref{thm:maininf} also build upon previous works on high-dimensional product codes. Indeed, related procotols were constructed in \cite{bombin_dimensional_2016,golowich_constant-overhead_2025,xu_batched_2025,tan_single-shot_2025}. However, our decoder construction needed to obtain fault-tolerance for these operations is move involved than such prior works due to the hierarchical structure of our codes based on \cite{spielman_linear-time_1996}.

\subsection{Applications to Fault-Tolerant Computation}
\label{sec:apps}

\begin{figure}
  \centering
  
  \begin{tikzpicture}[
    x=1cm,y=1cm, font=\small,
    cs input/.style={draw=blue!65!black, rounded corners=3pt, line width=.7pt},
    cs output/.style={draw=teal!75!black, rounded corners=3pt, line width=.7pt},
    cs inner/.style={draw=orange!70!black, fill=orange!8,
      dashed, rounded corners=3pt, line width=.7pt},
    cs physical/.style={circle, fill=black!75, inner sep=1.7pt},
    cs logical/.style={circle, draw=black!75, fill=white,
      line width=.65pt, inner sep=1.7pt},
    cs arrow/.style={->, >=stealth, line width=.85pt},
    cs label/.style={font=\footnotesize, align=center}
  ]
    % Rows are outer-code blocks; columns are distinct inner-code blocks.
    \node at (1.025,2.05) {Input};
    \node at (4.575,2.05) {After injection};
    \node at (10.525,2.05) {Before ejection};
    \node at (14.625,2.05) {Output};

    % Q_2 holds one qubit from each outer-code block in each column.
    \foreach \base/\last in {3.55/3,9.25/4} {
      \begin{scope}[shift={(\base,0)}]
        \foreach \j in {1,...,\last} {
          \pgfmathsetmacro{\qx}{.24+.48*\j}
          \draw[cs inner] (\qx-.19,-1.22) rectangle (\qx+.19,1.22);
          \node[cs label, text=orange!70!black] at (\qx,1.48) {$Q_2$};
        }
      \end{scope}
    }

    % The input and its concatenated encoding have identical row structure.
    \foreach \base in {0,3.55} {
      \begin{scope}[shift={(\base,0)}]
        \foreach \qy in {-.85,0,.85} {
          \draw[cs input] (0,\qy-.29) rectangle (2.05,\qy+.29);
          \node[cs label, text=blue!65!black] at (.27,\qy) {$Q_1$};
          \foreach \qx in {.72,1.20,1.68} {
            \ifdim\base pt=0pt
              \node[cs physical] at (\qx,\qy) {};
            \else
              \node[cs logical] at (\qx,\qy) {};
            \fi
          }
        }
      \end{scope}
    }

    % A different row length emphasizes that Q_1' may be a different code.
    \foreach \base in {9.25,13.35} {
      \begin{scope}[shift={(\base,0)}]
        \foreach \qy in {-.85,0,.85} {
          \draw[cs output] (0,\qy-.29) rectangle (2.55,\qy+.29);
          \node[cs label, text=teal!75!black] at (.27,\qy) {$Q_1'$};
          \foreach \qx in {.72,1.20,1.68,2.16} {
            \ifdim\base pt=13.35pt
              \node[cs physical] at (\qx,\qy) {};
            \else
              \node[cs logical] at (\qx,\qy) {};
            \fi
          }
        }
      \end{scope}
    }

    \draw[cs arrow] (2.19,0) -- (3.37,0)
      node[midway,above=5pt,cs label] {Inject}
      node[midway,below=5pt,cs label] {into $Q_2$};
    \draw[cs arrow] (5.80,0) -- (9.05,0)
      node[midway,above=6pt,cs label] {Decode $Q_1$\\then encode $Q_1'$}
      node[midway,below=6pt,cs label] {Fault-tolerant operations\\on $Q_2$ blocks};
    \draw[cs arrow] (11.99,0) -- (13.17,0)
      node[midway,above=5pt,cs label] {Eject}
      node[midway,below=5pt,cs label] {from $Q_2$};

    \node[cs label] at (1.025,-1.52) {Physical qubits};
    \node[cs label] at (4.575,-1.52) {Logical qubits of $Q_2$};
    \node[cs label] at (10.525,-1.52) {Logical qubits of $Q_2$};
    \node[cs label] at (14.625,-1.52) {Physical qubits};
    \node[cs label] at (7.95,-2.15)
      {Each row is one $Q_1$ or $Q_1'$ block; each dashed column is one $Q_2$ block.};
  \end{tikzpicture}
  
  \caption{\label{fig:codeswitch} Use of injection and ejection to switch from many blocks of a code $Q_1$ to blocks of a different code $Q_1'$. We assume that $Q_2$ supports fault-tolerant injection and ejection, along with $\ket{0}$ and $\ket{+}$ state preparation and transversal $\gCNOT$ gates.}
\end{figure}
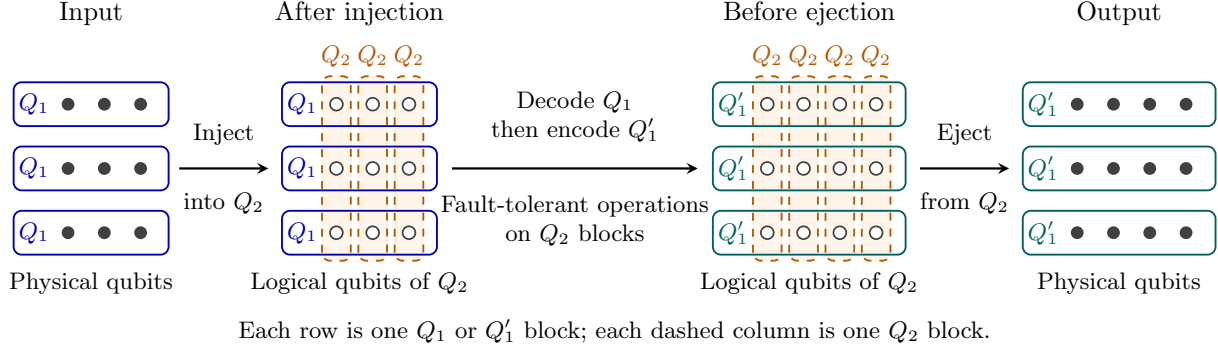

We now provide more details on various applications of low-overhead injection/ejection, as given by \Cref{thm:maininf}, to fault-tolerant quantum computation. As indicated in \Cref{sec:motivation}, one core example application is given by code switching. Specifically, given polynomially many disjoint blocks of a polynomialy-sized code $Q_1$, we can switch to blocks of a different code $Q_1'$ as illustrated in \Cref{fig:codeswitch}. The idea is to inject the physical qubits of our $Q_1$-blocks into many $Q_2$-blocks, and then run decoding for $Q_1$ and encoding for $Q_1'$ using fault-tolerant transversal gates on our $Q_2$-blocks. Finally, we eject out of our $Q_2$-blocks, so that our final output state just consists of $Q_1$-blocks.

As illustrated in \Cref{fig:codeswitch}, no two qubits in the same $Q_1$- or $Q_1'$-block are encoded into the same $Q_2$-block. Therefore even if the qubits encoded in a single $Q_2$-block experience correlated errors, we only need error independence across distinct $Q_2$-blocks in order to ensure error independence among the physical qubits within a given $Q_1$- or $Q_1'$-block. Hence we can instantiate $Q_2$ using our constant-overhead construction in \Cref{thm:maininf}. Specifically, we apply the guarantee described in \Cref{sec:maininf} that the qubits encoded in each $Q_2$-block have low marginal error probabilities under injection/ejection. The space-time overhead of the resulting code switching procedure is then dominated by the size of the decoding and encoding circuits of $Q_1$ and $Q_1'$ respectively.

Crucially, these decoding/encoding circuits for $Q_1$ and $Q_1'$ do \emph{not} need to be fault-tolerant, as fault-tolerance is provided by the gadgets described in \Cref{thm:maininf} acting on $Q_2$. Hence we may for instance choose $Q_1$ and $Q_1'$ to be arbitrary LDPC codes of polynomial distance, which are well-known to support (non-fault-tolerant) encoding and decoding with constant (quantum) space-time overhead via appropriate stabilizer measurements, followed by classical processing and Pauli corrections (see e.g.~\cite[Section~6]{gottesman_fault-tolerant_2014}). Thus our injection/ejection-based code switching procedure in \Cref{fig:codeswitch} provides an alternative to the \emph{batched} code switching procedure of \cite{xu_batched_2025}, which provided an alternative fault-tolerant protocol for fault-tolerantly switching between arbitrary quantum LDPC codes.

However, we emphasize certain advantages of our procedure in \Cref{fig:codeswitch}. For instance, as long as $Q_1$ has a polynomial-time classical decoding algorithm, then the $Q_1$-stabilizer measurement decoding, and hence all side classical computation in our scheme, can be done in polynomial time. More generally, our code switching procedure in \Cref{fig:codeswitch} can also be applied to non-LDPC codes, and will maintain low overhead for such codes with small encoding/decoding circuits.

In \Cref{thm:maininf}, we stated $\ket{0}$ and $\ket{+}$ state preparation and transversal $\gCNOT$ gates as examples of fault-tolerant Clifford gates that are sufficient for applications such as code switching. However, many more fault-tolerant gates can also be implemented using our codes. For instance, because our codes in \Cref{thm:maininf} are tensor (i.e.~hypergraph) product codes supporting single-shot state preparation, they also support single-shot dimensional jump protocols similarly as described in \cite{bombin_dimensional_2016,golowich_constant-overhead_2025,xu_batched_2025,tan_single-shot_2025}. In particular, \cite{golowich_constant-overhead_2025} show how to use such dimensional-jump techniques to perform targeted Hadamard or $\gCNOT$ gates on any constant number of logical qubits, or parallel Hadamard gates on all logical qubits, with just constant space-time overhead. The same techniques also apply to our codes in \Cref{thm:maininf}.

More generally, our injection/ejection-based code switching procedure described above (see \Cref{fig:codeswitch}) provides a key primitive for efficient fault-tolerant computation. Indeed, we may use this protocol to repeatedly switch between different codes that are efficiently encodable/decodable as described above, and support efficient implementations of different fault-tolerant gates. Such an approach may be particularly useful for obtaining fault-tolerant non-Clifford gates, which are noticably absent from our operations listed in \Cref{thm:maininf}. % We leave the development of a full universal fault-tolerance scheme based on this approach for future work.

Another approach to performing fault-tolerant non-Clifford gates is to simply inject noisy magic states using \Cref{thm:maininf}, and then perform magic state distillation. We can again deal with potential intra-block error correlations from injection by simultaneously performing injection across many independent code blocks. Indeed, we can then run the distillation protocol using transversal gates across code blocks, so that no two (potentially correlated) magic states injected into the same code block ever interact.

\subsection{Fault-Tolerance Formalism}
Formally defining proving fault-tolerance of a protocol, especially in the quantum setting, is a notoriously delicate task. As described in \Cref{sec:prelim} (see in particular \Cref{sec:faulttol}), we use a framework of quantum fault-tolerance based on the definitions developed in \cite{nguyen_quantum_2025,he_composable_2025}, and subsequently in \cite{breuckmann_fault-tolerant_2026}. This framework characterizes correctability of errors using deterministic conditions, which we show hold with high probability under locally stochastic noise.

As described in \Cref{sec:faulttol}, we add some notions to the earlier frameworks of \cite{nguyen_quantum_2025,he_composable_2025}, including \emph{reweightings} and \emph{postselections}. In brief, reweightings\footnote{\cite{breuckmann_fault-tolerant_2026} also defines a notion closely related to our reweightings, which they unfortunately call ``postselections.''} are a form of error in which the probabilities of different measurement outcomes can be ``reweighted.'' Such reweightings can appear when considering non-physical noise given by general non-channel superoperators. We must consider such non-physical noise because we decompose general error superoperators into linear combinations of Pauli superoperators. These individual Pauli terms may in turn be non-physical superoperators, e.g.~of the form $\rho\mapsto\rho\gX$, even if the original undecomposed error superoperator was a well-defined channel.

Meanwhile, we consider postselections on measurement outcomes for a separate purpose, namely, to obtain a well-defined classical probability distribution over the support of the error on the output state of a gadget. Indeed, by conditioning on certain measurement outcomes along with the choice of input error and fault, we are able to precisely bound the support of the output error. See \Cref{sec:faulttol} for more details.

\subsection{Roadmap}
The remainder of this paper is organized as follows. \Cref{sec:tech} provides a more detailed technical overview of the main ideas behind our construction used to prove \Cref{thm:maininf}. We then provide a detailed review of relevant definitions and basic notions in \Cref{sec:prelim}. We describe our code construction in \Cref{sec:codeconstruct}, and our decoder in \Cref{sec:decoder}. In \Cref{sec:ftnonu}, we present our fault-tolerant gadgets listed in \Cref{thm:maininf}. We provide concluding remarks in \Cref{sec:conclusion}. Appendix~\ref{sec:concat} describes how we reduce fault-tolerance under non-uniform locally stochastic noise to fault-tolerance under uniform noise, as mentioned in \Cref{sec:maininf} above.

\section{Technical Overview}
\label{sec:tech}
In this section, we provide an overview of our construction that proves \Cref{thm:maininf}.

\subsection{Underlying Classical LDPC Codes}
We begin by describing the classical LDPC codes based on those of \cite{spielman_linear-time_1996} that underlie our construction. We will subsequently take tensor (i.e.~hypergraph) products of chain complexes associated to these classical codes in order to obtain our desired quantum codes.

Our classical LDPC codes are based on \emph{lossless expanders}, as defined below.

\begin{definition}
  \label{def:lossless}
  Let $G=(V_0\sqcup V_1,E)$ be a bipartite graph of left-degree $\Delta_0$. For $\mu,\epsilon>0$, we say $G$ is a \emph{(onesided) $(\mu,\epsilon)$-lossless expander} if it holds for every $S\subseteq V_0$ that $|N_G(S)|\geq(1-\epsilon)\Delta_0|S|$, where $N_G(S)\subseteq V_1$ denotes the neighborhood of $S$.
\end{definition}

Throughout the remainder of this exposition in \Cref{sec:tech}, we let the lossless expansion parameters $\mu,\epsilon>0$ be sufficiently small fixed constants. We then refer to a $(\mu,\epsilon)$-lossless expander as simply a \emph{lossless expander}.

It is well-known that random biregular constant-degree graphs provide lossless expanders with high probability (see e.g.~\cite{hoory_expander_2006}). Explicit constructions have also been developed \cite{capalbo_randomness_2002,golowich_new_2024,cohen_hdx_2023,hsieh_explicit_2025}.

Bipartite adjacency matrices of appropriate lossless expanders provide parity-check matrices for asymptotically good classical LDPC codes (see e.g.~\cite{sipser_expander_1996}). However, we instead construct our classical LDPC codes by stringing together a sequence of lossless expanders of exponentially increasing sizes, as defined below. Our construction is not asymptotically good, i.e.~the distance is sublinear in the block length. However, it does support linear-time encoding, which intuitively will help facilitate our injection/ejection procedures on the resulting quantum codes.

\begin{definition}
  \label{def:classinf}
  Fix an infinite sequence of vertex sets $V_0,V_1,\dots$ with each $|V_\ell|=2|V_{\ell-1}|$. Then for $\ell\in\bN$, let $G^{(\ell)}=(V_\ell\sqcup V_{\ell-1},E^{(\ell)})$ be a biregular lossless expander of some fixed left-degree $\Delta$. Let $\delta^{(\ell)}\in\bF_2^{V_{\ell-1}\times V_\ell}$ denote the bipartite adjacency matrix of $G^{(\ell)}$.

  Then for $\bar{\ell}\in\bN$, we define a classical LDPC code $\cC=\cC(\bar{\ell})$ of dimension $k=|V_{\bar{\ell}}|$ and length $n=\sum_{\ell=0}^{\bar{\ell}}|V_\ell|$ to have parity-check matrix $\delta^{\cC}$ given by
  \begin{equation*}
    \delta^{\cC}(c) = (c_0+\delta^{(1)}(c_1),\; c_1+\delta^{(2)}(c_2),\;\dots\;,c_{\bar{\ell}-1}+\delta^{(\bar{\ell})}(c_{\bar{\ell}})).
  \end{equation*}
  That is, $c=(c_0,\dots,c_{\bar{\ell}})\in\bigoplus_{\ell=0}^{\bar{\ell}}\bF_2^{V_\ell}$ is a codeword of $\cC$ iff $\delta^{\cC}(c)=0$.
\end{definition}

Hence a bit string $c=(c_0,\dots,c_{\bar{\ell}})$, where each $c_\ell\in\bF_2^{V_\ell}$, is a codeword of $\cC=\cC(\bar{\ell})$ iff we have $c_{\ell-1}=\delta^{(\ell)}(c_\ell)$ for each $\ell\in\{1,\dots,\bar{\ell}\}$. Thus we can view $c_{\bar{\ell}}=c|_{V_{\bar{\ell}}}\in\bF_2^{V_{\bar{\ell}}}$ as our message. To encode this message, we inductively compute each $c_{\ell-1}=\delta^{(\ell)}(c_\ell)$ for $\ell=\bar{\ell},\dots,1$. As each constant-degree graph $G^{(\ell)}$ has $O(2^\ell)$ vertices, we can apply $\delta^{(\ell)}$ using a $O(2^\ell)$-sized circuit. Thus our encoding circuit for $\cC$ has size $O(2^{\bar{\ell}})$, which grows linearly with the block length. Indeed, this property was cruicial for the construction of linear-time encodable and decodable codes in \cite{spielman_linear-time_1996}.

However, our construction in \Cref{def:classinf} can be viewed as just ``half'' of the construction of \cite{spielman_linear-time_1996}. Specifically, we can view each $c_{\ell-1}=c|_{V_{\ell-1}}$ as containing ``checks'' on the longer string $c_\ell$. In particular, assuming we know the true value of $c_{\ell-1}$, we can decode errors on a bounded fraction of bits in $c_\ell$. This observation suggests that there is a benefit to having a lower noise rate on $c_\ell$ for smaller values of $\ell$, in order to correct errors on
$c_\ell$ for larger $\ell$, and ultimately correct errors on the message $c_{\bar{\ell}}$. Intuitively, \cite{spielman_linear-time_1996} achieves low error rates on $c_\ell$ for small $\ell$ by using a more involved encoding than given in \Cref{def:classinf}, which provides additional redundant information on $c_\ell$ for small $\ell$.

However, in our quantum code construction described below, we instead directly assume a lower physical noise rate for qubits associated to smaller $\ell$. We then show how to realize this non-uniform noise model via non-uniform concatenation. That is, we show how to encode each qubit associated to a given value of $\ell$ into an inner code that is larger, and hence provides more robust error protection, for smaller $\ell$.

Note that our resulting codes still only have polynomially large distance, instead of the linear distance achieved by the classical codes of \cite{spielman_linear-time_1996}, and by their quantum analogues in \cite{wills_linear-time_2026}. However, as our goal is to construct fault-tolerant protocols with logical error probability $2^{-\poly(n)}$, such polynomial distance is sufficient for our purposes.

\subsection{2-Dimensional Quantum Product Codes}
\label{sec:2diminf}
We now describe our quantum codes given by 2-dimensional tensor (i.e.~hypergraph) products of chain complexes associated that classical codes in \Cref{def:classinf}. As described in \Cref{sec:highdiminf} below, we will ultimately use higher (i.e.~4) dimensional tensor products to obtain all gadgets in \Cref{thm:maininf}. However, for simplicity we will illustrate ejection and error correction in the 2-dimensional case. Hence we first provide more details on the 2-dimensional construction.

Below, we assume familiarity with chain complexes and their tensor products; the relevant definitions can be found in \Cref{sec:prelim}.

\begin{figure}
  \centering
  
  \begin{tikzpicture}[
    x=1cm,y=1cm,font=\small,
    tp grid/.style={draw=black!45,line width=.35pt},
    tp qubits/.style={draw=blue!65!black,fill=blue!7,line width=.8pt},
    tp checks/.style={draw=teal!70!black,fill=teal!7,line width=.8pt},
    tp map/.style={->,>=stealth,line width=.8pt},
    tp index/.style={font=\scriptsize,inner sep=1.5pt},
    tp title/.style={align=center,font=\small}
  ]
    % Exact scale for bar{ell}=4: |V_ell| is represented by .18*2^ell cm.
    % C^0 has total width .18*(16+8+4+2+1)=5.58 cm;
    % C^1 has total width .18*(8+4+2+1)=2.70 cm.
    % The largest blocks are at the lower left of each sector.
    \def\tpunit{.18}
    \def\tpgap{2}
    \pgfmathsetmacro{\tpbits}{\tpunit*(2^5-1)}
    \pgfmathsetmacro{\tpchecks}{\tpunit*(2^4-1)}
    \pgfmathsetmacro{\tpoffset}{\tpbits+\tpgap}

    % Horizontal maps act on C; downward maps act on C^vee.
    \foreach \xx/\yy/\imax/\jmax/\kind in {
      0/\tpoffset/4/3/tp checks,
      \tpoffset/\tpoffset/3/3/tp qubits,
      0/0/4/4/tp qubits,
      \tpoffset/0/3/4/tp checks} {
      \begin{scope}[shift={(\xx,\yy)}]
        \pgfmathsetmacro{\ww}{\tpunit*(2^(\imax+1)-1)}
        \pgfmathsetmacro{\hh}{\tpunit*(2^(\jmax+1)-1)}
        \path[\kind] (0,0) rectangle (\ww,\hh);
        \foreach \i in {0,...,\imax} {
          \pgfmathsetmacro{\edge}{\tpunit*(2^(\imax+1)-2^\i)}
          \pgfmathsetmacro{\mid}{\edge-\tpunit*2^\i/2}
          \ifnum\i>0
            \draw[tp grid] (\edge,0) -- (\edge,\hh);
          \fi
          \node[tp index,below] at (\mid,-.05) {$\i$};
        }
        \foreach \j in {0,...,\jmax} {
          \pgfmathsetmacro{\edge}{\tpunit*(2^(\jmax+1)-2^\j)}
          \pgfmathsetmacro{\mid}{\edge-\tpunit*2^\j/2}
          \ifnum\j>0
            \draw[tp grid] (0,\edge) -- (\ww,\edge);
          \fi
          \node[tp index,left] at (-.05,\mid) {$\j$};
        }
      \end{scope}
    }

    \node[tp title] at ({\tpbits/2},{\tpoffset+\tpchecks+.8})
      {$C^1\times C^0=A^0$\\[2pt]
       \textcolor{teal!70!black}{$X$ stabilizers}};
    \node[tp title] at ({\tpoffset+\tpchecks/2},{\tpoffset+\tpchecks+.8})
      {$C^1\times C^1$\\[2pt]
       \textcolor{blue!65!black}{Qubits ($A^1$)}};
    \node[tp title] at ({\tpbits/2},-.9)
      {$C^0\times C^0$\\[2pt]
       \textcolor{blue!65!black}{Qubits ($A^1$)}};
    \node[tp title] at ({\tpoffset+\tpchecks/2},-.9)
      {$C^0\times C^1=A^2$\\[2pt]
       \textcolor{teal!70!black}{$Z$ stabilizers}};

    \foreach \y in {\tpbits/2,\tpoffset+\tpchecks/2}
      \draw[tp map] ({\tpbits+.2},\y) -- ({\tpoffset-.45},\y)
        node[midway,above=4pt] {$I\otimes\delta^{\cC}$};
    \foreach \x in {\tpbits/2,\tpoffset+\tpchecks/2}
      \draw[tp map] (\x,{\tpoffset-.2}) -- (\x,{\tpbits+.4})
        node[midway,right=5pt] {$(\delta^{\cC})^{\top}\otimes I$};
  \end{tikzpicture}

  \caption{\label{fig:2dim} 2-dimensional tensor (i.e.~hypergraph) product $\cA$ of (co)chain complexes $\cC$ associated to classical codes in \Cref{def:classinf}. We have drawn the $\bar{\ell}=4$ case. A rectangle with row number $\ell_1$ and column number $\ell_2$ represents a grid of qubits labeled by the set $V_{\ell_1}\times V_{\ell_2}$, which lie at level $(\bar{\ell}-\ell_1)+(\bar{\ell}-\ell_2)$.}
\end{figure}
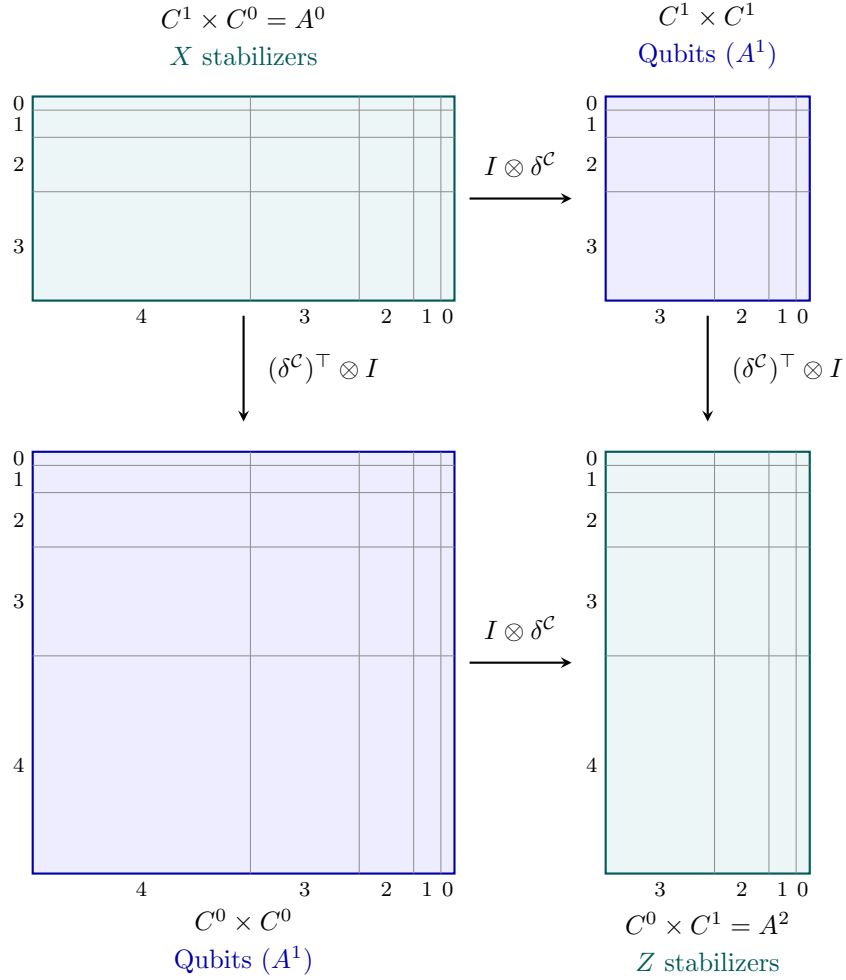

Our 2-dimensional tensor (i.e.~hypergraph) product codes are illustrated in \Cref{fig:2dim}. Specifically, we have illustrated the tensor product $\cA^*=\cC^*\otimes{\cC^\vee}^*$ of a 1-dimensional cochain complex $\cC^*=\cC^*(\ell)$ with its dualized (i.e.~transposed) complex ${\cC^\vee}^*\cong\cC_*$, where
\begin{equation*}
  \cC^* = (\cC^0\xrightarrow{\delta^{\cC}}\cC^1)
\end{equation*}
with $\cC^i=\bF_2^{C^i}$ for basis sets
\begin{align*}
  C^0 &= \bigsqcup_{\ell=0}^{\bar{\ell}}V_\ell, \hspace{3em} C^1 = \bigsqcup_{\ell=0}^{\bar{\ell}-1}V_\ell,
\end{align*}
and with $\delta^{\cC}$ defined as in \Cref{def:classinf}.

Write
\begin{equation*}
  \cA^* = (\cA^0 \rightarrow \cA^1 \rightarrow \cA^2),
\end{equation*}
and let $A^i$ be our chosen basis for each $\cA^i=\bF_2^{A^i}$, so that each $A^i$ is the disjoint union of sets of the form $V_{\ell_1}\times V_{\ell_2}$. We define the \emph{level} of a basis element $a\in V_{\ell_1}\times V_{\ell_2}$ to be
\begin{equation*}
  \Lev(a) = (\bar{\ell}-\ell_1) + (\bar{\ell}-\ell_2).
\end{equation*}
In particular, our physical code qubits are labeled by $A^1$, and hence have levels ranging from~$0$ to~$2\bar{\ell}$. Meanwhile, our $X$ (resp.~$Z$) stabilizers are labeled by $A^0$ (resp.~$A^2$).

We first prove fault-tolerance of our gadgets under a \emph{non-uniform locally stochastic noise model}, in which each qubit at level $\ell$ is corrupted with probability $p^{2^{\Theta(\ell)}}$ for a small constant $p>0$. More formally, the probability that a given set $E\subseteq A^1$ lies inside the support of the corruption at a given timestep is bounded by
\begin{equation*}
  \Pr[E\text{ lies in support of corruption}] \leq \prod_{a\in E}p^{2^{\Theta(\Lev(a))}}.
\end{equation*}
In Appendix~\ref{sec:concat}, we show how to simulate this non-uniform noise model under uniform locally stochastic noise, by encoding each level-$\ell$ qubit into an inner code of length $2^{\Theta(\ell)}$. We specifically concatenate with the construction of \cite{golowich_constant-overhead_2025}, which supports all necessary fault-tolerant operations to implement our gadgets.

\subsection{Fault-Tolerant Ejection}
We now describe how to perform fault-tolerant ejection on the 2-dimensional tensor product codes described in \Cref{sec:2diminf} above. This technique will generalize directly to higher-dimensional tensor product codes as well. Our ejection procedure will crucially leverage the hierarchy of lossless expanders in \Cref{def:classinf}.

First, we observe that each logical $\gX$ operator of our quantum code has the form $\gX^c$ for some codeword $c\in\cC^0$ of our classical code in \Cref{def:classinf} that is placed inside a single row in the rectangle $V_{\bar{\ell}}\times C^0$ at the bottom left of \Cref{fig:2dim}. Similarly, each logical $\gZ$ operator has the form $\gZ^c$ for a codeword $c\in\cC^0$ of our classical code in \Cref{def:classinf} that is placed inside a single column in the rectangle $C^0\times V_{\bar{\ell}}$. In particular, given $(a_1,a_2)\in V_{\bar{\ell}}\times V_{\bar{\ell}}$, we have an associated logical operator $\gX^{c_X(a_1,a_2)}$ supported inside the row $\{a_1\}\times C^0$, where $c_X(a_1,a_2)\in\cC^0$ is given by the encoding of $\1_{a_2}\in\bF_2^{V_{\bar{\ell}}}$ into our classical code from \Cref{def:classinf}. Similarly, we have an associated logical operator $\gZ^{c_Z(a_1,a_2)}$ supported inside the column $C^0\times\{a_2\}$, where $c_Z(a_1,a_2)$ is given by the encoding of $\1_{a_1}\in\bF_2^{V_{\bar{\ell}}}$ into our classical code.

Hence to perform ejection, we perform Pauli $\gX$ measurements on all qubits in the rectangle $V_{\bar{\ell}}\times(C^0\setminus V_{\bar{\ell}})$, and Pauli $\gZ$ measurements on all qubits in the rectangle $(C^0\setminus V_{\bar{\ell}})\times V_{\bar{\ell}}$. We then apply a Pauli correction to the rectangle $V_{\bar{\ell}}\times V_{\bar{\ell}}$, which will contain our ejected state, and we discard all other physical qubits.

The Pauli correction is computed as follows. For each qubit $(a_1,a_2)\in V_{\bar{\ell}}\times V_{\bar{\ell}}$, we compute the parity (i.e.~sum$\mod 2$) of the $\gX$ (resp.~$\gZ$) measurement outcomes in the support of $c_X(a_1,a_2)$ (resp.~$c_Z(a_1,a_2)$). We then apply a Pauli $\gZ$ (resp.~$\gX$) correction to qubit $(a_1,a_2)$ if this parity is odd. Note that we measure all qubits in the support of the logical operator $\gX^{c_X(a_1,a_2)}$ (resp.~$\gZ^{c_Z(a_1,a_2)}$) except for $(a_1,a_2)$. Hence after the Pauli correction, the resulting qubit $(a_1,a_2)$ precisely contains the ejection of the logical qubit specified by the logical operators $\gX^{c_X(a_1,a_2)}$ and $\gZ^{c_Z(a_1,a_2)}$. Indeed, measuring $\gX$ (resp.~$\gZ$) on this qubit after the Pauli correction has the same effect as measuring $\gX^{c_X(a_1,a_2)}$ (resp.~$\gZ^{c_Z(a_1,a_2)}$) prior to the ejection procedure.

The above analysis shows that our ejection procedure works in the absence of errors. To see that it is fault-tolerant, we observe that for each $\ell\in\bN$, the Pauli correction on a given ejected qubit only depends on the measurement outcomes of at most $2^{O(\ell)}$ distinct level-$\ell$ qubits. Hence assuming each level-$\ell$ qubit is corrupted with probability $p^{2^{\Theta(\ell)}}\ll 1/2^{O(\ell)}$ (see our discussion of the noise model above), the probability that a given ejected qubit is corrupted remains a small constant that decays with $p$.

While this heuristic argument suggests our ejection scheme is fault-tolerant, the full proof is significantly more involved. For instance, we must account for residual errors of our error-correction procedure, which may not simply be independently drawn with non-uniform probabilities. See \Cref{sec:ftnonu} for details.

\subsection{Fault-Tolerant Error Correction}
\label{sec:ecinf}
We now describe how to perform fault-tolerant error correction on the 2-dimensional tensor product codes described in \Cref{sec:2diminf} above. Like ejection, in \Cref{sec:ftnonu} we generalize this approach to error correction to higher-dimensional product codes. In fact, we ultimately use the error-correction ideas described here for logical $\ket{0}$ and $\ket{+}$ state preparation, which we then use to correct errors by repeatedly teleporting our state into fresh code blocks. However, for simplicity in this overview, we restrict attention to the problem of performing error correction directly on 2-dimensional product codes. Here for simplicity we also assume that syndrome measurements are noiseless. We describe the additional ideas needed for code state preparation, which also accounts for syndrome measurement errors, in \Cref{sec:highdiminf} below.

We specifically describe correction of $\gX$ errors via $\gZ$ stabilizer measurements followed by $\gX$ Pauli corrections. The key challenge is to construct a classical decoding algorithm that determines a valid Pauli correction given the stabilizer measurement outcomes. We again assume that a qubit at level $\ell$ experiences an $\gX$ error with probability $p^{2^{\Theta(\ell)}}$. The procedure to correct $\gZ$ errors via $\gX$ stabilizer measurements is analogous.

At a high level, our approach to decoding is to ``push'' all $\gX$ errors down into the bottom rectangle of qubits in \Cref{fig:2dim} labeled $V_{\bar{\ell}}\times C^0$. Within this rectangle, we then apply an ordinary ``flip'' decoder for classical LDPC codes based on lossless expanders (see e.g.~\cite{sipser_expander_1996}).

For this purpose, we can assume that the only qubits in $C^0\times C^0\subseteq A^1$ that experience $\gX$ errors in fact lie in the ``bottom half'' $V_{\bar{\ell}}\times C^0$ in \Cref{fig:2dim}. Indeed, any $\gX$ error in the ``top half'' $C^1\times C^0\subseteq C^0\times C^0$ is, up to $\gX$-stabilizers, equivalent to $\gX$ errors in $(V_{\bar{\ell}}\times C^0)\sqcup(C^1\times C^1)$.

Therefore, we begin by first greedily flipping (i.e.~applying Pauli $\gX$ corrections to) qubits in the top-right rectangle of qubits $C^1\times C^1\subseteq A^1$ in \Cref{fig:2dim} to ensure that the ``top half'' $C^1\times C^1\subseteq C^0\times C^1=A^2$ of the $\gZ$-stabilizers are satisfied. That is, we can apply Pauli $\gX$ corrections to qubits in $C^1\times C^1\subseteq A^1$ to ensure that the only remaining unsatisfied $\gZ$-stabilizers lie in the bottom right rectangle labeled $V_{\bar{\ell}}\times C^1$.

Now our only remaining $\gX$ errors lie in the rectangle $V_{\bar{\ell}}\times C^0\subseteq C^0\times C^0$. For each row of qubits labeled $\{a_1\}\times C^0$ in this rectangle with corruption $\gX^{e(a_1)}$ for $e(a_1)\in\bF_2^{C^0}$, the associated $\gZ$-stabilizer row labeled $\{a_1\}\times C^1$ by definition equals $\delta^{\cC}(e(a_1))$. That is, our $\gZ$-stabilizer rows in $V_{\bar{\ell}}\times C^1$ now simply give parity-check syndromes of our $\gX$-error rows, for the parity-check matrix $\delta^{\cC}$ given in \Cref{def:classinf}.

Because our classical codes in \Cref{def:classinf} are based on lossless expanders, we may therefore simply run a classical ``flip'' decoder within each row to correct the remaining errors. That is, we greedily apply $\gX$ corrections to qubits one at a time to reduce the syndrome weight. Here we crucially rely on the assumption that level-$\ell$ qubits for large $\ell$ are only corrupted with small probability $p^{2^{\Theta(\ell)}}$, so that few flips are needed at high levels where we only have a small number of qubits and parity-checks.

To decode higher-dimensional tensor product codes as described in \Cref{sec:highdiminf} below, a classical flip decoder no longer suffices. Indeed, after ``pushing'' down all errors in this case, we still must decode a high-dimensional quantum product code. We therefore instead use a \emph{small-set flip} decoder, building on ideas from \cite{leverrier_quantum_2015,fawzi_efficient_2018,fawzi_constant_2020,golowich_constant-overhead_2025} for decoding quantum codes given by tensor products of lossless expander codes. See \Cref{sec:decoder} for details.

We briefly remark on a few interesting properties of the decoding algorithm described above. Our algorithm only relies on \emph{onesided} lossless expansion (see \Cref{def:lossless}), as each tensor product factor code is only used for decoding either $\gX$ or $\gZ$ errors, but not both. In contrast, prior small-set-flip decoders for quantum product codes against locally stochastic noise (e.g.~\cite{leverrier_quantum_2015,golowich_constant-overhead_2025}) leveraged lossless expansion of all factor codes for decoding both $\gX$ and $\gZ$ errors. Hence these prior works required \emph{twosided} lossless expanders, which provide expansion out of both vertex sets in a bipartite graph.

Furthermore, our decoding algorithm is entirely immune to arbitrary corruptions on the top-right rectangle of qubits $C^1\times C^1\subseteq A^1$ in \Cref{fig:2dim}. Indeed, as the first step of our decoder greedily applies Pauli corrections throughout $C^1\times C^1$, we will immediately revert the effect of any errors inside this rectangle. This approach is again in contrast to prior small-set-flip decoders, which were sensitive to errors in all sectors of the tensor product complex.

We additionally note that our algorithm ``pushes'' $\gX$ errors into the rectangle of qubits labeled $V_{\bar{\ell}}\times C^0\subseteq A^1$, whereas it ``pushes'' $\gZ$ errors into the rectangle $C^0\times V_{\bar{\ell}}\subseteq A^1$. Hence $\gX$ and $\gZ$ errors are pushed into different rectangles, so we maintain separates bounds for the supports of the residual $\gX$ and $\gZ$ errors following error correction.

\subsection{State Preparation and Injection on Higher-Dimensional Product Codes}
\label{sec:highdiminf}
While above we described the 2-dimensional case of our product codes for simplicity, we ultimately prove \Cref{thm:maininf} using higher (i.e.~4 or greater) dimensional product codes. We use such higher-dimensional products in order to obtain constant-overhead fault-tolerant preparation of $\ket{0}$ and $\ket{+}$ logical states.

Such state preparation along with ejection immediately yields a protocol for fault-tolerant injection. Specifically, we can prepare two code blocks, one with logical $\ket{0}$ states and one with logical $\ket{+}$ states, and perform transversal $\gCNOT$ gates between the two to prepare logical Bell pairs. We then eject out of one code block. Thus we are left with Bell pairs between bare physical qubits and the remaining code block, which provides a resource state to teleport bare qubits into the code block.

Our procedure for fault-tolerant state preparation with constant space-time overhead builds on \cite{golowich_constant-overhead_2025}. Specifically, \cite{golowich_constant-overhead_2025} provided such a state preparation protocol for $\geq 3$-dimensional tensor products of lossless expander codes. The main idea is to prepare all physical qubits in $\ket{+}$ or $\ket{0}$, then measure all code stabilizers, which are used to determine a Pauli correction that moves the qubits' state into the code space. The key challenge is to construct a classical decoding algorithm for the measurement outcomes that guarantees correctness of the resulting Pauli correction. Such decoding typically requires higher ($\geq 3$) dimensional chain complexes in order to obtain ``meta-checks,'' or parity-checks on the stabilizer measurements.

\cite{golowich_constant-overhead_2025} specifically showed how to generalize the small-set-flip decoder of \cite{leverrier_quantum_2015,fawzi_efficient_2018,fawzi_constant_2020} for 2-dimensional products to higher-dimensional product codes constructed from lossless expanders. Hence by placing qubits at the middle level of a 4-dimensional product chain complex, so that both the $\gX$ and $\gZ$ stabilizers have meta-checks, \cite{golowich_constant-overhead_2025} obtained quantum codes for which both logical $\ket{0}$ and $\ket{+}$ states can be prepared fault-tolerantly with constant overhead under locally stochastic noise. Similar results were also obtained by \cite{xu_batched_2025,tan_single-shot_2025}, though using a minimum-weight decoder, which may not have a polynomial-time classical implementation, unlike a small-set-flip decoder.

We therefore obtain state preparation gadgets for our codes by extending the ideas of \cite{golowich_constant-overhead_2025} to obtain a decoder for high-dimensional tensor products of our codes described in \Cref{def:classinf}. We apply this decoder in the same way as \cite{golowich_constant-overhead_2025}, i.e.~by preparing physical qubits in $\ket{0}$ or $\ket{+}$ states, measuring stabilizers, and decoding to compute a Pauli correction.

However, whereas \cite{golowich_constant-overhead_2025} used underlying classical codes given by a single large twosided lossless expander, our classical codes in \Cref{def:classinf} consist of many onesided lossless expanders of varying sizes strung together. To account for this differing structure of our codes, we develop a decoder that combines the high-dimensional small-set-flip ideas of \cite{golowich_constant-overhead_2025} with the techniques described in \Cref{sec:ecinf}. Specifically, as described in \Cref{sec:ecinf}, our decoder first ``pushes'' errors into a subset of the physical qubits by greedily applying corrections in certain sectors of the product code. We then apply a high-dimensional small-set-flip decoder similar to that of \cite{golowich_constant-overhead_2025} to correct the remaining errors.

% TODO: describe 1-dimensional classical (simultaneously weighted and unweighted flip) decoding algorithm, then describe more informally how to extend to 2-dimensional hypergraph product case, with picture for 2-dim case (specifically maybe illustrate logical operators for ejection section?

\section{Preliminaries}
\label{sec:prelim}
In this section, we present definitions and basic results.

\subsection{Notation}
For $n\in\bN$, we let $[n]=\{1,2,\dots,n\}$. We let $\bF_2$ denote the field of order $2$. For $x\in\bF_2^n$, we let $|x|=|\{i\in[n]:x_i\neq 0\}$ denote the Hamming weight. For $i\in[n]$, when $n$ is clear from context, we let $\1_i\in\bF_2^n$ be the indicator vector for index $i$, meaning that $(\1_i)_j=1$ iff $i=j$.

An $n$-qubit pure quantum state is specified by a vector $\ket{\psi}\in(\bC^2)^{\otimes n}=\bC^{2^n}$. We use the standard notation $\bC^2=\spn\{\ket{0},\ket{1}\}$, $\ket{+}=(\ket{0}+\ket{1})/\sqrt{2}$, and $\ket{-}=(\ket{0}-\ket{1})/\sqrt{2}$. For a set $S\subseteq\bF_2^n$, we let $\ket{S}=(1/\sqrt{|S|})\sum_{x\in S}\ket{x}$ denote the uniform superposition of elements in $S$. A $n$-qubit density operator is a self-adjoint positive semi-definite operator $\rho\in\bC^{2^n\times 2^n}$ of trace $1$. We refer to linear maps $A:\bC^{2^n\times 2^n}\rightarrow\bC^{2^n\times 2^n}$ as superoperators. A superoperator that is completely positive and trace-preserving (CPTP) is called a quantum channel. We let $I_n:\bC^{2^n\times 2^n}\rightarrow\bC^{2^n\times 2^n}$ denote the identity channel $I_n(\rho)=\rho$.

We let $I,X,Y,Z\in\bC^{2\times 2}$ denote the standard single-qubit Pauli matrices. An $n$-qubit Pauli matrix is a tensor product of $n$ single-qubit Pauli matrices. For $P\in\{I,X,Y,Z\}$ and $x\in\bF_2^n$, we write $P^x=\bigotimes_{i=1}^nP^{x_i}$. An $n$-qubit Pauli superoperator is a map of the form $\rho\mapsto A\rho B$ for $n$-qubit Paulis $A,B$.

For an $n$-qubit (super)operator $A$, the support $S=\supp(A)$ is the minimal subset $S\subseteq[n]$ for which there exists a (super)operator $A_S$ acting on qubits $S$ such that $A=A_S\otimes I_{[n]\setminus S}$. We say the weight of a (super)operator $A$ is the size $|A|=|\supp(A)|$ of its support. It is well-known that every (super)operator has a unique decomposition into a linear combination of Pauli (super)operators. Then the support is simply the set of qubits on which ever nonzero term in this decomposition acts as the identity.

% We may write every $n$-qubit Pauli superoperator in the form $P(\rho)\propto X^{p_X}Z^{p_Z}\rho Z^{p_Z'}X^{p_X'}$ for some $p_X,p_X',p_Z,p_Z'\in\bF_2^n$. Then we define the $X$-support $\supp_X(P)=\supp(p_X)\cup\supp(p_X')$ and the $Z$-support $\supp_Z(P)=\supp(p_Z)\cup\supp(p_Z')$. To extend this definition to a general $n$-qubit superoperator $A$, letting $A=\sum_P\alpha_PP$ be the unique decomposition into Pauli superoperators, then we define the $X$-support $\supp_X(A)=\bigcup_{P:\alpha_P\neq 0}\supp_X(P)$ and the $Z$-support $\supp_Z(A)=\bigcup_{P:\alpha_P\neq 0}\supp_Z(P)$.

For vectors $v,v'\in\bF^n$ over an arbitrary field $\bF$, we write $v\propto v'$ if there exists $\alpha\in\bF$ such that $v=\alpha v'$.

\subsection{Error-Correcting Codes}
In this section, we define classical and quantum error correcting codes. In this paper, we restrict attention to codes over the binary alphabet $\bF_2$.

\begin{definition}
  A \emph{classical (linear) code} of length $n$ and dimension $k$ is a $k$-dimensional linear subspace $C\subseteq\bF_2^n$. The distance is $d=\min_{x\in C\setminus\{0\}}|x|$. We summarize these parameters by saying $C$ is an $[n,k,d]$ code.

  The \emph{dual} $C^\perp$ of $C$ is defined as $C^\perp=\{y\in\bF_2^n : x\cdot y=0 \ \forall x\in C\}$.

  A \emph{parity-check matrix} for $C$ is a matrix $H\in\bF_2^{m\times n}$ such that $C=\ker(H)$. The \emph{locality} of $H$ is the maximum Hamming weight of any row or column.
\end{definition}

\begin{definition}
  A \emph{quantum (CSS) code} of length $n$ is a pair of length-$n$ classical codes $Q=(Q_X,Q_Z)$. The dimension is $k=\dim(Q_Z)-\dim(Q_X^\perp)$, and the distance is
  \begin{equation*}
    d = \min_{x\in(Q_X\setminus Q_Z^\perp)\cup(Q_Z\setminus Q_X^\perp)}|x|.
  \end{equation*}
  We summarize these parameters by saying $Q$ is an $[[n,k,d]]$ code.
\end{definition}

For an $[[n,k,d]]$ code, if the $n$ physical qubits are labeled by elements of a set $N$ of size $|N|=n$, and the $k$ logical qubits are labeled by elements of a set $K$ of size $|K|=k$, we say the code is a $[[N,K,d]]$ code.

A family of classical codes is said to be ``low-density parity-check'' (LDPC) if it permits a family of constant-locality parity-check matrices. A family of quantum CSS codes is then said to be LDPC if the $X$ and $Z$ codes $Q_X$ and $Q_Z$ are classical LDPC codes.

We will need to specify encoding maps for CSS codes as defined below.

\begin{definition}
  A \emph{CSS encoding map} for a $[[n,k,d]]$ quantum CSS code $Q$ is a linear isomorphism $\Enc:\bF_2^k\xrightarrow{\sim} Q_Z/Q_X^\perp$. This map induces the \emph{encoding isometry} $\Enc:\bC^{2^k}\rightarrow\bC^{2^n}$ given by $\Enc\ket{x}=\ket{\Enc(x)}$, which in turn induces the \emph{encoding channel} $\Enc:\bC^{2^k\times 2^k}\rightarrow\bC^{2^n\times 2^n}$ given by $\Enc(\rho)=\Enc\rho\Enc^\dagger$. The meaning of `$\Enc$' (CSS encoding map, isometry, or channel) will be made clear from context.
\end{definition}

Below we present the notion of dual encoding maps for a CSS code, which are related to each other by applying Hadamard gates (see \Cref{def:gates}) to all qubits.

\begin{definition}
  \label{def:encdual}
  Let $\Enc_Z:\bF_2^k\rightarrow Q_Z/Q_X^\perp$ and $\Enc_X:\bF_2^k\rightarrow Q_X/Q_Z^\perp$ be CSS encoding maps for an $[[n,k]]$ CSS code $Q=(Q_X,Q_Z)$ and its $[[n,k]]$ \emph{dual code} $Q=(Q_Z,Q_X)$, respectively. We say $\Enc_X,\Enc_Z$ are \emph{dual CSS encoding maps} for $Q$ if it holds for every $x,x'\in\bF_2^k$ that
  \begin{equation*}
    x\cdot x' = \Enc_X(x)\cdot\Enc_Z(x'),
  \end{equation*}
  where the RHS above applies the natural bilinear form $(Q_X/Q_Z^\perp)\times(Q_Z/Q_X^\perp)\rightarrow\bF_2$.
\end{definition}

\begin{lemma}
  \label{lem:encdual}
  Let $\Enc_X,\Enc_Z$ be dual CSS encoding maps for an $[[n,k]]$ CSS code $Q$. Then the associated encoding isometries satisfy
  \begin{equation*}
    \Enc_X\circ\gH^{\otimes k} = \gH^{\otimes n}\circ\Enc_Z.
  \end{equation*}
\end{lemma}
\begin{proof}
  For every $x\in\bF_2^k$, by definition
  \begin{align*}
    \gH^{\otimes n}\Enc_Z\ket{x}
    &= \sqrt{\frac{1}{|Q_X^\perp|}}\sum_{y\in\Enc_Z(x)}\gH^{\otimes n}\ket{y} \\
    &= \sqrt{\frac{1}{2^n\cdot|Q_X^\perp|}}\sum_{z\in\bF_2^n}\sum_{y\in\Enc_Z(x)}(-1)^{y\cdot z}\ket{z}.
  \end{align*}
  As $\Enc_Z(x)$ is a coset of $Q_Z/Q_X^\perp$, the inner sum on the RHS above vanishes for every $z\notin Q_X$, so
  \begin{align*}
    \gH^{\otimes n}\Enc_Z\ket{x}
    &= \sqrt{\frac{1}{2^n\cdot|Q_X^\perp|}}\sum_{z\in Q_X}\sum_{y\in\Enc_Z(x)}(-1)^{y\cdot z}\ket{z} \\
    &= \sqrt{\frac{|Q_X^\perp|}{2^n}}\sum_{z\in Q_X}(-1)^{\Enc_Z(x)\cdot z}\ket{z},
  \end{align*}
  where $\Enc_Z(x)\cdot z$ is a well-defined element of $\bF_2$ because $\Enc_Z(x)\in Q_Z/Q_X^\perp$ and $z\in Q_X$. Now $\Enc_Z(x)\cdot z$ has the same value for all $z$ within a given coset in $Q_X/Q_Z^\perp$, so
  \begin{align*}
    \gH^{\otimes n}\Enc_Z\ket{x}
    &= \sqrt{\frac{|Q_X^\perp||Q_Z^\perp|}{2^n}}\sum_{x'\in\bF_2^k}(-1)^{\Enc_Z(x)\cdot\Enc_X(x')}\ket{\Enc_X(x')} \\
    &= \sqrt{\frac{1}{2^k}}\sum_{x'\in\bF_2^k}(-1)^{x\cdot x'}\ket{\Enc_X(x')} \\
    &= H^{\otimes k}\Enc_X\ket{x},
  \end{align*}
  as desired, where the second equality above applies the duality of $\Enc_X,\Enc_Z$ as defined in \Cref{def:encdual}.
\end{proof}

\subsection{Chain Complexes}
\label{sec:chaincomplex}
In this section, we provide standard definitions and facts regarding chain complexes, which are useful for constructing quantum CSS codes. The specific definitions we use below almost exactly follow those in \cite[Section~3.5]{golowich_constant-overhead_2025}, though we repeat the definitions here for reference.

\begin{definition}
  \label{def:chaincom}
  An \emph{$r$-dimensional chain complex $\cC_*$ over $\bF_2$} is a graded $\bF_2$-vector space $\cC=\bigoplus_{i=0}^r\cC_i$ together with a \emph{boundary map} $\partial^{\cC}:\cC\rightarrow\cC$ satisfying $(\partial^{\cC})^2=0$ and $\partial^{\cC}(\cC_i)\subseteq\cC_{i-1}$. When clear from context, we write $\partial=\partial^{\cC}$. We also write $\partial_i=\partial|_{\cC_i}$, and we summarize this data by writing
  \begin{equation*}
    \cC_* = (\cC_r\xrightarrow{\partial_r}\cC_{r-1}\xrightarrow{\partial_{r-1}}\cdots\xrightarrow{\partial_1}\cC_0).
  \end{equation*}
  We call $\cC_i$ the space of \emph{$i$-chains} of $\cC_*$, and we define the
  \begin{align*}
    i\text{-cycles} \; Z_i(\cC) &= \{z\in\cC_i:\partial(z)=0\} \\
    i\text{-boundaries} \; B_i(\cC) &= \{\partial(c):c\in\cC_{i+1}\} \\
    i\text{-homology} \; H_i(\cC) &= Z_i(\cC)/B_i(\cC).
  \end{align*}
  In a slight abuse of notation, for $i\in\bZ\setminus\{0,\dots,r\}$ we write $\cC_i=\{0\}$.

  The \emph{cochain complex}
  \begin{equation*}
    \cC^* = (\cC^0\xrightarrow{\delta_0}\cC^1\xrightarrow{\delta_1}\cdots\xrightarrow{\delta_{r-1}}\cC^r)
  \end{equation*}
  associated to $\cC_*$ is the chain complex with the same vector space $\cC=\bigoplus_{i=0}^r\cC^i$ with each $\cC^i=\cC_i$, but whose boundary map is the \emph{coboundary map} $\delta=\delta^{\cC}:\cC\rightarrow\cC$ given by $\delta=\partial^\top$. We then write $\delta_i=\delta|_{\cC^i}=\partial_{i+1}^\top$. We call $\cC^i$ the space of \emph{$i$-cochains}, and we similarly define the
  \begin{align*}
    i\text{-cocycles} \; Z^i(\cC) &= \{z\in\cC^i:\delta(z)=0\} \\
    i\text{-coboundaries} \; B^i(\cC) &= \{\delta(c):c\in\cC_{i-1}\} \\
    i\text{-cohomology} \; H^i(\cC) &= Z^i(\cC)/B^i(\cC).
  \end{align*}

  We will assume our chain complexes are $\Upsilon$-\emph{based} for some positive integer (or finite set) $\Upsilon$, meaning that $\cC_i=(\bF_2^\Upsilon)^{C_i}$ for a specified set $C_i$. We then define supports and Hamming weights of elements of $\cC_i=\cC^i$ with respect to this basis. That is, for $a\in\cC_i$, we let $\supp(a)\subseteq C_i$ denote the set of all $c\in C_i$ for which $a_c\in\bF_2^\Upsilon$ is nonzero, and we let $|a|=|\supp(a)|$.

  For $c_{i-1}\in C_{i-1}$ and $c_i\in C_i$, we write $c_{i-1}\triangleleft c_i$ if there exists $a_i\in\cC_i$ with $\supp(a_i)=\{c_i\}$ such that $c_{i-1}\in\supp(\partial(a_i))$. For $c_j\in C_j$ and $c_i\in C_i$, with $i<j$, we write $c_i\prec c_j$ if there exists a sequence $c_{i+1},\dots,c_{j-1}$ such that $c_i\triangleleft c_{i+1}\triangleleft\cdots\triangleleft c_{j-1}\triangleleft c_j$. This partial order provides the set $C=\bigsqcup_iC_i$ with the structure of a graded poset. For $i<j$, we extend the partial order notation to apply for a basis element $c_i\in C_i$ and a chain $c_j\in\cC_j=\bF_2^{C_j}$, so that $c_i\prec c_j$ if every $c'\in\supp(c_j)$ satisfies $c_i\prec c'$.

  % \comment{in this case, shouldn't  $c'\in\supp(\partial(\1_{c_j}))$ instead of $c'\in\supp(c_j)$? also doesn't that mean that $c'\in C_{j-1}$? I don't see how this generalizes for $i<j$}

  The \emph{$i$-systolic distance $d_i(\cC)$} and \emph{$i$-cosystolic distance $d^i(\cC)$} are defined as
  \begin{align*}
    d_i(\cC) &= \min_{z\in Z_i(\cC)\setminus B_i(\cC)}|z| \\
    d^i(\cC) &= \min_{z\in Z^i(\cC)\setminus B^i(\cC)}|z|.
  \end{align*}

  We say that $\cC_*$ has \emph{locality $w$} if for every $c\in C$ there are $\leq w$ basis elements $c'\in C$ with $c'\preceq c$ or $c'\succeq c$.
\end{definition}

The partial order in \Cref{def:chaincom} intuitively describes incidences in a chain complex $\cC$. That is, we first say that $c_{i-1}\triangleleft c_i$ for a pair of basis elements $c_{i-1},c_i$ at consecutive levels $i-1,i$ if the boundary matrix $\partial$ has a nonzero entry at position $(c_{i-1},c_i)$. We then extend this definition to a partial order on pairs of basis elements $c_i,c_j$ at arbitrary levels in the natural way, i.e.~by considering sequences of basis elements at consecutive levels that are ordered by $\triangleleft$.

We will typically define our chain complexes to be 1-based, unless explicitly specified otherwise. Our main use of larger $\Upsilon$ is to consider direct sums of 1-based complexes, as defined below.

\begin{definition}
  \label{def:directsum}
  For a 1-based chain complex $\cC_*$ and a positive integer (or finite set) $\Upsilon$, we define the \emph{$\Upsilon$-direct sum complex} $\cC^{\oplus\Upsilon}_*$ to be $\Upsilon$-based chain complex consisting of $\Upsilon$ disjoint copies of $\cC_*$, but where $\cC^{\oplus\Upsilon}_*$ has the same basis set $C=\bigsqcup_iC_i$ and partial order $\prec$ as $\cC$. That is, $\cC^{\oplus\Upsilon}_*$ is obtained by replacing each $\bF_2$-basis element of $\cC$ with a $\bF_2^\Upsilon$-basis element, so that $(\cC^{\oplus\Upsilon})_i=(\cC_i)^{\oplus\Upsilon}$ and $\partial^{\cC^{\oplus\Upsilon}}=(\partial^{\cC})^{\oplus\Upsilon}$.
\end{definition}

We will use direct sum complexes in \Cref{sec:decoder} in order to prove that when our decoder is given different inputs with similar supports, the resulting outputs also have similar supports. This property allows us to prove a notion of fault-tolerance in which the support of the output error of a gadget is determined by the support of the input error and fault, rather than on the specific error value (see \Cref{sec:faulttol} below).

We will specify encoding maps for our codes using chain maps:

\begin{definition}
  \label{def:chainmap}
  For $r$-dimensional chain complexes $\cA_*,\cB_*$, a \emph{chain map $\phi:\cA_*\rightarrow\cB_*$} is a sequence of maps $(\phi_i:\cA_i\rightarrow\cB_i)_{i\in[r]}$ satisfying $\partial_i^{\cB}\circ\phi_i=\phi_{i-1}\circ\partial_i^{\cA}$. Every chain map naturally induces a map on homology, which in a slight abuse of notation we also denote $\phi:H_*(\cA)\rightarrow H_*(\cB)$.
\end{definition}

We will construct our codes from tensor products (also sometimes called hypergraph products \cite{tillich_quantum_2014})) of 1-dimensional chain complexes:

\begin{definition}
  \label{def:cctensor}
  For $\Upsilon$-based chain complexes $\cA_*,\cB_*$ of respective dimensions $r_{\cA},r_{\cB}$, the \emph{tensor product $\cC_*=\cA_*\otimes\cB_*$} is the $\Upsilon^2$-based chain complex of dimension $r_{\cC}=r_{\cA}+r_{\cB}$ for which the $i$-chain basis is $C_i=\bigsqcup_{j\in\bZ}A_j\times B_{i-j}$, so that $\cC_i=\bigoplus_{j\in\bZ}\cA_j\otimes\cB_{i-j}$, and the boundary map is given by $\partial^{\cC}=\partial^{\cA}\otimes I+I\otimes\partial^{\cB}$. The tensor product of cochain complexes is defined analogously, so that $\cA^*\otimes\cB^*=(\cA_*\otimes\cB_*)^*$.
\end{definition}

The following formula describing how homology behaves under tensor products is well known.

\begin{proposition}[K\"{u}nneth formula]
  \label{prop:kunneth}
  For chain complexes $\cA_*,\cB_*$ with tensor product $\cC_*=\cA_*\otimes\cB_*$, then there is an isomorphism
  \begin{align*}
    H_i(\cC) &\cong \bigoplus_{j\in\bZ}H_j(\cA)\otimes H_{i-j}(\cB),
  \end{align*}
  which for $a\in Z_i(\cA)$, $b\in Z_i(\cB)$ is given by
  \begin{align*}
    a\otimes b+B_i(\cC) &\mapsfrom (a+B_i(\cA))\otimes(b+B_i(\cB)).
  \end{align*}
\end{proposition}

For a product complex $\cC_*=\cC^{(1)}_*\otimes\cdots\otimes\cC^{(r)}_*$ with basis elements $c=(c_1,\dots,c_r),\;c'=(c_1',\dots,c_r')\in C^{(1)}\times\cdots\times C^{(r)}=C$, the graded poset structure of $C$ by definition has $c\preceq c'$ iff it holds for every $i\in[r]$ that $c_i\preceq c'_i$.

\begin{definition}
  \label{def:cctoCSS}
  The \emph{quantum CSS code associated to level $i$ of a cochain complex $\cC^*$} is given by $Q=(Q_X=\ker(\partial_i),\;Q_Z=\ker(\delta_i))$.
\end{definition}

The quantum code $Q$ in Definition~\ref{def:cctoCSS} by definition has length $n=\dim(\cC^i)$, dimension $k=\dim(H^i(\cC))$, and distance $d=\min\{d^i(\cC),d_i(\cC)\}$. Furthermore, $Q$ has naturally associated $X,Z$ parity-check matrices $\partial_i,\delta_i$ respectively, whose locality must be at most the locality of  $\cC^*$.

\subsection{Circuits}
In this section, we describe the model of a quantum circuit that we use in this paper. We will ultimately consider quantum circuits with noisy qubits alongside classical circuits with noisless bits. The entire state of such a system is a quantum state, where the classical bits are simply qubits that lie in a classical state, as defined below.

\begin{definition}
  Let $N=\Cl{N}\sqcup\Qu{N}$ be a set. We say an operator $\rho\in\bC^{2^N\times 2^N}$ is \emph{classical on bits $\Cl{N}$} if $\rho$ can be expressed in the form $\rho=\sum_{x\in\bF_2^{\Cl{N}}}\ket{x}\bra{x}\otimes\rho_x$ for some operators $\rho_x\in\bC^{2^{\Qu{N}}\times 2^{\Qu{N}}}$.
\end{definition}

We are now ready to define our circuit model.

\begin{definition}
  \label{def:circuit}
  For a set $N$ and for $T\in\bN$, a \emph{quantum circuit} $\cR=(R_1,\dots,R_T)$ using space $N$ and time $T$ consists of superoperators $R_t:\bC^{2^{N_{t-1}}\times 2^{N_{t-1}}}\rightarrow\bC^{2^{N_t}\times 2^{N_t}}$ for some
  % \vnote{disjoint?} \lnote{no they can overlap; we use the same set $N$ of qubits at every timestep, though $N_t$ only contains the subset of the $N$ qubits that are ``active'' at timestep $t$ (meaning they have been initialized, or were input qubits that have not been terminated)
  subsets $N_0,\dots,N_T\subseteq N$. We let $\cR(\cdot):\bC^{2^{N_0}\times 2^{N_0}}\rightarrow\bC^{2^{N_T}\times 2^{N_T}}$ denote the superoperator
  \begin{equation*}
    \cR(\rho) = R_T\circ\cdots\circ R_1(\rho).
  \end{equation*}

  We assume the total space $N=\Cl{N}\sqcup\Qu{N}$ is partitioned into \emph{bits} (i.e.~\emph{classical space}) $\Cl{N}$ and \emph{qubits} (i.e.~\emph{quantum space}) $\Qu{N}$, and for $t\in[T]$ we similarly let $\Cl{N_t}=N_t\cap\Cl{N}$ and $\Qu{N_t}=N_t\cap\Qu{N}$. We require that $R_t$ map every $\rho\in\bC^{2^{N_{t-1}}\times 2^{N_{t-1}}}$ that is classical on bits $\Cl{N_{t-1}}$ to some $R_t(\rho)$ that is classical on bits $\Cl{N_t}$. We call an operator $\rho\in\bC^{2^{N_0}\times 2^{N_0}}$ a \emph{valid input} to the circuit $\cR$ if $\rho$ is classical on bits $\Cl{N_0}$. We will only ever consider valid inputs to circuits.

  We call $\Cl{N_t}$ and $\Qu{N_t}$ the \emph{active bits and qubits at time $t$}, respectively, and we call $\Cl{N_0},\Cl{N_T}$ the \emph{input, output bits} and $\Qu{N_0},\Qu{N_T}$ the \emph{input, output qubits} of $\cR$, respectively.
\end{definition}

Naturally, we want to construct circuits out of gates, which are simply circuits consisting of a single timestep:

\begin{definition}
  A \emph{gate} is a circuit using time $T=1$.
\end{definition}

\begin{definition}
  \label{eq:def:gateset}
  We say a circuit $\cR$ (\Cref{def:circuit}) uses a \emph{gate set $\cG$} if each $R_t$ can be decomposed into a tensor product of gates $G_{t,i}\in\cG$ acting on disjoint sets of (qu)bits. That is, for each $t\in[T]$, there exists $b_t\in\bN$ and gates $G_{t,1},\dots,G_{t,b_t}\in\cG$ such that $R_t=\bigotimes_{i\in[b_t]}G_{t,i}$ for some assignment of the (qu)bits acted upon by the gates $G_{t,i}$ to disjoint sets of (qu)bits in $N$, such that the disjoint union of the input (resp.~output) (qu)bits of the $G_{t,i}$ equals the set of (qu)bits in $N_{t-1}$ (resp.~$N_t$).
\end{definition}

Below we define the standard gates that we will use to construct our circuits.

\begin{definition}
  \label{def:gates}
  We define the following \emph{unitary gates}, which apply $\rho\mapsto U\rho U^\dagger$ for some unitary $U\in\bC^{2^n\times 2^n}$ acting on some number $n$ of qubits:
  \begin{itemize}
  \item The single-qubit Pauli gates
    \begin{equation*}
      \gI=\begin{pmatrix}1&0\\0&1\end{pmatrix}, \hspace{1em} \gX=\begin{pmatrix}0&1\\1&0\end{pmatrix}, \hspace{1em} \gY=\begin{pmatrix}0&-i\\i&0\end{pmatrix}, \hspace{1em} \gZ=\begin{pmatrix}1&0\\0&-1\end{pmatrix}.
    \end{equation*}
  \item For $\alpha\in\{\gI,\gX,\gY,\gZ\}$ and $x\in\bF_2^n$, the $n$-qubit Pauli gate $\alpha^x=\bigotimes_{i\in[n]}\alpha^i$.
  \item The single-qubit Hadamard gate $\gH=\frac{1}{\sqrt{2}}\begin{pmatrix}1&1\\1&-1\end{pmatrix}$.
  \item The two-qubit gate $\gCNOT:\ket{x_1,x_2}\mapsto\ket{x_1,x_1+x_2}$.
  \end{itemize}

   We also define the following non-unitary gates:
  \begin{itemize}
  \item The $\gX$-initialization gate $\gInitX:1\mapsto\ket{+}\bra{+}$ with zero input qubits, and one output qubit initialized to $\ket{+}$.
  \item The $\gZ$-initialization gate $\gInitZ:1\mapsto\ket{0}\bra{0}$ with zero input qubits, and one output qubit (or bit) initialized to $\ket{0}$.
  \item The termination gate $\gTerm:\rho\mapsto\tr(\rho)$ with one input qubit (or bit) and zero output qubits.
  \item The $\gX$-measurement gate $\gMX$ that takes as input one qubit (and no bits), and outputs one bit (and no qubits) containing the outcome of measuring the input qubit in the $\gX$ basis. That is,
    \begin{equation*}
      \gMX(\rho) = \bra{+}\rho\ket{+}\cdot\ket{0}\bra{0} + \bra{-}\rho\ket{-}\cdot\ket{1}\bra{1}.
    \end{equation*}
    % \vnote{Should be $\ket{+}\bra{+}$ above?} \lnote{In this case we want the output bit to be classical (added a note below)}
    For $\xi\in\bF_2$, we also define the \emph{$\xi$-postselected} $\gX$-measurement gate $\gMX[\xi]$ to also take as input one qubit and output one bit, but where
    \begin{equation*}
      \gMX[\xi](\rho) = \bra{\xi}\gH\rho\gH\ket{\xi}\cdot\ket{\xi}\bra{\xi}.
    \end{equation*}
  \item The $\gZ$-measurement gate $\gMZ$ that takes as input one qubit (and no bits), and outputs one bit (and no qubits) containing the outcome of measuring the input qubit in the $\gZ$ basis. That is,
    \begin{equation*}
      \gMZ(\rho) = \bra{0}\rho\ket{0}\cdot\ket{0}\bra{0} + \bra{1}\rho\ket{1}\cdot\ket{1}\bra{1}.
    \end{equation*}
    For $\xi\in\bF_2$, we also define the \emph{$\xi$-postselected} $\gX$-measurement gate $\gMX[\xi]$ to also take as input one qubit and output one bit, but where
    \begin{equation*}
      \gMX[\xi](\rho) = \bra{\xi}\rho\ket{\xi}\cdot\ket{\xi}\bra{\xi}.
    \end{equation*}
  \item For an arbitrary gate $G$ with some number $n$ of input and output (qu)bits, the $(n+1)$-(qu)bit classically-controlled gate $\gCt{G}$ that acts on the $n$ (qu)bits of $G$ along with one additional bit, and applies $G$ iff the additional bit is $1$. That is, for $x\in\bF_2$ and $\rho\in\bC^{2^n\times 2^n}$,
    \begin{equation*}
      \gCt{G}(\ket{x}\bra{x}\otimes\rho) = \begin{cases}
        \ket{x}\bra{x}\otimes\rho,&x=0 \\
        \ket{x}\bra{x}\otimes G(\rho),&x=1.
      \end{cases}
    \end{equation*}
    For $\zeta\in\bF_2$, we also define the \emph{$\zeta$-reweighted} classically-controlled gate $\gCt{G}[\zeta]$ to also act on the $n$ (qu)bits of $G$ along with one additional bit, but by
    \begin{equation*}
      \gCt{G}[\zeta](\ket{x}\bra{x}\otimes\rho) = (-1)^{\zeta\cdot x}\cdot\ket{x}\bra{x}\otimes G^x(\rho) = \begin{cases}
        \ket{x}\bra{x}\otimes\rho,&x=0 \\
        (-1)^\zeta\cdot\ket{x}\bra{x}\otimes G(\rho),&x=1.
      \end{cases}
    \end{equation*}
  \item For an arbitrary (and in particular, not necessarily linear) function $f:\bF_2^{n_0}\rightarrow\bF_2^{n_1}$, the classical function gate $\gCF_f$ that takes as input $n_0$ bits, and outputs $n_1$ bits containing the output of $f$ applied to the input. That is, for $x\in\bF_2^{n_0}$,
    \begin{equation*}
\gCF_f(\ket{x}\bra{x}) = \ket{f(x)}\bra{f(x)}.
    \end{equation*}
    We will only use classical function gates for which the function $f$ can be computed in $\poly(n_0+n_1)$ time. We write $\gCF_*$ to denote the set of all classical function gates, across all classical functions $f$.
  \end{itemize}
\end{definition}

Note that the $\gX$-measurement gate returns the measurement outcome in the $\gZ$-basis, as it outputs a classical bit, which we require to be a mixture of $\gZ$-basis states $\ket{0}\bra{0}$ and $\ket{1}\bra{1}$.

When describing circuits, we will often specify the output (qu)bits as the ``returned'' state; it is implicitly assumed that all other qubits that are not returned are terminated.

As described in \Cref{sec:faulttol}, we consider reweightings and postselections for technical reasons when defining and proving fault-tolerance. Postselections, as well as linear combinations of different reweightings, can yield non-CPTP superoperators, which are not physically implementable. We emphasize that our fault-tolerant circuits only apply CPTP (and hence physically implementable) gates. However, we prove fault-tolerance under more general noise that is allowed to be non-CPTP. More generally, even under CPTP noise, non-CPTP superoperators still arise in various intermediate decompositions throughout our analysis. As a basic example, we decompose some measurement gates $\gMX=\gMX[0]+\gMX[1]$ (and likewise $\gMZ=\gMZ[0]+\gMZ[1]$), and consider each postselected branch separately.

\subsection{Fault-Tolerance}
\label{sec:faulttol}
In this section, we describe the model of fault-tolerance that we use in this paper.
We generally follow the type of formalism used in \cite{nguyen_quantum_2025,he_composable_2025}, though we present definitions that are suitably tailored for our purposes.

We begin by defining faults, which cause errors during the execution of a circuit.

\begin{definition}
  \label{def:fault}
  Let $\cR=(R_1,\dots,R_T)$ be a circuit with each $R_t:\bC^{2^{N_{t-1}}\times 2^{N_{t-1}}}\rightarrow\bC^{2^{N_t}\times 2^{N_t}}$. A \emph{fault} on the circuit $\cR$ is a sequence $\cF=(F_1,\dots,F_T)$ of superoperators $F_t:\bC^{2^{N_t}\times 2^{N_t}}\rightarrow\bC^{2^{N_t}\times 2^{N_t}}$. We define $\supp(\cF)\subseteq N_1\sqcup\cdots\sqcup N_T\subseteq N^{\sqcup T}=N\times[T]$ by $\supp(\cF)=\supp(F_1)\sqcup\cdots\sqcup\supp(F_T)$.
  A \emph{Pauli fault} $\cF$ is a fault in which each $F_t$ is a Pauli superoperator. % meaning that there exist $f_{t,X},f_{t,X}',f_{t,Z},f_{t,Z}'\in\bF_2^{N_t}$ such that $F_t(\tau)=\gX^{f_{t,X}}\gZ^{f_{t,Z}}\tau\gZ^{f_{t,Z}'}\gX^{f_{t,X}'}$.
  % For a Pauli fault, we define the \emph{$X$-support} $\supp_X(\cF)=\supp(f_{t,X})\cup\supp(f_{t,X}')$ and the \emph{$Z$-support} $\supp_Z(\cF)=\supp(f_{t,Z})\cup\supp(f_{t,Z}')$.

  The \emph{$\cF$-corrupted circuit} $\cR[\cF]$ is defined by
  \begin{equation*}
    \cR[\cF] = (R_1,F_1,R_2,F_2,\dots,R_T,F_T).
  \end{equation*}

  In this paper, we require that faults act trivially on the classical bits, meaning that each $F_t$ acts as the identity on bits $\Cl{N_t}$. Hence we can view $F_t$ as a superoperator $F_t:\bC^{2^{\Qu{N_t}}\times 2^{\Qu{N_t}}}\rightarrow\bC^{2^{\Qu{N_t}}\times 2^{\Qu{N_t}}}$.
\end{definition}

We now also define reweightings, which are another type of corruption to a circuit.

\begin{definition}
  Let $\cR=(R_1,\dots,R_T)$ be a circuit using a gate set $\cG$. Let $n_{\mathsf{C}}$ denote the number of classically-controlled gates in $\cR$. A \emph{reweighting} for $\cR$ is a vector $\zeta=(\zeta_1,\dots,\zeta_{n_{\mathsf{C}}})\in\bF_2^{n_{\mathsf{C}}}$. The \emph{$\zeta$-reweighted circuit} $\cR[\zeta]$ is defined by replacing the $i$th classically-controlled gate $\gCt{G}$ in $\cR$ with its $\zeta_i$-reweighted version $\gCt{G}[\zeta_i]$ for every $i\in[n_{\mathsf{C}}]$, and otherwise leaving $\cR$ unchanged.

  It will always be clear from context whether an argument to $\cR[\cdot]$ is a fault or reweighting. When both are present, we apply the fault last, so that $\cR[\cF,\zeta]=(\cR[\zeta])[\cF]$. We also let $\cR[*]=\{\cR[\zeta]:\zeta\in\bF_2^{n_{\mathsf{C}}}\}$ denote the set of all reweighted versions of $\cR$.
\end{definition}

Reweighting in logical circuits that include classically-controlled gates can be induced when non-physical (i.e.~non-channel superoperator) noise is applied to an associated fault-tolerant physical circuit. For example, consider a state $\sigma\in\bC^{4\times 4}$ consisting of one classical bit and one qubit, so that
\begin{equation*}
  \sigma = \ket{0}\bra{0}\otimes\sigma_0 + \ket{1}\bra{1}\otimes\sigma_1
\end{equation*}
for some $\sigma_0,\sigma_1\in\bC^{2\times 2}$.
Assume that a non-physical Pauli $\gZ$ error occurs on the qubit, yielding the state $\sigma(I\otimes\gZ)$. If we apply the classically-controlled-$\gX$ unitary $\gCt{X}$, we obtain the state
\begin{equation*}
  \gCt{\gX}\sigma(I\otimes\gZ)\gCt{\gX} = \gCt{\gX}\sigma\gCt{\gX}(\gZ\otimes\gZ).
\end{equation*}
Thus applying $\gCt{X}$ caused the $\gZ$ error to propagate to the classical control bit! As defined in \Cref{def:gates}, we call this $\gZ$ error on the classical-control bit a \emph{reweighting}. Reweightings are problematic because we would like to assume that our circuit's classical bits are never corrupted.

Note that we cannot avoid reweightings by requiring our faults to be physical (i.e.~CPTP) superoperators. Indeed, we will decompose all faults into linear combinations of Pauli faults. Even a CPTP superoperator can have a Pauli decomposition that includes non-CPTP Pauli terms.

Furthermore, at first glance it seems we only need to consider reweightings on the underlying logical circuit bits, but not on the physical classically-controlled gates we apply. However, in Appendix~\ref{sec:concat} we will concatenate fault-tolerance schemes, specifically by simulating a physical qubit in an outer scheme using a logical qubit of an inner scheme. Hence logical reweightings of the inner scheme become physical reweightings of the outer scheme.

To address these issues, our notion of fault-tolerance in \Cref{def:faulttol} below requires the correct output under arbitrary reweightings on the physical circuit. However, we will typically prove fault-tolerance with respect to reweightings of the logical circuit, meaning that a ``correct output'' may include the effect of an arbitrary reweighting (or linear combination of reweightings) of our logical circuit. We emphasize that logical circuits that are entirely quantum, i.e.~perform all computation on qubits without using logical classically-controlled gates, are not susceptible to such logical reweighting, even if our fault-tolerant physical circuit uses physical classically-controlled gates that experience reweighting. Hence if we restrict attention to entirely quantum logical circuits (which can still simulate classical control by replacing bits with qubits), then our scheme is robust to physical reweightings without inducing logical reweightings.

Our notion of reweightings is adapted from \cite{breuckmann_fault-tolerant_2026}, who observed a related issue arising from classical control with non-physical faults. In an unfortunate clash of notation, \cite{breuckmann_fault-tolerant_2026} called their notion of a reweighting, which is similar to ours, a ``postselection.''

% We address this issue by allowing for reweightings, which by \Cref{def:gates} can be viewed as $\gZ$ errors on classical bits that occur when a classically-controlled gate is performed.

% Therefore when we compose fault-tolerance schemes, specifically by simulating a physical qubit in an outer scheme using a logical qubit of an inner scheme, we must consider reweighting on physical qubits of the outer scheme, as they correspond to logical qubits of the inner scheme.
% Note that logical circuits that are entirely quantum, i.e.~perform all computation on qubits without using logical classically-controlled gates, are not susceptible to logical reweighting, even if our fault-tolerant physical circuit uses physical classically-controlled gates that experience reweighting.

In contrast, we define postselections to be the following notion, which simply postselect on the outcome of a measurement gate.

\begin{definition}
  Let $\cR=(R_1,\dots,R_T)$ be a circuit using a gate set $\cG$. Fix some subset $\ps$ of the set of measurement gates (i.e.~gates $\gMX$ or $\gMZ$) in $\cR$. We call $\ps$ the set of \emph{postselectable gates} in $\cR$. Then a \emph{postselection} for $\cR$ is a vector $\xi\in\bF_2^{\ps}$. The \emph{$\xi$-postselected circuit} $\cR[\xi]$ is defined by replacing the measurement gate labeled $i$ in $\cR$ with its $\xi_i$-postselected version $\gMX[\xi_i]$ or $\gMZ[\xi_i]$ for every $i\in\ps$, and otherwise leaving $\cR$ unchanged.

  As with faults and reweightings, it will always be clear from context when an argument to $\cR[\cdot]$ is a postselection. When multiple are present, we apply the fault last, so $\cR[\cF,\zeta,\xi]=(\cR[\zeta,\xi])[\cF]$.
\end{definition}

We will typically specify postselections for syndrome measurement outcomes in our error-correction gadget. As syndrome measurements do not affect the code's logical qubits, the error-correction properties are preserved under such postselections. Rather, the postselections here allow us to determine the precise set inside which the residual error on the gadget's output state is contained. Without postselections, the residual error could instead be a linear combination of terms supported on different sets.

We now define the notion of \emph{bad sets}, which correspond to sets of qubits on which an error (on the entire bad set) could cause an uncorrectable corruption.

\begin{definition}
  For a set $N$ and a family of subsets $\cE\subseteq 2^N$ called \emph{bad sets}, we say that a set $S\subseteq N$ is \emph{$\cE$-avoiding} if no element of $\cE$ is a subset of $S$, that is, $E\not\subseteq S$ for every $E\in\cE$.

  % We extend this definition to Pauli faults in the natural way, so that a Pauli fault $\cF$ on a circuit using qubits $\Qu{N}\subseteq N$ and time $\leq T$ is \emph{$\cE$-avoiding} for $\cE\subseteq 2^{N\times[T]}$ if $\supp(\cF)\subseteq\Qu{N}\times[T]\subseteq N\times[T]$ is $\cE$-avoiding.

  % We then say that a fault $\cF=(F_1,\dots,F_T)$ is $\cE$-avoiding if each $F_t$ can be expressed as a linear combination $F_t=\sum_{i=1}^{j_t}\alpha_{t,i}F_{t,i}$ of Pauli superoperators $F_{t,1},\dots,F_{t,j_t}$ for coefficients $\alpha_{t,i}\in\bC$ such that the Pauli fault $(F_{1,i_1},\dots,F_{T,i_T})$ is $\cE$-avoiding for every tuple $(i_1,\dots,i_T)\in[j_1]\times\cdots\times[j_T]$.

  For families $\cE_1\subseteq 2^{N_1}$, $\cE_2\subseteq 2^{N_2}$, we let $\cE_1\sqcup\cE_2\subseteq 2^{N_1\sqcup N_2}$ denote the disjoint union of~$\cE_1$ and~$\cE_2$ on separate sets of underlying elements. Therefore for $\cE\subseteq 2^N$ and $T\in\bN$, we write $\cE^{\sqcup T}\subseteq 2^{N\times[T]}$ to denote $\cE\sqcup\cdots\sqcup\cE$ (with $T$ copies of $\cE$).
\end{definition}

% \begin{remark}
%   Because every superoperator can be decomposed into a linear combination of Pauli errors, an equivalent characterization of a $\cE$-avoiding fault is a fault $\cF$ in which each $F_t$ can be expressed as a linear combination $F_t=\sum_{i=1}^{j_t}\alpha_{t,i}F_{t,i}$ for some (possibly non-Pauli) superoperators $F_{t,i}$, such that $\supp(F_{1,i})\sqcup\cdots\sqcup\supp(F_{T,i_T})$ is $\cE$-avoiding for every tuple $(i_1,\dots,i_T)\in[j_1]\times\cdots\times[j_T]$.
% \end{remark}

We study fault-tolerance schemes with quantum inputs and quantum outputs. As the goal of injection is to fault-tolerantly output a logical state encoded in a code, our fault-tolerant schemes must specify quantum codes with associated encoding maps and bad sets of errors for the input and output states. Below, we define the notion of a \emph{decorated code} to capture this information.

\begin{definition}
  \label{def:deccode}
  Let $Q=(Q_X,Q_Z)$ be a $[[n,k]]$ quantum code with encoding isometry $\Enc:\bC^{2^k}\rightarrow\bC^{2^n}$. Let $\cE_X,\cE_Z\subseteq 2^{[n]}$ be families of bad sets, and let $\cE=(\cE_X,\cE_Z)$. We say the data $D=(Q,\Enc,\cE)$ forms a \emph{decorated quantum CSS code} (or simply \emph{decorated code}). We refer to $\cE$ as a family of bad sets for $Q$. If $\cE_X=\cE_Z$, as a shorthand we may write $D=(Q,\Enc,\cE_Z)$.

  When $Q=\emptyset$ is the trivial code with no physical or logical qubits, we let $D=\emptyset$ denote the associated trivial decorated code.

  For $[[n_i,k_i]]$ decorated codes $D_i=(Q_i,\Enc_i,\cE_i)$ for $i\in\{1,2\}$, we let
  \begin{equation*}
    D_1\sqcup D_2 = (Q_1\sqcup Q_2=(Q_{1,X}\oplus Q_{2,X},Q_{1,Z}\oplus Q_{2,Z}),\; \Enc_1\oplus\Enc_2,\; \cE_1\sqcup\cE_2)
  \end{equation*}
  denote the $[[n_1+n_2,k_1+k_2]]$ decorated code given by the combining (i.e.~appending) the disjoint codes $Q_1$ and $Q_2$.

  For an operator $\rho\in\bC^{2^k\times 2^k}$, we say an operator $\sigma\in\bC^{2^n\times 2^n}$ is a \emph{Pauli $\cE$-deviation of $\Enc(\rho)$} if there exists a $n$-qubit Pauli superoperator $F$ given for some $f_X,f_X',f_Z,f_Z'\in\bF_2^n$ by $F(\sigma')=\gX^{f_X}\gZ^{f_Z}\sigma'\gZ^{f_Z'}\gX^{f_X'}$, such that $\sigma\propto F\circ\Enc(\rho)$, and such that $\supp(f_X)\cup\supp(f_X')$ is $\cE_Z$-avoiding and $\supp(f_Z)\cup\supp(f_Z')$ is $\cE_X$-avoiding.

  We extend this definition to sets $S\subseteq\bC^{2^k\times 2^k}$ of operators by letting $\sigma$ be a \emph{Pauli $\cE$-deviation of $\Enc(S)$} if there exists some $\rho\in S$ such that $\sigma$ is a Pauli $\cE$-deviation of $\Enc(\rho)$. In particular, we will often fix a set $\bar{\cO}$ of superoperators, and then take $S=\bar{\cO}(\rho_0):=\{\bar{O}(\rho_0):\bar{O}\in\bar{\cO}\}$.

  We say $\sigma$ is a \emph{$\cE$-deviation of $\Enc(\rho)$ (resp.~$\Enc(S)$)} if $\sigma$ can be expressed as a linear combination $\sigma=\sum_{i=1}^j\sigma_i$ of Pauli $\cE$-deviations $\sigma_1,\dots,\sigma_j$ of $\Enc(\rho)$ (resp.~$\Enc(S)$).

  Similarly, we say a pair of sets $E=(E_X\subseteq[n],\; E_Z\subseteq[n])$ is \emph{$\cE$-avoiding} if $E_X$ is $\cE_X$-avoiding and $E_Z$ is $\cE_Z$-avoiding, and we let $\comp{E}=(\comp{E_X},\comp{E_Z})$ denote the complement sets of $(E_X,E_Z)$. Viewing $\comp{E_X}$ (resp.~$\comp{E}_Z$) as a family of singleton sets $\comp{E_X}=\{\{c\}:c\in\comp{E_X}\}$ (resp.~$\comp{E_Z}=\{\{c\}:c\in\comp{E_Z}\}$), then the notions of deviations above extend to $\comp{E}$-deviations. For instance, $\sigma$ is a Pauli $\comp{E}$-deviation of $\Enc(\rho)$ if there exists a $n$-qubit Pauli superoperator $F$ given for some $f_X,f_X',f_Z,f_Z'\in\bF_2^n$ by $F(\sigma')=\gX^{f_X}\gZ^{f_Z}\sigma'\gZ^{f_Z'}\gX^{f_X'}$, such that $\sigma\propto F\circ\Enc(\rho)$, and such that $\supp(f_X)\cup\supp(f_X')\subseteq E_Z$ and $\supp(f_Z)\cup\supp(f_Z')\subseteq E_X$.
\end{definition}

When specifying families $\cE_\alpha$ of bad sets for $\alpha\in\{X,Z\}$, such as for a length-$n$ code, it will sometimes be convenient to define $\cE_\alpha\subseteq 2^{N'}$ for some superset $N'$ of $[n]$. That is, some bad sets $E\in\cE_\alpha$ can contain elements not corresponding to any code component. Such a bad set $E$ by definition cannot lie inside $[n]$, and hence $E$ has no effect on the definition of a $\cE$-deviation.

We are now ready to present our main definition of a fault-tolerant gadget. Such a gadget implements some superoperator $\bar{O}$ on a logical quantum input state encoded in the decorated code $D_{\mathrm{in}}$, and returns the resulting output state encoded in the decorated code $D_{\mathrm{out}}$. We also allow the gadget to have classical inputs and outputs, which we leave unencoded as we do not allow errors on classical bits in this paper.

\begin{definition}
  \label{def:faulttol}
  For $\alpha\in\{\mathrm{in},\mathrm{out}\}$, fix sets $K_\alpha=\Cl{K_\alpha}\sqcup\Qu{K_\alpha}$ and $N_\alpha=\Cl{N_\alpha}\sqcup\Qu{N_\alpha}$ with $\Cl{K_\alpha}=\Cl{N_\alpha}$. Let $D_\alpha=(C_\alpha,\Enc_\alpha,\cE_\alpha)$ be a decorated $[[\Qu{N_\alpha},\Qu{K_\alpha}]]$ code. Let $\bar{\cO}$ be a set of superoperators $\bar{O}:\bC^{2^{K_{\mathrm{in}}}\times 2^{K_{\mathrm{in}}}}\rightarrow\bC^{2^{K_{\mathrm{out}}}\times 2^{K_{\mathrm{out}}}}$ such that for every input $\rho$ that is classical on bits $\Cl{K_{\mathrm{in}}}$, then $\bar{\cO}(\rho)$ is classical on bits $\Cl{K_{\mathrm{out}}}$. Let $\cR$ be a quantum circuit using space $N$ and time $T$ with input (qu)bits $N_{\mathrm{in}}$ and output (qu)bits $N_{\mathrm{out}}$. Let $\cE_{\mathrm{run}}\subseteq 2^{N\times[T]}$ be a family of bad sets for $\cR$. Let $\ps$ be a set of postselectable gates for $\cR$.

  \begin{itemize}
  \item We say that data $\bfG=(\cR,\cE_{\mathrm{run}},D_{\mathrm{in}},D_{\mathrm{out}},\ps)$ provides a \emph{fault-tolerant gadget for $\bar{\cO}$} if for every $\cE_{\mathrm{in}}$-avoiding pair of sets $E_{\mathrm{in}}=(E_{\mathrm{in},X}\subseteq\Qu{N_{\mathrm{in}}},\; E_{\mathrm{in},Z}\subseteq\Qu{N_{\mathrm{in}}})$, every $\cE_{\mathrm{run}}$-avoiding set $E_{\mathrm{run}}\subseteq N\times[T]$, and every postselection $\xi\in\bF_2^{\ps}$,
    there exists a $\cE_{\mathrm{out}}$-avoiding pair of sets $E_{\mathrm{out}}=(E_{\mathrm{out},X}\subseteq\Qu{N_{\mathrm{out}}},\; E_{\mathrm{out},Z}\subseteq\Qu{N_{\mathrm{out}}})$
    such that for every finite set $\Rf{K}$, every operator $\rho\in\bC^{2^{K_{\mathrm{in}}\sqcup\Rf{K}}\times 2^{K_{\mathrm{in}}\sqcup\Rf{K}}}$ that is classical on bits $\Cl{K_{\mathrm{in}}}$, every $\comp{E_{\mathrm{in}}}$-deviation $\sigma\in\bC^{2^{N_{\mathrm{in}}\sqcup\Rf{K}}\times 2^{N_{\mathrm{in}}\sqcup\Rf{K}}}$ of $\Enc_{\mathrm{in}}(\rho)$, every $\comp{E_{\mathrm{run}}}$-avoiding fault $\cF$ and reweighting $\zeta$ for $\cR$, then $\cR[\cF,\zeta,\xi](\sigma)\in\bC^{2^{N_{\mathrm{out}}\sqcup\Rf{K}}\times 2^{N_{\mathrm{out}}\sqcup\Rf{K}}}$ is a $\comp{E_{\mathrm{out}}}$-deviation of $\Enc_{\mathrm{out}}\circ\bar{\cO}(\rho)$.
  \item We say that the fault-tolerant gadget $\bfG$ is \emph{mending} if for every $\cE_{\mathrm{run}}$-avoiding set $E_{\mathrm{run}}\subseteq N\times[T]$, and every postselection $\xi\in\bF_2^{\ps}$,
    there exists a $\cE_{\mathrm{out}}$-avoiding pair of sets $E_{\mathrm{out}}=(E_{\mathrm{out},X}\subseteq\Qu{N_{\mathrm{out}}},\; E_{\mathrm{out},Z}\subseteq\Qu{N_{\mathrm{out}}})$ such that for every finite set $\Rf{K}$, every operator $\rho\in\bC^{2^{K_{\mathrm{in}}\sqcup\Rf{K}}\times 2^{K_{\mathrm{in}}\sqcup\Rf{K}}}$ that is classical on bits $\Cl{K_{\mathrm{in}}}$, every $\Cl{K}\sqcup\Rf{K}$-deviation $\sigma\in\bC^{2^{N_{\mathrm{in}}\sqcup\Rf{K}}\times 2^{N_{\mathrm{in}}\sqcup\Rf{K}}}$ of $\Enc_{\mathrm{in}}(\rho)$ and every $\comp{E_{\mathrm{run}}}$-avoiding fault $\cF$ and reweighting $\zeta$ for $\cR$, then $\cR[\cF,\zeta,\xi](\sigma)\in\bC^{2^{N_{\mathrm{out}}\sqcup\Rf{K}}\times 2^{N_{\mathrm{out}}\sqcup\Rf{K}}}$ is of the form $\cR[\cF,\zeta,\xi](\sigma)=\sum_{j=1}^j\sigma_i$, where each $\sigma_i$ is a $\comp{E_{\mathrm{out}}}$-deviation of $\Enc_{\mathrm{out}}\circ\bar{\cO}\circ L_i(\rho)$ for some superoperator $L_i:\bC^{2^{\Qu{K_{\mathrm{in}}}}\times 2^{\Qu{K_{\mathrm{in}}}}}\rightarrow\bC^{2^{\Qu{K_{\mathrm{in}}}}\times 2^{\Qu{K_{\mathrm{in}}}}}$.
  \end{itemize}

  We say that the fault-tolerant (resp.~mending) gadget $\bfG$ is \emph{refreshing} if the choice of $E_{\mathrm{out}}$ in the first (resp.~both) item(s) above does not depend on $E_{\mathrm{in}}$, but rather only depends on $E_{\mathrm{run}}$ and $\xi$.

\end{definition}

At a high level, \Cref{def:faulttol} captures the the following intuitive notion: a gadget is fault-tolerant if for every input with a sufficiently small corruption, the execution under a sufficiently small fault produces an output state with a sufficiently small corruption. Most of the subtlety then lies in formalizing appropriate notions of ``sufficiently small.'' We provide some more context around this definition below.

We begin with a couple remarks on the notation in \Cref{def:faulttol}. Every application of $\Enc_\alpha$ to a state on (qu)bits $K_\alpha\sqcup\Rf{K}$ acts nontrivially only on qubits in $\Qu{K_\alpha}$, and implicitly applies the identity to bits $\Cl{K_\alpha}$ and qubits $\Rf{K}$. Similarly, the noisy reweighted circuit superoperator $\cR[\cF,\zeta](\cdot)$ implicitly applies the identity to qubits $\Rf{K}$. Furthermore, recall that we define faults $\cF$ to have support $\supp(\cF)\subseteq\Qu{N}\times[T]$, and hence have no support inside $\Cl{N}\times[T]$, even if we do not explicitly include elements of $\Cl{N}\times[T]$ in $\cE_{\mathrm{run}}$.

The additional $\Rf{K}$-qubit system in \Cref{def:faulttol} is called a \emph{reference system} in \cite{nguyen_quantum_2025}, and is needed to ensure that fault-tolerance is preserved under parallel composition of gadgets (see \Cref{lem:parcomp}).

Meanwhile, the \emph{mending} property in \Cref{def:faulttol} is called the \emph{friendly} property in \cite{nguyen_quantum_2025,he_composable_2025}, and is used to ensure that states corrupted by bad errors (i.e.~errors whose support contains bad sets) can be returned to the code space, while possibly incurring a logical error. Such bad errors need to be considered when composing fault-tolerance schemes in a simulative manner, meaning that an instance of a fault-tolerance scheme is used to simulate physical qubits in another fault-tolerance scheme.

We emphasize that in \Cref{def:faulttol}, we require the support of the output error of a gadget to lie inside a set $E_{\mathrm{out}}$ that depends only on the \emph{supports} $E_{\mathrm{in}}$ and $E_{\mathrm{run}}$ of the input error and fault respectively (along with the postselection $\xi$), rather than on the precise choice of input error and fault. This requirement helps ensure that under locally stochastic faults, our ejection gadget has a low \emph{probability} of corrupting each (ejected) output qubit. We will thereby rule out superpositions of low-weight Pauli errors on the output state that have no interpretation as a classical probability distribution. The refreshing property also helps in this regard, by allowing us to ``refresh'' the probability distribution of the support of the error.

% \begin{definition}
%   Let $\bfG=(\cR,\cE_{\mathrm{run}},D_{\mathrm{in}},D_{\mathrm{out}},M)$ be a fault-tolerant gadget for $\bar{\cO}$, where we define all variables as in \Cref{def:faulttol}. For a subset $E_{\mathrm{ref}}\subseteq N\times[T]$, we say a fault $\cF$ is \emph{$E_{\mathrm{ref}}$-independent} if each $F_t$ decomposes as a tensor product 

%   $\bfG$ is \emph{$E_{\mathrm{ref}}$-refreshing} if whenever the fault $\cF$ 
% \end{definition}

The following basic lemma shows that for proving fault-tolerance, it suffices to consider Pauli deviations and faults. Similar results are for instance shown in \cite{nguyen_quantum_2025,he_composable_2025}.

\begin{lemma}
  \label{lem:paulift}
  The data $\bfG=(\cR,\cE_{\mathrm{run}},D_{\mathrm{in}},D_{\mathrm{out}},\ps)$ provides a fault-tolerant gadget for a set $\bar{\cO}$ of superoperators (see \Cref{def:faulttol}) if and only if for every $\cE_{\mathrm{in}}$-avoiding pair of sets $E_{\mathrm{in}}=(E_{\mathrm{in},X}\subseteq\Qu{N_{\mathrm{in}}},\; E_{\mathrm{in},Z}\subseteq\Qu{N_{\mathrm{in}}})$, every $\cE_{\mathrm{run}}$-avoiding set $E_{\mathrm{run}}\subseteq N\times[T]$, and every postselection $\xi\in\bF_2^{\ps}$,
  there exists a $\cE_{\mathrm{out}}$-avoiding pair of sets $E_{\mathrm{out}}=(E_{\mathrm{out},X}\subseteq\Qu{N_{\mathrm{out}}},\; E_{\mathrm{out},Z}\subseteq\Qu{N_{\mathrm{out}}})$
  such that for every finite set $\Rf{K}$, every operator $\rho\in\bC^{2^{K_{\mathrm{in}}\sqcup\Rf{K}}\times 2^{K_{\mathrm{in}}\sqcup\Rf{K}}}$ that is classical on bits $\Cl{K_{\mathrm{in}}}$, every Pauli $\comp{E_{\mathrm{in}}}$-deviation $\sigma\in\bC^{2^{N_{\mathrm{in}}\sqcup\Rf{K}}\times 2^{N_{\mathrm{in}}\sqcup\Rf{K}}}$ of $\Enc_{\mathrm{in}}(\rho)$, every $\comp{E_{\mathrm{run}}}$-avoiding Pauli fault $\cF$ and reweighting $\zeta$ for $\cR$, then $\cR[\cF,\zeta,\xi](\sigma)\in\bC^{2^{N_{\mathrm{out}}\sqcup\Rf{K}}\times 2^{N_{\mathrm{out}}\sqcup\Rf{K}}}$ is a $\comp{E_{\mathrm{out}}}$-deviation of $\Enc_{\mathrm{out}}\circ\bar{\cO}(\rho)$.

  Similarly, $\bfG$ is mending if and only if for every $\cE_{\mathrm{run}}$-avoiding set $E_{\mathrm{run}}\subseteq N\times[T]$, and every postselection $\xi\in\bF_2^{\ps}$,
  there exists a $\cE_{\mathrm{out}}$-avoiding pair of sets $E_{\mathrm{out}}=(E_{\mathrm{out},X}\subseteq\Qu{N_{\mathrm{out}}},\; E_{\mathrm{out},Z}\subseteq\Qu{N_{\mathrm{out}}})$ such that for every finite set $\Rf{K}$, every operator $\rho\in\bC^{2^{K_{\mathrm{in}}\sqcup\Rf{K}}\times 2^{K_{\mathrm{in}}\sqcup\Rf{K}}}$ that is classical on bits $\Cl{K_{\mathrm{in}}}$, every Pauli $\Cl{K}\sqcup\Rf{K}$-deviation $\sigma\in\bC^{2^{N_{\mathrm{in}}\sqcup\Rf{K}}\times 2^{N_{\mathrm{in}}\sqcup\Rf{K}}}$ of $\Enc_{\mathrm{in}}(\rho)$ and every $\comp{E_{\mathrm{run}}}$-avoiding Pauli fault $\cF$ and reweighting $\zeta$ for $\cR$, then $\cR[\cF,\zeta,\xi](\sigma)\in\bC^{2^{N_{\mathrm{out}}\sqcup\Rf{K}}\times 2^{N_{\mathrm{out}}\sqcup\Rf{K}}}$ is of the form $\cR[\cF,\zeta,\xi](\sigma)=\sum_{j=1}^j\sigma_i$, where each $\sigma_i$ is a $\comp{E_{\mathrm{out}}}$-deviation of $\Enc_{\mathrm{out}}\circ\bar{\cO}\circ L_i(\rho)$ for some superoperator $L_i:\bC^{2^{\Qu{K_{\mathrm{in}}}}\times 2^{\Qu{K_{\mathrm{in}}}}}\rightarrow\bC^{2^{\Qu{K_{\mathrm{in}}}}\times 2^{\Qu{K_{\mathrm{in}}}}}$.

  Furthermore, $\bfG$ is refreshing if the choice of $E_{\mathrm{out}}$ does not depend on $E_{\mathrm{in}}$, but rather only depends on $E_{\mathrm{run}}$.

\end{lemma}
\begin{proof}
  The ``only if'' direction is immediate, as Pauli deviations and faults are special cases of general deviations and faults. To prove the ``if'' direction, we simply decompose the relevant deviation and fault into a linear combination of Pauli deviations and faults.

  Specifically, to show (non-mending) fault-tolerance, consider an $\cE_{\mathrm{in}}$-avoiding pair of sets $E_{\mathrm{in}}=(E_{\mathrm{in},X},E_{\mathrm{in},Z})$, an $\cE_{\mathrm{run}}$-avoiding set $E_{\mathrm{run}}$, and a postselection $\xi$, and let $E_{\mathrm{out}}=(E_{\mathrm{out},X},E_{\mathrm{out},Z})$ be the $\cE_{\mathrm{out}}$-avoiding pair of sets given in the lemma statement.
  Now consider an arbitrary $\comp{E_{\mathrm{in}}}$-deviation $\sigma\in\bC^{2^{N_{\mathrm{in}}\sqcup\Rf{K}}\times 2^{N_{\mathrm{in}}\sqcup\Rf{K}}}$ of $\Enc_{\mathrm{in}}(\rho)$, and an arbitrary $\comp{E_{\mathrm{run}}}$-avoiding fault $\cF$ and reweighting $\zeta$ for $\cR$. Then by definition we can write $\sigma=\sum_{i=1}^{j_0}\sigma_i$ for Pauli $\comp{E_{\mathrm{in}}}$-deviations $\sigma_i$ of $\Enc_{\mathrm{in}}(\rho)$, and we can write each $F_t=\sum_{i=1}^{j_t}\alpha_{t,i}F_{t,i}$ for Pauli superoperators $F_{t,i}$ with coefficients $\alpha_{t,i}\in\bC$, such that the Pauli fault $\cF_{i_1,\dots,i_T}=(F_{1,i_1},\dots,F_{T,i_T})$ is $\comp{E_{\mathrm{run}}}$-avoiding for every tuple $(i_1,\dots,i_T)\in[j_1]\times\cdots\times[j_T]$. Then
  \begin{align*}
    \cR[\cF,\zeta,\xi](\sigma)
    &= \sum_{(i_0,\dots,i_T)\in[j_0]\times\cdots\times[j_T]} \alpha_{1,i_1}\cdots\alpha_{T,i_T} \cdot \cR[\cF_{i_1,\dots,i_T},\zeta,\xi](\sigma_{i_0}).
  \end{align*}
  Assuming the hypothesis in the lemma statement, each term in the sum on the RHS above is a $\comp{E_{\mathrm{out}}}$-deviation of $\Enc_{\mathrm{out}}\circ\bar{\cO}(\rho)$, so the sum is also an $\comp{E_{\mathrm{out}}}$-deviation of $\Enc_{\mathrm{out}}\circ\bar{\cO}(\rho)$. Thus $\bfG$ is a fault-tolerant gadget for $\bar{\cO}$, as desired.
  
  % Specifically, to show (non-mending) fault-tolerance, consider an arbitrary $\cE_{\mathrm{in}}\sqcup\Cl{K}\sqcup\Rf{K}$-deviation $\sigma\in\bC^{2^{N_{\mathrm{in}}\sqcup\Rf{K}}\times 2^{N_{\mathrm{in}}\sqcup\Rf{K}}}$ of $\Enc_{\mathrm{in}}(\rho)$, and an arbitrary $\cE_{\mathrm{run}}$-avoiding Pauli fault $\cF$ and reweighting $\zeta$ for $\cR$. Then by definition we can write $\sigma=\sum_{i=1}^{j_0}\sigma_i$ for Pauli $\cE_{\mathrm{in}}\sqcup\Cl{K}\sqcup\Rf{K}$-deviations $\sigma_i$ of $\Enc_{\mathrm{in}}(\rho)$, and we can write each $F_t=\sum_{i=1}^{j_t}\alpha_{t,i}F_{t,i}$ for Pauli superoperators $F_{t,i}$ with coefficients $\alpha_{t,i}\in\bC$, such that the Pauli fault $\cF_{i_1,\dots,i_T}=(F_{1,i_1},\dots,F_{T,i_T})$ is $\cE_{\mathrm{run}}$-avoiding for every tuple $(i_1,\dots,i_T)\in[j_1]\times\cdots\times[j_T]$. Then
  % \begin{align*}
  %   \cR[\cF,\zeta](\sigma)
  %   &= \sum_{(i_0,\dots,i_T)\in[j_0]\times\cdots\times[j_T]} \alpha_{1,i_1}\cdots\alpha_{T,i_T} \cdot \cR[\cF_{i_1,\dots,i_T},\zeta](\sigma_{i_0}).
  % \end{align*}
  % Assuming the hypothesis in the lemma statement, each term in the sum on the RHS above is a $\cE_{\mathrm{out}}\sqcup\Cl{K}\sqcup\Rf{K}$-deviation of $\Enc_{\mathrm{out}}\circ\bar{\cO}(\rho)$, so the sum is also $\cE_{\mathrm{out}}\sqcup\Cl{K}\sqcup\Rf{K}$-deviation of $\Enc_{\mathrm{out}}\circ\bar{\cO}(\rho)$. Thus $\bfG$ is a fault-tolerant gadget for $\bar{\cO}$, as desired.

  The refreshing property is then simply as stated in \Cref{def:faulttol}. Meanwhile, the proof of the mending property is analogous to the proof above of non-mending fault-tolerance, so we omit the details to avoid redundancy.
\end{proof}

\subsection{Gadget Composition}
In this section, we present basic lemmas on gadget composition, which simply adapt lemmas found in \cite{nguyen_quantum_2025,he_composable_2025} to our notation. We present proofs for completeness.

We begin with sequential composition, which applies two gadgets in sequence.

\begin{lemma}[Sequential composition]
  \label{lem:seqcomp}
  For $i\in\{1,2\}$, let
  \begin{equation*}
    \bfG^{(i)} = (\cR^{(i)}=(R^{(i)}_1,\dots,R^{(i)}_{T^{(i)}}),\; \cE_{\mathrm{run}}^{(i)},\; D_{\mathrm{in}}^{(i)},\; D_{\mathrm{out}}^{(i)},\; \ps^{(i)})
  \end{equation*}
  be a fault-tolerant gadget for a set $\bar{\cO}^{(i)}$ of superoperators $\bar{O}:\bC^{2^{K^{(i)}_{\mathrm{in}}}\times 2^{K^{(i)}_{\mathrm{in}}}}\rightarrow\bC^{2^{K^{(i)}_{\mathrm{out}}}\times 2^{K^{(i)}_{\mathrm{out}}}}$, such that $K^{(1)}_{\mathrm{out}}=K^{(2)}_{\mathrm{in}}$ and $D_{\mathrm{out}}^{(1)}=D_{\mathrm{in}}^{(2)}$. Letting
  \begin{equation*}
    \cR^{(2)}\circ\cR^{(1)} = (R^{(1)}_1,\dots,R^{(1)}_{T^{(1)}},R^{(2)}_1,\dots,R^{(2)}_{T^{(2)}})
  \end{equation*}
  denote the \emph{sequential composition} of the circuits $\cR^{(1)},\cR^{(2)}$, then the \emph{sequential composition}
  \begin{equation*}
    \bfG^{(2)}\circ\bfG^{(1)} := (\cR^{(2)}\circ\cR^{(1)},\; \cE_{\mathrm{run}}^{(1)}\sqcup\cE_{\mathrm{run}}^{(2)},\; D_{\mathrm{in}}^{(1)},\; D_{\mathrm{out}}^{(2)},\; \ps^{(1)}\sqcup\ps^{(2)})
  \end{equation*}
  of the gadgets $\bfG^{(1)},\bfG^{(2)}$ forms a fault-tolerant gadget for the set of superoperators
  \begin{equation*}
    \bar{\cO}^{(2)}\circ\bar{\cO}^{(1)}:=\{\bar{O}^{(2)}\circ \bar{O}^{(1)}:\bar{O}^{(1)}\in\bar{\cO}^{(1)},\bar{O}^{(2)}\in\bar{\cO}^{(2)}\}.
  \end{equation*}
  Furthermore, if $\bfG^{(1)}$ is mending (resp.~refreshing), then $\bfG^{(2)}\circ\bfG^{(1)}$ is mending (resp.~refreshing).
\end{lemma}

In \Cref{lem:seqcomp}, the requirement that $K^{(1)}_{\mathrm{out}}=K^{(2)}_{\mathrm{in}}$ and $D_{\mathrm{out}}^{(1)}=D_{\mathrm{in}}^{(2)}$ implies that the sets of output bits and qubits of $\cR^{(1)}$ equal the sets of output bits and qubits respectively of $\cR^{(2)}$. Hence the sequential composition circuit $\cR^{(2)}\circ\cR^{(1)}$ is well-defined.

\begin{proof}[Proof of \Cref{lem:seqcomp}]
  For $\alpha\in\{\mathrm{in},\mathrm{out}\}$ and $i\in\{1,2\}$, let
  \begin{equation*}
    D_\alpha^{(i)} = (Q_\alpha^{(i)},\; \Enc_\alpha^{(i)},\; \cE_\alpha^{(i)}).
  \end{equation*}
  Fix a $\cE_{\mathrm{in}}^{(1)}$-avoiding pair of sets $E_{\mathrm{in}}^{(1)}$, a $\cE_{\mathrm{run}}^{(1)}\sqcup\cE_{\mathrm{run}}^{(2)}$-avoiding set $E_{\mathrm{run}}=E_{\mathrm{run}}^{(1)}\sqcup E_{\mathrm{run}}^{(2)}$, and a postselection $\xi=(\xi^{(1)},\xi^{(2)})\in\bF_2^{\ps^{(1)}\sqcup\ps^{(2)}}$, so that for $i\in\{1,2\}$ then $E_{\mathrm{run}}^{(i)}$ is $\cE_{\mathrm{run}}^{(i)}$-avoiding, and $\xi^{(i)}\in\bF_2^{\ps^{(i)}}$. Let $E_{\mathrm{out}}^{(1)}=E_{\mathrm{in}}^{(2)}$ be the $\cE_{\mathrm{out}}^{(1)}=\cE_{\mathrm{in}}^{(2)}$-avoiding output error set guaranteed to exist by the fault-tolerance of $\bfG^{(1)}$ (see \Cref{def:faulttol}) for our chosen input error set $E_{\mathrm{in}}^{(1)}$, fault error set $E_{\mathrm{run}}^{(1)}$, and postselection $\xi^{(1)}$. Then let $E_{\mathrm{out}}^{(2)}$ be the $\cE_{\mathrm{out}}^{(2)}$-avoiding output error set guaranteed to exist by the fault-tolerance of $\bfG^{(2)}$ for input error set $E_{\mathrm{in}}^{(2)}$, fault error set $E_{\mathrm{run}}^{(2)}$, and postselection $\xi^{(2)}$.

  By \Cref{def:faulttol}, to show that $\bfG^{(2)}\circ\bfG^{(1)}$ is fault-tolerant, it suffices to show that for every finite set $\Rf{K}$, every operator $\rho\in\bC^{2^{K_{\mathrm{in}}^{(1)}\sqcup\Rf{K}}\times 2^{K_{\mathrm{in}}^{(1)}\sqcup\Rf{K}}}$ that is classical on bits $\Cl{K_{\mathrm{in}}^{(1)}}$, every $\comp{E_{\mathrm{in}}^{(1)}}$-deviation $\sigma$ of $\Enc_{\mathrm{in}}^{(1)}(\rho)$, every $\comp{E_{\mathrm{run}}}$-avoiding fault $\cF$ and reweighting $\zeta$ for $\cR^{(2)}\circ\cR^{(1)}$, then $\cR^{(2)}\circ\cR^{(1)}[\cF,\zeta,\xi](\sigma)\in\bC^{2^{N_{\mathrm{out}}\sqcup\Rf{K}}\times 2^{N_{\mathrm{out}}\sqcup\Rf{K}}}$ is a $\comp{E_{\mathrm{out}}^{(2)}}$-deviation of $\Enc_{\mathrm{out}}^{(2)}\circ\bar{\cO}^{(2)}\circ\bar{\cO}^{(1)}(\rho)$.

  Here by definition $\cF=(F_1^{(1)},\dots,F_{T^{(1)}}^{(1)},F_1^{(2)},\dots,F_{T^{(2)}}^{(2)})$, where for $i\in\{1,2\}$, then $\cF^{(i)}=(F_1^{(i)},\dots,F_{T^{(i)}}^{(i)})$ is a $\comp{E_{\mathrm{run}}^{(i)}}$-avoiding fault for $\cR^{(i)}$. Similarly, $\zeta=(\zeta^{(1)},\zeta^{(2)})$, where each $\zeta^{(i)}$ is a reweighting for $\cR^{(i)}$. Hence the fault-tolerance of $\bfG^{(1)}$ implies that $\cR^{(1)}[\cF^{(1)},\zeta^{(1)},\xi^{(1)}](\sigma)$ is a $\comp{E_{\mathrm{out}}^{(1)}}=\comp{E_{\mathrm{in}}^{(2)}}$-deviation of $\Enc_{\mathrm{out}}^{(1)}\circ\bar{\cO}^{(1)}(\rho)=\Enc_{\mathrm{in}}^{(2)}\circ\bar{\cO}^{(1)}(\rho)$. Then the fault-tolerance of $\bfG^{(2)}$ implies that
  \begin{equation}
    \label{eq:scboth}
    \cR^{(2)}[\cF^{(2)},\zeta^{(2)},\xi^{(1)}](\cR^{(1)}[\cF^{(1)},\zeta^{(1)},\xi^{(1)}](\sigma)) = \cR^{(2)}\circ\cR^{(1)}[\cF,\zeta,\xi](\sigma)
  \end{equation}
  is a $\comp{E_{\mathrm{out}}^{(2)}}$-deviation of $\Enc_{\mathrm{out}}^{(2)}\circ\bar{\cO}^{(2)}\circ\bar{\cO}^{(1)}(\rho)$, as desired, so $\bfG^{(2)}\circ\bfG^{(1)}$ is indeed a fault-tolerant gadget for $\bar{\cO}^{(2)}\circ\bar{\cO}^{(1)}$.

  An analogous argument implies that if $\bfG^{(1)}$ is mending, then $\bfG^{(2)}\circ\bfG^{(1)}$ is also mending. The difference is that now we instead assume that $\sigma$ is a $\Cl{K_{\mathrm{in}}^{(1)}}\sqcup\Rf{K}$-deviation of $\Enc_{\mathrm{in}}^{(1)}(\rho)$, so that the mending property of $\bfG^{(1)}$ implies that $\cR^{(1)}[\cF^{(1)},\zeta^{(1)},\xi^{(1)}](\sigma)$ is a linear combination of $\comp{E_{\mathrm{out}}^{(1)}}=\comp{E_{\mathrm{in}}^{(2)}}$-deviations of $\Enc_{\mathrm{out}}^{(1)}\circ\bar{\cO}^{(1)}\circ L(\rho)=\Enc_{\mathrm{in}}^{(2)}\circ\bar{\cO}^{(1)}\circ L(\rho)$ for superoperators $L$ acting on qubits $\Qu{K_{\mathrm{in}}^{(1)}}$. Hence by \Cref{eq:scboth}, the fault-tolerance of $\bfG^{(2)}$ implies that $\cR^{(2)}\circ\cR^{(1)}[\cF,\zeta,\xi](\sigma)$ is a linear combination of $\comp{E_{\mathrm{out}}^{(2)}}$-deviations of $\Enc_{\mathrm{out}}^{(2)}\circ\bar{\cO}^{(2)}\circ\bar{\cO}^{(1)}\circ L(\rho)$ for superoperators $L$ acting on qubits $\Qu{K_{\mathrm{in}}^{(1)}}$. Thus $\bfG^{(2)}\circ\bfG^{(1)}$ indeed satisfies the mending definition in \Cref{def:faulttol}.

  Meanwhile, if $\bfG^{(1)}$ is refreshing, then the choice of $E_{\mathrm{out}}^{(1)}=E_{\mathrm{in}}^{(2)}$ in the argument above depends only on $E_{\mathrm{run}}^{(1)}$ and on $\xi^{(1)}$, and hence the choice of $\cE_{\mathrm{out}}^{(2)}$ depends only on $E_{\mathrm{run}}=E_{\mathrm{run}}^{(1)}\sqcup E_{\mathrm{run}}^{(2)}$ and on $\xi=(\xi^{(1)},\xi^{(2)})$. Thus $\bfG^{(2)}\circ\bfG^{(1)}$ is refreshing.
\end{proof}

We now present parallel composition, which applies two gadgets in parallel. Note in sequential composition above, we used notation such as $\cE^{(1)}_{\mathrm{run}}\sqcup\cE_{\mathrm{run}}^{(2)}$ to denote the disjoint union of the families of bad sets $\cE^{(1)}_{\mathrm{run}},\cE_{\mathrm{run}}^{(2)}$ composed across time, so that $\cE^{(1)}_{\mathrm{run}}$ is supported on timesteps $1,\dots,T^{(1)}$ and $\cE^{(2)}_{\mathrm{run}}$ is supported on timesteps $T^{(1)}+1,\dots,T^{(1)}+T^{(2)}$. In contrast, in parallel composition below we let $\cE^{(1)}_{\mathrm{run}}\sqcup\cE_{\mathrm{run}}^{(2)}$ denote the disjoint union of $\cE^{(1)}_{\mathrm{run}},\cE_{\mathrm{run}}^{(2)}$ composed across space, so that $\cE^{(1)}_{\mathrm{run}}$ and $\cE^{(2)}_{\mathrm{run}}$ are supported on the same timesteps $1,\dots,T$, but on disjoint sets of qubits associated to the respective circuits $\cR^{(1)}$ and $\cR^{(2)}$.

\begin{lemma}[Parallel composition]
  \label{lem:parcomp}
  For $i\in\{1,2\}$, let
  \begin{equation*}
    \bfG^{(i)} = (\cR^{(i)}=(R^{(i)}_1,\dots,R^{(i)}_T),\; \cE_{\mathrm{run}}^{(i)},\; D_{\mathrm{in}}^{(i)},\; D_{\mathrm{out}}^{(i)},\; \ps^{(i)})
  \end{equation*}
  be a fault-tolerant gadget for a set $\bar{\cO}^{(i)}$ of superoperators $\bar{O}:\bC^{2^{K^{(i)}_{\mathrm{in}}}\times 2^{K^{(i)}_{\mathrm{in}}}}\rightarrow\bC^{2^{K^{(i)}_{\mathrm{out}}}\times 2^{K^{(i)}_{\mathrm{out}}}}$, such that $\cR^{(1)}$ and $\cR^{(2)}$ have the same time usage $T$. Letting
  \begin{equation*}
    \cR^{(1)}\otimes\cR^{(2)} := (R_1^{(1)}\otimes R_1^{(2)},\dots,R_T^{(1)}\otimes R_T^{(2)})
  \end{equation*}
  denote the \emph{parallel composition} of the circuits $\cR^{(1)},\cR^{(2)}$, then the \emph{parallel composition}
  \begin{equation*}
    \bfG^{(1)}\otimes\bfG^{(2)} := (\cR^{(1)}\otimes\cR^{(2)},\; \cE_{\mathrm{run}}^{(1)}\sqcup\cE_{\mathrm{run}}^{(2)},\; D_{\mathrm{in}}^{(1)}\sqcup D_{\mathrm{in}}^{(2)},\; D_{\mathrm{out}}^{(1)}\sqcup D_{\mathrm{out}}^{(2)},\; \ps^{(1)}\sqcup\ps^{(2)})
  \end{equation*}
  of the gadgets $\bfG^{(1},\bfG^{(2)}$ forms a fault-tolerant gadget for the set of superoperators
  \begin{equation*}
    \bar{\cO}^{(1)}\otimes\bar{\cO}^{(2)}:=\{\bar{O}^{(1)}\otimes \bar{O}^{(2)}:\bar{O}^{(1)}\in\bar{\cO}^{(1)},\bar{O}^{(2)}\in\bar{\cO}^{(2)}\}.
  \end{equation*}
  Furthermore, if both $\bfG^{(1)}$ and $\bfG^{(2)}$ are mending (resp.~refreshing), then $\bfG^{(1)}\otimes\bfG^{(2)}$ is mending (resp.~refreshing).
\end{lemma}
\begin{proof}
  For $\alpha\in\{\mathrm{in},\mathrm{out}\}$ and $i\in\{1,2\}$, let
  \begin{equation*}
    D_\alpha^{(i)} = (Q_\alpha^{(i)},\; \Enc_\alpha^{(i)},\; \cE_\alpha^{(i)}).
  \end{equation*}
  Fix a $\cE_{\mathrm{in}}^{(1)}\sqcup\cE_{\mathrm{in}}^{(2)}$-avoiding pair of sets $E_{\mathrm{in}}=E_{\mathrm{in}}^{(1)}\sqcup E_{\mathrm{in}}^{(2)}$, a $\cE_{\mathrm{run}}^{(1)}\sqcup\cE_{\mathrm{run}}^{(2)}$-avoiding set $E_{\mathrm{run}}=E_{\mathrm{run}}^{(1)}\sqcup E_{\mathrm{run}}^{(2)}$, and a postselection $\xi=(\xi^{(1)},\xi^{(2)})\in\bF_2^{\ps^{(1)}\sqcup\ps^{(2)}}$, so that for $i\in\{1,2\}$ then $E_{\mathrm{in}}^{(i)}$ is $\cE_{\mathrm{in}}^{(i)}$-avoiding, $E_{\mathrm{run}}^{(i)}$ is $\cE_{\mathrm{run}}^{(i)}$-avoiding, and $\xi^{(i)}\in\bF_2^{\ps^{(i)}}$. For $i\in\{1,2\}$, let $E_{\mathrm{out}}^{(i)}$ be the $\cE_{\mathrm{out}}^{(i)}$-avoiding output error set guaranteed to exist by the fault-tolerance of $\bfG^{(i)}$ (see \Cref{def:faulttol}) for our chosen input error set $E_{\mathrm{in}}^{(i)}$, fault error set $E_{\mathrm{run}}^{(i)}$, and postselection $\xi^{(i)}$. Then let $E_{\mathrm{out}}=E_{\mathrm{out}}^{(1)}\sqcup E_{\mathrm{out}}^{(2)}$.

  By \Cref{lem:paulift}, to show that $\bfG^{(1)}\otimes\bfG^{(2)}$ is fault-tolerant, it suffices to show that for every finite set $\Rf{K}$, every operator $\rho\in\bC^{2^{K_{\mathrm{in}}^{(1)}\sqcup K_{\mathrm{in}}^{(2)}\sqcup\Rf{K}}\times 2^{K_{\mathrm{in}}^{(1)}\sqcup K_{\mathrm{in}}^{(2)}\sqcup\Rf{K}}}$ that is classical on bits in $\Cl{K_{\mathrm{in}}^{(1)}}\sqcup\Cl{K_{\mathrm{in}}^{(2)}}$, every Pauli $\comp{E_{\mathrm{in}}}$-deviation $\sigma$ of $\Enc_{\mathrm{in}}^{(1)}\otimes\Enc_{\mathrm{in}}^{(2)}(\rho)$, every $\comp{E}_{\mathrm{run}}$-avoiding Pauli fault $\cF$ and reweighting $\zeta$ for $\cR^{(1)}\otimes\cR^{(2)}$, then $\cR^{(1)}\otimes\cR^{(2)}[\cF,\zeta,\xi](\sigma)$ is a $\comp{E}_{\mathrm{out}}$-deviation of $(\Enc_{\mathrm{out}}^{(1)}\otimes\Enc_{\mathrm{out}}^{(2)})\circ(\bar{\cO}^{(1)}\otimes\bar{\cO}^{(2)})(\rho)$.

  Here by definition we can write $\cF=(F_1^{(1)}\otimes F_1^{(2)},\dots,F_T^{(1)}\otimes F_T^{(2)})$, where for $i\in\{1,2\}$, then $\cF^{(i)}=(F_1^{(i)},\dots,F_T^{(i)})$ is a $\comp{E_{\mathrm{run}}^{(i)}}$-avoiding fault for $\cR^{(i)}$. Similarly, $\zeta=(\zeta^{(1)},\zeta^{(2)})$, where each $\zeta^{(i)}$ is a reweighting for $\cR^{(i)}$. Furthermore, we can write
  \begin{equation*}
    \sigma = (F_{\mathrm{in}}^{(1)}\otimes F_{\mathrm{in}}^{(2)})(\Enc_{\mathrm{in}}^{(1)}\otimes\Enc_{\mathrm{in}}^{(2)})(\rho),
  \end{equation*}
  where each $F_{\mathrm{in}}^{(i)}$ is a $\comp{E_{\mathrm{in}}^{(i)}}$-avoiding Pauli error. Also let $N^{(i)}$ denote the total space used by $\cR^{(i)}$. Then applying the fault-tolerance of $\bfG^{(1)}$ with reference system $N^{(2)}\sqcup\Rf{K}$ implies that $(\cR^{(1)}[\cF^{(1)},\zeta^{(1)},\xi^{(1)}]\otimes I_{N^{(2)}})(\sigma)$ is a linear combination of states in
  \begin{equation}
    \label{eq:pchalf}
    (F_{\mathrm{out}}^{(1)}\otimes F_{\mathrm{in}}^{(2)})(\Enc_{\mathrm{in}}^{(1)}\otimes\Enc_{\mathrm{in}}^{(2)})(\bar{\cO}^{(1)}\otimes I_{K_{\mathrm{in}}^{(2)}})(\rho)
  \end{equation}
  for $\comp{E_{\mathrm{out}}^{(1)}}$-avoiding Pauli errors $F_{\mathrm{out}}^{(1)}$. Now applying the fault-tolerance of $\bfG^{(2)}$ with reference system $N^{(1)}\sqcup\Rf{K}$ implies that
  \begin{equation*}
    (I_{N^{(1)}}\otimes\cR^{(2)}[\cF^{(2)},\zeta^{(2)},\xi^{(2)}]) (\cR^{(1)}[\cF^{(1)},\zeta^{(1)},\xi^{(1)}]\otimes I_{N^{(2)}}) (\sigma) = \cR^{(1)}\otimes\cR^{(2)}[\cF,\zeta,\xi](\sigma)
  \end{equation*}
  is a linear combination of states in
  \begin{equation}
    \label{eq:pcfull}
    (F_{\mathrm{out}}^{(1)}\otimes F_{\mathrm{out}}^{(2)})(\Enc_{\mathrm{in}}^{(1)}\otimes\Enc_{\mathrm{in}}^{(2)})(\bar{\cO}^{(1)}\otimes \bar{\cO}^{(2)})(\rho)
  \end{equation}
  for $\comp{E_{\mathrm{out}}}=\comp{(E_{\mathrm{out}}^{(1)}\sqcup E_{\mathrm{out}}^{(2)})}$-avoiding Pauli errors $F_{\mathrm{out}}^{(1)}\otimes F_{\mathrm{out}}^{(2)}$. Thus $\cR^{(1)}\otimes\cR^{(2)}[\cF,\zeta,\xi](\sigma)$ is indeed a $\comp{E_{\mathrm{out}}}$-deviation of $(\Enc_{\mathrm{in}}^{(1)}\otimes\Enc_{\mathrm{in}}^{(2)})(\bar{\cO}^{(1)}\otimes \bar{\cO}^{(2)})(\rho)$, as desired, so $\bfG^{(1)}\otimes\bfG^{(2)}$ is a fault-tolerant gadget for $\bar{\cO}^{(1)}\otimes\bar{\cO}^{(2)}$.

  An analogous argument implies that if both $\bfG^{(1)}$ and $\bfG^{(2)}$ are mending, then $\bfG^{(1)}\otimes\bfG^{(2)}$ is also mending. The difference is that we now assume that $\sigma$ is a $\Cl{K_{\mathrm{in}}^{(1)}}\sqcup\Cl{K_{\mathrm{in}}^{(2)}}\sqcup\Rf{K}$-deviation of $\Enc_{\mathrm{in}}^{(1)}\otimes\Enc_{\mathrm{in}}^{(2)}(\rho)$, so that the mending property of $\bfG^{(1)}$ implies that \Cref{eq:pchalf} in the argument above is replaced by a linear combination of states
  \begin{equation*}
    (F_{\mathrm{out}}^{(1)}\otimes F_{\mathrm{in}}^{(2)})(\Enc_{\mathrm{in}}^{(1)}\otimes\Enc_{\mathrm{in}}^{(2)})(\bar{\cO}^{(1)}\otimes I_{K_{\mathrm{in}}^{(2)}})(L^{(1)}\otimes I_{K_{\mathrm{in}}^{(2)}})(\rho)
  \end{equation*}
  for superoperators $L^{(1)}$ acting on qubits $\Qu{K_{\mathrm{in}}^{(1)}}$. Then by the mending property of $\bfG^{(2)}$, \Cref{eq:pcfull} in the argument above is replaced by a linear combination of states in
  \begin{equation*}
    (F_{\mathrm{out}}^{(1)}\otimes F_{\mathrm{out}}^{(2)})(\Enc_{\mathrm{in}}^{(1)}\otimes\Enc_{\mathrm{in}}^{(2)})(\bar{\cO}^{(1)}\otimes \bar{\cO}^{(2)})(L^{(1)}\otimes L^{(2)})(\rho)
  \end{equation*}
  for superoperators $L^{(1)}\otimes L^{(2)}$ acting on qubits $\Qu{K_{\mathrm{in}}^{(1)}}\sqcup\Qu{K_{\mathrm{in}}^{(2)}}$. Thus $\bfG^{(1)}\otimes\bfG^{(2)}$ indeed satisfies the mending definition in \Cref{def:faulttol}.

  Meanwhile, if $\bfG^{(1)}$ and $\bfG^{(2)}$ are both refreshing, then the choice of $E_{\mathrm{out}}^{(1)},E_{\mathrm{out}}^{(2)}$ in the argument above depends only on $E_{\mathrm{run}}=E_{\mathrm{run}}^{(1)}\sqcup E_{\mathrm{run}}^{(2)}$ and on $\xi=(\xi^{(1)},\xi^{(2)})$. Thus $\bfG^{(1)}\otimes\bfG^{(2)}$ is refreshing.
\end{proof}

\section{Code Construction}
\label{sec:codeconstruct}
In this section we describe the codes we consider. Our codes are constructed as tensor (i.e.~hypergraph products) of chain complexes associated to a variant of the asymptotically good linear-time encodable/decodable codes of \cite{spielman_linear-time_1996}. Specifically, we omit some code bits in the construction of \cite{spielman_linear-time_1996} to obtain a simplified version with lower distance that is sufficient for our purposes. We compensate for this simplification by assuming that a small fraction of the code (qu)bits have sub-constant error probabilities. We will later show how to concatenate these qubits with appropriate inner codes to return to a uniform error model across qubits.

All (co)chain complexes in this section are assumed to be 1-based.

\subsection{Classical LDPC Codes}
\label{sec:classcodes}
We begin by defining the classical LDPC codes, or equivalently 1-dimensional chain complexes, that we will use to construct quantum codes via tensor products. For this purpose, we first need to define lossless expanders, which for convenience we present in the language of 1-dimensional cochain complexes. As defined below, we will only need the onesided notion, which expands from left to right.

\begin{definition}
  Let $G=(V=V_0\sqcup V_1,\; E\subseteq V_0\times V_1)$ be a bipartite graph of maximum left and right degrees $\Delta_0$ and $\Delta_1$ respectively. For a set of vertices $S\subseteq V$, we let $N_G(S)\subseteq V$ denote the set of vertices connected that share an edge with some vertex in $S$. Then for $\mu,\epsilon>0$, we say $G$ is a \emph{(onesided) $(\mu,\epsilon)$-lossless expander} if it holds for every set $S\subseteq V_0$ that $|N_G(S)|\geq(1-\epsilon)\Delta_0|S|$.
\end{definition}

% \begin{definition}
%   For $\mu,\epsilon>0$ and $\Delta\in\bN$, a \emph{$(\mu,\epsilon)$-lossless expander} of left-degee $\Delta$ is a 1-dimensional cochain complex $\cC=(\cC^0\xrightarrow{\delta}\cC^1)$ satisfying the following. For every $c^0\in C^0$, there are $\Delta$ distinct $c^1\triangleright c^0$. Furthermore, for every set $S\subseteq C^0$ of size $|S|\leq\mu|C^0|$, there are at least $(1-\epsilon)\Delta|S|$ distinct $c^1\in C^1$ such that $c^1\triangleright c^0$ for some $c^0\in S$.
% \end{definition}

% In words, if we view construct a bipartite graph with vertex sets $C^0,C^1$ and bipartite adjacency matrix $\delta$, then a $(\mu,\epsilon,\Delta)$-lossless expander is a $\Delta$-left-regular graph such that every set $S$ of left vertices of size $<\mu|C^0|$ has a neighborhood of size at least $1-\epsilon$ times the maximum possible size $\Delta|S|$.

\begin{lemma}[Well known; see e.g.~Theorem 4.16 in \cite{hoory_expander_2006}]
  \label{lem:lossless}
  For every $\epsilon>0$, there exist $\mu=\mu(\epsilon)>0$ and $\Delta=\Delta(\epsilon)\geq 1/\epsilon$ such that for every $n\geq\Delta$, then there exists a $(\Delta,2\Delta)$-biregular $(\mu,\epsilon)$-lossless expander $G=(V_0\sqcup V_1,E)$ with $|V_0|=2n$ left vertices and $|V_1|=n$ right vertices.
\end{lemma}

It is well-known that random biregular graphs yield the lossless expanders in \Cref{lem:lossless} with probability $\rightarrow 1$ as $n\rightarrow\infty$.\footnote{We can choose $\mu(\epsilon)$ sufficiently small such that for values of $n$ too small for a nontrivial bound on random graphs to hold, then $\mu n<1$, and hence all graphs with $n$ right vertices are $(\mu,\epsilon)$-lossless expanders.} Lossless expanders have also be constructed explicitly \cite{capalbo_randomness_2002,golowich_new_2024,cohen_hdx_2023,hsieh_explicit_2025}.

\begin{definition}
  \label{def:classcode}
  For $\epsilon>0$, define $\mu=\mu(\epsilon)$ and $\Delta=\Delta(\epsilon)$ as in \Cref{lem:lossless}. Then for $\bar{\ell}\in\bN$, we define a 1-dimensional cochain complex $\cC^*=\cC^*(\epsilon,\bar{\ell})$ as follows. Let $\ell_0=2^{\lceil\log_2\Delta\rceil}$ be the least power of 2 above $\Delta$. For $\ell\geq 0$, we let $V_\ell$ be a set of $2^{\ell+\ell_0}$ vertices, and we let $G^{(\ell)}=(V_\ell\sqcup V_{\ell-1},E^{(\ell)})$ be a $(\mu,\epsilon)$-lossless expander as given by \Cref{lem:lossless}. We then let
  \begin{align*}
    C^0 &= V_0 \sqcup \cdots \sqcup V_{\bar{\ell}-1} \sqcup V_{\bar{\ell}} \\
    C^1 &= V_0 \sqcup \cdots \sqcup V_{\bar{\ell}-1},
  \end{align*}
  and we define $\delta^{\cC}\in\bF_2^{C^1\times C^0}$ such that for $c=(c_0,c_1,\dots,c_{\bar{\ell}})\in\cC^0$, letting $\delta^\ell\in\bF_2^{V_{\ell-1}\times V_\ell}$ denote the bipartite adjacency matrix for $G^{(\ell)}$, then
\begin{equation}
  \label{eq:czero}
  \delta^{\cC}(c) = (c_0+\delta^1(c_1),\; c_1+\delta^2(c_2), \dots, c_{\bar{\ell}-1}+\delta^{\bar{\ell}}(c_{\bar{\ell}})) \in \cC^1.
\end{equation}
  % For $i\in\{0,1\}$, we let $V^i_\ell$ denote the copy of $V_\ell$ in $C^i$.
\end{definition}

Equivalently, for $c^0\in C^0$ and $c^1\in C^1$, we have
\begin{equation}
  \label{eq:classcode}
  \delta^{\cC}_{c^1,c^0} = \begin{cases}
    1,& c^0=c^1 \text{ or } (c^0,c^1)\in E^{(\ell)} \text{ for some } \ell\in[\bar{\ell}] \\
    0,&\text{otherwise}.
  \end{cases}
\end{equation}

It follows from \Cref{eq:czero} that cocycles $c\in Z^0(\cC)=H^0(\cC)$ (so that $\delta^{\cC}(c)=0$) are obtained by fixing some $c_{\bar{\ell}}\in\bF_2^{V_{\bar{\ell}}}$, and then letting $c_{\ell-1}=\delta^\ell(c_\ell)$ for each $\ell=\bar{\ell},\dots,1$. It also follows from \Cref{eq:czero} that $\delta^{\cC}$ is surjective, so that $H^1(\cC)=0$. Indeed, given a target $c'=(c_0',\dots,c_{\bar{\ell}-1}')\in\cC^1$, we can construct a preimage $c=(c_0,\dots,c_{\bar{\ell}})\in\cC^0$ with $\delta^{\cC}(c)=c'$ by setting $c_{\bar{\ell}}=0$, and then inductively setting $c_\ell=c_\ell'+\delta^{\ell+1}(c_{\ell+1})$ for each $\ell=\bar{\ell}-1,\dots,0$.

These observations motivate the following natural encoding maps for $\cC^*$ as well as for the dual complex $\cC_*$. We specifically express our encoding maps as chain maps (see \Cref{def:chainmap}).

\begin{definition}
  \label{def:classenc}
  Define $\cC^*=\cC^*(\epsilon,\bar{\ell})$ as in \Cref{def:classcode}. We define a 1-dimensional cochain complex $\cM^*=\cM^*(\epsilon,\bar{\ell})$ by $M^0=V_{\bar{\ell}}\subseteq C^0$, $M^1=0$, and $\delta^{\cM}=0$. We then define chain maps
  \begin{align*}
    \Enc^*:\cM^*&\rightarrow\cC^* \\
    \Enc_*:\cM_*&\rightarrow\cC_*
  \end{align*}
  such that for $c_{\bar{\ell}}\in\bF_2^{V_{\bar{\ell}}}$, then
  \begin{align*}
\Enc^0(c_{\bar{\ell}}) &= (c_0,\dots,c_{\bar{\ell}-1},c_{\bar{\ell}}) \in \cC^0  \text{ where each } c_{\ell-1}=\delta^\ell(c_\ell) \\
    \Enc_0(c_{\bar{\ell}}) &= (0,\dots,0,c_{\bar{\ell}}) \in \cC_0=\cC^0. 
  \end{align*}
\end{definition}

\begin{lemma}
  \label{lem:classenc}
  The map $\Enc^*$ (resp.~$\Enc_*$) in \Cref{def:classenc} induces an isomorphism on cohomology (resp.~homology).
\end{lemma}
\begin{proof}
  \Cref{eq:czero} implies that $\im(\Enc^0)=Z^0(\cC)$, so the result follows by \cite[Lemma~3.33]{golowich_constant-overhead_2025}.
\end{proof}

As a side remark, because it takes $O(\Delta\cdot|V_\ell|)$ time to apply $\delta^\ell$ to $c_\ell$, then if we view $Z^0(\cC)$ as a classical code with message $c_{\bar{\ell}}\in\bF_2^{V_{\bar{\ell}}}$, this code has linear encoding time $O(\Delta\cdot(|V_{\bar{\ell}}|+\cdots+|V_1|))=O(|C^0|)$ assuming $\Delta=O(1)$.

\subsection{Decorated Quantum Product Codes}
\label{sec:quantumprod}
We will consider quantum codes associated to cochain complexes given by tensor products of the complex $\cC^*(\epsilon,\bar{\ell})$ from \Cref{def:classcode} along with the dual complex $\cC_*(\epsilon,\bar{\ell})$. In this section, we formally present these codes, and decorate them with encoding maps and families of bad sets.

\subsubsection{Product Codes and Encoding Map}
\label{sec:qpcode}
We begin by defining the product codes we use along with their associated encoding map. While quantum codes can be constructed from tensor products of just two 1-dimensional cochain complexes, we use 4-dimensional products to prove \Cref{thm:maininf}, in order to obtain single-shot (i.e.~constant-overhead) gadgets for preparing logical $\ket{0}$ and $\ket{+}$ states. Intuitively, the extra (co)chain complex levels provide redundancies in stabilizer measurements, which allow us to correct measurement errors in the state preparation gadget without needing to repeat measurements. Related ideas were for instance used in \cite{bombin_dimensional_2016,tan_single-shot_2025,golowich_constant-overhead_2025}.

\begin{definition}
  \label{def:prodcode}
  For $\epsilon>$, $\bar{\ell}\in\bN$, and for $r_X,r_Z\in\bZ_{\geq 0}$, we define the $r=r_X+r_Z$ dimensional cochain complex
  \begin{equation*}
    \cC^*(\epsilon,\bar{\ell},r_X,r_Z) := \cC_*(\epsilon,\bar{\ell})^{\otimes r_X}\otimes\cC^*(\epsilon,\bar{\ell})^{\otimes r_Z}.
  \end{equation*}
  Similarly defining
  \begin{equation*}
    \cM^*(\epsilon,\bar{\ell},r_X,r_Z) := \cM_*(\epsilon,\bar{\ell})^{\otimes r_X}\otimes\cM^*(\epsilon,\bar{\ell})^{\otimes r_Z},
  \end{equation*}
  we then define the associated product encoding chain map $\Enc_{\cC(\epsilon,\bar{\ell},r_X,r_Z)}:\cM^*(\epsilon,\bar{\ell},r_X,r_Z)\rightarrow\cC^*(\epsilon,\bar{\ell},r_X,r_Z)$ by
  \begin{equation*}
    \Enc_{\cC(\epsilon,\bar{\ell},r_X,r_Z)} = (\Enc_*)^{\otimes r_X} \otimes (\Enc^*)^{\otimes r_Z}.
  \end{equation*}

  Note that when $r_X=r_Z=0$, then $\cC^*(\epsilon,\bar{\ell},0,0)\cong\cM^*(\epsilon,\bar{\ell},0,0)$ is the $0$-dimensional cochain complex with a 1-dimensional space of cochains.
\end{definition}

We will specifically consider the quantum code at level $r_X$ of $\cC^*(\epsilon,\bar{\ell},r_X,r_Z)$, with encoding map $\Enc_{\cC(\epsilon,\bar{\ell},r_X,r_Z)}^{r_X}$. The following basic lemma shows that level $r_X$ is indeed the only level yielding a code of nonzero dimension:

\begin{lemma}
  \label{lem:prodenc}
  Let $\cC^*=\cC^*(\epsilon,\bar{\ell},r_X,r_Z)$ and $\cM^*=\cM^*(\epsilon,\bar{\ell},r_X,r_Z)$. Then $\Enc_{\cC}:\cM^*\rightarrow\cC^*$ induces an isomorphism on cohomology $\Enc_{\cC}:H^*(\cM)\xrightarrow{\sim}H^*(\cC)$, so that
  \begin{equation}
    \label{eq:prodenc}
    H^i(\cC) \cong H^i(\cM) = \begin{cases}
      \bF_2^{(V_{\bar{\ell}})^r},&i=r_X\\
      0,&i\neq r_X.
    \end{cases}
  \end{equation}
\end{lemma}
\begin{proof}
  Because $\Enc^*,\Enc_*$ induce isomorphisms on cohomology by \Cref{lem:classenc}, the product map $\Enc_{\cC}$ must also induce an isomorphism on cohomology by the K\"{u}nneth formula (\Cref{prop:kunneth}).
  Then \Cref{eq:prodenc} follows because $\delta^{\cM}=0$, so that $H^0(\cM(\epsilon,\bar{\ell}))\cong\cM^0(\epsilon,\bar{\ell})=\bF_2^{V_{\bar{\ell}}}$ and $H^1(\cM(\epsilon,\bar{\ell}))\cong\cM^1(\epsilon,\bar{\ell})=0$.
\end{proof}

\begin{remark}
  \label{remark:prodenccoh}
  Throughout this paper, we will consider the quantum code at level $r_X$ of $\cC^*$ with associated encoding map $\Enc_{\cC}^{r_X}:H^{r_X}(\cM)\xrightarrow{\sim}H^{r_X}(\cC)$. When encoding a message, it is implicitly assumed that $\Enc_{\cC}$ refers to this induced map acting on cohomology arising from the chain map $\Enc_{\cC}:\cM^*\rightarrow\cC^*$.
\end{remark}

\begin{remark}
  \label{remark:encdual}
  By \Cref{def:classenc}, it holds for every $c_{\bar{\ell}},c_{\bar{\ell}}'\in\bF_2^{V_{\bar{\ell}}}$ that
  \begin{equation*}
    c_{\bar{\ell}}\cdot c_{\bar{\ell}}' = \Enc^0(c_{\bar{\ell}})\cdot\Enc_0(c_{\bar{\ell}}').
  \end{equation*}
  Hence letting
  \begin{align*}
    \cC^* &= \cC_*(\epsilon,\bar{\ell})^{\otimes r_X}\otimes\cC^*(\epsilon,\bar{\ell})^{\otimes r_Z} = \cC^*(\epsilon,\bar{\ell},r_X,r_Z) \\
    \cM^* &= \cM_*(\epsilon,\bar{\ell})^{\otimes r_X}\otimes\cM^*(\epsilon,\bar{\ell})^{\otimes r_Z} = \cM^*(\epsilon,\bar{\ell},r_X,r_Z) \\
    \Enc_{\cC} &= (\Enc_*)^{\otimes r_X}\otimes(\Enc^*)^{\otimes r_Z}
  \end{align*}
  and
  \begin{align*}
    {\cC^\vee}^* &= \cC^*(\epsilon,\bar{\ell})^{\otimes r_X}\otimes\cC_*(\epsilon,\bar{\ell})^{\otimes r_Z} \cong \cC^*(\epsilon,\bar{\ell},r_Z,r_X) \\
    {\cM^\vee}^* &= \cM^*(\epsilon,\bar{\ell})^{\otimes r_X}\otimes\cM_*(\epsilon,\bar{\ell})^{\otimes r_Z} \cong \cM^*(\epsilon,\bar{\ell},r_Z,r_X) \\
    \Enc_{\cC^\vee} &= (\Enc^*)^{\otimes r_X}\otimes(\Enc_*)^{\otimes r_Z},
  \end{align*}
  where the isomorphisms above simply swap the first $r_X$ tensor factors with the latter $r_Z$ factors,
  then it follows that for every $x,x'\in\bF_2^{(V_{\bar{\ell}})^r}$, we similarly have
  \begin{equation*}
    x\cdot x' = \Enc_{\cC}(x)\cdot\Enc_{\cC^\vee}(x').
  \end{equation*}
  Thus $\Enc_{\cC}$ and $\Enc_{\cC^\vee}$ are dual encoding maps in the sense of \Cref{def:encdual}. Hence by \Cref{lem:encdual}, when we apply Hadamard gates to all qubits to switch between the $X$ and $Z$ bases, we go from the code and encoding map associated to $\cC^*(\epsilon,\bar{\ell},r_X,r_Z)$ to those associated to $\cC^*(\epsilon,\bar{\ell},r_Z,r_X)$, up to a reordering of the tensor factors. Throughout the paper we are therefore able to analyze gadgets in one of the two bases, and then the same analysis applies to an analogous gadget in the oppose basis.
\end{remark}

\subsubsection{Families of Bad Sets for Codes}
\label{sec:codebadsets}
To complete a decoration of $\cC^*(\epsilon,\bar{\ell},r_X,r_Z)$, we need to specify families of bad sets $\cE_X,\cE_Z\subseteq 2^{C^{r_X}(\epsilon,\bar{\ell},r_X,r_Z)}$. For this purpose, it will suffice to define $\cE_Z$; we will then obtain $\cE_X$ by applying our definition of $\cE_Z$ to the dual complex $\cC_*(\epsilon,\bar{\ell},r_X,r_Z)\cong\cC^*(\epsilon,\bar{\ell},r_X,r_Z)$.
We will first define $\cE_Z$ for $\cC^*(\epsilon,\bar{\ell},0,r_Z)$ with $r_X=0$. We will then inductively define $\cE_Z$ for $\cC^*(\epsilon,\bar{\ell},r_X,r_Z)$ from $\cC^*(\epsilon,\bar{\ell},r_X-1,r_Z)$.

Our definition of $\cE_Z$ will be somewhat non-standard, as we will use a non-uniform error model. That is,
% Assuming $1/\epsilon,\;r=O(1)$ are constant, this code can be shown to have $\poly(2^{\bar{\ell}})$ distance.
to facilitate our decoding and ejection procedures, we will assume that a small fraction of the code's physical qubits have subconstant error rates, meaning that they are very unlikely to experience an error. We will present gadgets for state preparation, error correction, and in/ejection under this unequal error model.

This error model is realizable because the fraction of qubits needing low error rates is sufficiently small that we can afford to concatenate with an inner code to simulate such low-error qubits, while only increasing the overall block length by a constant factor. We will formalize this intuition by showing how to concatenate with the construction of \cite{golowich_constant-overhead_2025}, which supports all of the operations we need for the inner code with constant space-time overhead.

We begin by defining the following ``reducer linear'' map, which allows us to ``reduce'' elements of $\bF_2^{V_0\sqcup\cdots\sqcup V_{\bar{\ell}}}$ down into elements of the subspace $\bF_2^{V_{\bar{\ell}}}$ by adding in appropriate coboundaries.

\begin{definition}
  \label{def:pushdown}
  Let $\cA^*=\cC^*(\epsilon,\bar{\ell},1,0)\cong\cC_*(\epsilon,\bar{\ell})$, so that $A^0=V_0\sqcup\cdots\sqcup V_{\bar{\ell}-1}$ and $A^1=V_0\sqcup\cdots\sqcup V_{\bar{\ell}}$. We inductively define a linear \emph{reducer map} map $\Rd:\cA^1\rightarrow\cA^0$ as follows. For the base case, for $v^1\in\bF_2^{V_{\bar{\ell}}}\subseteq\cA^1$ we let $\Rd(v^1)=0$. For the inductive step, for $v^1\in\bF_2^{V_\ell}\subseteq\cA^1$ where $\ell<\bar{\ell}$, we let
  \begin{equation*}
    \Rd(v^1) = \Rd(\delta^{\cA}(v^0)+v^1)+v^0,
  \end{equation*}
  where $v^0\in\bF_2^{V_\ell}\subseteq\cA^0$ denotes the vector $v^1$ moved to components $V_\ell$ in $A^0$, instead of in $A^1$.
\end{definition}

Note that by the definition of $\delta^{\cA}$ (as the transpose of the map $\delta^{\cC}$ defined in \Cref{eq:classcode}), in \Cref{def:pushdown} we have $\delta^{\cA}(v^0)\in\bF_2^{V_\ell\sqcup V_{\ell+1}}$ with $\delta^{\cA}(v^0)|_{V_\ell}=v^1$, and hence $\delta^{\cA}(v^0)+v^1\in\bF_2^{V_{\ell+1}}$, so our induction on $\ell$ is well-defined.

The following lemma presents the key property of the reducer map.

\begin{lemma}
  \label{lem:pdkey}
  Defining all variables as in \Cref{def:pushdown}, we have $\delta^{\cA}(\Rd(v^1))+v^1\in\bF_2^{V_{\bar{\ell}}}\subseteq\cA^1$.
\end{lemma}
\begin{proof}
  By definition
  \begin{equation}
    \label{eq:pdend}
    \delta^{\cA}(\Rd(v^1))+v^1 = \delta^{\cA}(\Rd(\delta^{\cA}(v^0)+v^1))+(\delta^{\cA}(v^0)+v^1),
  \end{equation}
  where $v^1\in\bF_2^{V_\ell}$ and $\delta^{\cA}(v^0)+v^1\in\bF_2^{V_{\ell+1}}$. Therefore while the LHS of \Cref{eq:pdend} lies in $\bF_2^{V_\ell\sqcup\cdots\sqcup V_{\bar{\ell}}}$, the RHS lies in $\bF_2^{V_{\ell+1}\sqcup\cdots\sqcup V_{\bar{\ell}}}$ Therefore we may repeatedly substitute $\delta^{\cA}(v^0)+v^1$ in for $v^1$ in \Cref{eq:pdend}, and we conclude that $\delta^{\cA}(\Rd(v^1))+v^1$ in fact lies in $\bF_2^{V_{\bar{\ell}}}$, as desired.
\end{proof}

We now generalize the reducer map to higher-dimensional complexes $\cC^*(\epsilon,\bar{\ell},r_X,r_Z)$, by sequentially applying the reducer map for the 1-dimensional case in \Cref{def:pushdown} to each of the $r_X$ factors of $\cC_*(\epsilon,\bar{\ell})$ in the product complex $\cC^*(\epsilon,\bar{\ell},r_X,r_Z)$.

\begin{definition}
  \label{def:reducer}
  Let $\cC^*=\cC^*(\epsilon,\bar{\ell},r_X,r_Z)$. For $0\leq i\leq r=r_X+r_Z$, we define the linear \emph{reducer map} $\Rd:\cC^i\rightarrow\cC^{i-1}$ inductively as follows ($\cC$ and $i$ will always be made clear from context when writing $\Rd$). Let $c\in\cC^i$. For the base case, if $r_X=0$ or $i=0$, we let $\Rd(c)=0$.

  For the inductive step, if $r_X\geq 1$ and $i\geq 1$, we decompose $\cC^*=\cA^*\otimes\cB^*$ for $\cA^*=\cC^*(\epsilon,\bar{\ell},1,0)\cong\cC_*(\epsilon,\bar{\ell})$ and $\cB^*=\cC^*(\epsilon,\bar{\ell},r_X-1,r_Z)$, so that $C^i=A^0\times B^i\sqcup A^1\times B^{i-1}$. Letting
  \begin{align*}
    f^0 &:= \Rd\otimes I_{B^{i-1}}(c|_{A^1\times B^{i-1}}) \in \cA^0\otimes\cB^{i-1} \subseteq \cC^{i-1} \\
    c' &:= (c+\delta^{\cC}(f^0))|_{V_{\bar{\ell}}\times B^{i-1}} = c+(\delta^{\cA}\otimes I_B)(f^0)|_{V_{\bar{\ell}}\times B^{i-1}} \in \bF_2^{V_{\bar{\ell}}\times B^{i-1}} \\
    f^1 &:= I_{V_{\bar{\ell}}}\otimes\Rd(c') \in \bF_2^{V_{\bar{\ell}}\times B^{i-2}} \subseteq \cA^1\otimes\cB^{i-2} \subseteq \cC^{i-1},
  \end{align*}
  we then define
  \begin{align*}
    \Rd(c) = (f^0,f^1) \in \cA^0\otimes\cB^{i-1}\oplus\cA^1\otimes\cB^{i-2} = \cC^{i-1}.
  \end{align*}
\end{definition}

We make a few remarks regarding \Cref{def:reducer}. When $r_X=1$ and $r_Z=0$, then $\cB^*=\cC^*(\epsilon,\bar{\ell},0,0)$ is the $0$-dimensional complex with $|B^0|=1$, so $\Rd$ in \Cref{def:reducer} equals $\Rd$ in \Cref{def:pushdown}.

\begin{lemma}
  \label{lem:redprops}
  Define all variables as in \Cref{def:reducer}. Then the following hold:
  \begin{enumerate}
  \item\label{it:rpsupp}
    The vector $c+\delta^{\cC}(\Rd(c))\in\cC^i=\cA^1\otimes\cB^{i-1}\oplus\cA^0\otimes\cB^i$ is supported on components $V_{\bar{\ell}}\times B^{i-1}\sqcup A^0\times B^i$; that is, $(c+\delta^{\cC}(\Rd(c)))|_{(V_0\sqcup\cdots\sqcup V_{\bar{\ell}})\times B^{i-1}}=0$. % technical condition, helps prove the next item
    % The vector $c+\delta^{\cC}(f^0)\in\cA^1\otimes\cB^{i-1}\oplus\cA^0\otimes\cB^i$ is supported on components $V_{\bar{\ell}}\times B^{i-1}\sqcup A^0\times B^i$; that is, $(c+\delta^{\cC}(f^0))|_{(V_0\sqcup\cdots\sqcup V_{\bar{\ell}})\times B^{i-1}}=0$. % technical condition, helps prove the next item
  \item\label{it:rpwelldef} $\Rd(c+\delta^{\cC}(\Rd(c)))=0$. % ensures reduction is reduced
  \item If $\Rd(c)\neq 0$, then $\delta^{\cC}(\Rd(c))|_{A^1\times B^{i-1}}\neq 0$. % says that if not reduced, then reduction is different; follows from previous item
  \item\label{it:rpredred} $\Rd(\Rd(c))=0$ % Says the reducer is reduced, so when running state prep, the reducer we add in turns into a harmless error one level up (in fact want the slightly stronger statement that it is reduced and has no support inside the logical sector $(V_{\bar{\ell}})^r$?)
  \end{enumerate}
\end{lemma}
\begin{proof}
  \begin{enumerate}
  \item By \Cref{lem:pdkey}, $(c+\delta^{\cC}(f^0))|_{A^1\times B^{i-1}}$ is supported inside $V_{\bar{\ell}}\times B^{i-1}$, and by definition $\delta^{\cC}(f^1)$ is also supported inside $V_{\bar{\ell}}\times B^{i-1}$. Hence $c+\delta^{\cC}(\Rd(c))=c+\delta^{\cC}(f^0)+\delta^{\cC}(f^1)$ is supported inside $V_{\bar{\ell}}\times B^{i-1}\sqcup A^0\times B^i$, as desired.
  \item We show the result by induction. For the base case, if $r_X=0$ or $i=0$, then $\Rd(c)=0$ and the result holds. For the inductive step, if $r_X\geq 1$ and $i\geq 1$, let $\bar{c}=c+\delta^{\cC}(\Rd(c))$, and let $\bar{f}^0,\bar{c}',\bar{f}^1$ be the associated variables for computing $\Rd(\bar{c})$ in \Cref{def:reducer}. By \Cref{it:rpsupp} above, $\bar{c}|_{A^1\times B^{i-1}}$ is supported inside $V_{\bar{\ell}}\times B^{i-1}$, so $\bar{f}^0=0$, and hence $\bar{c}'=\bar{c}|_{V_{\bar{\ell}}\times B^{i-1}}$. Then
    \begin{align*}
      \Rd(\bar{c})
      &= \bar{f}^1 = I_{V_{\bar{\ell}}}\otimes\Rd(\bar{c}') = I_{V_{\bar{\ell}}}\otimes\Rd(c+\delta^{\cC}(\Rd(c))|_{A^1\times B^{i-1}}) \\
      &= I_{V_{\bar{\ell}}}\otimes\Rd(c'+I_{V_{\bar{\ell}}}\otimes\delta^{\cB}(f^1)) = I_{V_{\bar{\ell}}}\otimes\Rd(c'+I_{V_{\bar{\ell}}}\otimes(\delta^{\cB}\circ\Rd)(c'))
      % \bar{c}'+\delta^{\cC}(f^1) = \bar{c}'+I_{V_{\bar{\ell}}}\otimes\delta^{\cB}(\Rd(\bar{c}')).
    \end{align*}
    Applying the inductive hypothesis for $\cB$ to the rows of $c'$ (each of which is a vector in $\cB^{i-1}$), we conclude that the RHS above, and hence LHS above, vanishes.
  \item If $\Rd(c)\neq 0$ but $\delta^{\cC}(\Rd(c))|_{A^1\times B^{i-1}}=0$, then because $\Rd(c+\delta^{\cC}(\Rd(c)))$ only depends on components of $c+\delta^{\cC}(\Rd(c))$ in $A^1\times B^{i-1}\subseteq C^i$, it follows that $\Rd(c+\delta^{\cC}(\Rd(c)))=\Rd(c)\neq 0$, contradicting \Cref{it:rpwelldef} above.
  \item We show the result by induction. For the base case, if $r_X=0$ or $i=0$, then $\Rd(c)=0$ and the result holds. For the inductive step, if $r_X\geq 1$ and $i\geq 1$, let $\bar{c}=\Rd(c)$, and let $\bar{f}^0,\bar{c}',\bar{f}^1$ be the associated variables for computing $\Rd(\bar{c})$ in \Cref{def:reducer}. Then by definition $\bar{c}|_{A^1\times B^{i-2}}=f^1$ is supported inside $V_{\bar{\ell}}\times B^{i-2}$, so $\bar{f}^0=0$, and hence
    \begin{equation*}
      \Rd(\bar{c}) = \bar{f}^1 = I_{V_{\bar{\ell}}}\otimes\Rd(f^1) = I_{V_{\bar{\ell}}}\otimes(\Rd\circ\Rd)(c').
    \end{equation*}
    The inductive hypothesis for $\cB$ implies that $\Rd\circ\Rd=0$, so the RHS above, and hence the LHS above, vanishes.
  \end{enumerate}
\end{proof}

\Cref{lem:redprops} motivates the following definition:

\begin{definition}
  \label{def:reduced}
  For $\cC^*=\cC^*(\epsilon,\bar{\ell},r_X,r_Z)$, we define the linear \emph{reduction map} $\Rdn:\cC^i\rightarrow\cC^i$ by
  \begin{equation*}
    \Rdn(c) = c+\delta^{\cC}(\Rd(c)).
  \end{equation*}
  We call $\Rdn(c)$ the \emph{reduction} of $c$, and we say $c$ is \emph{reduced} if $\Rd(c)=0$, or equivalently (by \Cref{lem:redprops}) if $\Rdn(c)=c$.
\end{definition}

The following lemma then follows immediately from \Cref{def:reducer,def:reduced}:

\begin{lemma}
  \label{lem:redfactor}
  Fix $\epsilon>0$, $\bar{\ell}\in\bN$, and for $r_X,r_Z\in\bZ_{\geq 0}$ let $\Rdn^{r_X,r_Z}:\cC^*(\epsilon,\bar{\ell},r_X,r_Z)\rightarrow\cC^*(\epsilon,\bar{\ell},r_X,r_Z)$ denote the reduction map in \Cref{def:reduced}. Define $V_Z\subseteq C(\epsilon,\bar{\ell},r_X,r_Z)$ by $V_Z:=(V_{\bar{\ell}})^{r_X}\times C(\epsilon,\bar{\ell},0,r_Z)$. Then for $c\in\cC^i(\epsilon,\bar{\ell},r_X,r_Z)$, we have
  \begin{equation}
    \label{eq:redfactor}
    \Rdn^{r_X,r_Z}(c)|_{V_Z} = (\Rdn^{1,0})^{\otimes r_X} \otimes I_{C(\epsilon,\bar{\ell},0,r_Z)}(c)|_{V_Z}.
  \end{equation}
\end{lemma}
\begin{proof}
  We show the result by induction on $r_X$. For the base case, when $r_X=0$ then $\Rdn^{r_X,r_Z}$ is simply the identity map, so \Cref{lem:redfactor} holds. For the inductive step, let $r_X\geq 1$ and assume the result holds for $r_X'<r_X$. Define all variables as in \Cref{def:reducer}. Then by definition $c'=(c+\delta^{\cC}(f^0))|_{V_{\bar{\ell}}\times B}$ equals
  \begin{equation*}
    c'=\Rdn^{1,0}\otimes I_B(c)|_{V_{\bar{\ell}}\times B}.
  \end{equation*}
  By the inductive hypothesis
  \begin{equation*}
    c'+\delta^{\cC}(f^1)|_{V_Z} = I_{V_{\bar{\ell}}}\otimes\Rdn^{r_X-1,r_Z}(c')|_{V_Z} = I_{V_{\bar{\ell}}}\otimes(\Rdn^{1,0})^{\otimes r_X-1}\otimes I_{C(\epsilon,\bar{\ell},0,r_Z)}(c')|_{V_Z},
  \end{equation*}
  so combining the two equations above, we have
  \begin{align*}
    \Rdn^{r_X,r_Z}(c)|_{V_Z}
    &= (c+\delta^{\cC}(f^0)+\delta^{\cC}(f^1))|_{V_Z} \\
    &= c'+\delta^{\cC}(f^1)|_{V_Z} \\
    &= (\Rdn^{1,0})^{\otimes r_X}\otimes I_{C(\epsilon,\bar{\ell},0,r_Z)}(c')|_{V_Z},
  \end{align*}
  as desired.
\end{proof}

\Cref{lem:redsupp} below shows that a cochain is reduced iff its support avoids an appropriate \emph{reducer set}, which we now define in \Cref{def:redsupp}:

\begin{definition}
  \label{def:redsupp}
  Let $\cC^*=\cC^*(\epsilon,\bar{\ell},r_X,r_Z)$. We inductively define the \emph{reducer set} $\cE_{Z,\Rd}^{\cC}\subseteq C$ as follows. For the base case, if $r_X=0$, we let $\cE_{Z,\Rd}^{\cC}=\emptyset$.

  For the inductive step, if $r_X\geq 1$, as in \Cref{def:reducer} we decompose $\cC^*=\cA^*\otimes\cB^*$ for $\cA^*=\cC^*(\epsilon,\bar{\ell},1,0)\cong\cC_*(\epsilon,\bar{\ell})$ and $\cB^*=\cC^*(\epsilon,\bar{\ell},r_X-1,r_Z)$, so that $C=A\times B=(A^0\times B)\sqcup(A^1\times B)$. Recalling that $V_{\bar{\ell}}\subseteq A^1$, we then define $\cE_{Z,\Rd}^{\cC}\subseteq A^1\times B\subseteq C$ by
  \begin{equation*}
    \cE_{Z,\Rd}^{\cC} = ((A^1\setminus V_{\bar{\ell}})\times B) \sqcup (V_{\bar{\ell}}\times\cE_{Z,\Rd}^{\cA}).
  \end{equation*}

  We then define $\cE_{X,\Rd}^{\cC}\subseteq C$ to equal the reducer set $\cE_{Z,\Rd}^{\cC^\vee}$ for the dual complex ${\cC^\vee}^*=\cC_*\cong\cC^*(\epsilon,\bar{\ell},r_Z,r_X)$.
\end{definition}

\begin{claim}
  \label{claim:redcapVZ}
  Let $\cC^*=\cC^*(\epsilon,\bar{\ell},r_X,r_Z)$ and $V_Z=(V_{\bar{\ell}})^{r_X}\times C(\epsilon,\bar{\ell},0,r_Z)$. Then $\cE_{Z,\Rd}^{\cC}\cap V_Z=\emptyset$.
\end{claim}
\begin{proof}
  By the inductive definition of $\cE_{Z,\Rd}^{\cC}$ in \Cref{def:redsupp}, every element of $\cE_{Z,\Rd}^{\cC}\subseteq C(\epsilon,\bar{\ell})^{\times r_X}\times C(\epsilon,\bar{\ell})^{\times r_Z}$ has projection onto one of the first $r_X$ coordinates lying outside of $V_{\bar{\ell}}$. Hence no such element can lie inside $V_Z=(V_{\bar{\ell}})^{\times r_X}\times C(\epsilon,\bar{\ell})^{\times r_Z}$.
\end{proof}

\begin{lemma}
  \label{lem:redsupp}
  Let $\cC^*=\cC^*(\epsilon,\bar{\ell},r_X,r_Z)$. A cochain $c\in\cC^i$ is reduced iff $\supp(c)\cap\cE_{Z,\Rd}^{\cC}=\emptyset$.
\end{lemma}
\begin{proof}
  We show the result by induction, following the inductive definition of $\Rd$ in \Cref{def:reducer}. For the base case, if $r_X=0$, then all cochains are reduced, and indeed $\cE_{Z,\Rd}^{\cC}=\emptyset$.

  For the inductive step, assume $r_X\geq 1$, and that the lemma holds for $r_X-1$. As in \Cref{def:reducer,def:redsupp}, let $\cC^*=\cA^*\otimes\cB^*$ for $\cA^*=\cC^*(\epsilon,\bar{\ell},1,0)\cong\cC_*(\epsilon,\bar{\ell})$ and $\cB^*=\cC^*(\epsilon,\bar{\ell},r_X-1,r_Z)$. We will use the following claim, which follows immediately from \Cref{def:pushdown}:
  
  \begin{claim}
    \label{claim:rs1d}
    The reducer map for the 1-dimensional complex $\cA^*=\cC^*(\epsilon,\bar{\ell},1,0)$ has kernel $\bF_2^{V_{\bar{\ell}}}$, and hence is nonvanishing on only those inputs whose support has nonempty intersection with $A^1\setminus V_{\bar{\ell}}$.
  \end{claim}

  Define the variables $f^0,f^0,c'$ as in \Cref{def:reducer}. First assume that $\supp(c)\cap\cE_{Z,\Rd}^{\cC}=\emptyset$, i.e.~$c|_{\cE_{Z,\Rd}^{\cC}}=0$. Then $c|_{(A^1\setminus V_{\bar{\ell}})\times B}=0$, so \Cref{claim:rs1d} implies that $f^0=0$. Hence $c'=c|_{V_{\bar{\ell}}\times B}$. But we also have that $c|_{V_{\bar{\ell}}\times\cE_{Z,\Rd}^{\cA}}=0$, so by the inductive hypothesis $f^1=0$. Thus $\Rd(c)=(f^0,f^1)=0$, so $c$ is reduced.

  For the opposite implication, now assume that $c$ is reduced, so that $\Rd(c)=(f^0,f^1)=0$. Because $f^0=0$, \Cref{claim:rs1d} implies that $c|_{(A^1\setminus V_{\bar{\ell}})\times B}=0$. Then because $c'=c|_{V_{\bar{\ell}}\times B}$ and $f^1=0$, the inductive hypothesis also implies that $c|_{V_{\bar{\ell}}\times\cE_{Z,\Rd}^{\cB}}=0$. Thus $c|_{\cE_{Z,\Rd}^{\cC}}=0$, as desired.
\end{proof}

% TODO: actually we may want to run decoder on corrupted codeword, rather than on syndrome, in which case we'd be reducing the entire corrupted codeword; should check if that's ok? (or could run decoder on syndrome and then can assume error is corrupted; that might be cleaner, as small-set flip already runs on syndrome, and clearing out the top sector can also be done using syndrome)

We will include the reducer set $\cE_{Z,\Rd}^{\cC}$ in our families of bad sets for our codes, in order to force all errors to be reduced; our decoders will always act on reduced errors. By definition the quantum code at level $i$ of a cochain complex $\cC^*$ is invariant under $\gX$-stabilizers, which are of the form $\gX^{\delta^{\cC}(c)}$ for $c\in\cC^{i-1}$. Therefore an error $\gX^e$ for $e\in\cC^i$ has the same effect on the code state as its reduction $\gX^{\Rd(e)}$.

We begin by defining families of bad sets for $\cC^*(\epsilon,\bar{\ell},0,r_Z)$ with $r_X=0$; in this case, by definition all cochains are reduced. We will then inductively extend the definition to larger $r_X$. Intuitively, bad sets will be given by sets of basis elements in $C$ that either:
\begin{enumerate}
\item Are dense within a large connected subgraph of a \emph{connectivity graph} (defined below), or
\item Have large \emph{weighted} Hamming norm for an appropriate weighting (defined below).
\end{enumerate}
We follow the notation of \cite{golowich_constant-overhead_2025} to present the first condition, i.e.~density within a large connected subgraph. However, this idea of characterizing bad sets to be dense within large connected subgraphs had been previously used, such as in \cite{kovalev_fault_2013,gottesman_fault-tolerant_2014,bombin_single-shot_2015,fawzi_efficient_2018}.

We begin by defining bad sets that are dense within a large connected subgraph of some graph.

% \lnote{TODO: in definition below, could remove requirement that $S\subseteq V(G')$, and instead require that $|S\cap V(G')|/|V(G')|\geq\gamma$, so then then family of bad sets would be upward-closed}

\begin{definition}
  \label{def:clusterfam}
  For a graph $G=(V,E)$ and for $\eta,\gamma>0$, we define $\cE(G,\eta,\gamma)\subseteq 2^V$ to be the family of sets $S\subseteq V$ for which there exists some connected subgraph $G'$ of $G$ such that $S\subseteq V(G')$, $|V(G')|\geq\eta$, and $|S|/|V(G')|\geq\gamma$.
  % If $G=(V,E,W)$ additionally has a vertex-weight function $W:V\rightarrow\bR$, we let $\cE(G,\eta,\gamma,\omega)\subseteq 2^V$ be the family containing every $S\in\cE(G,\eta,\gamma)$, along with every $S\subseteq V$ of total weight $|S|_W:=\sum_{v\in S}W(v)\geq\omega$.
  % If $\cC^*=\cC^*(\epsilon,\bar{\ell},r_X,r_Z)$, we additionally define a weight function $W^{\cC}:C\rightarrow\bR$ such that for every $c\in \prod_{i\in[r]}V_{\ell_i}\subseteq C$, then $W^{\cC}(c)=\beta^{\sum_{i\in[r]}(\bar{\ell}-\ell_i)}$ for $\beta=\beta^{\cC}=8\Delta(\epsilon)$ where $\Delta(\epsilon)$ is defined as in \Cref{def:classcode}.
\end{definition}

We will specifically use the connectivity graph defined as follows.

\begin{definition}
  \label{def:conngraph}
  For an $r$-dimensional cochain complex $\cC^*$, we define a \emph{connectivity graph} $G^{\cC}$ with vertex set $C=\bigsqcup_{i\in[r]}C^i$, and an edge connecting every $c,c'\in C$ for which either there exists some $c^0\in C^0$ with $c^0\preceq c,c'$, or there exists some $c^r\in C^r$ with $c^r\succeq c,c'$.
\end{definition}

We also define the following weighted Hamming norm for $\cC^*(\epsilon,\bar{\ell},r_X,r_Z)$.

\begin{definition}
  \label{def:weightnorm}
  Let $\cC^*=\cC^*(\epsilon,\bar{\ell},r_X,r_Z)$ and let $\beta\geq 1$.
  % If $\cC^*=\cC^*(\epsilon,\bar{\ell},r_X,r_Z)$, we additionally let $\beta=\beta(\epsilon)=8\Delta(\epsilon)$ where $\Delta(\epsilon)$ is defined as in \Cref{def:classcode}.
  For every $c\in C$, recalling that there exist $\ell_1,\dots,\ell_r\in\{0,\dots,\bar{\ell}\}$ such that $c\in\prod_{i\in[r]}V_{\ell_i}$, we define the \emph{level} of $c$ to be $\Lev(c):=\sum_{i\in[r]}(\bar{\ell}-\ell_i)$. Then we define the \emph{$\beta$-weighted Hamming norm} $|\cdot|_\beta:\cC\rightarrow\bR$ by $|c|_\beta=\sum_{c'\in\supp(c)}\beta^{\Lev(c')}$.
  We will also write $|\supp(c)|_\beta=|c|_\beta$, so that we can apply the $\beta$-weighted Hamming norm to subsets $\supp(c)\subseteq C$.
\end{definition}

We use the following notation for families of subsets with a given lower bound on the weighted Hamming norm:

\begin{definition}
  Let $\cC^*=\cC^*(\epsilon,\bar{\ell},r_X,r_Z)$ and let $\beta\geq 1$. For $\omega\in\bR$ and for a set $S\subseteq C$, we let
  \begin{equation*}
    2^S|_\beta^{\geq\omega} = \{E\subseteq S:|E|_\beta\geq\omega\}.
  \end{equation*}
  When $\beta=1$, we write $2^S|_1^{\geq\omega}=2^S|^{\geq\omega}$.
\end{definition}

We are now ready to present our families of bad sets for errors on $\cC^*(\epsilon,\bar{\ell},r_X,r_Z)$. Our families of bad sets will be specified by three additional parameters $\eta,\gamma,\omega$. We will typically set the subgraph size threshold $\eta$ and the weighted norm threshold $\omega$ to grow as small polynomials in $|C|$, or equivalently, in $2^{\bar{\ell}}$. Meanwhile, we will set the subgraph density threshold $\gamma$ to be a small constant.

\begin{definition}
  \label{def:badfams}
  Let $\cC^*=\cC^*(\epsilon,\bar{\ell},r_X,r_Z)$, and let $\eta,\gamma\in\bR$. We inductively define a family of bad sets $\cE_Z^{\cC}(\eta,\gamma)\subseteq 2^C$ as follows. For the base case, if $r_X=0$, we let
  \begin{equation}
    \label{eq:badbase}
    \cE_Z^{\cC}(\eta,\gamma) = \cE(G^{\cC},\eta,\gamma).
  \end{equation}
  For the inductive step, if $r_X\geq 1$, we decompose $\cC^*=\cA^*\otimes\cB^*$ for $\cA^*=\cC^*(\epsilon,\bar{\ell},1,0)\cong\cC_*(\epsilon,\bar{\ell})$ and $\cB^*=\cC^*(\epsilon,\bar{\ell},r_X-1,r_Z)$, and we let $\cE_Z^{\cC}(\eta,\gamma)$ be the disjoint union over all $v\in V_{\bar{\ell}}\subseteq A^1$ of the families $\cE_Z^{\cB}(\eta,\gamma)\subseteq 2^B\cong 2^{\{v\}\times B}$.
  
  Define subsets $V_X,V_Z\subseteq C$ by $V_Z:=(V_{\bar{\ell}})^{r_X}\times C(\epsilon,\bar{\ell},0,r_Z)$ and $V_X:=C(\epsilon,\bar{\ell},r_X,0)\times(V_{\bar{\ell}})^{r_Z}$. Then for $\omega\in\bR$, we define a family of bad sets $\cE_Z^{\cC}(\eta,\gamma,\omega;\beta)\subseteq 2^C$ by
  \begin{equation*}
    \cE_Z^{\cC}(\eta,\gamma,\omega;\beta) = \cE_Z^{\cC}(\eta,\gamma) \cup (2^{V_Z}|_\beta^{\geq\omega}) \cup \cE_{Z,\Rd}^{\cC},
  \end{equation*}
  where here we view the set $\cE_{Z,\Rd}^{\cC}$ as a family of singleton sets $\{\{c\}:c\in\cE_{Z,\Rd}^{\cC}\}$.

  We then define $\cE_X^{\cC}(\eta,\gamma,\omega;\beta)\subseteq 2^C$ to equal the family $\cE_Z^{\cC^\vee}(\eta,\gamma,\omega;\beta)$ for the dual complex ${\cC^\vee}^*=\cC_*\cong\cC^*(\epsilon,\bar{\ell},r_Z,r_X)$.

  Given the base case in \Cref{eq:badbase}, this inductive definition can be equivalently stated as
  \begin{align}
    \label{eq:baddirect}
    \begin{split}
      \cE_Z^{\cC}(\eta,\gamma,\omega;\beta) &= \cE_Z^{\cC(\epsilon,\bar{\ell},0,r_Z)}(\eta,\gamma)^{\sqcup(V_{\bar{\ell}})^{r_X}} \cup (2^{V_Z}|_\beta^{\geq\omega}) \cup \cE_{Z,\Rd}^{\cC} \\ % \subseteq 2^{V_Z} \\
      \cE_X^{\cC}(\eta,\gamma,\omega;\beta) &= \cE_X^{\cC(\epsilon,\bar{\ell},r_X,0)}(\eta,\gamma)^{\sqcup(V_{\bar{\ell}})^{r_Z}} \cup (2^{V_X}|_\beta^{\geq\omega}) \cup \cE_{X,\Rd}^{\cC}. % \subseteq 2^{V_X}.
    \end{split}
  \end{align}
  % \begin{align*}
  %   \cE_Z^{\cC}(\eta,\gamma,\omega) &= \cE(G^{\cC(\epsilon,\bar{\ell},0,r_Z)},\eta,\gamma,\omega)^{\sqcup(V_{\bar{\ell}})^{r_X}} \subseteq 2^{V_Z} \\
  %   \cE_X^{\cC}(\eta,\gamma,\omega) &= \cE(G^{\cC(\epsilon,\bar{\ell},r_X,0)},\eta,\gamma,\omega)^{\sqcup(V_{\bar{\ell}})^{r_Z}} \subseteq 2^{V_X}.
  % \end{align*}

  We also let
  \begin{equation*}
    \cE^{\cC}(\eta,\gamma,\omega;\beta) = (\cE_X^{\cC}(\eta,\gamma,\omega;\beta),\cE_Z^{\cC}(\eta,\gamma,\omega;\beta)).
  \end{equation*}
  When $\beta$ is clear from context we may write $\cE^{\cC}(\eta,\gamma,\omega)=\cE^{\cC}(\eta,\gamma,\omega;\beta)$.
\end{definition}

\begin{remark}
  \label{remark:badred}
  In \Cref{eq:baddirect}, we have
  \begin{align*}
    \cE_Z^{\cC(\epsilon,\bar{\ell},0,r_Z)}(\eta,\gamma)^{\sqcup(V_{\bar{\ell}})^{r_X}} \cup (2^{V_Z}|_\beta^{\geq\omega}) &\subseteq 2^{V_Z} \\
    \cE_X^{\cC(\epsilon,\bar{\ell},r_X,0)}(\eta,\gamma)^{\sqcup(V_{\bar{\ell}})^{r_Z}} \cup (2^{V_X}|_\beta^{\geq\omega}) &\subseteq 2^{V_X}.
  \end{align*}
  Thus if we for instance want to show that a reduced error $e\in\cC^i$ is $\cE_Z^{\cC}(\eta,\gamma,\omega;\beta)$-avoiding, \Cref{lem:redsupp} implies that $\supp(e)\cap\cE_{Z,\Rd}^{\cC}=\emptyset$, so it suffices to show that $e|_{V_Z}$ is $\cE_Z^{\cC(\epsilon,\bar{\ell},0,r_Z)}(\eta,\gamma)^{\sqcup(V_{\bar{\ell}})^{r_X}} \cup (2^{V_Z}|_\beta^{\geq\omega})$-avoiding. (The analogous statement for the dual case of $\cE_X^{\cC}(\eta,\gamma,\omega;\beta)$-avoiding errors also holds.)
\end{remark}

Then for the cochain complex $\cC^*=\cC^*(\epsilon,\bar{\ell},r_X,r_Z)$ with product encoding map $\Enc_{\cC}^*$ as defined in \Cref{sec:qpcode} (see in particular \Cref{remark:prodenccoh}), letting $Q=(Q_X,Q_Z)$ denote the CSS code at level $r_X$ of $\cC^*$, we will consider decorated codes of the form
\begin{equation*}
  D = (Q,\; \Enc_{\cC}^{r_X},\; \cE^{\cC}(\eta,\gamma,\omega)).
\end{equation*}

\section{Decoder}
\label{sec:decoder}
In this section, we present decoders for the cochain complexes $\cC^*(\epsilon,\bar{\ell},r_X,r_Z)$. These decoders will allow us to perform fault-tolerant state preparation and logical Pauli measurements for the code at level $r_X$ of $\cC^*(\epsilon,\bar{\ell},r_X,r_Z)$.

As described briefly in \Cref{sec:chaincomplex}, we in fact construct these decoders for direct sum complexes $\cC^*(\epsilon,\bar{\ell},r_X,r_Z)^{\oplus\Upsilon}$ for arbitrary $\Upsilon$ (see \Cref{def:directsum}). Decoding such direct sum complexes will allow us to argue that when our decoder is given many (i.e.~$\Upsilon$) different errors/syndromes with similar supports, the decoder's outputs all have similar supports. This property will be crucial for ensuring that our state preparation gadget in \Cref{sec:stateprep} below satisfies our notion of fault-tolerance in \Cref{def:faulttol}, which requires the output error support to be determined by the input and fault error supports, rather than their specific values. This use of direct sum complexes is inspired by the related notion of \emph{collective cosystolic expansion} introduced in \cite{kaufman_new_2021}, and the notion of \emph{collective robustness} subsequently studied in \cite{dinur_expansion_2024}.

Throughout this section, we extend the reducer and reduction maps defined in \Cref{def:reducer,def:reduced} on the complexes $\cC^*(\epsilon,\bar{\ell},r_X,r_Z)$ to their direct sums $\cC^*(\epsilon,\bar{\ell},r_X,r_Z)^{\oplus\Upsilon}$ in the natural way, i.e.~by applying the reducer/reduction map to each direct sum term.

Our decoders in this section use some ideas adapted from the decoders in \cite[Section~4]{golowich_constant-overhead_2025}, which applied to high-dimensional products of 1-dimensional complexes constructed from a single lossless expander. However, there is significantly more subtlety in our decoders in this section, which apply to high-dimensional products of our 1-dimensional complexes $\cC^*(\epsilon,\bar{\ell})$ constructed from many lossless expanders of varying sizes pasted together. (For an overview of the main ideas behind the 2-dimensional case of our decoder, see \Cref{sec:ecinf}.)

We now briefly summarize the main steps in our decoder. Similarly as described in \Cref{sec:ecinf}, we first ``push'' all errors into an appropriate subset of the qubits, on which we can then run a small-set flip decoder that is more analogous to the decoder in \cite[Section~4]{golowich_constant-overhead_2025}. For the error ``pushing,'' we may first assume the error is reduced, as every error is equivalent to its reducer up to stabilizers. We then greedily apply flips to eliminate any remaining errors in undesirable regions, similarly as described in \Cref{sec:ecinf}. At this point, we have ``pushed'' all errors into a set of qubits on which we can run a small-set flip decoder.

Yet there is an additional complication in adapting the higher-dimensional small-set flip decoder from \cite[Section~4]{golowich_constant-overhead_2025}. Because we need to bound the error weight under two separate norms (weighted and unweighted; see \Cref{def:weightnorm,def:badfams}), we choose which small sets to flip by analyzing coboundaries in two different (but related) cochain complexes $\cC^*$ and ${\cC'}^*$ (see in particular \Cref{lem:ssflip}). The complete analysis, which is somewhat delicate, is given in the sections below.

% \lnote{TODO: Somewhere should try to emphasize that there really is increased subtlety over \cite[Section~4]{golowich_constant-overhead_2025}; currently might be written to seem too similar/straightforward of an adaptation. Specifically, we first need to assume the error is reduced, and perform the initial flips to get down to the $r_X=0$ case. Then we need to simultaneously reduce the error or syndrome weight under two separate norms (unweighted and weighted), and the solution is to compare coboundaries of two different (but related) cochain complexes $\cC^*$ and ${\cC'}^*$; is a somewhat delicate argument.}

\subsection{Small-Set Flip Decoder}
In this section, we present a small-set flip decoder for cochain complexes of the form $\cC^*(\epsilon,\bar{\ell},r_X=0,r_Z=r)$ (and their associated direct sum complexes). Such a decoder guarantees that for every sufficiently low-weight error, there exists a a constant-sized perturbation of the error that either reduces the error weight while preserving the syndrome, or else reduces the syndrome weight. Repeatedly applying such perturbations allows us to reduce the weight of the syndrome to $0$, and hence correct the error.

To construct our decoder, we will need the following modification of the cochain complex $\cC^*(\epsilon,\bar{\ell},r_X=0,r_Z=r)$.

\begin{definition}
  For $\epsilon>0$ and $\bar{\ell}\in\bN$, we define a cochain complex ${\cC'}^*(\epsilon,\bar{\ell})$ with the same space of cochains ${\cC'}^i(\epsilon,\bar{\ell})=\cC^i=\cC^i(\epsilon,\bar{\ell})$ as in \Cref{def:classcode}, but with coboundary map $\delta^{\cC'}\in\bF_2^{C^1\times C^0}$ defined for $c^0\in C^0$ and $c^1\in C^1$ by
  \begin{equation*}
    % \label{eq:classcodemod}
    \delta^{\cC'}_{c^1,c^0} = \begin{cases}
      1,& (c^0,c^1)\in E^{(\ell)} \text{ for some } \ell\in[\bar{\ell}] \\
      0,&\text{otherwise}.
    \end{cases}
  \end{equation*}
  Then for $r_X,r_Z\in\bZ_{\geq 0}$, we define the $r=r_X+r_Z$ dimensional cochain complex
  \begin{equation*}
    {\cC'}^*(\epsilon,\bar{\ell},r_X,r_Z) := \cC'_*(\epsilon,\bar{\ell})^{\otimes r_X}\otimes{\cC'}^*(\epsilon,\bar{\ell})^{\otimes r_Z}.
  \end{equation*}
\end{definition}

That is, $\delta^{\cC'(\epsilon,\bar{\ell})}$ is defined similarly as $\delta^{\cC(\epsilon,\bar{\ell})}$, but where we have removed the $c^0=c^1$ option from \Cref{eq:classcode}. Letting $\prec$ and $\prec'$ denote the partial orders for $\cC^*(\epsilon,\bar{\ell})$ and ${\cC'}^*(\epsilon,\bar{\ell})$, it follows that $c^i\prec' c^j\implies c^i\prec c^j$, but the converse does not necessarily hold.

Because we have removed at most one out of every $\Delta$ nonzero entries from each column of the coboundary matrix, lossless expansion properties are roughly preserved when moving from $\cC(\epsilon,\bar{\ell})$ to $\cC'(\epsilon,\bar{\ell})$. However, for $c\in\bF_2^{V_\ell}\subseteq\cC^0$, whereas $\delta^{\cC(\epsilon,\bar{\ell})}(c)$ may be supported across $V_\ell$ and $V_{\ell-1}\subseteq C^1$, we always have $\delta^{\cC'(\epsilon,\bar{\ell})}(c)$ supported entirely inside $V_{\ell-1}\subseteq C^1$. This property will be useful for designing decoders that limit the blowup in the weighted Hamming norm of the error.

The following key lemma shows that for $\epsilon>0$ sufficiently small, the cochain complex $\cC^*(\epsilon,\bar{\ell},0,r)$ (and its associated direct sum complexes) supports a \emph{small-set flip decoder}. That is, for a sufficiently low-weight error on the code associated to some level $i$, we can further reduce the syndrome weight by performing a constant number of bit flips.

\begin{lemma}
  \label{lem:ssflip}
  For every $r,\bar{\ell}\in\bN$ and every $\beta>1$, there exists $\epsilon=\epsilon(r)>0$ such that the following holds. For arbitrary $\Upsilon\in\bN$, let $\cC^*=\cC^*(\epsilon,\bar{\ell},0,r)^{\oplus\Upsilon}$ and ${\cC'}^*={\cC'}^*(\epsilon,\bar{\ell},0,r)^{\oplus\Upsilon}$ with partial orders $\prec$ and $\prec'$ respectively. Also let $\mu=\mu(\epsilon)$ be the value from \Cref{lem:lossless}. For every $0\leq i\leq r-1$ and every nonzero $e\in\cC^i$ satisfying $|e|\leq\mu 2^{\bar{\ell}/2}$ and $|e|_\beta\leq\beta^{\bar{\ell}/2}$, there exists a basis element $\bar{c}^0\in C^0$ for which at least one of the following holds:
  \begin{enumerate}
  \item There exists $c^{i-1}\in\cC^{i-1}$ with $\bar{c}^0\preceq' c^{i-1}$ such that:
    \begin{align}
      \label{eq:ssfup}
      \begin{split}
        |e+\delta^{\cC'}(c^{i-1})| &< |e|-4r|c^{i-1}|.
        % |e+\delta^{\cC}(c^{i-1})| &< |e|-3r|c^{i-1}| \\
        % |e+\delta^{\cC}(c^{i-1})|_\beta &< |e|_\beta-3r\beta|c^{i-1}|_\beta
      \end{split}
    \end{align}
  \item There exists $c^i\in\cC^i$ with $\bar{c}^0\preceq' c^i$ such that:
    \begin{align}
      \label{eq:ssfdown}
      \begin{split}
        |\delta^{\cC}(e)+\delta^{\cC'}(c^i)| &< |\delta^{\cC}(e)|-4r|c^i|.
        % |\delta^{\cC}(e+c^i)| &< |\delta^{\cC}(e)|-3r|c^i| \\
        % |\delta^{\cC}(e+c^i)|_\beta &< |\delta^{\cC}(e)|_\beta-3r\beta|c^i|_\beta \\
      \end{split}
    \end{align}
  \end{enumerate}
\end{lemma}

Before proving \Cref{lem:ssflip}, we will present some remarks and implications of the lemma.
Note that the factor of $4r$ in the bounds in \Cref{lem:ssflip} is not special; any factor greater than $r$ would suffice for our purposes, and the proof of \Cref{lem:ssflip} in fact implies a factor growing linearly in $\Delta$ for fixed $r$.

The mismatch between $\delta^{\cC}$ and $\delta^{\cC'}$ in \Cref{eq:ssfup} and \Cref{eq:ssfdown} in \Cref{lem:ssflip} is introduced to ensure that we can prove the bound on the weighted Hamming norm in \Cref{lem:ssfimp} below. This weighted Hamming norm bound will ensure that $|e|_\beta$ does not blow up when we repeatedly apply \Cref{lem:ssflip} to add the cochains $\delta^{\cC'}(c^{i-1})$ and $c^i$ in to $e$.

\begin{lemma}
  \label{lem:ssfimp}
  For $\epsilon>0$, $\beta\geq 1$, $\bar{\ell},r\in\bN$, $\Upsilon\in\bN$, let $\cC^*=\cC^*(\epsilon,\bar{\ell},0,r)^{\oplus\Upsilon}$ and ${\cC'}^*={\cC'}^*(\epsilon,\bar{\ell},0,r)^{\oplus\Upsilon}$ with partial orders $\prec$ and $\prec'$ respectively. Let $0\leq i\leq r-1$, let $e\in\cC^i$, and let $\bar{c}^0\in C^0$. Then the following hold:
  \begin{enumerate}
  \item Every $c^{i-1}\in\cC^{i-1}$ with $\bar{c}^0\preceq' c^{i-1}$ that satisfies \Cref{eq:ssfup} must also satisfy
    \begin{align}
      \label{eq:ssfupimp}
      \begin{split}
        |e+\delta^{\cC}(c^{i-1})| &< |e|-3r|c^{i-1}| \\
        |e+\delta^{\cC}(c^{i-1})|_\beta &< |e|_\beta-3r\beta|c^{i-1}|_\beta.
      \end{split}
    \end{align}
  \item Every $c^i\in\cC^i$ with $\bar{c}^0\preceq' c^i$ that satisfies \Cref{eq:ssfdown} must also satisfy
    \begin{align}
      \label{eq:ssfdownimp}
      \begin{split}
        |\delta^{\cC}(e)+\delta^{\cC}(c^i)| &< |\delta^{\cC}(e)|-3r|c^i| \\
        |\delta^{\cC}(e)+\delta^{\cC}(c^i)|_\beta &< |\delta^{\cC}(e)|_\beta-3r\beta|c^i|_\beta.
      \end{split}
    \end{align}
  \end{enumerate}
\end{lemma}
\begin{proof}
  For every $b=(b_1,\dots,b_r)\in\supp(c^{i-1})$ and every $b'\triangleright b$ such that $b'\not\triangleright' b$, then there must be some $j\in[r]$ such that $b'=(b_1',\dots,b_r')$ satisfies $b_{j'}'=b_{j'}$ for every $j'\neq j$, and if $b_j\in V_\ell\subseteq C^0(\epsilon,\bar{\ell})$, then $b_j'\in V_\ell\subseteq C^1(\epsilon,\bar{\ell})$ is the copy of $b_j$ in $C^1(\epsilon,\bar{\ell})$; recall here that a copy of $V_\ell$ lies inside each of $C^0(\epsilon,\bar{\ell})$ and $C^1(\epsilon,\bar{\ell})$. Therefore there are at most $r$ distinct $b'\triangleright b$ such that $b'\not\triangleright' b$. Summing over all $b\in\supp(c^{i-1})$ gives
  \begin{equation*}
    |\delta^{\cC}(c^{i-1})-\delta^{\cC'}(c^{i-1})| \leq r|c^{i-1}|,
  \end{equation*}
  so the first inequality in \Cref{eq:ssfupimp} holds. By analogous reasoning we also have that
  \begin{equation*}
    |\delta^{\cC}(c^i)-\delta^{\cC'}(c^i)| \leq r|c^i|,
  \end{equation*}
  so the first inequality in \Cref{eq:ssfdownimp} holds.

  Furthermore, by the definition of $\cC'$ along with \Cref{def:weightnorm}, every $b\in\supp(c^{i-1})\subseteq\cC^{i-1}$ has $b\succeq \bar{c}^0$, and therefore has level $\Lev(b)=\Lev(\bar{c}^0)+i-1$. Similarly, every $b'\triangleright' b$ has level $\Lev(b')=\Lev(b)+1=\Lev(\bar{c}^0)+i$. Meanwhile, every $b'\triangleright b$ with $b'\not\triangleright' b$ has $\Lev(b')=\Lev(b)=\Lev(\bar{c}^0)+i-1$ (by the characterization above of such $b'$). Thus $\supp(\delta^{\cC'}(c^{i-1}))$ lies entirely with level $\Lev(\bar{c}^0)+i$, and by \Cref{eq:ssfup},
  the restriction of $\delta^{\cC'}(c^{i-1})$ to its support
  agrees with $e$ on at least $4r|c^{i-1}|$ more components than it disagrees with $e$. It follows by the definition of the $\beta$-weighted norm that
  \begin{equation*}
    |e+\delta^{\cC'}(c^{i-1})|_\beta < |e|_\beta-4r\beta^{\Lev(\bar{c}^0)+i}|c^{i-1}| = |e|_\beta-4r\beta|c^{i-1}|_\beta.
  \end{equation*}
  Then because as shown above, $\supp(\delta^{\cC}(c^{i-1}))\setminus\supp(\delta^{\cC'}(c^{i-1}))$ contains at most $r|c^{i-1}|$ elements, all of which lie in level $\Lev(\bar{c}^0)+i-1$, we have
  \begin{equation*}
    |\delta^{\cC}(c^{i-1})-\delta^{\cC'}(c^{i-1})|_\beta \leq r|c^{i-1}|_\beta.
  \end{equation*}
  Because $\beta\geq 1$, the two inequalities above then imply the second inequality in \Cref{eq:ssfupimp}. The second inequality in \Cref{eq:ssfdownimp} follows from analogous reasoning, but where we replace the variables $c^{i-1}$ and $e$ with $c^i$ and $\delta^{\cC}(e)$, respectively, and we increase all levels by $1$ accordingly. We omit the details to avoid redundancy.
\end{proof}

The proof of \Cref{lem:ssflip} is similar to the derivation of the small-set flip decoder in \cite[Section~4]{golowich_constant-overhead_2025}, which in turn was a high-dimensional generalization of decoders in \cite{leverrier_quantum_2015,fawzi_efficient_2018,fawzi_constant_2020}. The main difference between our proof of \Cref{lem:ssflip} and the decoder analysis in \cite{golowich_constant-overhead_2025} is that our 1-dimensional complexes $\cC^*(\epsilon,\bar{\ell})$ are constructed by combining lossless expanders of various sizes, including some constant-sized ones $G^{(\ell)}$ for $\ell=O(1)$. In contrast, \cite{golowich_constant-overhead_2025} used 1-dimensional complexes constructed only from large lossless expanders. Hence for errors of weight polynomial in the block length, our small graphs do not expand losslessly. Instead, in \Cref{lem:ssflip} we assume the \emph{weighted} Hamming norm $|e|_\beta$ is bounded, which implies $e$ has no support on these small graphs $G^{(\ell)}$.

Hence we recover lossless expansion, though there is an additional complication: we also need to ensure that our decoder's perturbation reduces the error's weighted norm, in addition to its unweighted norm. Without such a property, we would not be able to repeatedly apply perturbations. \Cref{lem:ssfimp} ensures the desired weighted norm bound; the use of $\delta^{\cC'}$ instead of $\delta^{\cC}$ in \Cref{lem:ssflip} was needed so that we could prove \Cref{lem:ssfimp}. Because the coboundary matrices $\delta^{\cC'}$ and $\delta^{\cC}$ differ only in a small number of entries, we are able to prove \Cref{lem:ssflip} using similar techniques as \cite{golowich_constant-overhead_2025}, though with some additional bookkeeping.

We will first need the following lemma, which follows from \cite{dinur_expansion_2024,kalachev_maximally_2025}, as described in \cite{golowich_constant-overhead_2025}. This lemma is often called the ``(collective) robustness'' or ``(collective) coboundary expansion'' of repetition codes, as it places a lower bound on the weight of coboundaries in (direct sums of) the cochain complex given by a product of 1-dimensional repetition code complexes.

\begin{lemma}[Follows from \cite{dinur_expansion_2024,kalachev_maximally_2025}]
  \label{lem:robust}
  For every $r\in\bN$ and $\alpha>0$, there exists $\kappa(r,\alpha)>0$ such that the following holds for every $n\in\bN$. For $i\in[r]$, let
  \begin{equation*}
    \cA^{(i)} = \left(\bF_2\xrightarrow{\delta^{(i)}}\bF_2^{n_i}\right)
  \end{equation*}
  be the 1-dimensional cochain complex with $\delta^{(i)}(1)=(1,\dots,1)$ equal to the length-$n_i$ all-1s vector for some $n_i\in[\alpha n,n]$, so that $B^1(\cA^{(i)})$ is a length-$n_i$ repetition code. Let
  \begin{equation*}
    \cA^* = \cA^{(1)}\otimes\cdots\otimes\cA^{(r)}.
  \end{equation*}
  Then for every $\Upsilon\in\bN$, every $0\leq i\leq r-1$ and every $a\in(\cA^{\oplus\Upsilon})^i$, there exists $a'\in B^i(\cA^{\oplus\Upsilon})$ such that
  \begin{equation*}
    \kappa(r,\alpha)\cdot n\cdot |a+a'| \leq |\delta(a)|.
  \end{equation*}
\end{lemma}

In particular, the proof of \cite[Lemma~4.8]{golowich_constant-overhead_2025} shows how the $\Upsilon=1$ case of \Cref{lem:robust} follows from results in \cite{dinur_expansion_2024,kalachev_maximally_2025}. Meanwhile, \cite[Proposition~5.14]{dinur_expansion_2024} shows that if \Cref{lem:robust} holds for $\Upsilon=1$, then it holds for every $\Upsilon\in\bN$, albeit with a smaller constant $\kappa(r,\alpha)>0$.

We are now ready to prove \Cref{lem:ssflip}.

\begin{proof}[Proof of \Cref{lem:ssflip}]%; similar to proof of Lemma~4.10 in \cite{golowich_constant-overhead_2025}]
  Define $\kappa=\kappa(r,1/2)>0$ to be the value given by letting $\alpha=1/2$ in \Cref{lem:robust}. Set
  \begin{equation}
    \label{eq:ssfeps}
    \epsilon = \min\left\{\frac{1}{32},\; \frac{\kappa}{64r}\right\}.
  \end{equation}
  
  Our goal is to show that either \Cref{eq:ssfup} holds, or else \Cref{eq:ssfdown} holds.
  For this purpose, let $\cB^*=\cC^*(\epsilon,\bar{\ell})$ and ${\cB'}^*={\cC'}^*(\epsilon,\bar{\ell})$, so that $\cC^*=((\cB^*)^{\otimes r})^{\oplus\Upsilon}$ and ${\cC'}^*=(({\cB'}^*)^{\otimes r})^{\oplus\Upsilon}$. For a set $S\subseteq B^0$, let $N^B(S)\subseteq B^1$ and $N^{B'}(S)\subseteq{B'}^1$ denote the neighborhood of $S$ in the bipartite graph with bipartite adjacency matrix given by $\delta^{\cB}$ and $\delta^{\cB'}$, respectively. For instance, $N^B(S)$ contains every $b^1\in B^1$ such that there exists $b^0\in S$ with $b^0\triangleleft b^1$.
  
  For each $i\in[r]$, we first fix an appropriate element $v_i\in B^0$ and subset $U_i\subseteq B^1$ satisfying $v_i\triangleleft' U_i$ (i.e.~$U_i\subseteq N^{B'}(v_i)$) and $|U_i|\geq(1-5\epsilon)\Delta$. We will then use the fact that the restriction of ${\cC'}^*$ to basis elements in $\prod_{i\in[r]}(\{v_i\}\cup U_i)$ has the structure of the product of repetition codes in \Cref{lem:robust}.

  For $i\in[r]$, we will inductively define $v_i$ and $U_i$ assuming we have defined $v_1,\dots,v_{i-1}$ and $U_1,\dots,U_{i-1}$. Recalling that $C=B^r$, for $j\in[r]$ we let $\Pi^{(j)}:C\rightarrow B$ denote projection onto the $j$th coordinate. Define $D_i\subseteq C$ by
  \begin{equation*}
    D_i = \prod_{j\in[i-1]}(\{v_j\}\times U_j) \times B^{r-i+1}
    % B_i = \{c\in C:\Pi^{(j)}(c)\in(\{v_j\}\cup U_j)\forall j\in[i-1]\},
  \end{equation*}
  and define $F_i\subseteq B$ by
  \begin{equation*}
    F_i = \Pi^{(i)}(\supp(e)\cap D_i).
  \end{equation*}

  We will enforce the inductive hypothesis that $\supp(e)\cap D_i$ (and hence $F_i$) is nonempty, and that for every $c=(c_1,\dots,c_r)\in D_i$, $j\in[i-1]$, and $v\in B^0$ with $v\triangleleft c_j$ and $(c_1,\dots,c_{j-1},v,c_{j+1},\dots,c_r)\in\supp(e)$, then $v=v_j$.

  For the base case, when $i=1$ then $D_1=C$, so $\supp(e)\cap D_1\neq\emptyset$ by the assumption that $e\neq 0$. Now assume the inductive hypothesis holds for some $i\geq 1$. We define $v_i,U_i$ separately in the cases where $F_i\cap B^0$ is empty vs nonempty:
  \begin{enumerate}
  \item First assume $F_i\cap B^0=\emptyset$. Then we must have $F_i\cap B^1\neq\emptyset$, so we set $v_i$ to be any element of $B^0$ such that $v_i\triangleleft' u$ for some $u\in F_i$, and then we let $U_i=\{u\in B^1:v_i\triangleleft u\}=N^{B'}(u)$, so that $F_i\cap U_i\neq\emptyset$ and $|U_i|=\Delta$.
  \item Now assume $F_i\cap B^0\neq\emptyset$. By definition $|F_i|\leq|e|\leq\mu 2^{\bar{\ell}/2}\leq\mu|V_{\bar{\ell}/2}|$. Furthermore, because $|e|_\beta\leq\beta^{\bar{\ell}/2}$, every $f\in F_i$ must have $\Lev(f)\leq\bar{\ell}/2$ (see \Cref{def:weightnorm}). Therefore $F_i\cap B^0$ is entirely supported inside subsets $V_\ell\subseteq B^0$ for $\ell\geq\bar{\ell}/2$. Each graph $G^{(\ell)}=(V_\ell\sqcup V_{\ell-1},E^{(\ell)})$ by definition has $(\mu,\epsilon)$-lossless expansion, so for $\ell\geq\bar{\ell}/2$ we have $|F_i|\leq\mu|V_{\ell}|$, and therefore the neighborhood of $F_i\cap V_\ell$ in $G^{(\ell)}$ has size $\geq(1-\epsilon)|F_i\cap V_\ell|$. Summing over all $\ell$, we have $|N^{B'}(F_i\cap B^0)|\geq(1-\epsilon)\Delta|F_i\cap B^0|$. Because the bipartite graph associated to $B$ is given by the graph associated to $B'$ with at most one additional edge incident to each vertex, we also have $|N^B(F_i\cap B^0)|\geq(1-\epsilon)\Delta|F_i\cap B^0|$. Then because each $v\in F_i\cap B^0$ has $|N^B(v)|\leq\Delta+1$, we have a total of $(\Delta+1)|F_i\cap B^0|$ outgoing edges from $F_i\cap B^0$, leading to at least $(1-\epsilon)\Delta|F_i\cap B^0|$ distinct right vertices. Thus at most $2((\Delta+1)-(1-\epsilon)\Delta)|F_i\cap B^0|=2(\epsilon\Delta+1)|F_i\cap B^0|$ of the outgoing edges can lead to right vertices shared by another outgoing edge. That is, at most $2(\epsilon\Delta+1)/(\Delta+1)$-fraction of the outgoing edges fom $F_i\cap B^0$ lead to right vertices shared by another outgoing edge. Thus there must exist some $v_i\in F_i\cap B^0$ such that the set
    \begin{equation*}
      \tilde{U}_i := N^B(v_i)\setminus N^B((F_i\cap B^0)\setminus\{v_i\})
    \end{equation*}
    of neighbors of $v_i$ that are not neighbors of any other vertices in $F_i\cap B^0$ satisfies
    \begin{equation*}
      |\tilde{U}_i| \geq \left(1-\frac{2(\epsilon\Delta+1)}{\Delta+1}\right)\Delta \geq \left(1-2\epsilon-\frac{2}{\Delta+1}\right)\Delta \geq (1-4\epsilon)\Delta,
    \end{equation*}
    where the final inequality above uses the fact that $\Delta\geq 1/\epsilon$ (see \Cref{lem:lossless}). Then because $N^B(v_i)$ consists of all elements of $N^{B'}(v_i)$ along with at most one additional element, the set
    \begin{equation*}
      U_i := N^{B'}(v_i)\setminus N^B((F_i\cap B^0)\setminus\{v_i\})
    \end{equation*}
    satisfies
    \begin{equation*}
      |U_i| \geq |\tilde{U}_i|-1 \geq (1-5\epsilon)\Delta,
    \end{equation*}
    where again we use the fact that $\Delta\geq 1/\epsilon$.
  \end{enumerate}

  In both cases above, we chose $v_i,U_i$ so that $(\{v_i\}\cup U_i)\cap F_i\neq\emptyset$. Therefore by the definition of $F_i$, there exists some $c\in\supp(e)\cap D_i$ with $\Pi^{(i)}(c)\in(\{v_i\}\cup U_i)\cap F_i$, so $D_{i+1}=\{c\in D_i:\Pi^{(i)}(c)\in(\{v_i\}\cup U_i)$ contains $c$ and thus is nonempty. To complete the proof of the inductive hypothesis, we must show that for every $c=(c_1,\dots,c_r)\in D_{i+1}$, every $v\in B^0$ satisfying $v\triangleleft c_i$ and $(c_1,\dots,c_{i-1},v,c_{i+1},\dots,c_r)\in\supp(e)$ must equal $v=v_i$. But if $v\neq v_i$, then because $v\in F_i\cap B^0$ by definition, we have that $v\in(F_i\cap B^0)\setminus\{v_i\}$, so $c_i\in N^B(v)\subseteq N^B((F_i\cap B^0)\setminus\{v_i\})$, and therefore $c_i\notin U_i$, contradicting the assumption that $c\in D_{i+1}$. Therefore the inductive hypothesis must hold, as desired.

    Having completed our inductive definition of $v_1,\dots,v_r$ and $U_1,\dots,U_r$, we let $D=D_r\subseteq C$, and we define an $r$-dimensional cochain complex $\cD^*$ with $j$-cochain basis $D^j=D\cap C^j$, and with coboundary map $\delta^{\cD}(c)=\delta^{\cC}(c)|_D=\delta^{\cC'}(c)|_D$ for $c\in\cD\subseteq\cC$. By definition $\cD^*$ is the $\Upsilon$-direct-sum of the tensor product of $r$ 1-dimensional complexes obtained from the restriction of $\cB^*$ to $0$-cochains supported in $\{v_i\}$ and $1$-cochains supported in $U_i$ for $i\in[r]$. Because $U_i\subseteq N^{B'}(v_i)$, each of these 1-dimensional complexes is of the form described in \Cref{lem:robust}, so we can apply \Cref{lem:robust} to the $\Upsilon$-direct-sum of the product complex $\cD^*$.

    Also note that the final $(i=r)$ iteration of the inductive hypothesis above implies the following:

    \begin{claim}
      \label{claim:ssDprops}
      \begin{enumerate}
      \item $\supp(e)\cap D\neq\emptyset$, that is, $e|_D\neq 0$, and
      \item For every $c\in\supp(e)$ and $c'\in D$ satisfying $c\triangleleft c'$, then $c\in D$, that is, for some $i\in[r]$ we have
        \begin{equation*}
          c = (c_1',\dots,c_{i-1}',v_i,c_{i+1}',\dots,c_r').
        \end{equation*}
      \end{enumerate}
    \end{claim}

    Now define $\bar{c}^0=(v_1,\dots,v_r)\in C^0$. Our goal is to show that $\bar{c}^0$ satisfies either \Cref{eq:ssfup} or \Cref{eq:ssfdown}. We consider two cases separately:

    \begin{enumerate}
    \item $|\delta^{\cD}(e|_D)|\geq 16r\epsilon\Delta|e|_D|$: In this case, we show that \Cref{eq:ssfdown} holds. Intuitively, because $\delta^{\cD}(e|_D)$ is large enough, we will show that flipping all bits of $e$ in $\supp(e|_D)$ reduces the weight of the image under $\delta^{\cC}$, as we will zero out all components in $D$, while changing at most a small number of components outside of $D$.

      Specifically, we define $c^i\in\cC^i$ by $c^i=e|_D\in(\bF_2^\Upsilon)^{D^i}\subseteq(\bF_2^\Upsilon)^{C^i}$.
      % Note that if $i=0$, we have defined $c^0$ twice, but both definitions agree, as they set $c^0$ to be the indicator of the unique element $(v_1,\dots,v_r)\in D^0$.
      Now we have
      \begin{equation}
        \label{eq:ssfin}
        \delta^{\cC}(e)|_D = \delta^{\cC}(c^i)|_D = \delta^{\cD}(c^i) = \delta^{\cC'}(c^i)|_D.
      \end{equation}
      Specifically, the second and third equalities above hold by definition, while the first equality holds because if it were violated at some component (i.e.~basis element) $c'\in D^{i+1}$, then there must be some $c\triangleleft c'$ such that $e_c\neq(c^i)_c$, and hence $c\in\supp(e)\setminus D$ so that $e_c\neq 0$ and $(c^i)_c=0$. But the existence of such $c,c'$ violates \Cref{claim:ssDprops}. Furthermore, we have
      \begin{equation}
        \label{eq:ssfout}
        |\delta^{\cC'}(c^i)|_{C\setminus D}| \leq 5\epsilon r\Delta|c^i|,
      \end{equation}
      as for every $c=(c_1,\dots,c_r)\in\supp(c^i)\subseteq D$, then every $c'\triangleright' c$ must be of the form $c'=(c_1,\dots,c_{j-1},c_j',c_{j+1},\dots,c_r)$ for some $j\in[r]$ and some $c_j'\triangleright c_j$. But there $\leq r$ such choices of $j\in[r]$, and we will only have $c'\notin D$ if $c_j'$ is one of the $\leq 5\epsilon\Delta$ elements of $N^{B'}(c_j)\setminus U_j$. Hence there are $\leq 5\epsilon r\Delta$ choices of $c'\triangleright' c$ that lie outside of $D$, so \Cref{eq:ssfout} follows by summing over all $c\in\supp(c^i)$.
      Now it follows that
      \begin{align*}
        |\delta^{\cC}(e)+\delta^{\cC'}(c^i)|
        &= |(\delta^{\cC}(e)+\delta^{\cC'}(c^i))|_D| + |(\delta^{\cC}(e)+\delta^{\cC'}(c^i))|_{C\setminus D}| \\
        &\leq 0 + |\delta^{\cC}(e)|_{C\setminus D}|+|\delta^{\cC'}(c^i)|_{C\setminus D}| \\
        &\leq |\delta^{\cC}(e)|-|\delta^{\cC}(e)|_D|+5\epsilon r\Delta|c^i| \\
        &= |\delta^{\cC}(e)|-|\delta^{\cD}(c^i)|+5\epsilon r\Delta|c^i| \\
        &< |\delta^{\cC}(e)|-10\epsilon r\Delta|c^i|,
      \end{align*}
      where the first inequality above holds by \Cref{eq:ssfin}, the second inequality holds by \Cref{eq:ssfout}, the subsequent equality holds by \Cref{eq:ssfin}, and the final inequality holds by the definition of $c^i=e|_D$ along with the initial assumption that $|\delta^{\cD}(e|_D)|\geq 16r\epsilon\Delta|e|_D|$. Thus \Cref{eq:ssfdown} follows because $\Delta\geq 1/\epsilon$.
    \item $|\delta^{\cD}(e|_D)|<16r\epsilon\Delta|e|_D|$: In this case, we show that \Cref{eq:ssfup} holds. Intuitively, because $|\delta^{\cD}(e|_D)|$ is small enough, we apply \Cref{lem:robust} to $\cD^*$ to show that $e|_D$ is close to a coboundary $\delta^{\cD}(c^{i-1})$. We then argue that we can reduce the weight of $e$ by adding in the coboundary $\delta^{\cC'}(c^{i-1})$.

      First, if $i=0$, then as $D^0=\{(v_1,\dots,v_r)\}$, we have $e|_D\in\bF_2^\Upsilon$ equal to some nonzero vector, with $\supp(e|_D)=(v_1,\dots,v_r)$ the unique basis element in $D^0$. Hence by the definition of $\cD^*$,
      \begin{equation*}
        |\delta^{\cD}(e|_D)| = \sum_{i\in[r]}|U_i| \geq r(1-5\epsilon)\Delta.
      \end{equation*}
      Because $1-5\epsilon>16\epsilon$ as $\epsilon\leq 1/32$ by \Cref{eq:ssfeps}, and $|e|_D|=1$, the RHS above is greater than $16r\epsilon\Delta|e|_D|$, contradicting the assumption that $|\delta^{\cD}(e|_D)|<16r\epsilon\Delta|e|_D|$. Thus we may assume that $i\geq 1$.

      Now because each $|U_i|\geq(1-5\epsilon)\Delta>\Delta/2$ as $\epsilon\leq 1/32$, then recalling the definition of $\kappa=\kappa(r,1/2)>0$ from \Cref{lem:robust}, it follows from \Cref{lem:robust} that there exists some $e'\in B^i(\cD)$ such that
      \begin{equation*}
        \kappa\Delta\cdot|e|_D+e'| \leq |\delta^{\cD}(e|_D)| < 16r\epsilon\Delta|e|_D|,
      \end{equation*}
      that is,
      \begin{equation*}
        |e|_D+e'| \leq \frac{16r\epsilon}{\kappa}\cdot|e|_D|.
      \end{equation*}
      Now applying \Cref{lem:robust} to any element of $\cD^{i-1}$ whose coboundary equals $e'$, we obtain some $c^{i-1}\in\cD^{i-1}\subseteq\cC^{i-1}$ with $\delta^{\cD}(c^{i-1})=e'$ and
      \begin{equation}
        \label{eq:ssfcim1}
        |c^{i-1}| \leq \frac{1}{\kappa\Delta}|e'| \leq \frac{1}{\kappa\Delta}\left(\frac{16r\epsilon}{\kappa}+1\right)|e|_D|.
      \end{equation}
      By definition $|e'|\geq(1-16r\epsilon/\kappa)|e|_D|>0$ because $e|_D\neq 0$ and $\epsilon<\kappa/16r$ by \Cref{eq:ssfeps}, so $c^{i-1}\neq 0$.
      % As a point of notation, if $i=1$ then it appears we have overloaded the variable $c^0=c^{i-1}$ but in this case both $c^0$ and $c^{i-1}$ must equal the unique nonzero element $\1_{(v_1,\dots,v_r)}$ of $\cD^0$.

      Recall that our goal is to show \Cref{eq:ssfup} holds. For this purpose,  we have that
      \begin{equation*}
        |(e+\delta^{\cC'}(c^{i-1}))|_D| = |e|_D+e'| \leq \frac{16r\epsilon}{\kappa}\cdot|e|_D|
      \end{equation*}
      and
      \begin{equation*}
        |(e+\delta^{\cC'}(c^{i-1}))|_{C\setminus D}-e|_{C\setminus D}| = |\delta^{\cC'}(c^{i-1})|_{C\setminus D}| \leq 5\epsilon r\Delta|c^{i-1}|
        %\leq \frac{5\epsilon r}{\kappa}\left(\frac{16r\epsilon}{\kappa}+1\right)|e|_D|
      \end{equation*}
      where the inequality above holds by the definition of $D$, by similar reasoning used to show \Cref{eq:ssfout} above.
      % the second inequality above holds by \Cref{eq:ssfcim1}.
      Combining the above inequalities then gives the desired inequality
      \begin{align*}
        |e+\delta^{\cC'}(c^{i-1})|
        &= |(e+\delta^{\cC'}(c^{i-1}))|_D| + |(e+\delta^{\cC'}(c^{i-1}))|_{C\setminus D}| \\
        &\leq \frac{16r\epsilon}{\kappa}\cdot|e|_D| + |e|_{C\setminus D}| + 5\epsilon r\Delta|c^{i-1}| \\
        &= \frac{16r\epsilon}{\kappa}\cdot|e|_D| + |e| - |e_D| + 5\epsilon r\Delta|c^{i-1}| \\
        &\leq |e| - \left(1-\frac{16r\epsilon}{\kappa}\right)|e|_D| + 5\epsilon r\Delta|c^{i-1}| \\
        &\leq |e| - \frac12|e|_D| + 5\epsilon r\Delta|c^{i-1}| \\
        &< |e| - \frac{\kappa\Delta}{4}|c^{i-1}| + 5\epsilon r\Delta|c^{i-1}| \\
        &\leq |e| - 4r|c^{i-1}|
      \end{align*}
      where the third inequality above holds because $\epsilon\leq\kappa/32r$ by \Cref{eq:ssfeps}, the fourth inequality holds by \Cref{eq:ssfcim1} because $\epsilon\leq\kappa/32r$, and the fifth inequality holds because $\kappa\Delta/4-5\epsilon r\Delta\geq 4r$ as $\Delta\geq 1/\epsilon$ and $\epsilon\leq\kappa/40r$ by \Cref{eq:ssfeps}. \qedhere
    \end{enumerate}
\end{proof}

\subsection{Small-Set Flip Subroutines}
In this section, we present subroutines that repeatedly apply constant-sized perturbations (i.e.~``flips'') to an error to reduce the error weight or its syndrome weight. These algorithms are specifically presented in \Cref{alg:ssflip}. We will subsequently apply \Cref{lem:ssflip} and \Cref{lem:ssfimp} to show that these subroutines indeed successfully decode a large class of errors. These subroutines are adapted from those presented in \cite[Section~4]{golowich_constant-overhead_2025}, which in turn are high-dimensional generalizations of those used in \cite{leverrier_quantum_2015,fawzi_efficient_2018,fawzi_constant_2020}. However, in our setting we require the flip to satisfy either \Cref{eq:ssfup} or \Cref{eq:ssfdown}, so that we will be able to apply \Cref{lem:ssfimp} to argue that both the unweighted and $\beta$-weighted Hamming norms decrease with each flip.

\begin{algorithm}
  \caption{\label{alg:ssflip} (Classical) decoding algorithms that apply a small-set flip decoder in the sense of \Cref{lem:ssflip}. These algorithms are adapted from related ones in \cite{golowich_constant-overhead_2025}. For $\epsilon>0$, $\bar{\ell},r\in\bN$, and $\Upsilon\in\bN$, we let $\cC^*=\cC^*(\epsilon,\bar{\ell},0,r)^{\oplus\Upsilon}$ and ${\cC'}^*={\cC'}^*(\epsilon,\bar{\ell},0,r)^{\oplus\Upsilon}$ with partial orders $\prec$ and $\prec'$ respectively. The decoding algorithms here implicitly depend on the choice of $\epsilon,\bar{\ell},r,\Upsilon$, as well as a choice of $0\leq i\leq r-1$, which will always be made clear from context.}
  \SetKwInOut{Input}{Input}
  \SetKwInOut{Output}{Output}
  \SetKwProg{Fn}{Function}{:}{}

  \SetKwFunction{FnSSFlipSyn}{SSFlipSyn}

  \Input{$(i+1)$-coboundary $s\in B^{i+1}(\cC)$ of the cochain complex $\cC^*$}
  \Output{$a^i\in\cC^i$ with $s=\delta^{\cC}(a^i)$}
  \Fn{\FnSSFlipSyn{$s$}}{
  % \Fn{\FnSSFlipSyn{$s;i,\cC^*$}}{
    Initialize $a^i\gets 0\in\cC^i$ \\
    \While{$\exists \bar{c}^0\in C^0,c^i\in\cC^i$ with $\bar{c}^0\preceq' c^i$ and $|s+\delta^{\cC}(a^i)+\delta^{\cC'}(c^i)|<|s+\delta^{\cC}(a^i)|-4r|c^i|$}{
      $a^i\gets a^i+c^i$
    }
    \If{$s\neq\delta^{\cC}(a^i)$}{
      \Return{FAIL}
    }\Else{
      \Return{$a^i$}
    }
  }
  
  \SetKwFunction{FnSSFlipErr}{SSFlipErr}

  \Input{$i$-cochain $e\in\cC^i$ of the cochain complex $\cC^*$}
  \Output{$a^{i-1}\in\cC^{i-1},a^i\in\cC^i$ with $e=a^i+\delta^{\cC}(a^{i-1})$}

  \Fn{\FnSSFlipErr{$e$}}{
  % \Fn{\FnSSFlipErr{$e;i,\cC^*$}}{
    Initialize $a^{i-1}\gets 0\in\cC^{i-1}$, $a^i\gets 0\in\cC^i$ \\
    \While{$e\neq a^i+\delta^{\cC}(a^{i-1})$}{ \label{li:fewhile}
      \If{$\exists \bar{c}^0\in C^0,c^{i-1}\in\cC^{i-1}$ with $\bar{c}^0\preceq' c^{i-1}$ and $|e+a^i+\delta^{\cC}(a^{i-1})+\delta^{\cC'}(c^{i-1})|<|e+a^i+\delta^{\cC}(a^{i-1})|-4r|c^{i-1}|$}{ \label{li:im1flip}
        $a^{i-1} \gets a^{i-1}+c^{i-1}$
      }\ElseIf{$\exists \bar{c}^0\in C^0,c^i\in\cC^i$ with $\bar{c}^0\preceq' c^i$ and $|\delta^{\cC}(e+a^i)+\delta^{\cC'}(c^i)|<|\delta^{\cC}(e+a^i)|-4r|c^i|$}{ \label{li:iflip}
        $a^i \gets a^i+c^i$
      }\Else{
        \Return{FAIL}
      }
    }
    \Return{$a^{i-1},a^i$}
  }
\end{algorithm}

In \Cref{alg:ssflip}, and particularly in \FnSSFlipSyn, we assume we are given a noiseless syndrome $s=\delta^{\cC}(e)$ of some error $e$. While \cite{golowich_constant-overhead_2025} applied a similar algorithm to noisy syndromes, in this paper for simplicity we will always consider sufficiently high-dimensional complexes $\cC^*$ that we will be able to use ``Knill-type'' error-correction, which never requires decoding from noisy syndromes.

We will now formalize the fact that \FnSSFlipSyn and \FnSSFlipErr are \emph{local algorithms}, meaning that they act independently on disjoint connected components of the connectivity graph $G^{\cC}$ defined in \Cref{def:conngraph}. This locality property of the algorithms was previously used in \cite{fawzi_efficient_2018}, and later in \cite{golowich_constant-overhead_2025}; we generally follow the notation used in \cite{golowich_constant-overhead_2025}. Specifically, we first define the \emph{footprint} of these algorithms to be the set of cochain basis elements that appear in the support of some vector during their execution:

\begin{definition}
  \label{def:footprint}
  Define $\cC^*=\cC^*(\epsilon,\bar{\ell},0,r)^{\oplus\Upsilon}$ and ${\cC'}^*={\cC'}^*(\epsilon,\bar{\ell},0,r)^{\oplus\Upsilon}$ as in \Cref{alg:ssflip}.

  For $e\in\cC^i$, we say the \emph{footprint} at time $t$ of \FnSSFlipSyn{$\delta^{\cC}(e)$} is the subset of $C^i\sqcup C^{i+1}$ given by the union of the values of the sets
  % $\supp(e),\supp(c^i)\subseteq C^i$ and $\supp(\delta^{\cC}(e)),\supp(\delta^{\cC}(c^i))\subseteq C^{i+1}$
  $\supp(e+a^i)\subseteq C^i$ and $\supp(\delta^{\cC}(e+a^i))\subseteq C^{i+1}$
  across the first $t$ iterations of the while loop.

  Similarly, for $e\in\cC^i$, we say the \emph{footprint} at time $t$ of \FnSSFlipErr{$e$} is the subset of $C^{i-1}\sqcup C^i\sqcup C^{i+1}$ given by the union of the values of the sets
  $\supp(a^{i-1})\subseteq C^{i-1}$, $\supp(e+a^i+\delta^{\cC}(a^{i-1}))\subseteq C^i$, and $\supp(\delta^{\cC}(e+a^i))\subseteq C^{i+1}$
  across the first $t$ iterations of the while loop.

  In both cases, by the \emph{footprint} we mean the footprint at time $t=\infty$ after the respective while loop has terminated.
\end{definition}

\begin{remark}
  Every $c^{i-1}$ or $c^i$ that is chosen in some while loop in \Cref{alg:ssflip} equals the sum of two vectors whose supports are included in the footprint in \Cref{def:footprint}; for instance, each $c^i$ in \FnSSFlipSyn is the difference of the value of $e+a^i$ at two successive timesteps. Hence $\supp(c^{i-1})$ and $\supp(c^i)$ lie in the footprint. By similar reasoning $\supp(\delta^{\cC}(c^{i-1}))$ and $\supp(\delta^{\cC}(c^i))$ lie in the footprint, as for instance $\delta^{\cC}(c^i)$ in \FnSSFlipSyn is the difference of the value of $\delta^{\cC}(e+a^i)$ at two successive timesteps. Furthermore, $\supp(\delta^{\cC'}(c^{i-1}))$ and $\supp(\delta^{\cC'}(c^i))$ also lie in the footprint because by construction $\supp(\delta^{\cC'}(c^{i-1}))\subseteq\supp(\delta^{\cC}(c^{i-1}))$ and $\supp(\delta^{\cC'}(c^i))\subseteq\supp(\delta^{\cC}(c^i))$.
\end{remark}

\Cref{lem:sslocal} below shows that our small-set flip algorithms act independently on disjoint connected components of subgraph of the connectivity graph $G^{\cC}$ (see \Cref{def:conngraph}) induced by the footprint. Below, we slightly extend the standard notation of restriction of vectors. Specifically, for $x\in(\bF_2^\Upsilon)^{C^i}$ and $V\subseteq C$, we let $x|_V\in(\bF_2^\Upsilon)^{C^i}$ denote the vector obtained from $x$ by replacing the values of every component in $C^i\setminus V$ with $0$. In other words, we write $x|_V$ as a shorthand for $(x|_{V\cap C^i},0^{C^i\setminus V})$.

\begin{lemma}[Similar to Lemma~4.4 of \cite{golowich_constant-overhead_2025}]
  \label{lem:sslocal}
  Define $\cC^*=\cC^*(\epsilon,\bar{\ell},0,r)^{\oplus\Upsilon}$ and ${\cC'}^*={\cC'}^*(\epsilon,\bar{\ell},0,r)^{\oplus\Upsilon}$ as in \Cref{alg:ssflip}, and let $\beta\geq 1$, $0\leq i\leq r-1$ and $e\in\cC^i$. Then the following hold:
  \begin{enumerate}
  \item Let $S_{\mathrm{syn}}\subseteq C^i\sqcup C^{i+1}$ and $a^i\in\cC^i$ denote the footprint and output, respectively, of \FnSSFlipSyn{$\delta^{\cC}(e)$}. For every connected component $V$ of the subgraph of $G^{\cC}$ induced by $S_{\mathrm{syn}}$, the output of \FnSSFlipSyn{$\delta(e|_V)$} must equal $a^i|_V$.
    Furthermore, we have
    \begin{align}
      \label{eq:sslunif}
      \begin{split}
        |V| &\leq 4r\Delta^2\cdot|e|_V| \\
        |V|_\beta &\leq 4r\beta\Delta^2\cdot|e|_V|_\beta.
      \end{split}
    \end{align}
  \item Let $S_{\mathrm{err}}\subseteq C^{i-1}\sqcup C^i\sqcup C^{i+1}$ and $a^{i-1}\in\cC^{i-1}$, $a^i\in\cC^i$ denote the footprint and output, respectively, of \FnSSFlipErr{$e$}. For every connected component $V$ of the subgraph of $G^{\cC}$ induced by $S_{\mathrm{err}}$, the output of \FnSSFlipErr{$e|_V$} must equal $a^{i-1}|_{V\cap C^{i-1}},\;a^i|_{V\cap C^i}$.
    Furthermore, we have
    \begin{align}
      \label{eq:sslweigh}
      \begin{split}
        |V| &\leq 8r\Delta^2\cdot|e|_V| \\
        |V|_\beta &\leq 8r\beta\Delta^2\cdot|e|_V|_\beta.
      \end{split}
    \end{align}
  \end{enumerate}
\end{lemma}
\begin{proof}
  The proof of the first claim in each item, i.e.~that the decoders output $a^{i-1}|_V,\;a^i|_V$ on input $e|_V$, is identical to the proof of \cite[Lemma~4.4]{golowich_constant-overhead_2025}. That is, every update $c^i$ or $c^{i-1}$ in any iteration of the while loop in either \FnSSFlipSyn{$\delta^{\cC}(e)$} or \FnSSFlipErr{$e$} must be contained in the neighborhood in $G^{\cC}$ of the footprint from the previous step. Thus different connected components of the subgraph induced by the footprint do not interact at all during the algorithm's execution. It also follows that each $\supp(c^i)$ and each $\supp(c^{i-1})$ is either contained in $V$, or else disjoint from~$V$.

  We now turn to proving \Cref{eq:sslunif,eq:sslweigh}. Note that it suffices to prove the second inequality (regarding $\beta$-weighted norms) in each of \Cref{eq:sslunif,eq:sslweigh}, as then the first inequality (regarding unweighted norms) follows by setting $\beta=1$. We consider \FnSSFlipSyn and \FnSSFlipErr separately. The proof is similar to that of \cite[Lemma~4.5]{golowich_constant-overhead_2025}, though in our setting we must consider weighted as well as unweighted norms.
  \begin{enumerate}
  \item Fix a connected component $V$ of the subgraph of $G^{\cC}$ induced by $S_{\mathrm{syn}}$. Let $F^i$ be the (multi)set of values of $c^i$ in the while loop in \FnSSFlipSyn{$\delta^{\cC}(e)$} for which $\supp(c^i)\subseteq V$. The value $|\delta^{\cC}(e+a^i)|_V|_\beta$ initially equals $|\delta^{\cC}(e)|_V|_\beta$ (as we initialize $a^i=0$), and then by \Cref{eq:ssfdownimp} in \Cref{lem:ssfimp} must decrease by at least $3r\beta|c^i|_\beta$ in the iteration of the while loop at which each $c^i\in F^i$ is chosen. Therefore because $|\delta^{\cC}(e+a^i)|_V|_\beta$ cannot decrease to a value below $0$, we must have
    \begin{equation}
      \label{eq:sslsumci}
      \sum_{c^i\in F^i}3r\beta|c^i|_\beta \leq |\delta^{\cC}(e)|_V|_\beta.
    \end{equation}
    Now by definition every basis element $b\in C$ has at most $r(\Delta)+1$ distinct $b'\triangleright b$, each of which lies at level $\Lev(b')\leq\Lev(b)+1$. Therefore every $c\in\cC$ has
    \begin{equation}
      \label{eq:cobnormbeta}
      |\delta^{\cC}(c)|_\beta\leq r(\Delta+1)\beta|c|.
    \end{equation}
    Then as $V$ is the footprint of \FnSSFlipSyn{$\delta^{\cC}(e|_V)$}, by definition
    \begin{equation*}
      V = \supp(e|_V) \cup \supp(\delta^{\cC}(e)|_V) \cup \bigcup_{c^i\in F^i}(\supp(c^i) \cup \supp(\delta^{\cC}(c^i))),
    \end{equation*}
    so
    \begin{align}
      \label{eq:Vbetasyn}
      \begin{split}
        |V|_\beta
        &\leq |e|_V|_\beta + |\delta^{\cC}(e)|_V|_\beta + \sum_{c^i\in F^i}(|c^i|_\beta + |\delta^{\cC}(c^i)|_\beta) \\
        &\leq (1+r(\Delta+1)\beta)\left(|e|_V|_\beta+\sum_{c^i\in F^i}|c^i|_\beta\right) \\
        &\leq (1+r(\Delta+1)\beta)\left(|e|_V|_\beta+\frac{|\delta^{\cC}(e)|_V|_\beta}{3r\beta}\right) \\
        &\leq \frac{(1+r(\Delta+1)\beta)^2}{3r\beta}|e|_V|_\beta \\
        &\leq 4r\beta\Delta^2\cdot|e|_V|_\beta,
      \end{split}
    \end{align}
    as desired.
  \item Fix a connected component $V$ of the subgraph of $G^{\cC}$ induced by $S_{\mathrm{err}}$. Let $F^{i-1}$ (resp.~$F^i$) be the (multi)set of values of $c^{i-1}$ (resp.~$c^i$) in the while loop in \FnSSFlipErr{$e$} for which $\supp(c^{i-1})\subseteq V$ (resp.~$\supp(c^i)\subseteq V$). As in the analysis of \FnSSFlipSyn{$\delta^{\cC}(e)$} above, the value $|\delta^{\cC}(e+a^i)|_V|_\beta$ initially equals $|\delta^{\cC}(e)|_V|_\beta$ (as we initialize $a^i=0$), and then by \Cref{eq:ssfdownimp} in \Cref{lem:ssfimp} must decrease by at least $3r\beta|c^i|_\beta$ in the iteration of the while loop at which each $c^i\in F$ is chosen. Thus \Cref{eq:sslsumci} also holds here for \FnSSFlipErr{$e$}. Similarly, the value $|e+a^i+\delta^{\cC}(a^{i-1})|_\beta$ initially equals $|e|_\beta$, can increase by at most $|c^i|_\beta$ in the iteration of the while loop at which each $c^i\in F^i$ is chosen, and by \Cref{eq:ssfupimp} in \Cref{lem:ssfimp} must decrease by at least $3r\beta|c^{i-1}|_\beta$ in the iteration of the while loop at which each $c^{i-1}\in F^{i-1}$ is chosen. Therefore because $|e+a^i+\delta^{\cC}(a^{i-1})|_\beta$ cannot decrease to a value below $0$, we must have
    \begin{equation}
      \label{eq:sslsumcim1}
      \sum_{c^{i-1}\in F^{i-1}}3r\beta|c^{i-1}|_\beta \leq |e|_V|_\beta+\sum_{c^i\in F^i}|c^i|_\beta % \leq |e|_\beta+\frac{|\delta^{\cC}(e)|_V|_\beta}{3r\beta}.
    \end{equation}
    Now as $V$ is the footprint of \FnSSFlipErr{$e|_V$}, by definition
    \begin{align*}
      V = \supp(e|_V) \cup \supp(\delta^{\cC}(e)|_V) \cup \bigcup_{c^{i-1}\in F^{i-1}}&(\supp(c^{i-1})\cup\supp(\delta^{\cC}(c^{i-1}))) \\
      \cup \bigcup_{c^i\in F^i}&(\supp(c^i)\cup\supp(\delta^{\cC}(c^i))),
    \end{align*}
    so by \Cref{eq:sslsumcim1,eq:sslsumci,eq:cobnormbeta} we have
    \begin{align*}
      |V|_\beta
      &\leq (1+r(\Delta+1)\beta)\left(|e|_V|_\beta + \sum_{c^{i-1}\in F^{i-1}}|c^{i-1}|_\beta + \sum_{c^i\in F^i}|c^i|_\beta\right) \\
      &\leq 2(1+r(\Delta+1)\beta)\left(|e|_V|_\beta + \sum_{c^i\in F^i}|c^i|_\beta\right) \\
      &\leq 8r\beta\Delta^2\cdot|e|_V|_\beta.
    \end{align*}
    To see that the final inequality above holds, note that the expression on the second line above is $2$ times the expression on the second line of \Cref{eq:Vbetasyn}, and hence the final line above is $2$ times the final line of \Cref{eq:Vbetasyn}, as we showed above that \Cref{eq:sslsumci} holds for both \FnSSFlipSyn{$\delta^{\cC}(e)|_V$} and for \FnSSFlipErr{$e|_V$}. \qedhere
  \end{enumerate}
\end{proof}

In \Cref{lem:basedec} below, we apply the lemmas above to show that \FnSSFlipSyn successfully decodes $\cE_Z^{\cC}(\eta,\gamma,\omega)$-avoiding errors for appropriate choices of $\eta,\gamma,\omega$. Here we consider the cochain complex $\cC^*=\cC^*(\epsilon,\bar{\ell},0,r)$, which by \Cref{def:redsupp} has $\cE_{Z,\Rd}^{\cC}=\emptyset$. Hence all cochains are trivially $\cE_{Z,\Rd}^{\cC}$-avoiding, which simplifies the task of showing that a cochain is $\cE_Z^{\cC}(\eta,\gamma,\omega;\beta)$-avoiding (see \Cref{def:badfams}).

\begin{lemma}
  \label{lem:basedec}
  Let $r,\bar{\ell}\in\bN$, $\beta>1$, and define $\epsilon=\epsilon(r)$ as in \Cref{lem:ssflip}, define $\mu=\mu(\epsilon)$ and $\Delta=\Delta(\epsilon)$ as in \Cref{lem:lossless}.
  % and define $\beta=\beta(\epsilon)$ as in \Cref{def:weightnorm}.
  For $\Upsilon\in\bN$, let $\cC^*=\cC^*(\epsilon,\bar{\ell},0,r)^{\oplus\Upsilon}$, and for arbitrary
  \begin{align}
    \label{eq:bdinerr}
    \eta &\leq \mu 2^{\bar{\ell}/2}, \hspace{1em} \gamma \leq \frac{1}{16r\Delta^2}, \hspace{1em} \omega \leq \frac{\beta^{\bar{\ell}/2}}{8r\beta\Delta^2},
  \end{align}
  let
  \begin{align}
    \label{eq:bdouterr}
    \tilde{\eta} &= \eta, \hspace{1em} \tilde{\gamma} = 8r\Delta^2\cdot\gamma, \hspace{1em} \tilde{\omega} = 8r\beta\Delta^2\cdot\omega.
  \end{align}
  For $0\leq i\leq r-1$, let $e\in\cC^i$ be an error that is $\cE_Z^{\cC}(\eta,\gamma,\omega;\beta)$-avoiding (see \Cref{def:badfams}). Then the output $\tilde{a}^i=\FnSSFlipSyn{$\delta^{\cC}(e)$}$ (see \Cref{alg:ssflip}) is of the form $\tilde{a}^i=e+\delta^{\cC}(\tilde{a}^{i-1})$ for some $\tilde{a}^{i-1}\in\cC^{i-1}$ that is $\cE_Z^{\cC}(\tilde{\eta},\tilde{\gamma},\tilde{\omega};\beta)$-avoiding.
\end{lemma}
\begin{proof}
  We begin with the following basic claim relating \FnSSFlipSyn to \FnSSFlipErr.

  \begin{claim}
    \label{claim:bdsyntoerr}
    If \FnSSFlipSyn{$\delta^{\cC}(e)$} outputs $\tilde{a}^i$, then \FnSSFlipErr{$e$} outputs $\tilde{a}^{i-1},\tilde{a}^i$ for some $\tilde{a}^{i-1}\in\cC^{i-1}$.
  \end{claim}
  \begin{proof}
    The claim follows directly from the definition of \Cref{alg:ssflip}, as every update $c^i$ that is added into $a^i$ will be the same in the executions of \FnSSFlipSyn{$\delta^{\cC}(e)$} and \FnSSFlipErr{$e$}.
  \end{proof}

  Defining $\tilde{a}^{i-1},\tilde{a}^i$ as in \Cref{claim:bdsyntoerr}, it suffices to show that $\tilde{a}^i=e+\delta^{\cC}(\tilde{a}^{i-1})$ and that $\tilde{a}^{i-1}$ is $\cE_Z^{\cC}(\tilde{\eta},\tilde{\gamma},\tilde{\omega})$-avoiding. For this purpose, let $S_{\mathrm{err}}\subseteq C^{i-1}\sqcup C^i\sqcup C^{i+1}$ be the footprint of \FnSSFlipErr{$e$} (see \Cref{def:footprint}).

  \begin{claim}
    \label{claim:bdavoid}
    $S_{\mathrm{err}}$ is $\cE_Z^{\cC}(\tilde{\eta},\tilde{\gamma},\tilde{\omega})$-avoiding.
  \end{claim}
  \begin{proof}
    First, applying \Cref{eq:sslweigh} in \Cref{lem:sslocal} and summing over all connected components $V$, we conclude that
    \begin{equation*}
      |S_{\mathrm{err}}|_\beta \leq 8r\beta\Delta^2\cdot|e|_\beta < 8r\beta\Delta^2\cdot\omega = \tilde{\omega}.
    \end{equation*}
    
    Thus it remains to be shown that $S_{\mathrm{err}}$ is $\cE(G^{\cC},\tilde{\eta},\tilde{\gamma})$-avoiding. For this purpose, assume for a contradiction that there exists some subset $V\subseteq C$ of size $|V|\geq\tilde{\eta}$ that induces a connected subgraph of $G^{\cC}$ and has $|V\cap S_{\mathrm{err}}|/|V|\geq\tilde{\gamma}$. We may repeatedly add to $V$ any element $c\in S_{\mathrm{err}}\setminus V$ that is connectected to $V$ by an edge, until no more such elements $c$ exist. This procedure only increases $|V|$ as well as the density $|V\cap S_{\mathrm{err}}|/|V|$, so $V$ is still a connected subgraph of $G^{\cC}$ satisfying $|V|\geq\tilde{\eta}$ and $|V\cap S_{\mathrm{err}}|/|V|\geq\tilde{\gamma}$. Furthermore, now every connected component $V'$ of the subgraph of $G^{\cC}$ induced by $S_{\mathrm{err}}$ lies either entirely inside or entirely outside $V$.
    By \Cref{eq:sslweigh} in \Cref{lem:sslocal}, for each such $V'$, we have $|e|_{V'}|\geq|V'|/8r\Delta^2$. Summing over all such $V'$ inside $V$, we obtain
    \begin{equation}
      \label{eq:bdeV}
      |e|_V| = |e|_{V\cap S_{\mathrm{err}}}| \geq \frac{|V\cap S_{\mathrm{err}}|}{8r\Delta^2} \geq \frac{\tilde{\gamma}|V|}{8r\Delta^2} = \gamma|V|,
    \end{equation}
    where the first equality above holds because $\supp(e)\subseteq S_{\mathrm{err}}$ by the definition of the footprint.

    However, the assumption that $e$ is $\cE_Z^{\cC}(\eta=\tilde{\eta},\;\gamma,\;\omega)$-avoiding, and hence $\cE(G^{\cC},\tilde{\eta},\gamma)$-avoiding, implies that $|e|_V|<\gamma|V|$, which contradicts \Cref{eq:bdeV}. Thus the assumption that $S_{\mathrm{err}}$ is not $\cE(G^{\cC},\tilde{\eta},\tilde{\gamma})$-avoiding was false, as desired.
  \end{proof}

  Beacuse $\supp(\tilde{a}^{i-1})\subseteq S_{\mathrm{err}}$ by the definition of the footprint, \Cref{claim:bdavoid} implies that $\tilde{a}^{i-1}$ is $\cE_Z^{\cC}(\tilde{\eta},\tilde{\gamma},\tilde{\omega})$-avoiding. Therefore the following claim completes the proof of the lemma.

  \begin{claim}
    We have $\tilde{a}^i=e+\delta^{\cC}(\tilde{a}^{i-1})$; that is, \FnSSFlipErr{$e$} in \Cref{alg:ssflip} does not return FAIL.
  \end{claim}
  \begin{proof}
    Assume for a contradiction that \FnSSFlipErr{$e$} does return FAIL, and let $a^{i-1},a^i$ be the variables in the algorithm just prior to returning. Let $V$ be a connected component of the subgraph of $G^{\cC}$ induced by $S_{\mathrm{err}}$. By \Cref{claim:bdavoid}, because $\tilde{\gamma}<1$, we must have $|V|<\eta=\tilde{\eta}$. Then because \FnSSFlipErr acts independently on each connected component $V$ (see \Cref{lem:sslocal}), \FnSSFlipErr{$e|_V$} must run until it obtains the values $a^{i-1}|_V,\;a^i|_V$, and then must output FAIL. Indeed, if upon reaching the values $a^{i-1}|_V,\;a^i|_V$, \FnSSFlipErr{$e|_V$} found a valid update $c^{i-1}$ or $c^i$, then this update would by definition lie inside the associated connected component $V$ of the footprint, and hence the same valid update would apply to \FnSSFlipErr{$e$}, a contradiction.

    However, by \Cref{claim:bdavoid} we have
    \begin{equation*}
      |(e+a^i+\delta^{\cC}(a^{i-1}))|_V|_\beta \leq |V|_\beta \leq |S_{\mathrm{err}}|_\beta < \tilde{\omega} \leq \beta^{\bar{\ell}/2},
    \end{equation*}
    and because $\tilde{\gamma}<1$,
    \begin{equation*}
      |(e+a^i+\delta^{\cC}(a^{i-1}))|_V| \leq |V| < \eta = \mu 2^{\bar{\ell}/2}.
    \end{equation*}
    Therefore by \Cref{lem:ssflip}, there exists $\bar{c}^0\in C^0$ such that either there is some $c^{i-1}\in\cC^{i-1}$ with $\bar{c}^0\preceq' c^{i-1}$ and
    \begin{equation*}
      |(e+a^i+\delta^{\cC}(a^{i-1}))|_V+\delta^{\cC'}(c^{i-1})| < |(e+a^i+\delta^{\cC}(a^{i-1}))|_V|-4r|c^{i-1}|,
    \end{equation*}
    or else there is some $c^i\in\cC^i$ with $\bar{c}^0\preceq' c^i$ and
    \begin{equation*}
      |\delta^{\cC}(e+a^i)|_V+\delta^{\cC'}(c^i)|<|\delta^{\cC}(e+a^i)|_V|-4r|c^i|.
    \end{equation*}
    Thus upon reaching the values $a^{i-1}|_V,\;a^i|_V$, \FnSSFlipErr{$e|_V$} must find a valid update $c^{i-1}$ or $c^i$, contradicting our conclusion above that it cannot find such an update. Thus the assumption that \FnSSFlipErr{$e$} returns FAIL was false, as desired.
  \end{proof}
\end{proof}

\subsection{General Decoding Algorithm}
We now apply our small-set flip decoder for $\cC^*(\epsilon,\bar{\ell},0,r_Z)^{\oplus\Upsilon}$ to construct a decoder for complexes $\cC^*(\epsilon,\bar{\ell},r_X,r_Z)^{\oplus\Upsilon}$ with arbitrary $r_X,r_Z$. Specifically, we present our decoding algorithm in \Cref{alg:decode}, and we analyze it in \Cref{lem:decode}.

\begin{algorithm}
  \caption{\label{alg:decode} (Classical) decoding algorithm for $\cC^*=\cC^*(\epsilon,\bar{\ell},r_X,r_Z)^{\oplus\Upsilon}$. We apply the small-set flip decoder from \Cref{alg:ssflip} as a key subroutine. The decoding algorithm implicitly depends on the choice of $\epsilon,\bar{\ell},r_X,r_Z$, as well as a choice of $0\leq i\leq r-1$, which will always be made clear from context.
    Below, letting $\cA^*=\cC_*(\epsilon,\bar{\ell})=\cC^*(\epsilon,\bar{\ell},1,0)$, then we let $V^0_\ell\subseteq A^0$ and $V^1_\ell\subseteq A^1$ denote the copy of $V_\ell$ in $A^0$ and in $A^1$, respectively. Similarly, for $v\in V_\ell$, we let $v^0\in V^0_\ell$ and $v^1\in V^1_\ell$ denote the copy of $v$ in the respective copy of $V_\ell$.}
  \SetKwInOut{Input}{Input}
  \SetKwInOut{Output}{Output}
  \SetKwProg{Fn}{Function}{:}{}

  \SetKwFunction{FnDecode}{Decode}

  \Input{$(i+1)$-coboundary $s\in B^{i+1}(\cC)$ of the cochain complex $\cC^*=\cC^*(\epsilon,\bar{\ell},r_X,r_Z)^{\oplus\Upsilon}$}
  \Output{$a^i\in\cC^i$ with $s=\delta^{\cC}(a^i)$}
  \Fn{\FnDecode{$s;\cC^*=\cC^*(\epsilon,\bar{\ell},r_X,r_Z)^{\oplus\Upsilon}$}}{
    \If{$r_X=0$}{
      \Return{\FnSSFlipErr{$s$}}
    }\Else{
      Let $\cA^*=\cC_*(\epsilon,\bar{\ell})=\cC^*(\epsilon,\bar{\ell},1,0)$, $\cB^*=\cC^*(\epsilon,\bar{\ell},r_X-1,r_Z)^{\oplus\Upsilon}$ (so $\cC^*=\cA^*\otimes\cB^*$) \\
      Initialize $a^i\gets 0\in\cC^*=(\cA^*\otimes\cB^*)^i=\cA^0\otimes\cB^i\oplus\cA^{1}\otimes\cB^{i-1}$ \\
      \For{$\ell=0,\dots,\bar{\ell}-1$}{ \label{li:deczero}
        $a^i|_{V^0_\ell\times B^i}\gets (s+\delta^{\cC}(a^i)|_{V^1_\ell\times B^i}$ % (where we restrict $a^i$ to $V_\ell\times B^i\subseteq A^0\times B^i\subseteq C^i$, and we restrict $s+\delta^{\cC}(a^i)$ to $V_\ell\times B^i\subseteq A^1\times B^i\subseteq C^{i+1}$)
      }
      \For{$v\in V_{\bar{\ell}}$}{ \label{li:decrec}
        $a^i|_{\{v^1\}\times B^{i-1}}\gets\FnDecode{$(s+\delta^{\cC}(a^i))|_{\{v^1\}\times B^i};\cB^*$}$ \label{li:deccall} % (where we restrict $a^i$ to $\{v\}\times B^{i-1}\subseteq A^1\times B^{i-1}\subseteq C^i$, and we restrict $s+\delta^{\cC}(a^i)$ to $\{v\}\times B^i\subseteq A^1\times B^i\subseteq C^{i+1}$)
      }
      \Return $a^i$
    }
  }
\end{algorithm}

\begin{lemma}
  \label{lem:decode}
  For $r_X,r_Z\in\bZ_{\geq 0}$, let $r=r_X+r_Z$ and define $\epsilon=\epsilon(r)$ as in \Cref{lem:ssflip}. For $\Upsilon\in\bN$, let $\cC^*=\cC^*(\epsilon,\bar{\ell},r_X,r_Z)^{\oplus\Upsilon}$, and for $\eta,\gamma,\omega\geq 0$ satisfying \Cref{eq:bdinerr}, define $\tilde{\eta},\tilde{\gamma},\tilde{\omega}$ as in \Cref{eq:bdouterr} in \Cref{lem:basedec}. Define $V_Z:=(V_{\bar{\ell}})^{r_X}\times C(\epsilon,\bar{\ell},0,r_Z)\subseteq C$ as in \Cref{def:badfams}. Fix some $\beta>1$.
  For some $0\leq i\leq r-1$, let $e\in\cC^i$ be $\cE_Z^{\cC}(\eta,\gamma,\omega;\beta)$-avoiding, so that $e$ is reduced by \Cref{lem:redsupp,def:badfams} (see also \Cref{def:reduced,def:redsupp}). Then the output $\tilde{a}^i=\FnDecode{$\delta^{\cC}(e);\cC^*$}$ of \Cref{alg:decode} is of the form $\tilde{a}^i=e+\delta^{\cC}(\tilde{a}^{i-1})$ for some $\tilde{a}^{i-1}\in\cC^{i-1}$ that is $\cE_Z^{\cC}(\tilde{\eta},\tilde{\gamma},\tilde{\omega};\beta)$-avoiding and has $\supp(\tilde{a}^{i-1})\subseteq V_Z$ (so that in particular, $\tilde{a}^{i-1}$ is reduced).
\end{lemma}
\begin{proof}
  We prove the result by induction on $r_X$. For the base case, when $r_X=0$, then the result follows immediately by \Cref{lem:basedec}. For the inductive step, let $r_X\geq 1$, and assume the result holds for every $r_X'<r_X$.
  % By definition $\supp(\tilde{a}^i)$ contains no elements in $(A^1\setminus V_{\bar{\ell}})\times B^{i-1}$, so $\Rd\otimes I_{B^{i-1}}(\tilde{a}^i|_{A^1\times B^{i-1}})=0$. Furthermore, by the inductive hypothesis, for each $v\in V_{\bar{\ell}}$ then $a^i|_{\{v^1\}\times B^{i-1}}$ is reduced, so $I\otimes\Rd(a^i|_{V_{\bar{\ell}}\times B^{i-1}})=0$. Therefore by \Cref{def:reducer} we have $\Rd(\tilde{a}^i)=0$, so $\tilde{a}^i$ is reduced.
  % Define $V_Z:=(V_{\bar{\ell}})^{r_X}\times C(\epsilon,\bar{\ell},0,r_Z)$ as in \Cref{def:badfams}.
  The following claim shows that our decoder exactly determines the errors in $A^0\times B^i$. Below, recall that $C^i=A^0\times B^i\sqcup A^1\times B^{i-1}$.
  
  \begin{claim}
    \label{claim:deczero}
    Let $a'$ denote the value of $a^i$ just prior to beginning \Cref{li:decrec} in \Cref{alg:decode} on input $s=\delta^{\cC}(e)$. Then $a'|_{A^0\times B^i}=e|_{A^0\times B^i}$ and $a'|_{A^1\times B^{i-1}}=0$.
  \end{claim}
  \begin{proof}
    By definition $a'|_{A^1\times B^{i-1}}=0$, so it suffices to show that $a'|_{A^0\times B^i}=e|_{A^0\times B^i}$. For this purpose, we observe that because $e$ is reduced, for every $0\leq\ell\leq\bar{\ell}-1$ we by definition have $e|_{V^1_\ell\times B^{i-1}}=0$, and hence $(e+a^i)|_{V^1_\ell\times B^{i-1}}=0$. Therefore if it holds for every $0\leq\ell'\leq\ell-1$ that $(e+a^i)|_{V^0_\ell\times B^i}=0$, then by the definition of $\delta^{\cC}$ and $\delta^{\cA}$ we have
    \begin{equation*}
      \delta^{\cC}(e+a^i)|_{V^1_\ell\times B^i} = ((\delta^{\cA}\otimes I)(e+a^i) + (I\otimes\delta^{\cB})(e+a^i))|_{V^1_\ell\times B^i} = (e+a^i)|_{V^0_\ell\times B^i}.
    \end{equation*}
    Therefore recalling that $s=\delta^{\cC}(e)$, then after iteration $\ell$ of the for loop in \Cref{li:deczero}, for every $0\leq\ell'\leq\ell$ we must have $(e+a^i)|_{V^0_\ell\times B^i}=0$. Thus the desired claim follows because $A^0=\bigsqcup_{\ell=0}^{\bar{\ell}-1}V^0_\ell$.
  \end{proof}

  \Cref{claim:deczero} implies that each recursive call to \FnDecode{$(s+\delta^{\cC}(a^i))|_{\{v^1\}\times B^i};\cB^*$} in \Cref{li:deccall} of \Cref{alg:decode} receives input
  \begin{equation}
    \label{eq:recinp}
    (s+\delta^{\cC}(a^i))|_{\{v^1\}\times B^i} = \delta^{\cC}(e+a^i)|_{\{v^1\}\times B^i} = \delta^{\cB}(e|_{\{v^1\}\times B^{i-1}}).
  \end{equation}
  For $v\in V_{\bar{\ell}}$, let $\omega_v=|e|_{(\{v^1\}\times B^{i-1})\cap V_Z}|_\beta$, so that $\sum_{v\in V_{\bar{\ell}}}\omega_v=|e|_{V_Z}|_\beta<\omega$ because $e$ is $\cE_Z^{\cC}(\eta,\gamma,\omega)$-avoiding. Let
  \begin{equation*}
    \tilde{\omega}_v = 8r\beta\Delta^2\left(\omega_v + \frac{\omega-|e|_{V_Z}|_\beta}{|V_{\bar{\ell}}|}\right),
  \end{equation*}
  so that $\tilde{\omega}_v\geq 8r\beta\Delta^2\cdot\omega_v$ and
  \begin{equation}
    \label{eq:decsumto}
    \sum_{v\in V_{\bar{\ell}}}\tilde{\omega}_v=\tilde{\omega}.
  \end{equation}
  By the inductive hypothesis, the output of each recursive call to \FnDecode with input given by \Cref{eq:recinp} is of the form $\tilde{a}^i|_{\{v^1\}\times B^{i-1}}=e|_{\{v^1\}\times B^{i-1}}+\delta^{\cB}(\tilde{a}^{i-1}_{v^1})$ for some $\tilde{a}^{i-1}_{v^1}\in(\bF_2^\Upsilon)^{B^{i-2}}$ that is $\cE_Z^{\cC}(\tilde{\eta},\tilde{\gamma},\tilde{\omega}_v)$-avoiding and supported inside $V_{\bar{\ell}}^{r_X-1}\times C(\epsilon,\bar{\ell},0,r_Z)$. Define $\tilde{a}^{i-1}\in(\bF_2^\Upsilon)^{V_{\bar{\ell}}^1\times B^{i-2}}\subseteq(\bF_2^\Upsilon)^{C^{i-1}}=\cC^{i-1}$ by letting $\tilde{a}^{i-1}|_{\{v^1\}\times B^{i-2}}=\tilde{a}^{i-1}_{v^1}$ for each $v\in V_{\bar{\ell}}$, so that $\tilde{a}^{i-1}|_{C^{i-1}\setminus(V_{\bar{\ell}}\times B^{i-2})}=0$. Then $\supp(\tilde{a}^{i-1})\subseteq V_Z$, so that in particular $\tilde{a}^{i-1}$ is reduced by \Cref{def:reducer}. Also by \Cref{def:badfams,eq:decsumto}, $\tilde{a}^{i-1}$ is $\cE_Z^{\cC}(\tilde{\eta},\tilde{\gamma},\tilde{\omega})$-avoiding. Furthermore, by \Cref{claim:deczero} along with the definition of $\tilde{a}^{i-1}$ and the fact that $e$ is reduced so that $e|_{(A^1\setminus V^1_{\bar{\ell}})\times B^{i-1}}=0$, we have
  \begin{equation*}
    \tilde{a}^i = e+I\otimes\delta^{\cB}(\tilde{a}^{i-1}|_{V^1_{\bar{\ell}}\times B^{i-2}}) = e+\delta^{\cC}(\tilde{a}^{i-1}),
  \end{equation*}
  as desired.
\end{proof}

\Cref{cor:decode} below applies \Cref{lem:decode} to show that the 1-based complex $\cC^*=\cC^*(\epsilon,\bar{\ell},r_X,r_Z)$ (as opposed to the direct sum) has the following property: once we fix an $\cE_Z^{\cC}(\eta,\gamma,\omega;\beta)$-avoiding set $E$ and a syndrome $s$, there exists a single $\cE_Z^{\cC}(\tilde{\eta},\tilde{\gamma},\tilde{\omega};\beta)$-avoiding set $\tilde{E}$ such that for every error $e$ supported inside $E$ with syndrome $\delta^{\cC}(e)=s$, our decoder on input $s$ recovers $e$ up to a residual error given by the coboundary of some cochain supported inside $\tilde{E}$. The proof of \Cref{cor:decode} simply invokes \Cref{lem:decode} on the direct sum of all possible such errors $e$. We will use \Cref{cor:decode} in \Cref{sec:stateprep} to ensure that there exists an appropriate set $\tilde{E}$ containing the support of the residual error on the outputs state of our state preparation gadget (\Cref{lem:stateprep}), in order to show fault-tolerance in the sense of \Cref{def:faulttol}.

\begin{corollary}
  \label{cor:decode}
  For $r_X,r_Z\in\bZ_{\geq 0}$, let $r=r_X+r_Z$ and define $\epsilon=\epsilon(r)$ as in \Cref{lem:ssflip}. Let $\cC^*=\cC^*(\epsilon,\bar{\ell},r_X,r_Z)$, and for $\eta,\gamma,\omega\geq 0$ satisfying \Cref{eq:bdinerr}, define $\tilde{\eta},\tilde{\gamma},\tilde{\omega}$ as in \Cref{eq:bdouterr} in \Cref{lem:basedec}. Define $V_Z:=(V_{\bar{\ell}})^{r_X}\times C(\epsilon,\bar{\ell},0,r_Z)\subseteq C$ as in \Cref{def:badfams}. Fix some $\beta>1$.
  For some $0\leq i\leq r-1$, fix some $\cE_Z^{\cC}(\eta,\gamma,\omega;\beta)$-avoiding set $E\subseteq C^i$, and fix some $s\in\cC^{i+1}$.
  Then there exists a $\cE_Z^{\cC}(\tilde{\eta},\tilde{\gamma},\tilde{\omega};\beta)$-avoiding set $\tilde{E}\subseteq C^{i-1}\cap V_Z$ such that for every $e\in\cC^i$ with $\supp(e)\subseteq E$ (so that $e$ is reduced by \Cref{lem:redsupp,def:badfams}; see also \Cref{def:reduced,def:redsupp}) and with $\delta^{\cC}(e)=s$,
  the output $\tilde{a}^i=\FnDecode{$s;\cC^*$}$ of \Cref{alg:decode} can be expressed in the form $\tilde{a}^i=e+\delta^{\cC}(\tilde{a}^{i-1})$ for some $\tilde{a}^{i-1}\in\cC^{i-1}$ with $\supp(\tilde{a}^{i-1})\subseteq\tilde{E}$ (so that in particular, $\tilde{a}^{i-1}$ is reduced because $\tilde{E}\subseteq V_Z$).
\end{corollary}
\begin{proof}
  Let $\Upsilon\subseteq 2^E\subseteq\cC^i$ denote the set of all $e\in\cC^i$ with $\supp(e)\subseteq E$ such that $\delta^{\cC}(e)=s$. Then define $\bar{e}\in(\cC^i)^{\oplus\Upsilon}$ to simply be the concatenation of all $e\in\Upsilon$. That is, recalling that $(\cC^i)^{\oplus\Upsilon}=2^{\Upsilon\times C^i}$, then for every $e\in\Upsilon\subseteq 2^E\subseteq\cC^i$, we let $\bar{e}|_{\{e\}\times C^i}:=e\in\cC^i$. Then because by definition $\delta^{\cC^{\oplus\Upsilon}}(\bar{e})=s^{\oplus\Upsilon}$, \Cref{lem:decode} implies that the output $\bar{\tilde{a}}^i=\FnDecode{$s^{\oplus\Upsilon};\cC^{\oplus\Upsilon}$}$ of \Cref{alg:decode} is of the form $\bar{\tilde{a}}^i=\bar{e}+\delta^{\cC^{\oplus\Upsilon}}(\bar{\tilde{a}}^{i-1})$ for some $\bar{\tilde{a}}^{i-1}\in(\cC^{i-1})^{\oplus\Upsilon}$ that is $\cE_Z^{\cC}(\tilde{\eta},\tilde{\gamma},\tilde{\omega};\beta)$-avoiding and has $\supp(\bar{\tilde{a}}^{i-1})\subseteq V_Z$.

  Let $\tilde{E}=\supp(\bar{\tilde{a}}^{i-1})$, so that $\tilde{E}\subseteq V_Z$ is $\cE_Z^{\cC}(\tilde{\eta},\tilde{\gamma},\tilde{\omega};\beta)$-avoiding. By definition, \FnDecode{$s^{\oplus\Upsilon};\cC^{\oplus\Upsilon}$} simply runs $|\Upsilon|$ identical executions of \FnDecode{$s;\cC$} in parallel. Therefore for every $e\in\cC^i$ with $\supp(e)\subseteq E$ and $\delta^{\cC}(e)=s$, the output $\tilde{a}^i=\FnDecode{$s;\cC^*$}$ is given by
  \begin{equation*}
    \tilde{a}^i = \bar{\tilde{a}}^i|_{\{e\}\times C^i} = \bar{e}|_{\{e\}\times C^i}+\delta^{\cC}(\bar{\tilde{a}}^{i-1}|_{\{e\}\times C^i}) = e+\delta^{\cC}(\bar{\tilde{a}}^{i-1}|_{\{e\}\times C^i}).
  \end{equation*}
  Hence letting $\tilde{a}^{i-1}=\bar{\tilde{a}}^{i-1}|_{\{e\}\times C^i}$ gives the desired conclusion.
\end{proof}

The following basic lemma shows that the reducer of a cochain never contains any bad sets defined in \Cref{def:badfams}.

\begin{lemma}
  \label{lem:redgood}
  Let $\cC^*=\cC^*(\epsilon,\bar{\ell},r_X,r_Z)$, let $r=r_X+r_Z$, and let $V_Z:=(V_{\bar{\ell}})^{r_X}\times C(\epsilon,\bar{\ell},0,r_Z)$ as in \Cref{def:badfams}. For $0\leq i\leq r$ and for $e\in\cC^i$, we have
  \begin{equation*}
    \supp(\Rd(e))\cap(V_Z\cup\cE_{Z,\Rd}^{\cC}) = \emptyset,
  \end{equation*}
  so in particular $\Rd(e)$ is $\cE_Z^{\cC}(\eta,\gamma,\omega;\beta)$-avoiding for every $\eta,\gamma,\omega>0$.
\end{lemma}
\begin{proof}
  First, \Cref{lem:redprops} shows that $\Rd(\Rd(e))=0$ so that $\Rd(e)$ is reduced, and hence \Cref{lem:redsupp} implies that $\supp(\Rd(e))\cap\cE_{Z,\Rd}^{\cC}=\emptyset$.

  To show that $\supp(\Rd(e))\cap V_Z=\emptyset$, recall that
  \begin{equation*}
    \cC^*(\epsilon,\bar{\ell},r_X,r_Z) = \cC^*(\epsilon,\bar{\ell},1,0)^{\otimes r_X} \otimes \cC^*(\epsilon,\bar{\ell},0,1)^{\otimes r_Z}.
  \end{equation*}
  By \Cref{def:reducer}, the image of the reducer map $\Rd:\cC^i\rightarrow\cC^{i-1}$ is entirely supported on elements of $C=C(\epsilon,\bar{\ell},1,0)^{\times r_X}\times C(\epsilon,\bar{\ell},0,1)^{\times r_Z}$ whose projection onto at least one of the factors of $C(\epsilon,\bar{\ell},1,0)$ lies inside $C^0(\epsilon,\bar{\ell},1,0)$. However, by definition $V_Z\subseteq C^1(\epsilon,\bar{\ell},1,0)^{\times r_X}\times C(\epsilon,\bar{\ell},0,1)^{\times r_Z}$, so indeed $\supp(\Rd(e))\cap V_Z=\emptyset$.

  It then follows that $\Rd(e)$ is $\cE_Z^{\cC}(\eta,\gamma,\omega;\beta)$-avoiding for every $\eta,\gamma,\omega>0$ because by \Cref{def:badfams,remark:badred}, we have $\cE_Z^{\cC}(\eta,\gamma,\omega;\beta) \subseteq 2^{V_Z\cup\cE_{Z,\Rd}^{\cC}}$.
\end{proof}

% TODO: give decoder for reduced errors: exactly kill all errors in $A^0$ sectors, then run small-set-flip decoder from \cite{golowich_constant-overhead_2025} on (each $\cC^*(\epsilon,\bar{\ell},0,r_Z)$-layer of) the remaining error. Probably want to give two algorithms: one for state prep where we're given a noisy coboundary of an arbitrary vector, and we want to recover an approximation of the original vector (up to adding in an arbitrary cocycle, which is unavoidable); the other for error-correction where we're given a noisy cocyle and we want to recover the cocycle's cohomology class

\section{Fault-Tolerant Gadgets with Non-Uniform Errors}
\label{sec:ftnonu}
In this section, we describe our main gadgets for the codes in \Cref{sec:quantumprod}. These gadgets are designed to have low failure (i.e.~logical error) probability when the physical fault consists of non-uniform locally stochastic errors. That is, each physical qubit in our gadgets is labeled by some basis element of a cochain complex $\cC^*=\cC^*(\epsilon,\bar{\ell},r_X,r_Z)$ as defined in \Cref{sec:codeconstruct}. As described in \Cref{sec:bfls} below, our gadgets are designed for locally stochastic noise in which a physical qubit at level~$\ell$ (see \Cref{def:weightnorm}) is corrupted with probability $\leq p^{2^{\Omega(\ell)}}$, for a small constant $p>0$. In Appendix~\ref{sec:concat}, we show how to concatenate our gadgets with codes of varying sizes in order to obtain fault-tolerance under uniform locally stochastic noise, in which every qubit's error probability is bounded by the same constant. The key idea is that the fraction of qubits at level $\ell$ decays exponentially in $\ell$, so we can afford to encode the few qubits with large $\ell$ into larger codes, while preserving a constant encoding rate of the entire scheme.

Note that our ejection (\Cref{lem:eject}) and injection (\Cref{lem:inject}) gadgets below have some bare physical qubits labeled by $M^{r_X}=V_{\bar{\ell}}^{\times r}$, which we view as level-$0$ qubits (again see \Cref{def:weightnorm}) via the inclusion $M^{r_X}=V_{\bar{\ell}}^{\times r}\subseteq C^{r_X}$.

\subsection{Families of Bad Sets for Faults}
In this section, we define the families of bad sets $\cE_{\mathrm{run}}$ that we use for faults for our gadgets. Letting $\cC^*=\cC^*(\epsilon,\bar{\ell},r_X,r_Z)$, our gadgets will act on sets of qubits labeled by subsets of $C$, or in some cases, subsets of multiple copies of $C$. For each such copy of $C$, we will impose a family of bad sets as defined in \Cref{sec:bfdef} below. \Cref{sec:bfprop} presents various properties of these families of bad sets, and \Cref{sec:bfls} shows that they capture the behavior of appropriate locally stochastic noise with high probability.

\subsubsection{Definition of Families}
\label{sec:bfdef}
We begin by defining the \emph{shadow} of a vertex $v_\ell\in V_\ell$ to be the set of vertices in $V_{\bar{\ell}}$ that can be reached from $v_\ell$ by a path, with one edge in each graph $G_{\ell+1},\dots,G_{\bar{\ell}}$.

\begin{definition}
  Define all variables as in \Cref{def:classcode}. For $0\leq\ell\leq\bar{\ell}$ and for $v_\ell\in V_\ell$, we define the \emph{shadow} $\Shad(v_\ell)\subseteq V_{\bar{\ell}}$ to be the set of all vertices $v_\ell\in V_{\bar{\ell}}$ such that there exist $v_{\ell+1}\in V_{\ell+1},\dots,v_{\bar{\ell}-1}\in V_{\bar{\ell}-1}$ for which $(v_{\ell'},v_{\ell'-1})\in E^{(\ell')}$ for every $\ell<\ell'\leq\bar{\ell}$.

  We extend $\Shad$ to a mapping $\Shad:2^{C(\epsilon,\bar{\ell})}\rightarrow 2^{V_{\bar{\ell}}}$ such that for $S\subseteq C(\epsilon,\bar{\ell})$, then $\Shad(S)=\bigcup_{v\in S}\Shad(v)$, where we use the fact that $C(\epsilon,\bar{\ell})$ is by definition a disjoint union of sets $V_\ell$. Letting $V_Z=(V_{\bar{\ell}})^{\times r_X}\times C(\epsilon,\bar{\ell})^{\times r_Z}$, then because $C(\epsilon,\bar{\ell},r_X,r_Z)\cong C(\epsilon,\bar{\ell})^{\times r_X}\times C(\epsilon,\bar{\ell})^{\times r_Z}$, we have a mapping $\Shad^{\times r_X}:C(\epsilon,\bar{\ell},r_X,r_Z)\rightarrow 2^{V_Z}$ that maps $(u^{(1)},\dots,u^{(r)})$ to the set of all $(v^{(1)},\dots,v^{(r_X)},u^{(r_X+1)},\dots,u^{(r)})\in V_Z$ such that for $1\leq i\leq r_X$ we have $v^{(i)}\in\Shad(u^{(i)})$.
\end{definition}

Below, recall the definition of $\cE_Z^{\cC(\epsilon,\bar{\ell},r_X,r_Z)}(\eta,\gamma)$ from \Cref{def:badfams}.

\begin{definition}
  \label{def:badfault}
  Let $\cC^*=\cC^*(\epsilon,\bar{\ell},r_X,r_Z)$. For $\eta,\gamma$, define
  \begin{equation*}
    \cE_{\mathrm{run},Z}^{\cC}(\eta,\gamma) = \{E\subseteq 2^C:\exists E'\in\cE_Z^{\cC}(\eta,\gamma) \text{ with } E'\subseteq\Shad^{\times r_X}(E)\}.
  \end{equation*}
  Then for $\omega\geq 0$ and $\beta\geq 1$, define
  \begin{equation}
    \label{eq:Erundef}
    \cE_{\mathrm{run},Z}^{\cC}(\eta,\gamma,\omega;\beta) = \cE_{\mathrm{run},Z}^{\cC}(\eta,\gamma) \cup (2^C|_\beta^{\geq\omega}).
  \end{equation}
  Define $\cE_{\mathrm{run},X}^{\cC}(\eta,\gamma,\omega;\beta)\subseteq 2^C$ to equal the family $\cE_{\mathrm{run},Z}^{\cC^\vee}(\eta,\gamma,\omega;\beta)$ for the dual complex ${\cC^\vee}^*=\cC_*\cong\cC^*(\epsilon,\bar{\ell},r_Z,r_X)$. We let
  \begin{equation}
    \label{eq:ErunXZ}
    \cE_{\mathrm{run}}^{\cC}(\eta,\gamma,\omega;\beta) = \cE_{\mathrm{run},X}^{\cC}(\eta,\gamma,\omega;\beta) \cup \cE_{\mathrm{run},Z}^{\cC}(\eta,\gamma,\omega;\beta)).
  \end{equation}
  When $\beta$ is clear from context we may write $\cE_{\mathrm{run}}^{\cC}(\eta,\gamma,\omega)=\cE_{\mathrm{run}}^{\cC}(\eta,\gamma,\omega;\beta)$.
\end{definition}

\subsubsection{Properties of Families}
\label{sec:bfprop}
In this section, we show various properties of the families of bad sets defined in \Cref{sec:bfdef}.

The following basic claim shows how shadows bound reductions for 1-dimensional complexes.

\begin{claim}
  \label{claim:rf1d}
  Let $\cA^*=\cC^*(\epsilon,\bar{\ell},1,0)$ denote the 1-dimensional complex. Then for every $a\in\cA^1$, we have $\supp(\Rdn(a))\subseteq\Shad(\supp(a))$. Furthermore, if $\beta\geq 2\Delta$, then for every $S\subseteq A$ we have $|\Shad(S)|_\beta=|\Shad(S)|\leq|S|_\beta$.
\end{claim}
\begin{proof}
  By linearity and because $\Shad(S)=\bigcup_{b\in S}\Shad(b)$, it suffices to consider $a$ of weight $|a|=1$, and $S$ of size $|S|=1$. If $a=\1_{v_\ell}$ for some $v_\ell\in V_\ell\subseteq A^1$, then the inductive definition of $\Rd$ in \Cref{def:pushdown} implies that $\Rdn(a)=a+\delta^{\cA}(\Rd(a))$ equals the sum of $\1_{v_{\bar{\ell}}}\in\bF_2^{V_{\bar{\ell}}}\subseteq A^1$ over all paths $(v_\ell,\dots,v_{\bar{\ell}})$ such that each $(v_{\ell'},v_{\ell'-1})\in E^{(\ell')}$. Therefore by definition $\supp(\Rdn(a))\subseteq\Shad(a)$.

  Now if $S=\{v_\ell\}$ for some $v_\ell$ in some vertex set $V_\ell\subseteq A$, then because each $G^{(\ell')}$ has right-degree $2\Delta$, the number of paths $(v_\ell,\dots,v_{\bar{\ell}})$ such that each $(v_{\ell'},v_{\ell'-1})\in E^{(\ell')}$ is $|\Shad(v_{\ell})|\leq(2\Delta)^{\bar{\ell}-\ell}=|\{v_\ell\}|_{2\Delta}$, so the claim follows by the assumption that $\beta\geq 2\Delta$.
\end{proof}

\Cref{claim:rf1d} immediately implies the following two claims.

\begin{claim}
  \label{claim:rfsb}
  For every $S\subseteq C(\epsilon,\bar{\ell},r_X,r_Z)$ we have $|\Shad^{\times r_X}(S)|_\beta\leq|S|_\beta$.
\end{claim}
\begin{proof}
  By definition $\Shad^{\times r_X}(S)$ simply applies $\Shad$ to each of the first $r_X$ factors in the $r=r_X+r_Z$ dimensional product defining $C(\epsilon,\bar{\ell},r_X,r_Z)$. By \Cref{claim:rf1d}, each of these applications of $\Shad$ does not increase the $\beta$-weighted norm, so \Cref{claim:rfsb} follows.
\end{proof}

\begin{claim}
  \label{claim:rfsuppred}
  Letting $V_Z:=(V_{\bar{\ell}})^{r_X}\times C(\epsilon,\bar{\ell},0,r_Z)$, then for every $e\in\cC^i(\epsilon,\bar{\ell},r_X,r_Z)$ we have
  \begin{equation*}
    % \label{eq:rfsuppred}
    \supp(\Rdn(e)|_{V_Z})\subseteq\Shad^{\times r_X}(\supp(e)).
  \end{equation*}
\end{claim}
\begin{proof}
  The claim follows immediately from \Cref{lem:redfactor,claim:rf1d}.
\end{proof}

In \Cref{lem:redfault} below, we show that for $\beta\geq 2\Delta$, then the reduction of a $\cE_{\mathrm{run}}^{\cC}(\eta,\gamma,\omega;\beta)$-avoiding error (see \Cref{def:badfault} above) will be $\cE^{\cC}(\eta,\gamma,\omega;\beta)$-avoiding. The proof applies the claims above.

\begin{lemma}
  \label{lem:redfault}
  Let $\cC^*=\cC^*(\epsilon,\bar{\ell},r_X,r_Z)$, let $\Delta=\Delta(\epsilon)$ be defined as in \Cref{lem:lossless}, and let $V_Z=(V_{\bar{\ell}})^{r_X}\times C(\epsilon,\bar{\ell},0,r_Z)$. Let $\beta \geq 2\Delta$. For every $\cE_{\mathrm{run},Z}^{\cC}(\eta,\gamma,\omega;\beta)$-avoiding set $E\subseteq C^i$, then the set
  \begin{equation*}
    \bar{E} = \Shad^{\times r_X}(E) \cup (C^i\setminus(V_Z\cup\cE_{Z,\Rd}^{\cC}))
  \end{equation*}
  is $\cE_Z^{\cC}(\eta,\gamma,\omega;\beta)$-avoiding. Furthermore, for every $e\in\cC^i$ with $\supp(e)\subseteq E$, then $\supp(\Rdn(e))\subseteq\bar{E}$.
\end{lemma}
\begin{proof}
  By \Cref{claim:rfsb},
  \begin{equation*}
    % \label{eq:rfbsmall}
    |\Shad^{\times r_X}(E)|_\beta \leq |E|_\beta < \omega,
  \end{equation*}
  where the final inequality above holds by the assumption that $E$ is $\cE_{\mathrm{run},Z}^{\cC}(\eta,\gamma,\omega;\beta)$-avoiding. This assumption that $E$ is $\cE_{\mathrm{run},Z}^{\cC}(\eta,\gamma,\omega;\beta)$-avoiding also implies that $\Shad^{\times r_X}(E)$ does not contain any set $E'\in\cE_Z^{\cC}(\eta,\gamma)$, that is, $\Shad^{\times r_X}(E)$ is $\cE_Z^{\cC}(\eta,\gamma)$-avoiding. Then as $\Shad^{\times r_X}(E)\subseteq V_Z$ by definition and $V_Z\cap\cE_{Z,\Rd}^{\cC}=\emptyset$ by \Cref{claim:redcapVZ}, we have shown that $\Shad^{\times r_X}(E)$ is
  \begin{equation*}
    \cE_Z^{\cC}(\eta,\gamma) \cup (2^{V_Z}|_\beta^{\geq\omega}) \cup \cE_{Z,\Rd}^{\cC}=\cE_Z^{\cC}(\eta,\gamma,\omega;\beta)
  \end{equation*}
  avoiding. Furthermore, because $\cE_Z^{\cC}(\eta,\gamma,\omega;\beta)\subseteq 2^{V_Z\cup\cE_{Z,\Rd}^{\cC}}$ (see \Cref{def:badfams,remark:badred}), it follows that no element of $C^i\setminus(V_Z\cup\cE_{Z,\Rd}^{\cC})$ lies in any bad set in $\cE_Z^{\cC}(\eta,\gamma,\omega;\beta)$, so $\bar{E}$ is also $\cE_Z^{\cC}(\eta,\gamma,\omega;\beta)$-avoiding.

  It remains to be shown that every $e\in\cC^i$ with $\supp(e)\subseteq E$ has $\supp(\Rdn(e))\subseteq\bar{E}$. But this claim holds because \Cref{claim:rfsuppred} implies that every element of $\supp(\Rdn(e))\cap V_Z$ lies in $\Shad^{\times r_X}(E)\subseteq\bar{E}$, and \Cref{lem:redsupp} implies that every element of $\supp(\Rdn(e))\setminus V_Z$ lies in $C^i\setminus(V_Z\cup\cE_{Z,\Rd}^{\cC})\subseteq\bar{E}$.
\end{proof}

To analyze our state preparation gadget in \Cref{sec:stateprep}, we will need \Cref{lem:redpropfault} below, which provides a variant of \Cref{lem:redfault} that is robust to propagation of errors through a syndrome extraction circuit.
% We therefore first define the propagation of an error set:

% \begin{definition}
%   Let $\cC^*=\cC^*(\epsilon,\bar{\ell},r_X,r_Z)$. For a set $E\subseteq C^i$, 
% \end{definition}

\begin{lemma}
  \label{lem:redpropfault}
  Let $\cC^*=\cC^*(\epsilon,\bar{\ell},r_X,r_Z)$, let $r=r_X+r_Z$, and let $\Delta=\Delta(\epsilon)$ be defined as in \Cref{lem:lossless}. Let $\beta \geq 2\Delta$. For every $\eta,\gamma,\omega>0$ and every $\cE_{\mathrm{run},Z}^{\cC}(\eta,\gamma,\omega;\beta)$-avoiding set $E\subseteq C^i$, let
  \begin{equation*}
    E' = \{c'\in C^{i+1}:\exists c\in E\text{ with }c\triangleleft c'\},
  \end{equation*}
  and let
  \begin{equation*}
    \eta' = \eta, \hspace{1em} \gamma' = 16r\Delta^2 \cdot \gamma, \hspace{1em} \omega' = 4r\beta\Delta \cdot \omega.
  \end{equation*}
  Then the set
  \begin{equation*}
    \bar{E}' = \Shad^{\times r_X}(E') \cup (C^{i+1}\setminus(V_Z\cup\cE_{Z,\Rd}^{\cC}))
  \end{equation*}
  is $\cE_Z^{\cC}(\eta',\gamma',\omega';\beta)$-avoiding. Furthermore, for every $e'\in\cC^{i+1}$ with $\supp(e')\subseteq E'$, then $\supp(\Rdn(e'))\subseteq\bar{E}'$.
\end{lemma}
\begin{proof}
  The fact that every $e'\in\cC^{i+1}$ with $\supp(e')\in E'$ satisfies $\supp(\Rdn(e'))\subseteq\bar{E}'$ follows immediately from \Cref{claim:rfsuppred,lem:redsupp}, analogously as in the proof of \Cref{lem:redfault} above.
  
  Because $\cE_Z^{\cC}(\eta,\gamma,\omega;\beta)\subseteq 2^{V_Z\cup\cE_{Z,\Rd}^{\cC}}$ (see \Cref{def:badfams,remark:badred}), no element of $C^{i+1}\setminus(V_Z\cup\cE_{Z,\Rd}^{\cC})$ lies in any bad set in $\cE_Z^{\cC}(\eta',\gamma',\omega';\beta)$. Thus to show that $\bar{E}'$ is $\cE_Z^{\cC}(\eta',\gamma',\omega';\beta)$-avoiding, it suffices to show that $\Shad^{\times r_X}(E')$ is $\cE_Z^{\cC}(\eta',\gamma',\omega';\beta)$-avoiding.
  
  % Let $V_Z=(V_{\bar{\ell}})^{r_X}\times C(\epsilon,\bar{\ell},0,r_Z)$. If $\supp(e')\subseteq E'$, then \Cref{claim:rfsuppred} implies that $\supp(\Rdn(e')|_{V_Z})\subseteq\Shad^{\times r_X}(E')$. Hence to show that $\supp(\Rdn(e'))$ is $\cE_Z^{\cC}(\eta',\gamma',\omega';\beta)$-avoiding, by \Cref{remark:badred} it suffices to show that $\Shad^{\times r_X}(E')$ is $\cE_Z^{\cC}(\eta',\gamma',\omega';\beta)$-avoiding.

  For this purpose, by definition $\Shad^{\times r_X}(E')\subseteq V_Z$, and by \Cref{claim:redcapVZ} $V_Z\cap\cE_{Z,\Rd}^{\cC}=\emptyset$, so $\Shad^{\times r_X}(E')\cap\cE_{Z,\Rd}^{\cC}=\emptyset$.

  Furthermore, for each $c\in E$ there are at most $r(2\Delta+1)$ distinct $c'\triangleright c$, each of which lies at level $\Lev(c')\leq\Lev(c)+1$. Therefore
  \begin{align*}
    |\Shad^{\times r_X}(E')|_\beta \leq |E'|_\beta \leq r(2\Delta+1)\beta|E|_\beta < r(2\Delta+1)\beta\omega \leq \omega',
  \end{align*}
  where the first inequality above holds by \Cref{claim:rfsb} because $\beta\geq 2\Delta$.

  Thus it remains to be shown that $\Shad^{\times r_X}(E')$ is $\cE_Z^{\cC}(\eta',\gamma')$-avoiding. For some $v=(v_1,\dots,v_{r_X})\in(V_{\bar{\ell}})^{r_X}$, let $V'\subseteq\{v\}\times C(\epsilon,\bar{\ell},0,r_Z)\subseteq C$ be a set of size $|V'|\geq\eta'$ that induces a connected subgraph of $G^{C(\epsilon,\bar{\ell},0,r_Z)}$. Our goal is to show that
  \begin{equation}
    \label{eq:rpfgoal}
    |V'\cap\Shad^{\times r_X}(E')|\geq\gamma'|V'|.
  \end{equation}
  Define the set $V\subseteq\{v\}\times C(\epsilon,\bar{\ell},0,r_Z)$ by
  \begin{equation*}
    V = V' \cup \{c\in\{v\}\times C(\epsilon,\bar{\ell},0,r_Z):\exists c'\in V'\text{ with }c=c'\text{ or }c\triangleleft c'\},
  \end{equation*}
  so that
  \begin{equation}
    \label{eq:rpfVp}
    \eta \leq |V'| \leq |V| \leq r_Z(2\Delta+1)|V'|.
  \end{equation}
  This set $V$ induces a connected subgraph of $G^{C(\epsilon,\bar{\ell},0,r_Z)}$, as $V'$ induces a connected subgraph of $G^{C(\epsilon,\bar{\ell},0,r_Z)}$, and every element of $V$ is by definition equal or adjacent to some element of $V'$. The following key claim will allow us to apply the assumption that $E$ is $\cE_{\mathrm{run},Z}^{\cC}(\eta,\gamma)$-avoiding.

  \begin{claim}
    \label{claim:rpfkey}
    For every $c'\in V'\cap\Shad^{\times r_X}(E')$, there exists $c\in V\cap\Shad^{\times r_X}(E)$ such that either $c=c'$ or with $c\triangleleft c'$.
  \end{claim}
  \begin{proof}
    By definition there exists $a'\in E'$ with $c'\in\Shad^{\times r_X}(a')$, and then there exists $a\in E$ with $a\triangleleft a'$. Writing $a=(a_1,\dots,a_r)$ and $a'=(a'_1,\dots,a'_r)\in C(\epsilon,\bar{\ell})^{\times r_X}\times C(\epsilon,\bar{\ell})^{\times r_Z}$, then there exists a unique $j\in[r]$ such that $a_j\triangleleft a'_j$, and $a_{j'}=a'_{j'}$ for every $j'\in[r]\setminus\{j\}$. If $j'\leq r_X$, then by definition $c'\in\Shad^{\times r_X}(a)$, so we can let $c=c'$. If instead $j'>r_X$, then as we must have $c'=(v_1,\dots,v_{r_X},a'_{r_X+1},\dots,a'_r)\in\Shad^{\times r_X}(a')$, it follows that $c=(v_1,\dots,v_{r_X},a_{r_X+1},\dots,a_r)$ must lie in $\Shad^{\times r_X}(a)$, as desired.
  \end{proof}

  \Cref{claim:rpfkey} implies that
  \begin{equation*}
    |V\cap\Shad^{\times r_X}(E)| \geq \frac{|V'\cap\Shad^{\times r_X}(E')|}{\Delta+2},
  \end{equation*}
  as for every $c\in V\cap\Shad^{\times r_X}(E)$, there are at most $\Delta+2$ distinct $c'\in V'\cap\Shad^{\times r_X}(E')$ such that $c=c'$ or $c\triangleleft c'$. However, because $E$ is $\cE_{\mathrm{run},Z}^{\cC}(\eta,\gamma)$-avoiding so that $\Shad^{\times r_X}(E)$ is $\cE_Z^{\cC}(\eta,\gamma)$-avoiding (see \Cref{def:badfault}), and because $|V|\geq\eta$ by \Cref{eq:rpfVp}, we must have by \Cref{eq:rpfVp} that
  \begin{equation*}
    |V\cap\Shad^{\times r_X}(E)| < \gamma|V| \leq \gamma r_Z(2\Delta+1)|V'|.
  \end{equation*}
  The two inequalities above together imply that
  \begin{equation*}
    |V'\cap\Shad^{\times r_X}(E')| < \gamma r_Z(2\Delta+1)(\Delta+2)|V'| \leq \gamma'|V'|,
  \end{equation*}
  so \Cref{eq:rpfgoal} holds, as desired.
\end{proof}

\subsubsection{Behavior Under Locally Stochastic Noise}
\label{sec:bfls}
In this section, we show that the families of bad sets defined in \Cref{sec:bfdef} reflect the behavior of appropriate locally stochastic noise.
Specifically, \Cref{lem:lsfault} below shows that a locally stochastic fault that corrupts each level-$\ell$ qubit with probability $p^{2^{\Omega(\ell)}}$ will be $\cE_{\mathrm{run}}=\cE_{\mathrm{run}}^{\cC}(\Omega(2^{\bar{\ell}/2}),\Omega(1),\beta^{\Omega(\bar{\ell})};\beta)$-avoiding with high probability $1-2^{-2^{\Omega(\bar{\ell})}}$, where $p>0$ is a small constant. Hence throughout \Cref{sec:ftnonu} we consider such $\cE_{\mathrm{run}}$-avoiding faults, which with high probability capture the behavior of locally stochastic noise with non-uniform error probabilities. In Appendix~\ref{sec:concat}, we show how to concatenate our gadgets from \Cref{sec:ftnonu} with an ``inner'' fault-tolerance scheme, in order to obtain fault-tolerance against locally stochastic noise with uniform error probabilities.

% \begin{definition}
%   We say a sub-family $\cE_{\mathrm{run}}'\subseteq\cE_{\mathrm{run}}$ \emph{generates} $\cE_{\mathrm{run}}$ if for every $E\in\cE_{\mathrm{run}}$, there exists $E'\in\cE_{\mathrm{run}}'$ with $E'\subseteq E$.
% \end{definition}

\begin{lemma}
  \label{lem:lsfault}
  For integers $r_X,r_Z\geq 0$, $\bar{\ell}\geq 1$, let $r=r_X+r_Z$, define $\epsilon=\epsilon(r)$ as in \Cref{lem:ssflip}, and define $\mu=\mu(\epsilon)$ and $\Delta=\Delta(\epsilon)$ as in \Cref{lem:lossless}. Define $\cC^*=\cC^*(\epsilon,\bar{\ell},r_X,r_Z)$ as in \Cref{def:prodcode}. Fix some $\lambda>0$, which may be an arbitrarily small constant. We treat $r_X,r_Z,\lambda$ and therefore also $r,\epsilon,\mu,\Delta$ as fixed constants, while letting $\bar{\ell}\rightarrow\infty$ in the big-$O$ notation below.

  Then there exists a sufficiently large constant $\beta_0=\beta_0(r_X,r_Z,\lambda)\geq 2\Delta$ and a sufficiently small constant $p_0=p_0(r_X,r_Z,\lambda)>0$ such that the following holds for every fixed $\beta\geq\beta_0$ and $0<p\leq p_0$. Let $E_{\mathrm{run}}\subseteq C$ be sampled from some probability distribution with the property that for every $S\subseteq C$,
  \begin{equation*}
    \Pr[S\subseteq E_{\mathrm{run}}] \leq \prod_{a\in S}p^{2^{\lambda\cdot\Lev(a)}}.
  \end{equation*}
  Then for every $\eta,\gamma,\omega>0$ satisfying
  \begin{equation*}
    \eta\geq\lambda\cdot 2^{\bar{\ell}/2}, \hspace{1em} \gamma\geq\lambda, \hspace{1em} \omega\geq\beta^{\lambda\cdot\bar{\ell}},
  \end{equation*}
  letting $\cE_{\mathrm{run}}=\cE_{\mathrm{run}}^{\cC}(\eta,\gamma,\omega;\beta)$ as in \Cref{def:badfault}, we have
  \begin{equation*}
    \Pr[E_{\mathrm{run}}\text{ is not }\cE_{\mathrm{run}}\text{-avoiding}] \leq 2^{-2^{\Omega(\bar{\ell})}}.
  \end{equation*}
\end{lemma}
\begin{proof}
  By the definition of $\cE_{\mathrm{run}}^{\cC}(\eta,\gamma,\omega;\beta)$ in \Cref{eq:ErunXZ}, it suffices to show that the probability that $E_{\mathrm{run}}$ is not $\cE_{\mathrm{run},Z}^{\cC}(\eta,\gamma,\omega;\beta)$-avoiding is $2^{-2^{\Omega(\bar{\ell})}}$, as analogous reasoning then implies the same bound for $\cE_{\mathrm{run},X}^{\cC}(\eta,\gamma,\omega;\beta)$. By \Cref{eq:Erundef}, it in turn suffices to show the same probability bound for the event that $E_{\mathrm{run}}$ is not $\cE_{\mathrm{run},Z}^{\cC}(\eta,\gamma)$-avoiding, as well as the event that $E_{\mathrm{run}}$ is not $2^C|_{\beta}^{\geq\omega}$-avoiding.

  For the latter, if $E_{\mathrm{run}}$ is not $2^C|_{\beta}^{\geq\omega}$-avoiding, then $|E_{\mathrm{run}}|_\beta\geq\omega\geq\beta^{\lambda\cdot\bar{\ell}}$. Thus because $|C|=O(2^{\bar{\ell}\cdot r})=2^{O(\bar{\ell})}$, for sufficiently large $\beta$ we have $|C|\leq\beta^{(\lambda/2)\cdot\bar{\ell}}$. Therefore $|E_{\mathrm{run}}|_\beta\geq\beta^{(\lambda/2)\cdot\bar{\ell}}\cdot|C|$, so then $E_{\mathrm{run}}$ must contain some $a\in C$ at level $\Lev(a)\geq(\lambda/2)\cdot\bar{\ell}$. But each such $a$ by definition lies in $E_{\mathrm{run}}$ with probability $\leq p^{2^{(\lambda^2/2)\cdot\bar{\ell}}}$, so
  \begin{align*}
    \Pr[E_{\mathrm{run}}\text{ is not }2^C|_{\beta}^{\geq\omega}\text{-avoiding}]
    &\leq |C|\cdot p^{2^{(\lambda^2/2)\cdot\bar{\ell}}} \leq 2^{-2^{\Omega(\bar{\ell})}},
  \end{align*}
  where the final inequality above holds assuming $p<1$ because $|C|=O(2^{\bar{\ell}\cdot r})$.

  It remains to bound the probability that $E_{\mathrm{run}}$ is not $\cE_{\mathrm{run},Z}^{\cC}(\eta,\gamma)$-avoiding. For this purpose, recall from \Cref{def:badfams} that
  \begin{equation*}
    \cE_{\mathrm{run},Z}^{\cC}(\eta,\gamma) = \cE_{\mathrm{run},Z}^{\cC(\epsilon,\bar{\ell},0,r_Z)}(\eta,\gamma)^{\sqcup V_{\bar{\ell}}^{\times r_X}},
  \end{equation*}
  where
  \begin{equation*}
    \cE_{\mathrm{run},Z}^{\cC(\epsilon,\bar{\ell},0,r_Z)}(\eta,\gamma) = \cE(G^{\cC(\epsilon,\bar{\ell},0,r_Z)},\eta,\gamma),
  \end{equation*}
  with the RHS above defined as in \Cref{def:clusterfam,def:conngraph}. Recall that every set in $\cE(G^{\cC(\epsilon,\bar{\ell},0,r_Z)},\eta,\gamma)$ consists of $\geq\gamma$-fraction of the vertices in some $\geq\eta$-vertex connected subgraph $G'$ of the constant-degree graph $G^{\cC(\epsilon,\bar{\ell},0,r_Z)}$. By passing to smaller subgraphs if necessary, it follows that every set in $\cE(G^{\cC(\epsilon,\bar{\ell},0,r_Z)},\eta,\gamma)$ contains some set $E'\in\cE(G^{\cC(\epsilon,\bar{\ell},0,r_Z)},\eta,\gamma)$ of size $|E'|=\Theta(\eta)=\Theta(2^{\bar{\ell}/2})$. There are a total of $2^{O(2^{\bar{\ell}/2})}$ such sets $E'$, as $G^{\cC(\epsilon,\bar{\ell},0,r_Z)}$ is a constant-degree graph on $O(2^{\bar{\ell}\cdot r_Z})$ vertices, so it has at most $O(2^{\bar{\ell}\cdot r_Z})\cdot 2^{O(2^{\bar{\ell}/2})}=2^{O(2^{\bar{\ell}/2})}$ connected subgraphs of size $\Theta(2^{\bar{\ell}/2})$. For each such set $E'$, we have
  \begin{equation*}
    \Pr[E'\subseteq\Shad^{\times r_X}(E_{\mathrm{run}})] \leq \left(\sum_{\ell=0}^{\bar{\ell}}2^{O(\ell)}\cdot p^{2^{\lambda\cdot\ell}}\right)^{|E'|},
  \end{equation*}
  as for every $a\in E'$, there are at most $2^{O(\ell)}$ elements $a'\in C$ at level $\Lev(a')=\ell$ with $a\in\Shad^{\times r_X}(a')$. Union bounding over all $2^{O(2^{\bar{\ell}/2})}$ choices of $E'$, and using the fact that $|E'|=\Theta(2^{\bar{\ell}/2})$, we conclude that
  \begin{align*}
    \Pr[E_{\mathrm{run}}\text{ is not }\cE_{\mathrm{run},Z}^{\cC}(\eta,\gamma)\text{-avoiding}]
    &\leq 2^{O(2^{\bar{\ell}/2})} \cdot \left(\sum_{\ell=0}^{\bar{\ell}}2^{O(\ell)}\cdot p^{2^{\lambda\cdot\ell}}\right)^{\Omega(2^{\bar{\ell}/2})}.
  \end{align*}
  Assuming $p>0$ is a sufficiently small constant, the RHS above is at most $2^{-\Omega(2^{\bar{\ell}/2})}$, as desired.
\end{proof}

For our ejection and injection gadgets, we will in fact want to take $\omega=\lambda\cdot|C|$ for an arbitrarily small constant $\lambda>0$, so that only a small constant fraction of the ejected/injected qubits are corrupted. In \Cref{lem:lsex} below, we therefore show that a locally stochastic fault that corrupts each level-$\ell$ qubit with probability $p^{2^{\Omega(\ell)}}$ will in fact be $\cE_{\mathrm{run}}=\cE_{\mathrm{run}}^{\cC}(\Omega(2^{\bar{\ell}/2}),\Omega(1),\lambda\cdot|C|;\beta)$-avoiding with probability $\geq 1-p/\lambda$, where we can choose $p>0$ to be an arbitrarily small constant depending on $\lambda$. Thus with all but small constant probability, our injection and ejection gadgets will preserve all but a small constant fraction of qubits. It in particular follows that all but a small constant fraction of qubits have just a small constant fraction marginal probability of corruption during ejection and injection.

\begin{lemma}
  \label{lem:lsex}
  Define $r_X,r_Z,\bar{\ell},r,\epsilon,\mu,\Delta,\cC,\lambda$ as in \Cref{lem:lsfault}. Then there exists a sufficiently large constant take $\beta_0=\beta_0(r_X,r_Z,\lambda)\geq 2\Delta$ such that for every fixed $\beta\geq\beta_0$, there exists a sufficiently small constant $p_0=p_0(r_X,r_Z,\lambda,\beta)>0$ such that the following holds for every $0<p\leq p_0$.
  Define $E_{\mathrm{run}}$ as in \Cref{lem:lsfault}.
  Then for every $\eta,\gamma,\omega>0$ satisfying
  \begin{equation*}
    \eta\geq\lambda\cdot 2^{\bar{\ell}/2}, \hspace{1em} \gamma\geq\lambda, \hspace{1em} \omega\geq\lambda\cdot|C|,
  \end{equation*}
  letting $\cE_{\mathrm{run}}=\cE_{\mathrm{run}}^{\cC}(\eta,\gamma,\omega;\beta)$ as in \Cref{def:badfault}, we have
  \begin{equation*}
    \Pr[E_{\mathrm{run}}\text{ is not }\cE_{\mathrm{run}}\text{-avoiding}] \leq p/\lambda,
  \end{equation*}
  assuming $\bar{\ell}\geq\bar{\ell}_0$ for some sufficiently large constant $\bar{\ell}_0=\bar{\ell}_0(r_X,r_Z,\lambda,p,\beta)$.
\end{lemma}
\begin{proof}
  As in the proof of \Cref{lem:lsfault}, it suffices to bound the probability that $E_{\mathrm{run}}$ is not $\cE_{\mathrm{run},Z}^{\cC}(\eta,\gamma)$-avoiding, as well as the probability that $E_{\mathrm{run}}$ is not $2^C|_{\beta}^{\geq\omega}$-avoiding. The proof of \Cref{lem:lsfault} already bounds the former assuming $\bar{\ell}$ is at least some sufficiently large constant $\bar{\ell}_0$, so it remains to bound the latter. For this purpose, we will simply apply Markov's inequality. Specifically, by definition we have
  \begin{align*}
    \bE[|E_{\mathrm{run}}|_\beta]
    &= \sum_{a\in C}Pr[a\in E_{\mathrm{run}}]\cdot\beta^{\Lev(a)} \leq \sum_{a\in C} p^{2^{\lambda\cdot\Lev(a)}}\cdot\beta^{\Lev(a)}.
  \end{align*}
  Then assuming $p>0$ is sufficiently small relative to $\beta$, we have $p^{2^{\lambda\cdot\Lev(a)}}\cdot\beta^{\Lev(a)}\leq p$, so
  \begin{equation*}
    \bE[|E_{\mathrm{run}}|_\beta] \leq p\cdot|C|.
  \end{equation*}
  Therefore
  \begin{align*}
    \Pr[|E_{\mathrm{run}}|_\beta\geq\omega]
    &\leq \frac{p\cdot|C|}{\omega} \leq \frac{p}{\lambda},
  \end{align*}
  as desired.
\end{proof}

\subsection{Core Gadgets}
\label{sec:coregad}
In this section, we present our fault-tolerant state preparation, logical measurement, and ejection gadgets for our codes from \Cref{sec:quantumprod}. We view these gadgets as our ``core'' gadgets because we compose them, along with the basic gadgets in \Cref{sec:basicgad} below, in a black-box manner to construct our error-correction and injection gadgets.

Specifically, by running our logical $\ket{0}$ and $\ket{+}$ state preparation gadgets followed by transversal $\gCNOT$ gates, we can construct logical Bell pairs across two code blocks. Then given as input an arbitrary code state with an appropriately bounded corruption, we can perform transversal $\gCNOT$ gates followed by our logical measurement gadget on one half of the Bell pairs, in order to teleport the input state into the other half of the Bell pairs. This teleportation into a fresh code block naturally performs error correction, as we completely discard the input code block.

Meanwhile, to inject a given set of bare physical qubits into our code, we again begin by preparing logical Bell pairs across two code blocks. However, we now eject one half of the Bell pairs out of the code block into bare physical form. Using these bare physical qubits, we can then teleport the input qubits into the remaining code block containing the other half of the Bell pairs.

\subsubsection{State Preparation}
\label{sec:stateprep}
In this section, we present our state-preparation gadget for the codes in \Cref{sec:quantumprod}. This gadget shares similarities with the state preparation gadget in \cite{golowich_constant-overhead_2025}, though here we use our specialized decoder from \Cref{sec:decoder}.

\begin{lemma}
  \label{lem:stateprep}
  For $r_X,r_Z,\bar{\ell}\in\bN$ with $r_Z\geq 2$, let $r=r_X+r_Z$, define $\epsilon=\epsilon(r)$ as in \Cref{lem:ssflip}, and define $\mu=\mu(\epsilon)$ and $\Delta=\Delta(\epsilon)$ as in \Cref{lem:lossless}. Define $\cC^*=\cC^*(\epsilon,\bar{\ell},r_X,r_Z)$, $\cM^*=\cM^*(\epsilon,\bar{\ell},r_X,r_Z)$, and $\Enc_{\cC}:\cM^*\rightarrow\cC^*$ as in \Cref{def:prodcode}. Let $Q=(Q_X,Q_Z)$ be the $[[n,k]]$ code at level $r_X$ of $\cC^*$.

  Let $\beta\geq 2\Delta$, and for arbitrary positive
  \begin{equation}
    \label{eq:sprunbounds}
    \eta_{\mathrm{run}} \leq \mu 2^{\bar{\ell}/2}, \hspace{1em} \gamma_{\mathrm{run}} \leq \frac{1}{2^{12}r^3\Delta^5}, \hspace{1em} \omega_{\mathrm{run}} \leq \frac{\beta^{\bar{\ell}/2}}{2^9r^3\beta^2\Delta^4},
  \end{equation}
  let
  \begin{equation}
    \label{eq:spoutbounds}
    \eta_{\mathrm{out}} = \eta_{\mathrm{run}}, \hspace{1em} \gamma_{\mathrm{out}} = 2^{14}r^3\Delta^5\cdot\gamma_{\mathrm{run}}, \hspace{1em} \omega_{\mathrm{out}} = 2^{12}r^3\beta^2\Delta^4\cdot\omega_{\mathrm{run}},
  \end{equation}
  and let
  \begin{align*}
    \cE_{\mathrm{run}} &= \cE_{\mathrm{run}}^{\cC}(\eta_{\mathrm{run}},\gamma_{\mathrm{run}},\omega_{\mathrm{run}};\beta), \\
    \cE_{\mathrm{out}} &= \cE^{\cC}(\eta_{\mathrm{out}},\gamma_{\mathrm{out}},\omega_{\mathrm{out}};\beta)
  \end{align*}
  be the families of bad sets defined in \Cref{def:badfault}, \Cref{def:badfams}, respectively. Define the decorated code
  \begin{equation*}
    D_{\mathrm{out}} = (Q,\; \Enc_{\cC}^{r_X},\; \cE_{\mathrm{out}}).
  \end{equation*}

  Then there exists a fault-tolerant gadget $\bfG_X=(\cR_X,\cE_{\mathrm{run}}^{\sqcup T},\emptyset,D_{\mathrm{out}},\ps)$ for the channel $\bar{O}_X:\bC\rightarrow\bC^{2^k\times 2^k}$ given by $\bar{O}_X(1)=\ket{+^k}\bra{+^k}$, where $\cR_X$ is a quantum circuit using quantum space $\Qu{N}=C^{r_X}\sqcup C^{r_X+1}$ (so $|\Qu{N}|\leq 4r\Delta n$), time $T=r(2\Delta+2)+6\leq 8r\Delta$, and gate set $\cG_X=\{\gInitX,\gInitZ,\gTerm,\gCNOT,\gMZ,\gCt{\gX},\gCF_*\}$.

  Similarly, there exists a fault-tolerant gadget $\bfG_Z=(\cR_Z,\cE_{\mathrm{run}}^{\sqcup T},\emptyset,D_{\mathrm{out}},\ps)$ for the channel $\bar{O}_Z:\bC\rightarrow\bC^{2^k\times 2^k}$ given by $\bar{O}_Z(1)=\ket{0^k}\bra{0^k}$, where $\cR_Z$ is a quantum circuit using quantum space $\Qu{N}=C_{r_X}\sqcup C_{r_X-1}$ (so $|\Qu{N}|\leq 4r\Delta n$), time $T=r(2\Delta+2)+6\leq 8r\Delta$, and gate set $\cG_Z=\{\gInitX,\gInitZ,\gTerm,\gCNOT,\gMX,\gCt{\gZ},\gCF_*\}$.
\end{lemma}

To prove \Cref{lem:stateprep}, we follow prior works (including \cite{golowich_constant-overhead_2025}) in using the following well-known result providing edge-coloring on bounded-degree graphs, in order to obtain a bounded-depth circuit that implements syndrome extraction for our codes.

\begin{lemma}[Vizing's theorem (e.g.~\cite{misra_constructive_1992})]
  \label{lem:vizing}
  For every graph $G=(V,E)$ of maximum degree $\Delta$, there exists a $O(|V|\cdot|E|)$-time algorithm that computes a coloring of the edges with $\Delta+1$ colors, such that no two edges sharing a vertex have the same color.
\end{lemma}

% TODO: should there also be a state prep statement for the case where we're going to eject, so that $\omega$ is some small constant times the number of qubits? or perhaps that will follow from current version of the statement?

% We apply Vizing's theorem in \Cref{alg:stateprep}, which gives our state preparation circuit $\cR$ in \Cref{lem:stateprep}.

\begin{algorithm}
  \caption{\label{alg:stateprep} State preparation circuit $\cR_X$ in \Cref{lem:stateprep} for the code at level $r_X$ of $\cC^*=\cC^*(\epsilon,\bar{\ell},r_X,r_Z)$. By definition, the graph $G$ we define in \Cref{li:spgraph} has bipartite adjacency matrix given by the coboundary map $\delta^{\cC}_{r_X}$. The classical computations in \Cref{li:spgraph,li:spdecode,li:spge} are performed using classical function gates (see \Cref{def:gates}). The inner for loop in \Cref{li:spforin} can be run in a single timestep by the definition of the edge coloring given by \Cref{lem:vizing}.}
  \SetKwInOut{Input}{Input}
  \SetKwInOut{Output}{Output}
  \SetKwProg{Fn}{Function}{:}{}

  \SetKwFunction{FnStatePrep}{StatePrep}

  \Input{None}
  \Output{$\Enc_{\cC}^{r_X}(\ket{+^k}\bra{+^k})$}
  \Fn{\FnStatePrep{$\cC^*$}}{
    Run $\gInitX^{\otimes C^{r_X}}$ to initialize a block of $n=|C^{r_X}|$ data qubits to $\ket{+^{C^{r_X}}}\bra{+^{C^{r_X}}}$ \\
    Run $\gInitZ^{\otimes C^{r_X+1}}$ to initialize a block of $|C^{r_X+1}|$ ancilla qubits to $\ket{0^{C^{r_X+1}}}\bra{0^{C^{r_X+1}}}$ \\
    Letting $G=(C^{r_X}\sqcup C^{r_X+1},E(G))$ be the bipartite graph with an edge between every $c\in C^{r_X},c'\in C^{r_X+1}$ satisfying $c\triangleleft c'$, compute an edge coloring $\chi:E(G)\rightarrow[r(2\Delta+2)]$ using \Cref{lem:vizing} \\ \label{li:spgraph}
    \For{$i=1,\dots,r(2\Delta+2)$}{ \label{li:spforout}
      \For{$(c,c')\in E(G)$ with $\chi(c,c')=i$}{ \label{li:spforin}
        Apply $\gCNOT$ to data qubit $c$ and ancilla qubit $c'$ \label{li:spcnot}
      }
    }
    Run measurements $\gMZ^{\otimes C^{r_X+1}}$ on the ancilla block, and let $s\in\bF_2^{C^{r_X+1}}$ be the outcome \\ \label{li:spmeasure}
    Let $\tilde{a}^{r_X+1}\gets\FnDecode{$\delta^{\cC}(s);\cC^*$}$ be the output from running \Cref{alg:decode} on $\delta^{\cC}(s)$ \\ \label{li:spdecode}
    Run Gaussian elimination to find some $x\in\cC^{r_X}$ with $\delta^{\cC}(x)=s+\tilde{a}^{r_X+1}$ \\ \label{li:spge}
    Run classically-controlled gates $\gCt{\gX}^{\otimes C^{r_X}}$ to apply $\gX^x$ to the data qubits \\ \label{li:spcorr}
    \Return{the data qubits} \label{li:spreturn}
  }
\end{algorithm}

\begin{proof}[Proof of \Cref{lem:stateprep}]
  We will present and analyze $\bfG_X$. The construction and analysis of $\bfG_Z$ is exactly analogous, except that the argument is applied to the dual complex ${\cC^{\vee}}^*=\cC_*\cong\cC^*(\epsilon,\bar{\ell},r_Z,r_X)$. Hence in this dualized argument, we use qubits labeled by $C_{r_X}\sqcup C_{r_X-1}$ instead of $C^{r_X}\sqcup C^{r_X+1}$, and all quantum gates are conjugated by Hadamards, so all instances of $\gX$ and $\gZ$ are swapped, and all $\gCNOT$ gates are applied with control and target qubits reversed. We omit the details to avoid redundancy.
  
  Our desired state-preparation circuit $\cR_X$ is given in \Cref{alg:stateprep}. This circuit by definition uses quantum space $\Qu{N}=C^{r_X}\sqcup C^{r_X+1}$, so that $|\Qu{N}|\leq n+r(2\Delta+1)n\leq 4r\Delta n$, and uses time $T=r(2\Delta+2)+6\leq 8r\Delta$ and gate set $\cG_X$. We let the set $\ps$ of postselectable gates contain all measurement gates in $\cR$. These measurement gates are all performed in \Cref{li:spmeasure} of \Cref{alg:stateprep}, and are labeled by the set $\ps=C^{r_X+1}$.

  It remains to show that $\bfG_X$ forms a fault-tolerant gadget for the channel $\bar{O}_X:1\mapsto\ket{+^k}\bra{+^k}$. Fix a finite set $\Rf{K}$, a $\cE_{\mathrm{run}}^{\sqcup T}$-avoiding set $E_{\mathrm{run}}$, and a postselection\footnote{Here we denote our postselection $s$ because it represents the syndrome measurement.} $s\in\bF_2^{\ps}=\bF_2^{C^{r_X+1}}$. By \Cref{lem:paulift}, it suffices to show that there exists a $\cE_{\mathrm{out}}$-avoiding pair of sets $E_{\mathrm{out}}$ such that for every $\comp{E_{\mathrm{run}}}$-avoiding Pauli fault $\cF$ and reweighting $\zeta$ for $\cR_X$ then $\cR_X[\cF,\zeta,s](1)$ is a $\comp{E_{\mathrm{out}}}$-deviation of $\Enc_{\cC}^{r_X}(\ket{+^k}\bra{+^k})$. Because $\cR_X$ has no input qubits and $\cF$ is a Pauli fault, the reference system $\Rf{K}$ never becomes entangled with the physical qubits in $\Qu{N}$, so $\Rf{K}$ plays no role and can be ignored for the remainder of the proof.

  In the absence of a fault, then just prior to \Cref{li:spforout} in \Cref{alg:stateprep}, the qubits are in the state $(\ket{+}\bra{+})^{\otimes C^{r_X}}\otimes(\ket{0}\bra{0})^{\otimes C^{r_X+1}}$. Without a fault, the for loop in \Cref{li:spforout} then applies a $\gCNOT$ gate to every $c\in C^{r_X},c'\in C^{r_X+1}$ for which $c\triangleleft c'$, or equivalently, for which $(\delta^{\cC}_{r_X})_{c',c}=1$. Thus this for loop simply applies the map $\ket{y,z}\mapsto\ket{y,z+\delta^{\cC}_{r_X}(y)}$, so without a fault, the state just prior to the measurements in \Cref{li:spmeasure} of \Cref{alg:stateprep} is
  \begin{equation*}
    \sum_{y,y'\in C^{r_X}}\ket{y}\bra{y'}\otimes\ket{\delta^{\cC}_{r_X}(y)}\bra{\delta^{\cC}_{r_X}(y)}.
  \end{equation*}

  However, in the presence of our $\cE_{\mathrm{run}}^{\sqcup T}$-avoiding fault $\cF$, then in each of the $\leq T\leq 8r\Delta$ timesteps prior to the measurements in \Cref{li:spmeasure}, our qubits in $C^{r_X}\sqcup C^{r_X+1}$ can experience an arbitrary $\comp{E_{\mathrm{run}}}$-avoiding Pauli error. Furthermore, for every pair of qubits $c\in C^{r_X},c'\in C^{r_X+1}$ with $c\triangleleft c'$, the associated $\gCNOT$ gate in \Cref{li:spcnot} will propagate a Pauli $\gX$ error from $c$ to $c'$, and will propagate a Pauli $\gZ$ error from $c'$ to $c$.
  Thus define $E\subseteq C^{r_X}\sqcup C^{r_X+1}$ to be the set
  \begin{equation*}
    E = \bigcup_{t\in[T]}(E_{\mathrm{run}}\cap(N\times\{t\}))
  \end{equation*}
  of all qubits that may lie in the support of the fault across all timesteps, so that $E$ is $\cE_{\mathrm{run}}^{\cC}(\eta_{\mathrm{run}},T\gamma_{\mathrm{run}},T\omega_{\mathrm{run}};\beta)$-avoiding, and let
  \begin{align*}
    E'_X &= E \cup \{c'\in C^{r_X+1}:\exists c\in E\cap C^{r_X}\text{ with }c\triangleleft c'\} \\
    E'_Z &= E \cup \{c\in C^{r_X}:\exists c'\in E\cap C^{r_X+1}\text{ with }c\triangleleft c'\}.
  \end{align*}
  Then there exist $(e_X,f_X),(e_X',f_X')\in\cC^{r_X}\times\cC^{r_X+1}$ supported inside $E'_X$ and $(e_Z,f_Z),(e_Z',f_Z')\in\cC^{r_X}\times\cC^{r_X+1}$ supported inside $E'_Z$ such that the state of the qubits in $C^{r_X}\sqcup C^{r_X+1}$ just prior to the measurements in \Cref{li:spmeasure} of \Cref{alg:stateprep} is
  \begin{equation}
    \label{eq:sppremeas}
    \sum_{y,y'\in C^{r_X}}\gZ^{e_Z}\gX^{e_X}\ket{y}\bra{y'}\gX^{e_X'}\gZ^{e_Z'}\otimes\gZ^{f_Z}\gX^{f_X}\ket{\delta^{\cC}_{r_X}(y)}\bra{\delta^{\cC}_{r_X}(y')}\gX^{f_X'}\gZ^{f_Z'}.
  \end{equation}

  Furthermore, by definition $E'_X\cap C^{r_X}=E\cap C^{r_X}$ and $E'_Z\cap C^{r_X+1}=E\cap C^{r_X+1}$ are $\cE_{\mathrm{run}}^{\cC}(\eta_{\mathrm{run}},T\gamma_{\mathrm{run}},T\omega_{\mathrm{run}};\beta)$-avoiding, and hence are $\cE_{\mathrm{run}}^{\cC}(\eta_{\mathrm{run}},\; 8r\Delta\cdot\gamma_{\mathrm{run}},\; 8r\Delta\cdot\omega_{\mathrm{run}};\; \beta)$-avoiding as $T\leq 8r\Delta$.
  Therefore by \Cref{lem:redfault}, the set
  \begin{equation*}
    \bar{E}_X = \Shad^{\times r_X}(E\cap C^{r_X}) \cup (C^{r_X}\setminus(V_Z\cup\cE_{Z,\Rd}^{\cC}))
  \end{equation*}
  is $\cE_Z^{\cC}(\eta_{\mathrm{run}},\; 8r\Delta\cdot\gamma_{\mathrm{run}},\; 8r\Delta\cdot\omega_{\mathrm{run}};\; \beta)$-avoiding. Similarly, applying \Cref{lem:redfault} to the dual complex ${\cC^{\vee}}^*=\cC_*\cong\cC^*(\epsilon,\bar{\ell},r_Z,r_X)$, the set
  \begin{equation*}
    \bar{E}_Z = \Shad^{\times r_Z}(E\cap C^{r_X+1}) \cup (C^{r_X+1}\setminus(V_X\cup\cE_{X,\Rd}^{\cC}))
  \end{equation*}
  is $\cE_X^{\cC}(\eta_{\mathrm{run}},\; 8r\Delta\cdot\gamma_{\mathrm{run}},\; 8r\Delta\cdot\omega_{\mathrm{run}};\; \beta)$-avoiding.
  As a point of notation, here $\Shad^{\times r_X}$ applies $\Shad$ to the first $r_X$ components of the $r$-dimensional product space $C(\epsilon,\bar{\ell},r_X,r_Z)\cong C(\epsilon,\bar{\ell})^{\times r_X}\times C(\epsilon,\bar{\ell})^{\times r_Z}$, while $\Shad^{\times r_Z}$ applies $\Shad$ to the latter $r_Z$ components.
  Meanwhile, by \Cref{lem:redpropfault}, letting
  \begin{align}
    \label{eq:spprimed}
    \begin{split}
      \gamma' &= 16r\Delta^2\cdot T\gamma_{\mathrm{run}}+T\gamma_{\mathrm{run}} \leq 2^8r^2\Delta^3\cdot\gamma_{\mathrm{run}} \\
      \omega' &= 4r\beta\Delta\cdot T\omega_{\mathrm{run}}+T\omega_{\mathrm{run}} \leq 2^6r^2\beta\Delta^2\cdot\omega_{\mathrm{run}},
    \end{split}
  \end{align}
  then the set
  \begin{equation*}
    \bar{E}_X' = \Shad^{\times r_X}(E'_X\cap C^{r_X+1}) \cup (C^{r_X+1}\setminus(V_Z\cup\cE_{Z,\Rd}^{\cC}))
  \end{equation*}
  is $\cE_Z^{\cC}(\eta_{\mathrm{run}},\gamma',\omega';\beta)$-avoiding. Similarly, applying \Cref{lem:redpropfault} to the dual complex ${\cC^{\vee}}^*$, we also have that the set
  \begin{equation*}
    \bar{E}_Z' = \Shad^{\times r_Z}(E'_Z\cap C^{r_X}) \cup (C^{r_X}\setminus(V_X\cup\cE_{X,\Rd}^{\cC}))
  \end{equation*}
  is $\cE_X^{\cC}(\eta_{\mathrm{run}},\gamma',\omega';\beta)$-avoiding.

  Now applying the $s$-postselected measurements in \Cref{li:spmeasure} of \Cref{alg:stateprep} collapses the state in \Cref{eq:sppremeas} to
  \begin{equation}
    \label{eq:sppostmeas}
    \sum_{y,y'} \gZ^{e_Z}\gX^{e_X}\ket{y}\bra{y'}\gX^{e_X'}\gZ^{e_Z'}\otimes\gZ^{f_Z}\ket{s}\bra{s}\gZ^{f_Z'},
    % \sum_{s\in\cC^{r_X+1}} \sum_{y,y'} \gZ^{e_Z}\gX^{e_X}\ket{y}\bra{y'}\gX^{e_X'}\gZ^{e_Z'}\otimes\gZ^{f_Z}\ket{s}\bra{s}\gZ^{f_Z'},
  \end{equation}
  where the second sum above is over all $y,y'\in\cC^{r_X}$ for which $\delta^{\cC}_{r_X}(y)=s+f_X$ and $\delta^{\cC}_{r_X}(y')=s+f_X'$. If no such $y,y'$ exist, then the state collapses to $0$, and hence the algorithm outputs $0$, which is trivially a $\comp{E_{\mathrm{out}}}$-deviation of $\Enc_{\cC}^{r_X}(\ket{+^k}\bra{+^k})$ for every choice of $\comp{E_{\mathrm{out}}}$, as desired. Therefore assume that such $y,y'$ do exist. Equivalently, fixing some $y_0,y_0'\in\cC^{r_X}$ such that $\delta^{\cC}_{r_X}(y_0)=s+f_X$ and $\delta^{\cC}_{r_X}(y_0')=s+f_X'$, then the sum in \Cref{eq:sppostmeas} is over all $y\in y_0+Z^{r_X}(\cC)$ and $y'\in y_0'+Z^{r_X}(\cC)$.
  % Restricting attention to the term in the outer sum in \Cref{eq:sppostmeas} associated to a single $s\in\cC^{r_X+1}$,
  \Cref{li:spdecode} of \Cref{alg:stateprep} then computes $\tilde{a}^{r_X+1}\gets\FnDecode{$\delta^{\cC}_{r_X+1}(s);\cC^*$}$. Because $s=\delta^{\cC}_{r_X}(y_0)+f_X=\delta^{\cC}_{r_X}(y_0')+f_X'$, and because the reduction of a cochain differs from that cochain by a coboundary, we have
  \begin{equation*}
    \delta^{\cC}_{r_X+1}(s) = \delta^{\cC}_{r_X+1}(f_X) = \delta^{\cC}_{r_X+1}(\Rdn(f_X)) = \delta^{\cC}_{r_X+1}(f_X') = \delta^{\cC}_{r_X+1}(\Rdn(f_X')).
  \end{equation*}

  Now by \Cref{lem:redpropfault}, for every $f_X''\in\cC^{r_X+1}$ with $\supp(f_X'')\subseteq E_X'$, the reduction $\Rdn(f_X'')$ must be supported inside the $\cE_Z^{\cC}(\eta_{\mathrm{run}},\gamma',\omega';\beta)$-avoiding set $\bar{E}_X'$.
  % Now define $E_X''\subseteq C^{r_X+1}$ by
  % \begin{equation*}
  %   E_X'' = \bigcup_{f_X''\in\cC^{r_X+1}:\supp(f_X'')\subseteq E_X'}\supp(\Rdn(f_X'')).
  % \end{equation*}
  % Because every $\Rdn(f_X'')$ is reduced, \Cref{lem:redsupp} implies that $E_X''\cap\cE_{Z,\Rd}^{\cC}=\emptyset$. Furthermore, letting $V_Z=(V_{\bar{\ell}})^{r_X}\times C(\epsilon,\bar{\ell},0,r_Z)$, then \Cref{claim:rfsuppred} implies that $E_X''\cap V_Z\subseteq\Shad^{\times r_X}(E_X'\cap C^{r_X+1})$. As we showed above that $\Shad^{\times r_X}(E_X'\cap C^{r_X+1})$ is $\cE_Z^{\cC}(\eta_{\mathrm{run}},\gamma',\omega';\beta)$-avoiding, it follows by \Cref{def:badfams,remark:badred} that $E_X''$ is $\cE_Z^{\cC}(\eta_{\mathrm{run}},\gamma',\omega';\beta)$-avoiding.
  % % Then as $\supp(f_X),\supp(f_X')\subseteq E'_X\cap C^{r_X+1}$ and $\Shad^{\times r_X}(E'_X\cap C^{r_X+1})$ is $\cE_Z^{\cC}(\eta_{\mathrm{run}},\gamma',\omega';\beta)$-avoiding, the supports of the reductions $\Rdn(f_X),\Rdn(f_X')$ are $\cE_Z^{\cC}(\eta_{\mathrm{run}},\gamma',\omega';\beta)$-avoiding (see \Cref{claim:rfsuppred}).
  Therefore because \Cref{eq:sprunbounds,eq:spprimed} together imply
  \begin{align*}
    \eta_{\mathrm{run}} &\leq \mu 2^{\bar{\ell}/2} \\
    \gamma' &\leq 2^8r^2\Delta^3\cdot\gamma_{\mathrm{run}} \leq \frac{1}{16r\Delta^2} \\
    \omega' &\leq 2^6r^2\beta\Delta^2\cdot\omega_{\mathrm{run}} \leq \frac{\beta^{\bar{\ell}/2}}{8r\beta\Delta^2}
  \end{align*}
  so that $\eta_{\mathrm{run}},\gamma',\omega'$ satisfy the constraints in \Cref{eq:bdinerr}, we define
  \begin{align}
    \label{eq:sptilded}
    \begin{split}
      \tilde{\gamma} &= 8r\Delta^2\cdot\gamma' \leq 2^{11}r^3\Delta^5\cdot\gamma_{\mathrm{run}} \\
      \tilde{\omega} &= 8r\beta\Delta^2\cdot\omega' \leq 2^9r^3\beta^2\Delta^4\cdot\omega_{\mathrm{run}}.
    \end{split}
  \end{align}
  Then \Cref{cor:decode} implies that there exists a $\cE_Z^{\cC}(\eta_{\mathrm{run}},\tilde{\gamma},\tilde{\omega};\beta)$-avoiding set $\tilde{E}_X\subseteq C^{r_X}\cap V_Z$ depending only on $\bar{E}_X'$ (which in turn is determined by $E_{\mathrm{run}}$ via the definitions of $E$ and $E_X'$) and on $s$ such that the decoder's output $\tilde{a}^{r_X+1}\in\cC^{r_X+1}$ is of the form
  \begin{equation*}
    \tilde{a}^{r_X+1} = \Rdn(f_X)+\delta^{\cC}_{r_X}(\tilde{b}) = \Rdn(f_X')+\delta^{\cC}_{r_X}(\tilde{b}')
  \end{equation*}
  for some $\tilde{b},\tilde{b}'\in\cC^{r_X}$ that have $\supp(\tilde{b}),\supp(\tilde{b}')\subseteq\tilde{E}_X$ (so that in particular $\tilde{b},\tilde{b}'$ are reduced because they are supported inside $V_Z$).

  It follows that $y_0+\tilde{b}+\Rd(f_X)$ and $y_0'+\tilde{b}'+\Rd(f_X')$ are both valid choices for the value $x\in\cC^{r_X}$ computed in \Cref{li:spge} in \Cref{alg:stateprep}, so $y_0+\tilde{b}+\Rd(f_X)+Z^{r_X}(\cC)$ and $y_0'+\tilde{b}'+\Rd(f_X')+Z^{r_X}(\cC)$ must be the same coset, and whichever value of $x$ the algorithm chooses must lie in this coset. Hence after applying the classically-controlled gates in \Cref{li:spcorr} in \Cref{alg:stateprep} in order to apply $\gX^x$ to the state given by the term of \Cref{eq:sppostmeas} associated to measurement outcome $s$, \Cref{alg:stateprep} returns a state of the form (up to a global phase depending on $s$)
  \begin{align}
    \label{eq:spoutput}
    \begin{split}
      \hspace{1em}&\hspace{-1em} \sum_{y,y'\in Z^{r_X}(\cC)}\gZ^{\tilde{e}_Z}\gX^{\tilde{e}_X+\tilde{b}+\Rd(f_X)}\ket{y}\bra{y'}\gX^{\tilde{e}_X'+\tilde{b}'+\Rd(f_X')}\gZ^{\tilde{e}_Z'} \\
      &= \gZ^{\tilde{e}_Z}\gX^{\tilde{e}_X+\tilde{b}+\Rd(f_X)}\Enc_{\cC}^{r_X}\left(\ket{+^k}\bra{+^k}\right)\gX^{\tilde{e}_X'+\tilde{b}'+\Rd(f_X')}\gZ^{\tilde{e}_Z'}
    \end{split}
  \end{align}
  for some $\tilde{e}_X,\tilde{e}_X'\in\cC^{r_X}$ supported inside $E\cap C^{r_X}=E'_X\cap C^{r_X}$ and some $\tilde{e}_Z,\tilde{e}_Z'\in\cC^{r_X}$ supported inside $E'_Z\cap C^{r_X}$. Specifically, $\tilde{e}_X,\tilde{e}_X',\tilde{e}_Z,\tilde{e}_Z'$ differ from $e_X,e_X',e_Z,e_Z'$ only by the Pauli errors arising from the fault $\cF$ in \Cref{li:spdecode}--\Cref{li:spreturn} in \Cref{alg:stateprep}. The only gates applied during these timesteps are classical function gates (which do not touch qubits) and $\gCt{X}$ gates. Each of these $\gCt{X}$ gates only touches a single qubit, so it cannot propagate Pauli errors between qubits; the reweighting $\zeta$ on these $\gCt{X}$ gates simply introduces a global phase that depends on $x$ (and hence on $s$). Hence during these timesteps, the Pauli errors added on to $e_X,e_X',e_Z,e_Z'$ are simply those occurring in the fault $\cF$ (with no additional propagation), and hence they by definition lie inside $E\cap C^{r_X}\subseteq E'\cap C^{r_X}$.

  By \Cref{lem:redgood}, every $f_X''\in C^{r_X+1}$ has $\supp(\Rd(f_X''))\subseteq C^{r_X}\setminus(V_Z\cup\cE_{Z,\Rd}^{\cC})$. Recall that by \Cref{def:badfams,remark:badred}, every set in $\cE_Z^{\cC}(\eta_{\mathrm{run}},\tilde{\gamma},\tilde{\omega};\beta)$ is a subset of $V_Z\cup\cE_{Z,\Rd}^{\cC}$. Hence because we showed above that $\tilde{E}_X$ is $\cE_Z^{\cC}(\eta_{\mathrm{run}},\tilde{\gamma},\tilde{\omega};\beta)$-avoiding, it follows that $\tilde{E}_X\cup(C^{r_X}\setminus(V_Z\cup\cE_{Z,\Rd}^{\cC}))$ is also $\cE_Z^{\cC}(\eta_{\mathrm{run}},\tilde{\gamma},\tilde{\omega};\beta)$-avoiding.
  We showed above that $\tilde{b},\tilde{b}'$ are supported inside $\tilde{E}_X$, so it follows that $\tilde{b}+\Rd(f_X),\tilde{b}'+\Rd(f_X')$ are supported inside $\tilde{E}_X\cup(C^{r_X}\setminus(V_Z\cup\cE_{Z,\Rd}^{\cC}))$.
  % Hence $\supp(\tilde{b}+\Rd(f_X))\cup\supp(\tilde{b}'+\Rd(f_X'))$ is $\cE_Z^{\cC}(\eta_{\mathrm{run}},2\tilde{\gamma},2\tilde{\omega};\beta)$-avoiding.
  
  Furthermore, because $\tilde{e}_X,\tilde{e}_X'$ are supported inside $E\cap C^{r_X}$, \Cref{lem:redfault} implies that $\Rdn(\tilde{e}_X),\Rdn(\tilde{e}_X')$ are supported inside $\bar{E}_X$. Meanwhile, because $\tilde{e}_Z,\tilde{e}_Z'$ are supported inside $E'_Z\cap C^{r_X}$, applying \Cref{lem:redpropfault} to the dual complex ${\cC^\vee}^*=\cC_*\cong\cC^*(\epsilon,\bar{\ell},r_Z,r_X)$ implies that $\Rdn^\vee(\tilde{e}_Z),\Rdn^\vee(\tilde{e}_Z')$ are supported inside $\bar{E}_Z'$, where $\Rdn^\vee$ denotes the reduction map on this dual complex $\cC^*(\epsilon,\bar{\ell},r_Z,r_X)$.
  % That is, $\Rdn^\vee$ applies the isomorphism ${\cC^\vee}^*\xrightarrow{\sim}\cC^*(\epsilon,\bar{\ell},r_Z,r_X)$, then applies the reduction map on $\cC^*(\epsilon,\bar{\ell},r_Z,r_X)$, and subsequently applies the inverse isomorphism.
  
  Now because the reduction map $\Rdn$ (resp.~$\Rdn^\vee$) simply adds a coboundary (resp.~boundary) of $\cC^*$ to its input, Pauli errors acting on the codespace have the same effect as their reductions, so we can reduce the Pauli error on the output state of \Cref{alg:stateprep} in \Cref{eq:spoutput}:
  \begin{align*}
    \hspace{1em}&\hspace{-1em} \gZ^{\tilde{e}_Z}\gX^{\tilde{e}_X+\tilde{b}+\Rd(f_X)}\Enc_{\cC}^{r_X}\left(\ket{+^k}\bra{+^k}\right)\gX^{\tilde{e}_X'+\tilde{b}'+\Rd(f_X')}\gZ^{\tilde{e}_Z'} \\
    &\propto \gZ^{\Rdn^\vee(\tilde{e}_Z)}\gX^{\Rdn(\tilde{e}_X)+\tilde{b}+\Rd(f_X)}\Enc_{\cC}^{r_X}\left(\ket{+^k}\bra{+^k}\right)\gX^{\Rdn(\tilde{e}_X')+\tilde{b}'+\Rd(f_X')}\gZ^{\Rdn^\vee(\tilde{e}_Z')},
  \end{align*}
  where we recall that $\tilde{b}+\Rd(f_X),\tilde{b}'+\Rd(f_X')$ are already reduced. Thus the RHS above is a $\comp{E_{\mathrm{out}}}$-deviation of $\Enc_{\cC}^{r_X}\left(\ket{+^k}\bra{+^k}\right)$, where
  \begin{equation*}
    E_{\mathrm{out}} = (E_{\mathrm{out},X}=\bar{E}_Z',\; E_{\mathrm{out},Z}=(\tilde{E}_X\cup\bar{E}_X)).
  \end{equation*}
  Recall here that we do not need to explictly include the set $C^{r_X}\setminus(V_Z\cup\cE_{Z,\Rd}^{\cC})$ in the definition of $E_{\mathrm{out},Z}$, as it is already included through the definition of $\bar{E}_X$.
  Because we showed above that $\bar{E}_Z$ is $\cE_X^{\cC}(\eta_{\mathrm{run}},\gamma',\omega';\beta)$-avoiding, that $\tilde{E}_X$ is $\cE_Z^{\cC}(\eta_{\mathrm{run}},\tilde{\gamma},\tilde{\omega};\beta)$-avoiding, and that $\bar{E}_X$ is $\cE_Z^{\cC}(\eta_{\mathrm{run}},\; 8r\Delta\cdot\gamma_{\mathrm{run}},\; 8r\Delta\cdot\omega_{\mathrm{run}};\; \beta)$-avoiding, it follows that $E_{\mathrm{out}}$ is
  \begin{equation*}
    (\cE_X^{\cC}(\eta_{\mathrm{run}},\gamma',\omega';\beta),\; \cE_Z^{\cC}(\eta_{\mathrm{run}},\; \tilde{\gamma}+8r\Delta\gamma_{\mathrm{run}},\; \tilde{\omega}+8r\Delta\omega_{\mathrm{run}};\; \beta))
  \end{equation*}
  avoiding. Thus because \Cref{eq:spoutbounds,eq:spprimed,eq:sptilded} imply that
  \begin{align*}
    \max\{\gamma',\; \tilde{\gamma}+8r\Delta\cdot\gamma_{\mathrm{run}}\} &\leq 2^{14}r^3\Delta^5\cdot\gamma_{\mathrm{run}} = \gamma_{\mathrm{out}} \\
    \max\{\omega',\; \tilde{\omega}+8r\Delta\cdot\omega_{\mathrm{run}}\} &\leq 2^{12}r^3\beta^2\Delta^4\cdot\omega_{\mathrm{run}} = \omega_{\mathrm{out}},
  \end{align*}
  it follows that \Cref{alg:stateprep} outputs a $\comp{E_{\mathrm{out}}}$-deviation of $\Enc_{\cC}^{r_X}\left(\ket{+^k}\bra{+^k}\right)$, where $E_{\mathrm{out}}$ is a $\cE_{\mathrm{out}}=\cE^{\cC}(\eta_{\mathrm{out}},\gamma_{\mathrm{out}},\omega_{\mathrm{out}};\beta)$-avoiding set that depends only on $E_{\mathrm{run}}$ and on $s$, as desired.
\end{proof}

\subsubsection{Logical Measurement}
In this section, we present our logical-measurement gadget for the codes in \Cref{sec:quantumprod}. We use the standard approach of simply measuring all the physical qubits in the desired basis, and then decoding the measurement outcomes. However, like our state-preparation gadget in \Cref{sec:stateprep}, here we use our specialized decoder from \Cref{sec:decoder}.

Our logical measurement gadget takes as input a quantum code state but outputs a classical bit string giving the measurement outcome. Hence there is no output code. Our gadget simply measures all the qubits in the first timestep, and then applies a classical function gate to decode the measurement outcome. Hence there are no active qubits by the time the fault acts following the first timestep (see \Cref{def:fault}), so the fault has no effect on the output. Rather, we only need to deal with errors on the input code state.

\begin{lemma}
  \label{lem:logmeas}
   For $r_X,r_Z,\bar{\ell}\in\bN$, let $r=r_X+r_Z$, define $\epsilon=\epsilon(r)$ as in \Cref{lem:ssflip}, and define $\mu=\mu(\epsilon)$ and $\Delta=\Delta(\epsilon)$ as in \Cref{lem:lossless}. Define $\cC^*=\cC^*(\epsilon,\bar{\ell},r_X,r_Z)$, $\cM^*=\cM^*(\epsilon,\bar{\ell},r_X,r_Z)$, and $\Enc_{\cC}:\cM^*\rightarrow\cC^*$ as in \Cref{def:prodcode}. Let $Q=(Q_X,Q_Z)$ be the $[[n,k]]$ code at level $r_X$ of $\cC^*$.

  Let $\beta\geq 2\Delta$, and for arbitrary positive
  \begin{equation}
    \label{eq:lmrunbounds}
    \eta_{\mathrm{in}} \leq \mu 2^{\bar{\ell}/2}, \hspace{1em} \gamma_{\mathrm{in}} \leq \frac{1}{16r\Delta^2}, \hspace{1em} \omega_{\mathrm{in}} \leq \frac{\beta^{\bar{\ell}/2}}{8r\beta\Delta^2},
  \end{equation}
  let
  \begin{equation*}
    \cE_{\mathrm{in}} = \cE^{\cC}(\eta_{\mathrm{in}},\gamma_{\mathrm{in}},\omega_{\mathrm{in}};\beta)
  \end{equation*}
  be the family of bad sets defined in \Cref{def:badfams}. Define the decorated code
  \begin{equation*}
    D_{\mathrm{in}} = (Q,\; \Enc_{\cC}^{r_X},\; \cE_{\mathrm{in}}).
  \end{equation*}

  Let $\bar{O}_X=\gMX^{\otimes k}:\bC^{2^k\times 2^k}\rightarrow\bC^{2^k\times 2^k}$ be the $\gX$-measurement channel, which has $k$ input qubits (with no input bits) and $k$ output bits (with no output qubits).
  Then there exists a mending fault-tolerant gadget $\bfG_X=(\cR_X,\emptyset,D_{\mathrm{in}},\emptyset,\emptyset)$ for $\bar{O}_X$, where $\cR_X$ is a quantum circuit using quantum space $\Qu{N}=C^{r_X}$ (so $|\Qu{N}|=n$), time $T=2$, and gate set $\cG_Z=\{\gMX,\gCF_*\}$.
  
  Similarly, let $\bar{O}_Z=\gMZ^{\otimes k}:\bC^{2^k\times 2^k}\rightarrow\bC^{2^k\times 2^k}$ be the $\gZ$-measurement channel, which has $k$ input qubits (with no input bits) and $k$ output bits (with no output qubits).
  Then there exists a mending fault-tolerant gadget $\bfG_Z=(\cR_Z,\emptyset,D_{\mathrm{in}},\emptyset)$ for $\bar{O}_Z$, where $\cR_Z$ is a quantum circuit using quantum space $\Qu{N}=C^{r_X}$ (so $|\Qu{N}|=n$), time $T=2$, and gate set $\cG_Z=\{\gMZ,\gCF_*\}$.
\end{lemma}
\begin{proof}
  We will present and analyze $\bfG_Z$. Similarly as in \Cref{lem:stateprep}, the construction and analysis of $\bfG_X$ are analogous, except the argument is applied to the dual complex ${\cC^\vee}^*=\cC_*\cong\cC^*(\epsilon,\bar{\ell},r_Z,r_X)$, so that the roles of $\gX$ and $\gZ$ are swapped.

  Our desired circuit $\cR_Z$ uses $T=2$ timesteps. In the first timestep, we apply $\gMZ^{\otimes C^{r_X}}$ to measure all $n=|C^{r_X}|$ input qubits. Hence there are no remaining active qubits after this first timestep, so the total quantum space usage is $\Qu{N}=C^{r_X}$. Letting $z\in C^{r_X}$ denote the measurement outcome, then in the second timestep we apply a classical function gate to decode $z$. Specifically, we run \FnDecode{$\delta^{\cC}_{r_X}(z);\cC^*$} from \Cref{alg:decode}, and let $\tilde{a}^{r_X}\in\cC^{r_X}$ be the output. If the decoder returns FAIL, we instead let $\tilde{a}^{r_X}\in\cC^{r_X}$ be an arbitrary chosen element satisfying $\delta^{\cC}_{r_X}(\tilde{a}^{r_X})=\delta^{\cC}_{r_X}(z)$, so that $\tilde{a}^{r_X}$ is the same for all $z$ in a given coset of $Z^{r_X}(\cC)$. We then return the unique $\tilde{x}\in\bF_2^k=\cM^{r_X}=H^{r_X}(\cM)$ such that $\Enc_{\cC}^{r_X}(\tilde{x})=z+\tilde{a}^{r_X}+B^{r_X}(\cC)\in H^{r_X}(\cC)$. Assuming such an $\tilde{x}$ exists (which we show below is the case), we find it simply by running Gaussian elimination.

  We now show that $\bfG_Z$ is a fault-tolerant gadget for the channel $\bar{O}_Z=\gMZ^{\otimes k}$. As remarked above, we do not need to consider a fault $\cF$, as $\cR_Z$ has no active qubits for a fault to act on after the first timestep. Similarly, we do not need to consider a reweighting $\zeta$, as $\cR_Z$ uses no classically-controlled gates. We also do not need to consider a postselection $\xi$, as the set of postselectable gates of $\bfG_Z$ is $\ps=\emptyset$. Therefore fix a $\cE_{\mathrm{in}}$-avoiding pair of sets $E_{\mathrm{in}}=(E_{\mathrm{in},X}\subseteq C^{r_X},\;E_{\mathrm{in},Z}\subseteq C^{r_X})$, a finite set $\Rf{K}$, an operator $\rho\in\bC^{2^{M^{r_X}\sqcup\Rf{K}}\times 2^{M^{r_X}\sqcup\Rf{K}}}$, and a Pauli $\comp{E_{\mathrm{in}}}$-deviation $\sigma$ of $\Enc_{\cC}^{r_X}(\rho)$. By \Cref{lem:paulift}, our goal is then to show that $\cR_Z(\sigma)\propto\gMZ^{\otimes k}(\rho)$.

  For this purpose, we first write
  \begin{align}
    \label{eq:lmrho}
    \rho &= \sum_{x,x'\in\bF_2^k}\ket{x}\bra{x'}\otimes\rho_{x,x'},
  \end{align}
  where $\rho_{x,x'}\in\bC^{2^{\Rf{K}}\times 2^{\Rf{K}}}$. Therefore there exist $e_X,e_X',e_Z,e_Z'\in\cC^{r_X}$ such that $(\supp(e_X)\cup\supp(e_X'),\supp(e_Z)\cup\supp(e_Z'))$ is $\cE_{\mathrm{in}}$-avoiding, and such that
  \begin{align}
    \label{eq:lmpremeas}
    \begin{split}
      \sigma
      &= \sum_{x,x'\in\bF_2^k} \gZ^{e_Z}\gX^{e_X} \ket{\Enc_{\cC}^{r_X}(x)}\bra{\Enc_{\cC}^{r_X}(x')} \gX^{e_X'}\gZ^{e_Z'} \otimes \rho_{x,x'} \\
      &= \sum_{x,x'\in\bF_2^k} \gZ^{e_Z} \ket{\Enc_{\cC}^{r_X}(x)+e_X}\bra{\Enc_{\cC}^{r_X}(x')+e_X'} \gZ^{e_Z'} \otimes \rho_{x,x'}
    \end{split}
  \end{align}
  If $e_X+e_X'\notin Z^{r_X}(\cC)$, then for every $x,x'\in\bF_2^k$ we have $(\Enc_{\cC}^{r_X}(x)+e_X)\cap(\Enc_{\cC}^{r_X}(x')+e_X')=\emptyset$, so the measurements $\gMZ^{\otimes C^{r_X}}$ collapse the state $\sigma$ to $0$, and $\cR_Z(\sigma)=0$ is trivially proportional to $\gMZ^{\otimes k}(\rho)$. Therefore assume that $e_X+e_X'\in Z^{r_X}(\cC)$, so that $\delta^{\cC}_{r_X}(e_X+e_X')=0$.

  Now because $\supp(e_X+e_X')$ is $\cE_{\mathrm{in},Z}=\cE_Z^{\cC}(\eta_{\mathrm{in}},\gamma_{\mathrm{in}},\omega_{\mathrm{in}};\beta)$-avoiding, and by \Cref{eq:lmrunbounds} the constraints in \Cref{eq:bdinerr} are satisfied by $\eta_{\mathrm{in}},\gamma_{\mathrm{in}},\omega_{\mathrm{in}}$, it follows by \Cref{lem:decode} (with $\Upsilon=1$) that \FnDecode{$\delta^{\cC}_{r_X}(e_X+e_X')=0;\cC^*$} from \Cref{alg:decode} returns some element in $e_X+e_X'+B^{r_X}(\cC)$. But by definition \FnDecode{$0;\cC^*$} always returns $0$. Thus $e_X+e_X'\in B^{r_X}(\cC)$, so \Cref{eq:lmpremeas} simplifies to
  \begin{align}
    \label{eq:lmsig}
    \sigma
    &= \sum_{x,x'\in\bF_2^k} \gZ^{e_Z} \ket{\Enc_{\cC}^{r_X}(x)+e_X}\bra{\Enc_{\cC}^{r_X}(x')+e_X} \gZ^{e_Z'} \otimes \rho_{x,x'}.
  \end{align}
  Therefore because for $x\neq x'$ we have $\Enc_{\cC}^{r_X}(x)\cap\Enc_{\cC}^{r_X}(x')=\emptyset$, applying the measurements $\gMZ^{\otimes C^{r_X}}$ collapses $\sigma$ to
  \begin{align*}
    \label{eq:lmpostmeas}
    \gMZ^{\otimes C^{r_X}}(\sigma)
    &\propto \sum_{x\in\bF_2^k}\sum_{z\in\Enc_{\cC}^{r_X}(x)+e_X} \gZ^{e_Z} \ket{z}\bra{z} \gZ^{e_Z'} \otimes \rho_{x,x'} \\
    &= \sum_{x\in\bF_2^k}\sum_{z\in\Enc_{\cC}^{r_X}(x)+e_X} (-1)^{(e_Z+e_Z')\cdot z} \ket{z}\bra{z} \otimes \rho_{x,x'}
  \end{align*}
  Then because $e_X$ is $\cE_{\mathrm{in},Z}=\cE_Z^{\cC}(\eta_{\mathrm{in}},\gamma_{\mathrm{in}},\omega_{\mathrm{in}};\beta)$-avoiding, \Cref{lem:decode} implies that the output $\tilde{a}^{r_X}$ of $\FnDecode{$\delta^{\cC}_{r_X}(z)=\delta^{\cC}_{r_X}(e_X);\cC^*$}$ satisfies $\tilde{a}^{r_X}\in e_X+B^{r_X}(\cC)$. Therefore $z+\tilde{a}^{r_X}+B^{r_X}(\cC)=\Enc_{\cC}^{r_X}(x)$. Hence for each measurement outcome $z\in\Enc_{\cC}^{r_X}(x)+e_X$, our circuit $\cR_Z$ returns $\tilde{x}=x$, meaning that the final output state is
  \begin{align*}
    \cR_Z(\sigma)
    &\propto \sum_{x\in\bF_2^k} \left(\sum_{z\in\Enc_{\cC}^{r_X}(x)+e_X}(-1)^{(e_Z+e_Z')\cdot z}\right) \ket{x}\bra{x} \otimes \rho_{x,x} \\
    &\propto \sum_{x\in\bF_2^k} \left(\sum_{z\in\Enc_{\cC}^{r_X}(x)}(-1)^{(e_Z+e_Z')\cdot z}\right) \ket{x}\bra{x} \otimes \rho_{x,x}.
  \end{align*}

  We now simplify the expression in parentheses on the RHS above. If $e_Z+e_Z'\notin B^{r_X}(\cC)^\perp=Z_{r_X}(\cC)$, then because $\Enc_{\cC}^{r_X}$ is a coset of $B^{r_X}(\cC)$, it follows that $(e_Z+e_Z')\cdot z$ equals $0$ for half of all $z\in\Enc_{\cC}^{r_X}(x)$, and equals $1$ for the other half of all $z\in\Enc_{\cC}^{r_X}(x)$. Hence
  \begin{equation*}
    \sum_{z\in\Enc_{\cC}^{r_X}(x)}(-1)^{(e_Z+e_Z')\cdot z} = 0,
  \end{equation*}
  so that $\cR_Z(\sigma)=0$. Therefore assume that $e_Z+e_Z'\in Z_{r_X}(\cC)$. Then because $\supp(e_Z+e_Z')$ is $\cE_{\mathrm{in},X}=\cE_X^{\cC}(\eta_{\mathrm{in}},\gamma_{\mathrm{in}},\omega_{\mathrm{in}};\beta)$-avoiding, it follows by applying \Cref{lem:decode} to the dual complex ${\cC^\vee}^*=\cC_*\cong\cC^*(\epsilon,\bar{\ell},r_Z,r_X)$ that \FnDecode{$\partial^{\cC}_{r_X}(e_Z+e_Z')=0;\cC_*$} from \Cref{alg:decode} returns some element in $e_Z+e_Z'+B_{r_X}(\cC)$. But by definition \FnDecode{$0;\cC_*$} always returns $0$. Thus $e_Z+e_Z'\in B_{r_X}(\cC)=Z^{r_X}(\cC)^\perp$, so for every $z\in\Enc_{\cC}^{r_X}(x)\subseteq Z^{r_X}(\cC)$ we have $(e_Z+e_Z')\cdot z=0$. Hence
  \begin{equation*}
    \sum_{z\in\Enc_{\cC}^{r_X}(x)}(-1)^{(e_Z+e_Z')\cdot z} = |\Enc_{\cC}^{r_X}(x)| = |B^{r_X}(\cC)|
  \end{equation*}
  does not depend on the value of $x\in\bF_2^k$. Thus we conclude that
  \begin{align*}
    \cR_Z(\sigma)
    &\propto \sum_{x\in\bF_2^k} \ket{x}\bra{x} \otimes \rho_{x,x}.
  \end{align*}
  The RHS above is precisely the state $\gMZ^{\otimes k}(\rho)$ obtained by measuring $\rho$ in the $Z$-basis, as desired. Thus we have shown that $\bfG_Z$ is a fault-tolerant gadget for $\bar{O}_Z=\gMZ^{\otimes k}$.

  We now show that $\bfG_Z$ is mending. We fix a $\cE_{\mathrm{in}}$-avoiding pair of sets $E_{\mathrm{in}}=(E_{\mathrm{in},X}\subseteq C^{r_X},\;E_{\mathrm{in},Z}\subseteq C^{r_X})$, a finite set $\Rf{K}$, and an operator $\rho\in\bC^{2^{M^{r_X}\sqcup\Rf{K}}\times 2^{M^{r_X}\sqcup\Rf{K}}}$ as in the proof of fault-tolerance above. However, now we let $\sigma$ be a Pauli $\Rf{K}$-deviation of $\Enc_{\cC}^{r_X}(\rho)$. By \Cref{lem:paulift}, our goal is then to show that $\cR_Z(\sigma)$ is a linear combination of states of the form $\gMZ^{\otimes k}\circ L(\rho)$ for $k$-qubit superoperators $L$.
  % By definition we can decompse
  % \begin{equation*}
  %   \cR_Z(\sigma) = \sum_{\xi\in\bF_2^{C^{r_X}}}\cR_Z[\xi](\sigma),
  % \end{equation*}
  % and then analyze each postselected term $\cR_Z[\xi](\sigma)$ independently, where $\xi\in\bF_2^{C^{r_X}}$ specifies the postselection for the $\gMZ^{\otimes C^{r_X}}$ measurements in $\cR_Z$.

  Using the decomposition of $\rho$ in \Cref{eq:lmrho}, there must exist $e_X,e_X',e_Z,e_Z'\in\cC^{r_X}$ such that $\sigma$ is of the form in \Cref{eq:lmsig}. As in the fault-tolerance proof above, if $e_X+e_X'\notin Z^{r_X}(\cC)$, then $\cR_Z(\sigma)=0$, which is trivially of the desired form. Hence assume that $e_X+e_X'\in Z^{r_X}(\cC)$, so that $\delta^{\cC}_{r_X}(e_X+e_X')=0$. Let $\bar{x}\in\bF_2^k$ be the unique bitstring for which $e_X+e_X'\in\Enc_{\cC}^{r_X}(\bar{x})$. Then because for $x,x'\in\bF_2^k$, we have $(\Enc_{\cC}^{r_X}(x)+e_X)\cap(\Enc_{\cC}^{r_X}(x')+e_X')\neq\emptyset$ only when $x'=x+\bar{x}$, it follows that applying the measurements $\gMZ^{\otimes C^{r_X}}$ collapses the state $\sigma$ in \Cref{eq:lmsig} to
  \begin{align*}
    \gMZ^{\otimes C^{r_X}}(\sigma)
    &\propto \sum_{x\in\bF_2^k}\sum_{z\in\Enc_{\cC}^{r_X}(x)+e_X} \gZ^{e_Z} \ket{z}\bra{z} \gZ^{e_Z'} \otimes \rho_{x,x+\bar{x}} \\
    &= \sum_{x\in\bF_2^k}\sum_{z\in\Enc_{\cC}^{r_X}(x)+e_X} (-1)^{(e_Z+e_Z')\cdot z} \ket{z}\bra{z} \otimes \rho_{x,x+\bar{x}}.
  \end{align*}
  Now by definition, the value $\tilde{a}^{r_X}$ computed in $\cR_Z$ by running $\FnDecode{$\delta^{\cC}_{r_X}(z)=\delta^{\cC}_{r_X}(e_X)=\delta^{\cC}_{r_X}(e_X');\cC^*$}$ depends only on $e_X,e_X'$, and satisfies $z+\tilde{a}^{r_X}\in Z^{r_X}(\cC)$ for every $z\in\Enc_{\cC}^{r_X}(x)+e_X$. Recalling that our circuit $\cR_Z$ returns the unique $\tilde{x}\in\bF_2^k$ for which $z+\tilde{a}^{r_X}\in\Enc_{\cC}^{r_X}(\tilde{x})$, we define $\bar{x}'=x+\tilde{x}$. Because $z\in\Enc_{\cC}^{r_X}(x)+e_X$, we have $e_X+\tilde{a}^{r_X}\in\Enc_{\cC}^{r_X}(\bar{x}')$, so $\bar{x}'$ is well-defined as a function of $e_X,e_X'$, and in particular does not depend on $x$. Then our circuit's final output state is
  \begin{align*}
    \cR_Z(\sigma)
    &\propto \sum_{x\in\bF_2^k} \left(\sum_{z\in\Enc_{\cC}^{r_X}(x)+e_X}(-1)^{(e_Z+e_Z')\cdot z}\right) \ket{x+\bar{x}'}\bra{x+\bar{x}'} \otimes \rho_{x,x+\bar{x}} \\
    &\propto \sum_{x\in\bF_2^k} \left(\sum_{z\in\Enc_{\cC}^{r_X}(x)}(-1)^{(e_Z+e_Z')\cdot z}\right) \ket{x+\bar{x}'}\bra{x+\bar{x}'} \otimes \rho_{x,x+\bar{x}}.
  \end{align*}
  By the linearity of $\Enc_{\cC}^{r_X}:\bF_2^k\rightarrow H^{r_X}(\cC)$, there exists a linear map $z':\bF_2^k\rightarrow Z^{r_X}(\cC)$ such that for every $x\in\bF_2^k$, then $z'(x)\in\Enc_{\cC}^{r_X}$, and hence $\Enc_{\cC}^{r_X}(x)=z'(x)+B^{r_X}(\cC)$. Then because $x\mapsto(e_Z+e_Z')\cdot z'(x)$ is a linear map from $\bF_2^k\rightarrow\bF_2$, there exists $\bar{z}\in\bF_2^k$ such that $x\cdot\bar{z}=(e_Z+e_Z')\cdot z'(x)$. Thus
  \begin{align*}
    \cR_Z(\sigma)
    &\propto \sum_{x\in\bF_2^k} (-1)^{x\cdot\bar{z}} \left(\sum_{z\in B^{r_X}(\cC)}(-1)^{(e_Z+e_Z')\cdot z}\right) \ket{x+\bar{x}'}\bra{x+\bar{x}'} \otimes \rho_{x,x+\bar{x}} \\
    &\propto \sum_{x\in\bF_2^k} (-1)^{x\cdot\bar{z}} \ket{x+\bar{x}'}\bra{x+\bar{x}'} \otimes \rho_{x,x+\bar{x}} \\
    &\propto \gMZ^{\otimes k}(\gX^{\bar{x}'}\gZ^{\bar{z}}\rho\gX^{\bar{x}+\bar{x}'}).
  \end{align*}
  Thus indeed $\cR_Z(\sigma)$ is of the form $\gMZ^{\otimes k}\circ L(\rho)$ for the $k$-qubit superoperator $L(\rho)=\gX^{\bar{x}'}\gZ^{\bar{z}}\rho\gX^{\bar{x}+\bar{x}'}$, which completes the proof that $\bfG_Z$ is mending.
\end{proof}

\subsubsection{Ejection}
\label{sec:eject}
In \Cref{lem:eject} below, we present our ejection gadget.

\begin{lemma}
  \label{lem:eject}
  For $r_X,r_Z,\bar{\ell}\in\bN$ and $\epsilon>0$, define $\Delta=\Delta(\epsilon)$ as in \Cref{lem:lossless}. Define $\cC^*=\cC^*(\epsilon,\bar{\ell},r_X,r_Z)$, $\cM^*=\cM^*(\epsilon,\bar{\ell},r_X,r_Z)$, and $\Enc_{\cC}:\cM^*\rightarrow\cC^*$ as in \Cref{def:prodcode}. Let $Q=(Q_X,Q_Z)$ be the $[[n,k]]$ code at level $r_X$ of $\cC^*$.

  Let $\beta\geq 2\Delta$. For arbitrary $\eta_{\mathrm{in}},\gamma_{\mathrm{in}},\omega_{\mathrm{in}},\omega_{\mathrm{run}}>0$, let
  \begin{equation*}
    \omega_{\mathrm{out}} = 2\omega_{\mathrm{in}}+5\omega_{\mathrm{run}},
  \end{equation*}
  and (recalling that $k=|M^{r_X}|=\dim(H^{r_X}(\cM))$) let
  \begin{align*}
    \cE_{\mathrm{run}} &= 2^{M^{r_X}}|^{\geq\omega_{\mathrm{run}}}.
  \end{align*}
  We define the decorated input code
  \begin{align*}
    D_{\mathrm{in}} &= (Q,\; \Enc_{\cC}^{r_X},\; \cE_{\mathrm{in}} = \cE^{\cC}(\eta_{\mathrm{in}},\gamma_{\mathrm{in}},\omega_{\mathrm{in}};\beta))
  \end{align*}
  as given by \Cref{def:badfams}, and the decorated output code
  \begin{align*}
    D_{\mathrm{out}} &= ((\bF_2^{M^{r_X}},\bF_2^{M^{r_X}}),\; \id,\; \cE_{\mathrm{out}}=2^{M^{r_X}}|^{\geq\omega_{\mathrm{out}}}),
  \end{align*}
  whose encoding map is the $k=|M^{r_X}|$-qubit identity map $\id:\bF_2^{M^{r_X}}\rightarrow\bF_2^{M^{r_X}}$.

  Then there exists a fault-tolerant gadget $\bfG=(\cR,\cE_{\mathrm{run}}^{\sqcup T},D_{\mathrm{in}},D_{\mathrm{out}},\emptyset)$ for the $k$-qubit identity channel $\bar{O}=\id:\bC^{2^k\times 2^k}\rightarrow\bC^{2^k\times 2^k}$, where $\cR$ is a quantum circuit using quantum space $\Qu{N}=C^{r_X}$ (so $|\Qu{N}|=n$), time $T=5$, and gate set $\cG=\{\gTerm,\gMX,\gMZ,\gCF_*,\gCt{X},\gCt{Z}\}$.
\end{lemma}

In \Cref{lem:eject}, letting $r=r_X+r_Z$, then the first timestep $R_1$ of the circuit $\cR$ will measure or terminate all qubits in $C^{r_X}\setminus V_{\bar{\ell}}^{\times r}$. Hence the fault and output error only act on qubits labeled by $V_{\bar{\ell}}^{\times r}\cong M^{r_X}$.

To ensure that only a small constant fraction of qubits are corrupted during injection, we want to choose $\omega_{\mathrm{in}},\omega_{\mathrm{run}}=\lambda\cdot|M^{r_X}|$ for a small constant $\lambda>0$. Recall from \Cref{sec:bfls} that our gadgets in this section are intended to be resilient to locally stochastic noise that corrupts each level-$\ell$ qubit with probability $p^{2^{\Omega(\ell)}}$ for small constant $p>0$. When $\omega_{\mathrm{run}}=\lambda\cdot|M^{r_X}|$, then for a choice of the constant $p>0$ that is sufficiently small relative to $\lambda$, the probability that a locally stochastic fault is not $\cE_{\mathrm{run}}=2^{M^{r_X}}|^{\geq\omega_{\mathrm{run}}}$-avoiding is exponentially small in $|M^{r_X}|$.

Meanwhile, to ensure that $\omega_{\mathrm{in}}=\lambda\cdot|M^{r_X}|$, we can simply apply \Cref{lem:seqcomp} to precompose our ejection gadget in \Cref{lem:eject} with our error-correction gadget in \Cref{lem:errcorr} below. We can instantiate this error-correction gadget with its $\omega_{\mathrm{out}}$ value (which equals our ejection gadget's $\omega_{\mathrm{in}}$ value) equal to $\lambda\cdot|M^{r_X}|$ by setting the error-correction gadget's $\omega_{\mathrm{run}}$ value equal to $\lambda'\cdot|M^{r_X}|$ for a sufficiently small constant $\lambda'>0$. By \Cref{lem:lsex}, a locally stochastic fault on our error-correction gadget will only fail to satisfy this choice of $\omega_{\mathrm{run}}=\lambda'\cdot|M^{r_X}|$ with small constant probability. Hence with just a small constant failure probability, our gadget in \Cref{lem:eject} will successfully perform ejection with errors on just $\lambda$-fraction of the output qubits. Note that once we fix the fault's probability distribution, this set of $\lambda$-fraction corrupted qubits has a well-defined probability distribution, so we can identify a large constant fraction of our qubits with low marginal probability of corruption during ejection.

To prove \Cref{lem:eject}, we will repeatedly apply \Cref{lem:ejectone}, which shows how to eject out of one of the $r$ tensor factors defining $\cC$. As we will ultimately combine all $r$ applications of this subroutine into a single step, we do not present \Cref{lem:ejectone} using our general language of fault-tolerant gadgets, but rather describe the precise conditions that our subroutine satisfies.

\begin{lemma}
  \label{lem:ejectone}
  For integers $r_X\geq 1$, $r_Z\geq 0$, $\bar{\ell}\geq 1$, and a real number $\epsilon>0$, let $r=r_X+r_Z$, and define
  \begin{align*}
    \cA^* &= \cC^*(\epsilon,\bar{\ell},1,0), \hspace{1em} \cM_{\cA}=\cM^*(\epsilon,\bar{\ell},1,0), \hspace{1em} \Enc_{\cA}:\cM_{\cA}^*\rightarrow\cA^* \\
    \cB^* &= \cC^*(\epsilon,\bar{\ell},r_X-1,r_Z), \hspace{1em} \cM_{\cB}=\cM^*(\epsilon,\bar{\ell},r_X-1,r_Z), \hspace{1em} \Enc_{\cB}:\cM_{\cB}^*\rightarrow\cB^* \\
    \cC^* &= \cC^*(\epsilon,\bar{\ell},r_X,r_Z), \hspace{1em} \cM_{\cC}=\cM^*(\epsilon,\bar{\ell},r_X,r_Z), \hspace{1em} \Enc_{\cC}:\cM_{\cC}^*\rightarrow\cC^*,
  \end{align*}
  as given by \Cref{def:prodcode}, so that
  \begin{align*}
    \cC^* &= \cA^*\otimes\cB^*, \hspace{1em} \cM_{\cC}^* = \cM_{\cA}^*\otimes\cM_{\cB}^*, \hspace{1em} \Enc_{\cC} = \Enc_{\cA}\otimes\Enc_{\cB}.
  \end{align*}
  Let $e_X,e_X',e_Z,e_Z'\in\cC^{r_X}$ be such that $\supp(e_X)\cup\supp(e_X')$ is $\cE_{Z,\Rd}^{\cC}$-avoiding, and $\supp(e_Z)\cup\supp(e_Z')$ is $\cE_{X,\Rd}^{\cC}$-avoiding.

  Then there exists a circuit $\cR'$ using quantum space $C^{r_X}=(A^1\times B^{r_X-1})\sqcup(A^0\times B^{r_X})$ and time $T=4$, where the four timesteps apply:
  \begin{enumerate}
  \item $\gMZ$ gates to all qubits in the set $(A^0\times B^{r_X})\sqcup((A^1\setminus V_{\bar{\ell}})\times B^{r_X-1})$.
  \item An appropriate classical function gate, which takes as input the measurement outcomes, and outputs a set of bits labeled by $V_{\bar{\ell}}\times B^{r_X-1}$.
  \item Classically-controlled-$\gX$ gates $\gCt{\gX}^{\otimes V_{\bar{\ell}}\times B^{r_X-1}}$, with classical control bits given by the output of the classical function gate, and target qubits given by the remaining active qubits.
  \item Termination gates on all remaining classical bits.
  \end{enumerate}
  For every finite set $\Rf{K}$, every operator $\rho\in\bC^{2^{V_{\bar{\ell}}^r\sqcup\Rf{K}}\times 2^{V_{\bar{\ell}}^r\sqcup\Rf{K}}}$, every reweighting $\zeta\in\bF_2^{V_{\bar{\ell}}\times B^{r_X-1}}$, and every postselection $b\in\bF_2^{(A^0\times B^{r_X})\sqcup((A^1\setminus V_{\bar{\ell}})\times B^{r_X-1})}$, this circuit has the property that
  \begin{align}
    \label{eq:eogoal}
    \cR'[\zeta,b]\left(\gX^{e_X}\gZ^{e_Z}\Enc_{\cC}^{r_X}(\rho)\gZ^{e_Z'}\gX^{e_X'}\right)
    &\propto \gX^{e_{X,2}}\gZ^{e_{Z,2}}(I_{V_{\bar{\ell}}}\otimes\Enc_{\cB}^{r_X-1})(\rho)\gZ^{e_{Z,2}'}\gX^{e_{X,2}'},
  \end{align}
  where above for $e\in\cC^{r_X}$ we let $e_2=e|_{V_{\bar{\ell}}\times B^{r_X-1}}$, and where $I_{V_{\bar{\ell}}}\otimes\Enc_{\cB}^{r_X-1}$ denotes the channel associated to the linear encoding map
  \begin{equation*}
    I_{V_{\bar{\ell}}}\otimes\Enc_{\cB}^{r_X-1} : \bF_2^{V_{\bar{\ell}}^r}=\bF_2^{V_{\bar{\ell}}}\otimes\bF_2^{V_{\bar{\ell}}^{r-1}}\rightarrow\bF_2^{V_{\bar{\ell}}}\otimes H^{r_X-1}(\cB),
  \end{equation*}
  i.e.~$I_{V_{\bar{\ell}}}\otimes\Enc_{\cB}^{r_X-1}$ simply applies $\Enc_{\cB}^{r_X-1}$ to $|V_{\bar{\ell}}|$ disjoint blocks of qubits.
\end{lemma}
\begin{proof}
  To begin, we decompose
  \begin{align*}
    \rho
    &= \sum_{x,x'\in\bF_2^{V_{\bar{\ell}}^r}}\ket{x}\bra{x'}\otimes\rho_{x,x'},
  \end{align*}
  where each $\rho_{x,x'}\in\bF_2^{2^{\Rf{K}}\times 2^{\Rf{K}}}$. Recall that $\Enc_{\cC}=\Enc_{\cA}\otimes\Enc_{\cB}$, where by \Cref{def:classenc}, $\Enc_{\cA}$ simply applies the identity map in $V_{\bar{\ell}}$. Therefore
  \begin{align}
    \label{eq:eolog}
    \Enc_{\cC}(\rho)
    &= \sum_{x,x'\in\bF_2^{V_{\bar{\ell}}^r}} \ket{(I_{V_{\bar{\ell}}}\otimes\Enc_{\cB}^{r_X-1})(x)+B^{r_X}(\cC)}\bra{(I_{V_{\bar{\ell}}}\otimes\Enc_{\cB}^{r_X-1})(x')+B^{r_X}(\cC)}.
  \end{align}
  
  For $i\in\{r_X-1,r_X\}$, define the disjoint subsets $C^i_0,C^i_1,C^i_2\subseteq C^i$ by
  \begin{align*}
    C^i_0 &= A^0\times B^{i},\; \hspace{1em} C^i_1 = (A^1\setminus V_{\bar{\ell}})\times B^{i-1},\; \hspace{1em} C^i_2 = V_{\bar{\ell}}\times B^{i-1}.
  \end{align*}
  For $S\subseteq\{0,1,2\}$, let $C^i_S=\bigsqcup_{j\in S}C^i_j$, so that $C^i=C^i_{012}$. For $c\in\cC^i_S$ and $S'\subseteq S$, we also denote the restriction of $c$ to components in $C^i_{S'}$ by $c_{S'}=c|_{C^i_{S'}}$.

  With this notation,
  \begin{align*}
    B^{r_X}(\cC)
    &= \left\{(\delta^{\cA}\otimes I_B)(a_0)+(I_{A^1\setminus V_{\bar{\ell}}}\otimes\delta^{\cB})(a_1)+(I_{V_{\bar{\ell}}}\otimes\delta^{\cB})(a_2):a\in C^{r_X-1}_{012}\right\}.
  \end{align*}
  Because elements in the image of $\Enc_{\cB}^{r_X-1}$ in \Cref{eq:eolog} are cohomology classes in $H^{r_X-1}(\cB)$, the $(I_{V_{\bar{\ell}}}\otimes\delta^{\cB})(a_2)$ term above can be absorbed into the expression $(I_{V_{\bar{\ell}}}\otimes\Enc_{\cB}^{r_X-1})(x)$ in \Cref{eq:eolog}, so \Cref{eq:eolog} becomes
  \begin{align}
    \label{eq:eolog2}
    \begin{split}
      \Enc_{\cC}(\rho)
      &\propto \sum_{x,x',a,a'} \ket{(I_{A^0}\otimes\delta^{\cB})(a_0),\; (I_{V_{\bar{\ell}}}\otimes\Enc_{\cB}^{r_X-1})(x)+(\delta^{\cA}\otimes I_B)(a_0)+(I_{A^1\setminus V_{\bar{\ell}}}\otimes\delta^{\cB})(a_1)} \\
      &\hspace{4em} \bra{(I_{A^0}\otimes\delta^{\cB})(a_0'),\; (I_{V_{\bar{\ell}}}\otimes\Enc_{\cB}^{r_X-1})(x')+(\delta^{\cA}\otimes I_B)(a_0')+(I_{A^1\setminus V_{\bar{\ell}}}\otimes\delta^{\cB})(a_1')} \otimes \rho_{x,x'},
    \end{split}
  \end{align}
  where the sum above is over all $x,x'\in\bF_2^{V_{\bar{\ell}}^r}$ and all $a,a'\in C^{r_X-1}_{01}$.

  Now after running the $b$-postselected $\gMZ$ measurements in the first timestep of the execution of $\cR'[\zeta,b]\left(\gX^{e_X}\gZ^{e_Z}\Enc_{\cC}^{r_X}(\rho)\gZ^{e_Z'}\gX^{e_X'}\right)$, the resulting quantum state supported on qubits $C^{r_X}_2$ is
  \begin{align}
    \label{eq:eopm}
    \bra{b}\gX^{e_X}\gZ^{e_Z}\Enc_{\cC}^{r_X}(\rho)\gZ^{e_Z'}\gX^{e_X'}\ket{b},
  \end{align}
  where here we recall that $b\in\cC^{r_X}_{01}$ and $e_X,e_X',e_Z,e_Z'\in\cC^{r_X}_{012}$. By the assumption that $e_X,e_X'$ are $\cE_{Z,\Rd}^{\cC}$-avoiding (see \Cref{def:redsupp}), we have $e_{X,1},e_{X',1}=0$. Meanwhile, the Pauli errors $\gZ^{e_{Z,01}}$ and $\gZ^{e_{Z,012}'}$ can be absorbed into the bra $\bra{b}$ and ket $\ket{b}$ respectively, up to a global phase depending only on $b,b',e_X,e_X',e_Z,e_Z'$. Hence the state in \Cref{eq:eopm} is proportional to
  \begin{align*}
    \gX^{e_{X,2}}\gZ^{e_{Z,2}}\bra{b_0+e_{X,0},\;b_1} \Enc_{\cC}^{r_X}(\rho) \ket{b_0'+e_{X,0}',\;b_1'}\gZ^{e_{Z,2}'}\gX^{e_{X,2}'}.
  \end{align*}
  By \Cref{eq:eolog2}, this expression is in turn proportional to
  \begin{align}
    \label{eq:eopm2}
    \begin{split}
      \bra{b}\gX^{e_X}\gZ^{e_Z}\Enc_{\cC}^{r_X}(\rho)\gZ^{e_Z'}\gX^{e_X'}\ket{b}
      &\propto \sum_{x,x',a,a'} \gX^{e_{X,2}}\gZ^{e_{Z,2}} \ket{(I_{V_{\bar{\ell}}}\otimes\Enc_{\cB}^{r_X-1})(x)+(\delta^{\cA}\otimes I_B)(a_0)_2} \\
      &\hspace{3em} \bra{(I_{V_{\bar{\ell}}}\otimes\Enc_{\cB}^{r_X-1})(x')+(\delta^{\cA}\otimes I_B)(a_0')_2} \gZ^{e_{Z,2}'}\gX^{e_{X,2}'} \otimes \rho_{x,x'},
    \end{split}
  \end{align}
  where the sum above is over all $x,x'\in\bF_2^{V_{\bar{\ell}}^r}$ and all $a,a'\in C^{r_X-1}_{01}$ that satisfy
  \begin{align}
    \label{eq:eosumcond}
    \begin{split}
      (I_{A^0}\otimes\delta^{\cB})(a_0) = b_0+e_{X,0}, &\hspace{1em} (I_{A^0}\otimes\delta^{\cB})(a_0') = b_0+e_{X,0}' \\
      (\delta^{\cA}\otimes I_B)(a_0)_1+(I_{A^1\setminus V_{\bar{\ell}}}\otimes\delta^{\cB})(a_1) = b_1, &\hspace{1em} (\delta^{\cA}\otimes I_B)(a_0')_1+(I_{A^1\setminus V_{\bar{\ell}}}\otimes\delta^{\cB})(a_1') = b_1.
    \end{split}
  \end{align}
  If this sum in \Cref{eq:eopm2} vanishes, our circuit outputs $0$, which is trivially proportional to the desird output state. Therefore assume that the sum does not vanish.

  To obtain the desired output state in \Cref{eq:eogoal}, we want to apply classically-controlled-$\gX$ gates to cancel the $(\delta^{\cA}\otimes I_B)(a_0)_2,(\delta^{\cA}\otimes I_B)(a_0')_2$ term in \Cref{eq:eopm2}. We now show how to compute this $\gX$-correction from the measurement outcome $b$.

  Let $\Rd:\bF_2^{A^1\setminus V_{\bar{\ell}}}\rightarrow\bF_2^{A_0}$ denote the reducer map defined in \Cref{def:pushdown}, where we are able to replace the domain $\cA^1=\bF_2^{A^1}$ with $\bF_2^{A^1\setminus V_{\bar{\ell}}}$ because by definition $\Rd$ maps all elements in $\bF_2^{V_{\bar{\ell}}}$ to $0$. Also let $\delta^{\cA}|_1:\bF_2^{A^0}\rightarrow\bF_2^{A^1\setminus V_{\bar{\ell}}}$ (resp.~$\delta^{\cA}|_2:\bF_2^{A^0}\rightarrow\bF_2^{V_{\bar{\ell}}}$) be the map $\delta^{\cA}|_1(c)=\delta^{\cA}(c)_{A^1\setminus V_{\bar{\ell}}}$ (resp.~$\delta^{\cA}|_2(c)=\delta^{\cA}(c)_{V_{\bar{\ell}}}$) that simply applies $\delta^{\cA}$, and then restricts the output to components in $A^1\setminus V_{\bar{\ell}}$ (resp.~$V_{\bar{\ell}}$). Note that by definition $|A^0|=|A^1\setminus V_{\bar{\ell}}|$. In fact, \Cref{lem:pdkey} implies that $\delta^{\cA}|_1\circ\Rd=I$, that is, $\delta^{\cA}|_1$ and $\Rd$ are inverses of each other.

  It then follows by \Cref{eq:eosumcond} that
  \begin{align}
    \label{eq:eocorr}
    \begin{split}
      ((\delta^{\cA}|_2\circ\Rd)\otimes I_B)(b_1)
      &= (\delta^{\cA}|_2\otimes I_B)(a_0)+((\delta^{\cA}|_2\circ\Rd)\otimes\delta^{\cB})(a_1) \\
      &\in (\delta^{\cA}\otimes I_B)(a_0)_2+\bF_2^{V_{\bar{\ell}}}\otimes B^{r_X-1}(\cB),
    \end{split}
  \end{align}
  and also that
  \begin{align}
    \label{eq:eocorrp}
    \begin{split}
      ((\delta^{\cA}|_2\circ\Rd)\otimes I_B)(b_1)
      &= (\delta^{\cA}|_2\otimes I_B)(a_0')+((\delta^{\cA}|_2\circ\Rd)\otimes\delta^{\cB})(a_1') \\
      &\in (\delta^{\cA}\otimes I_B)(a_0')_2+\bF_2^{V_{\bar{\ell}}}\otimes B^{r_X-1}(\cB).
    \end{split}
  \end{align}
  Therefore we define the classical function gate in $\cR'$ to take as input $b=b_{01}\in\bF_2^{C^{r_X}_{01}}$, and to output $((\delta^{\cA}|_2\circ\Rd)\otimes I_B)(b_1)=((\delta^{\cA}\circ\Rd)\otimes I_B)(b_1)_2\in\bF_2^{C^{r_X}_2}$. The subsequent $\zeta$-reweighted classically-controlled-$\gX$ correction $\gCt{\gX}^{\otimes C^{r_X}_2}[\zeta]$ therefore applies $\gX^{((\delta^{\cA}\circ\Rd)\otimes I_B)(b_1)_2}$ to the state in \Cref{eq:eopm2}, while also multiplying by the phase $(-1)^{\zeta\cdot((\delta^{\cA}\circ\Rd)\otimes I_B)(b_1)_2}$ that depends only on $b$ and on $\zeta$. Hence the final output state of $\cR'$ is
  \begin{align*}
    \hspace{1em}&\hspace{-1em} \cR'[\zeta,b]\left(\gX^{e_X}\gZ^{e_Z}\Enc_{\cC}^{r_X}(\rho)\gZ^{e_Z'}\gX^{e_X'}\right) \\
                &\propto \sum_{x,x'\in\bF_2^{V_{\bar{\ell}}^r}} \gX^{e_{X,2}}\gZ^{e_{Z,2}} \ket{(I_{V_{\bar{\ell}}}\otimes\Enc_{\cB}^{r_X-1})(x)}\bra{(I_{V_{\bar{\ell}}}\otimes\Enc_{\cB}^{r_X-1})(x')} \gZ^{e_{Z,2}'}\gX^{e_{X,2}'} \otimes \rho_{x,x'} \\
                &= \gX^{e_{X,2}}\gZ^{e_{Z,2}} (I_{V_{\bar{\ell}}}\otimes\Enc_{\cB}^{r_X-1})(\rho) \gZ^{e_{Z,2}'}\gX^{e_{X,2}'},
  \end{align*}
  where the first proportionality above holds because the $\bF_2^{V_{\bar{\ell}}}\otimes B^{r_X-1}(\cB)$ terms in \Cref{eq:eocorr} and \Cref{eq:eocorrp} are absorbed into the $(I_{V_{\bar{\ell}}}\otimes\Enc_{\cB}^{r_X-1})(x)$ and $(I_{V_{\bar{\ell}}}\otimes\Enc_{\cB}^{r_X-1})(x')$ terms respectively. Thus \Cref{eq:eogoal} holds, as desired.
\end{proof}

We now apply \Cref{lem:ejectone} to prove \Cref{lem:eject}.

\begin{proof}[Proof of \Cref{lem:eject}]
  The desired circuit $\cR$ uses $T=5$ timesteps, and will perform a computation equivalent to $r$ iterations of the subroutine $\cR'$ from \Cref{lem:ejectone}, though with all $r$ iterations compressed into a single round of measurements followed by corrections.

  To begin, we define a circuit $\cR_0$ that first performs $r_X$ successive iterations of the subroutine $\cR'$ from \Cref{lem:ejectone}, where the $i$th iteration specifically performs $|V_{\bar{\ell}}^{i-1}|$ disjoint applications of \Cref{lem:ejectone} on the complex $\cA\otimes\cB=\cC^*(\epsilon,\bar{\ell},r_X-(i-1),r_Z)$. Then, $\cR_0$ performs $r_Z$ successive iterations of the dualized version of the subroutine $\cR'$. That is, the $i$th such iteration performs $|V_{\bar{\ell}}^{r_X+i-1}|$ disjoint applications of \Cref{lem:ejectone} on the dualized complex $\cA\otimes\cB={\cC^\vee}^*(\epsilon,\bar{\ell},0,r_Z-(i-1))\cong\cC^*(\epsilon,\bar{\ell},r_Z-(i-1),0)$.

  For every finite set $\Rf{K}$, every operator $\rho\in\bC^{2^{V_{\bar{\ell}}^r\sqcup\Rf{K}}\times 2^{V_{\bar{\ell}}^r\sqcup\Rf{K}}}$, every reweighting $\zeta$ and postselection $\xi$ for $\cR_0$, and every $e_X,e_X',e_Z,e_Z'\in\cC^{r_X}$ such that $\supp(e_X)\cup\supp(e_X')$ is $\cE_{Z,\Rd}^{\cC}$-avoiding and $\supp(e_Z)\cup\supp(e_Z')$ is $\cE_{X,\Rd}^{\cC}$-avoiding, then \Cref{lem:ejectone} implies that
  \begin{align}
    \label{eq:enofault}
    \cR_0[\zeta,\xi]\left(\gX^{e_X}\gZ^{e_Z}\Enc_{\cC}^{r_X}(\rho)\gZ^{e_Z'}\gX^{e_X'}\right)
    &\propto \gX^{\tilde{e}_X}\gZ^{\tilde{e}_Z}\rho\gZ^{\tilde{e}_Z'}\gX^{\tilde{e}_X'},
  \end{align}
  where above for $e\in\cC^{r_X}$ we let $\tilde{e}=e|_{V_{\bar{\ell}}^r}$.
  % Letting $\ps_0$ be the set of all measurement gates in $\cR_0$, then by definition $\cR_0[\zeta](\cdot)=\sum_{\xi\in\bF_2^{\ps_0}}\cR_0[\zeta,\xi](\cdot)$. Hence summing the above equation over all postselections $\xi\in\bF_2^{\ps_0}$ gives that
  % \begin{align}
  %   \label{eq:enofault}
  %   \cR_0[\zeta]\left(\gX^{e_X}\gZ^{e_Z}\Enc_{\cC}^{r_X}(\rho)\gZ^{e_Z'}\gX^{e_X'}\right)
  %   &\propto \gX^{\tilde{e}_X}\gZ^{\tilde{e}_Z}\rho\gZ^{\tilde{e}_Z'}\gX^{\tilde{e}_X'}.
  % \end{align}

  The circuit $\cR_0$ uses quantum space $\cC^{r_X}$ and time $T_0=4r$. These $T$ timesteps are partitioned into $r$ rounds of $4$ timesteps. Each such round consists of multiple independent applications of a circuit $\cR'$ given by \Cref{lem:ejectone}, meaning it applies $\gMX$ or $\gMZ$ measurements to some qubits, then applies a classical function gate to compute an appropriate Pauli $\gX$ or $\gZ$ correction, which is applied to the remaining qubits via a classically-controlled-$\gX$ or $\gZ$ gate. All remaining classical bits are then terminated.

  However, this circuit $\cR_0$ can be compressed into a circuit $\cR$ that still uses quantum space $\cC^{r_X}$, but now uses time $T=5$, and yet still satisfies \Cref{eq:enofault}, meaning that
  \begin{align}
    \label{eq:enofaultcomp}
    \cR[\zeta,\xi]\left(\gX^{e_X}\gZ^{e_Z}\Enc_{\cC}^{r_X}(\rho)\gZ^{e_Z'}\gX^{e_X'}\right)
    &\propto \gX^{\tilde{e}_X}\gZ^{\tilde{e}_Z}\rho\gZ^{\tilde{e}_Z'}\gX^{\tilde{e}_X'}.
  \end{align}
  Specifically, the first timestep of $\cR$ applies all $\gMX$ and $\gMZ$ measurements applied across all timesteps of $\cR_0$, so that all qubits in $C^{r_X}\setminus V_{\bar{\ell}}^r$ are measured. The second timestep of $\cR$ applies a classical function gate that takes as input these measurement outcomes, and outputs classical bitstrings $c_X,c_Z\in\bF_2^{V_{\bar{\ell}}^r}$, which specify the Pauli $\gX$ and $\gZ$ corrections to be applid in the third and fourth timesteps via the classically-controlled Pauli gates $\gCt{\gX}^{\otimes V_{\bar{\ell}}^r}$ and $\gCt{\gZ}^{\otimes V_{\bar{\ell}}^r}$, respectively. The fifth timestep of $\cR$ terminates all remaining classical bits.

  To ensure that \Cref{eq:enofaultcomp} holds, we must choose the classical function gate applied in the second timestep of $\cR$ to ensure that the resulting Pauli corrections $\gX^{c_X}$ and $\gZ^{c_Z}$ have the same affect on qubits $V_{\bar{\ell}}^r$ as the accumulated affect of all Pauli corrections applied during the execution of $\cR_0$. But such a classical function gate must exist, because any circuit consisting entirely of $\gMX$ and $\gMZ$ measurements and classically-controlled $\gX$ and $\gZ$ corrections can always be compressed into a a single round of measurements, followed by a single round of corrections. Indeed, we simply compute the desired Pauli corrections for $\cR$ in the order that they occur in $\cR_0$. If a $\gZ$ (resp.~$\gX$) correction is ever applied prior to a $\gMX$ (resp.~$\gMZ$) measurement in $\cR_0$, we flip the measurement outcome in the classical memory of $\cR$ before computing any subsequent corrections. No such adjustment is needed for $\gZ$ (resp.~$\gX$) corrections applied prior to $\gMZ$ (resp.~$\gMX$) measurements, as $\gZ$ (resp.~$\gX$) commutes with $\gMZ$ (resp.~$\gMX$).

  We are now ready to prove that $\bfG=(\cR,\cE_{\mathrm{run}}^{\sqcup T},D_{\mathrm{in}},D_{\mathrm{out}},\emptyset)$ is a fault-tolerant gadget for the $|V_{\bar{\ell}}^r|$-qubit identity channel, as stated in \Cref{lem:eject}. Fix a $\cE_{\mathrm{in}}$-avoiding pair of sets $E_{\mathrm{in}}=(E_{\mathrm{in},X}\subseteq C^{r_X},\;E_{\mathrm{in},Z}\subseteq C^{r_X})$, and a $\cE_{\mathrm{run}}^{\sqcup T}$-avoiding set $E_{\mathrm{run}}\subseteq C^{r_X}\times[T]$. Define the set $E_{\mathrm{out}}\subseteq V_{\bar{\ell}}^r$ by
  \begin{equation*}
    E_{\mathrm{out}} = (E_{\mathrm{in},X}\cap V_{\bar{\ell}}^r) \cup (E_{\mathrm{in},Z}\cap V_{\bar{\ell}}^r) \cup \bigcup_{t\in[T]}(E_{\mathrm{run}}\cap(V_{\bar{\ell}}^r\times\{t\})).
  \end{equation*}
  That is, $E_{\mathrm{out}}$ contains all qubits in $V_{\bar{\ell}}^r$ that lie in the support of the input or fault error sets across all timesteps. By the definition of $\cE_{\mathrm{in}},\cE_{\mathrm{run}}\cE_{\mathrm{out}}$ in the statement of \Cref{lem:eject}, we have
  \begin{equation*}
    |E_{\mathrm{out}}| = |E_{\mathrm{out}}|_\beta < 2\omega_{\mathrm{in}}+5\omega_{\mathrm{run}} = \omega_{\mathrm{out}},
  \end{equation*}
  so $E_{\mathrm{out}}$ is $\cE_{\mathrm{out}}$-avoiding.

  By \Cref{lem:paulift}, our goal is to show that for every finite set $\Rf{K}$, every operator $\rho\in\bC^{2^{V_{\bar{\ell}}^r\sqcup\Rf{K}}\times 2^{V_{\bar{\ell}}^r\sqcup\Rf{K}}}$, every Pauli $\comp{E_{\mathrm{in}}}$-deviation $\sigma$ of $\Enc_{\mathrm{in}}(\rho)$, and every $\comp{E_{\mathrm{run}}}$-avoiding Pauli fault $\cF$ and reweighting $\zeta$ for $\cR$, then $\cR[\cF,\zeta](\sigma)$ is a $\comp{E_{\mathrm{out}}}$-deviation of $\rho$. By definition there exist $e_X,e_X'$ supported in $E_{\mathrm{in},Z}$ and $e_Z,e_Z'$ supported in $E_{\mathrm{in},X}$ such that
  \begin{equation*}
    \sigma \propto \gX^{e_X}\gZ^{e_Z}\Enc_{\cC}^{r_X}(\rho)\gZ^{e_Z'}\gX^{e_X'}.
  \end{equation*}
  Because the input family of bad sets $\cE_{\mathrm{in}}$ contains $(\cE_{X,\Rd}^{\cC},\cE_{Z,\Rd}^{\cC})$ by \Cref{def:badfams}, $(\supp(e_Z)\cup\supp(e_Z'),\supp(e_X)\cup\supp(e_X'))$ is $(\cE_{X,\Rd}^{\cC},\cE_{Z,\Rd}^{\cC})$-avoiding. Therefore letting $\ps=C^{r_X}\setminus V_{\bar{\ell}}^r$ denote the set of all measurement gates in $\cR$, we may apply \Cref{eq:enofaultcomp} to conclude that for every postselection $\xi\in\bF_2^{\ps}$, we have
  \begin{align*}
    \cR[\zeta,\xi](\sigma) &\propto \gX^{\tilde{e}_X}\gZ^{\tilde{e}_Z}\rho\gZ^{\tilde{e}_Z'}\gX^{\tilde{e}_X'},
  \end{align*}
  where recall here that for $e\in\cC^{r_X}$ we let $\tilde{e}=e|_{V_{\bar{\ell}}^r}$.

  Now we consider the execution $\cR[\cF,\zeta,\xi](\sigma)$ that additionally includes the Pauli fault $\cF$. Recall that $\cF$ acts on qubits $V_{\bar{\ell}}^r$, and the only gates experienced by these qubits are classically-controlled-Pauli gates, where the values of the classical control bits are uniquely determined by the value of the postselection $\xi$. Hence we may propagate all the Pauli fault errors through to the end of the circuit, while simply introducing a phase depending only on $\cF$ and on $\xi$, so that
  \begin{align*}
    \cR[\cF,\zeta,\xi](\sigma) &\propto F_T\circ\cdots\circ F_1(\gX^{\tilde{e}_X}\gZ^{\tilde{e}_Z}\rho\gZ^{\tilde{e}_Z'}\gX^{\tilde{e}_X'}).
  \end{align*}
  By definition $\cR[\cF,\zeta](\cdot)=\sum_{\xi\in\bF_2^{\ps}}\cR_0[\cF,\zeta,\xi](\cdot)$. Hence summing the above equation over all postselections $\xi\in\bF_2^{\ps}$ gives that
  \begin{align*}
    \cR[\cF,\zeta](\sigma) &\propto F_T\circ\cdots\circ F_1(\gX^{\tilde{e}_X}\gZ^{\tilde{e}_Z}\rho\gZ^{\tilde{e}_Z'}\gX^{\tilde{e}_X'}).
  \end{align*}
  Here $F_T\circ\cdots\circ F_1$ is by definition a Pauli superoperator supported inside $\bigcup_{t\in[T]}(E_{\mathrm{run}}\cap(V_{\bar{\ell}}^r\times\{t\}))\subseteq E_{\mathrm{out}}$. Also recall that $\tilde{e}_X,\tilde{e}_X',\tilde{e}_Z,\tilde{e}_Z'$ are all supported inside $(E_{\mathrm{in},X}\cap V_{\bar{\ell}}^r) \cup (E_{\mathrm{in},Z}\cap V_{\bar{\ell}}^r)\subseteq E_{\mathrm{out}}$. Thus $\cR[\cF,\zeta](\sigma)$ is a $\comp{E_{\mathrm{out}}}$-deviation of $\rho$, as desired.
\end{proof}

\subsection{Basic Gadgets}
\label{sec:basicgad}
In this section, we present some basic gadgets that are generally well-known in the literature. However, the we present the details for our specific codes within our fault-tolerance framework.

\subsubsection{Idle (Without Error-Correction)}
In this section, we present the basic gadget that simply idles, i.e.~performs only identity gates, for some number $T$ of timesteps, while accumulating errors from a fault.

\begin{lemma}
  \label{lem:idle}
  Let $\cC^*=\cC^*(\epsilon,\bar{\ell},r_X,r_Z)$, $\cM^*=\cM^*(\epsilon,\bar{\ell},r_X,r_Z)$, and $\Enc_{\cC}:\cM^*\rightarrow\cC^*$ as in \Cref{def:prodcode}. Let $Q=(Q_X,Q_Z)$ be the $[[n,k]]$ code at level $r_X$ of $\cC^*$.

  For arbitrary $\eta,\gamma_{\mathrm{in}},\omega_{\mathrm{in}},\gamma_{\mathrm{run}},\omega_{\mathrm{run}}>0$ and $T\in\bN$, let
  \begin{equation*}
    \gamma_{\mathrm{out}} = \gamma_{\mathrm{in}}+T\gamma_{\mathrm{run}}, \hspace{1em} \omega_{\mathrm{out}} = \omega_{\mathrm{in}}+T\omega_{\mathrm{run}}
  \end{equation*}
  and
  \begin{align*}
    \cE_{\mathrm{run}} &= \cE_{\mathrm{run}}^{\cC}(\eta,\gamma_{\mathrm{run}},\omega_{\mathrm{run}};\beta).
  \end{align*}
  (see \Cref{def:badfault}). For $\alpha\in\{\mathrm{in},\mathrm{out}\}$, define the decorated code
  \begin{align*}
    D_\alpha = (Q,\; \Enc_{\cC}^{r_X},\; \cE_\alpha = \cE^{\cC}(\eta,\gamma_\alpha,\omega_\alpha;\beta)).
  \end{align*}
  (see \Cref{def:badfams}).

  Then there exists a fault-tolerant gadget $\bfG=(\cR,\cE_{\mathrm{run}}^{\sqcup T},D_{\mathrm{in}},D_{\mathrm{out}},\emptyset)$ for the channel $\bar{O}=I_{M^{r_X}}$, where $\cR$ is the quantum circuit using space $C^{r_X}$ and time $T$ that simply idles, i.e.~applies identity gates, in every timestep.
\end{lemma}
\begin{proof}
  Let $E_{\mathrm{in}}=(E_{\mathrm{in},X}\subseteq C^{r_X},\;E_{\mathrm{in},Z}\subseteq C^{r_X})$ be a $\cE_{\mathrm{in}}$-avoiding pair of sets, and let $E_{\mathrm{run}}\subseteq C^{r_X}\times T$ be a  $\cE_{\mathrm{run}}$-avoiding set. Let
  \begin{equation*}
    E = \bigcup_{t\in[T]}(E_{\mathrm{run}}\cap(C^{r_X}\times\{t\}))
  \end{equation*}
  be the union of all qubits in $E_{\mathrm{run}}$ across all timesteps, so that $E$ is $\cE_{\mathrm{run}}^{\cC}(\eta,T\gamma_{\mathrm{run}},T\omega_{\mathrm{run}};\beta)$-avoiding. Let
  \begin{align*}
    \bar{E}_X &= \Shad^{\times r_X}(E)\cup(C^{r_X}\setminus(V_Z\cup\cE_{Z,\Rd}^{\cC})) \\
    \bar{E}_Z &= \Shad^{\times r_Z}(E)\cup(C^{r_X}\setminus(V_X\cup\cE_{X,\Rd}^{\cC})),
  \end{align*}
  as in \Cref{lem:redfault},
  so that $\Shad^{\times r_X}$ (resp.~$\Shad^{\times r_Z}$) acts on the first $r_X$ (resp.~last $r_Z$) factors, and so that $(\bar{E}_{Z,i},\bar{E}_{X,i})$ is $\cE^{\cC}(\eta,T\gamma_{\mathrm{run}},T\omega_{\mathrm{run}};\beta)$-avoiding by \Cref{lem:redfault}. Then define
  \begin{align}
    \label{eq:idout}
    \begin{split}
      E_{\mathrm{out},X} &= E_{\mathrm{in},X} \cup \bar{E}_Z \\
      E_{\mathrm{out},Z} &= E_{\mathrm{in},Z} \cup \bar{E}_X ,
    \end{split}
  \end{align}
  so that $E_{\mathrm{out}}=(E_{\mathrm{out},X},E_{\mathrm{out},Z})$ is
  \begin{equation*}
    \cE^{\cC}(\eta,\; \gamma_{\mathrm{in}}+T\gamma_{\mathrm{run}},\; \omega_{\mathrm{in}}+T\omega_{\mathrm{run}};\; \beta) = \cE^{\cC}(\eta,\gamma_{\mathrm{out}},\omega_{\mathrm{out}};\beta)
  \end{equation*}
  avoiding.

  To show that $\bfG$ is a fault-tolerant gadget for $\bar{O}=I^{\otimes M^{r_X}}$, by \Cref{lem:paulift} it suffices to show that for every finite set $\Rf{K}$, every operator $\rho\in\bC^{2^{M^{r_X}\sqcup\Rf{K}}\times 2^{M^{r_X}\sqcup\Rf{K}}}$, every Pauli $\comp{E_{\mathrm{in}}}$-deviation $\sigma$ of $\Enc_{\cC}^{r_X}(\rho)$, and every $\comp{E_{\mathrm{run}}}$-avoiding Pauli fault $\cF$ for $\cR$, then $\cR[\cF](\sigma)$ is a $\comp{E_{\mathrm{out}}}$-deviation of $\Enc_{\cC}^{r_X}(\rho)$.

  For this purpose, let $F_{\mathrm{run}}=F_T\circ\cdots\circ F_1$ denote the accumulated Pauli error from the entire fault, so that $F$ is $\comp{E_{\mathrm{run}}}$-avoiding. Also let $F_{\mathrm{in}}$ denote the $\comp{E_{\mathrm{in}}}$-avoiding Pauli error, so that $\sigma=F_{\mathrm{in}}(\Enc_{\cC}^{r_X}(\rho))$. Then because the circuit $\cR$ simply applies identity gates in every timestep, we have
  \begin{align*}
    \cR[\cF](\sigma) &= F_{\mathrm{run}}\circ F_{\mathrm{in}}(\Enc_{\cC}^{r_X}(\rho)) \propto F_{\mathrm{in}}\circ F_{\mathrm{run}}(\Enc_{\cC}^{r_X}(\rho)),
  \end{align*}
  where the proportionality phase above depends only on the Paulis $F_{\mathrm{in}},F_{\mathrm{run}}$. Letting $\bar{F}_{\mathrm{run}}$ denote the Pauli error given by the reductions of the $X$ and $Z$ errors comprising $F_{\mathrm{run}}$, then $\bar{F}_{\mathrm{run}}(\Enc_{\cC}^{r_X}(\rho))=F_{\mathrm{run}}(\Enc_{\cC}^{r_X}(\rho))$, as the reduction of a Pauli error differs from the original Pauli error by a stabilizer. Furthermore, by \Cref{lem:redfault}, the $X$ and $Z$ errors comprising $\bar{F}_{\mathrm{run}}$ are supported inside $\bar{E}_X$ and $\bar{E}_Z$ respectively. Hence by \Cref{eq:idout}, the state
  \begin{equation*}
    \cR[\cF](\sigma) \propto F_{\mathrm{in}}\circ\bar{F}_{\mathrm{run}}(\Enc_{\cC}^{r_X}(\rho))
  \end{equation*}
  is a $\comp{E_{\mathrm{out}}}$-deviation of $\Enc_{\cC}^{r_X}(\rho)$, as desired.
\end{proof}

\subsubsection{CNOT}
In this section, we present a standard transversal $\gCNOT$ gadget for our codes. Below, recall from \Cref{def:deccode} that for a $[[n,k]]$ decorated code $D$, then $D^{\sqcup 2}$ denotes the $[[2n,2k]]$ decorated code given by taking two disjoint copies of $D$.

\begin{lemma}
  \label{lem:cnot}
  Let $\cC^*=\cC^*(\epsilon,\bar{\ell},r_X,r_Z)$, $\cM^*=\cM^*(\epsilon,\bar{\ell},r_X,r_Z)$, and $\Enc_{\cC}:\cM^*\rightarrow\cC^*$ as in \Cref{def:prodcode}. Let $Q=(Q_X,Q_Z)$ be the $[[n,k]]$ code at level $r_X$ of $\cC^*$.

  For arbitrary $\eta,\gamma_{\mathrm{in}},\omega_{\mathrm{in}},\gamma_{\mathrm{run}},\omega_{\mathrm{run}}>0$, let
  \begin{equation*}
    \gamma_{\mathrm{out}} = 2\gamma_{\mathrm{in}}+\gamma_{\mathrm{run}}, \hspace{1em} \omega_{\mathrm{out}} = 2\omega_{\mathrm{in}}+\omega_{\mathrm{run}}
  \end{equation*}
  and
  \begin{align*}
    \cE_{\mathrm{run}} &= \cE_{\mathrm{run}}^{\cC}(\eta,\gamma_{\mathrm{run}},\omega_{\mathrm{run}};\beta).
  \end{align*}
  (see \Cref{def:badfault}). For $\alpha\in\{\mathrm{in},\mathrm{out}\}$, define the decorated code
  \begin{align*}
    D_\alpha = (Q,\; \Enc_{\cC}^{r_X},\; \cE_\alpha = \cE^{\cC}(\eta,\gamma_\alpha,\omega_\alpha;\beta)).
  \end{align*}
  (see \Cref{def:badfams}).

  Then there exists a fault-tolerant gadget $\bfG=(\cR,\cE_{\mathrm{run}}^{\sqcup 2},D_{\mathrm{in}}^{\sqcup 2},D_{\mathrm{out}}^{\sqcup 2},\emptyset)$ for the channel $\bar{O}=\gCNOT^{\otimes k}$, where $\cR$ is the quantum circuit using space $(C^{r_X})^{\sqcup 2}$ and time $T=1$ that simply applies $R_1=\gCNOT^{\otimes C^{r_X}}$ in its single timestep.
\end{lemma}
\begin{proof}
  Let $E_{\mathrm{in}}=(E_{\mathrm{in},X}\subseteq(C^{r_X})^{\sqcup 2},\; E_{\mathrm{in},Z}\subseteq(C^{r_X})^{\sqcup 2})$ be a $\cE_{\mathrm{in}}^{\sqcup 2}$-avoiding pair of sets, and let $E_{\mathrm{run}}\subseteq(C^{r_X})^{\sqcup 2}$ be a $\cE_{\mathrm{run}}^{\sqcup 2}$-avoiding set. Here for a a set $E\subseteq(C^{r_X})^{\sqcup 2}$, we let $E=E_1\sqcup E_2$ denote the decomposition into the subsets $E_1,E_2\subseteq C^{r_X}$ in the two respective disjoint copies of $C^{r_X}$. Similarly, for $e\in\bF_2^{(C^{r_X})^{\sqcup 2}}$, we let $e=(e_1,e_2)\in(\cC^{r_X})^{\oplus 2}$. Then for $i\in\{1,2\}$ we define
  \begin{align*}
    \bar{E}_{X,i} &= \Shad^{\times r_X}(E_{\mathrm{run},i})\cup(C^{r_X}\setminus(V_Z\cup\cE_{Z,\Rd}^{\cC})) \\
    \bar{E}_{Z,i} &= \Shad^{\times r_Z}(E_{\mathrm{run},i})\cup(C^{r_X}\setminus(V_X\cup\cE_{X,\Rd}^{\cC})) \\
  \end{align*}
  as in \Cref{lem:redfault},
  so that $\Shad^{\times r_X}$ (resp.~$\Shad^{\times r_Z}$) acts on the first $r_X$ (resp.~last $r_Z$) factors, and so that $(\bar{E}_{Z,i},\bar{E}_{X,i})$ is $\cE^{\cC}(\eta,\gamma_{\mathrm{run}},\omega_{\mathrm{run}};\beta)$-avoiding by \Cref{lem:redfault}. For $i\in\{1,2\}$, we then define
  \begin{align*}
    E_{\mathrm{out},X,i} &= E_{\mathrm{in},X,1}\cup E_{\mathrm{in},X,2}\cup \bar{E}_{Z,i} \\
    E_{\mathrm{out},Z,i} &= E_{\mathrm{in},Z,1}\cup E_{\mathrm{in},Z,2}\cup \bar{E}_{X,i},
  \end{align*}
  so that
  \begin{equation*}
    E_{\mathrm{out}} = (E_{\mathrm{out},X,1}\sqcup E_{\mathrm{out},X,2}, E_{\mathrm{out},Z,1}\sqcup E_{\mathrm{out},Z,1})
  \end{equation*}
  is
  \begin{equation*}
    \cE^{\cC}(\eta,2\gamma_{\mathrm{in}}+\gamma_{\mathrm{run}},2\omega_{\mathrm{in}}+\omega_{\mathrm{run}};\beta)^{\sqcup 2} = \cE^{\cC}(\eta,\gamma_{\mathrm{out}},\omega_{\mathrm{out}};\beta)
  \end{equation*}
  avoiding.
  
  To show that $\bfG$ is a fault-tolerant gadget for $\bar{O}=\gCNOT^{\otimes M^{r_X}}$, by \Cref{lem:paulift} it suffices to show that for every finite set $\Rf{K}$, every operator $\rho\in\bC^{2^{M^{r_X}\sqcup\Rf{K}}\times 2^{M^{r_X}\sqcup\Rf{K}}}$, every Pauli $\comp{E_{\mathrm{in}}}$-deviation $\sigma$ of $\Enc_{\cC}^{r_X}(\rho)$, and every $\comp{E_{\mathrm{run}}}$-avoiding Pauli fault $\cF$ for $\cR=\gCNOT^{\otimes C^{r_X}}$, then $\cR[\cF](\sigma)$ is a $\comp{E_{\mathrm{out}}}$-deviation of $\Enc_{\cC}^{r_X}\circ\gCNOT^{\otimes M^{r_X}}(\rho)$.

  For this purpose, write
  \begin{equation*}
    \rho = \sum_{x,x'\in(\cC^{r_X})^{\oplus 2}}\ket{x_1,x_2}\bra{x'_1,x'_2}\otimes\rho_{x,x'},
  \end{equation*}
  where each $\rho_{x,x'}\in\bC^{2^{\Rf{K}}\times 2^{\Rf{K}}}$.
  Also write
  \begin{equation}
    \label{eq:cnotinit}
    \sigma = \gX^{e_X}\gZ^{e_Z}\Enc_{\cC}^{r_X}(\rho)\gZ^{e_Z'}\gX^{e_X'},
  \end{equation}
  where $(\supp(e_X)\cup\supp(e_X'),\supp(e_Z)\cup\supp(e_Z'))$ is $\comp{E_{\mathrm{in}}}$-avoiding. We define $\tilde{e}_X,\tilde{e}_X',\tilde{e}_Z,\tilde{e}_Z'\in(\cC^{r_X})^{\oplus 2}$ by
  \begin{align*}
    \tilde{e}_X = (e_{X,1},e_{X,1}+e_{X,2}), &\hspace{1em} \tilde{e}_X' = (e_{X,1}',e_{X,1}'+e_{X,2}') \\
    \tilde{e}_Z = (e_{Z,1}+e_{Z,2},e_{Z,2}), &\hspace{1em} \tilde{e}_Z' = (e_{Z,1}'+e_{Z,2}',e_{X,2}')
  \end{align*}
  to specify the conjugation of the Pauli error in \Cref{eq:cnotinit} by $\gCNOT^{\otimes C^{r_X}}$, so that without a fault, the output of $\cR(\sigma)$ is
  \begin{align*}
    \cR(\sigma)
    &= \sum_{x,x'\in(\cC^{r_X2})^{\oplus 2}} \gCNOT^{\otimes C^{r_X}} \gX^{e_X}\gZ^{e_Z} \ket{\Enc_{\cC}^{r_X}(x_1),\Enc_{\cC}^{r_X}(x_2)} \\
    &\hspace{7em} \bra{\Enc_{\cC}^{r_X}(x'_1),\Enc_{\cC}^{r_X}(x'_2)} \gZ^{e_Z'}\gX^{e_X'} \gCNOT^{\otimes C^{r_X}} \otimes \rho_{x,x'} \\
    &= \sum_{x,x'\in(\cC^{r_X2})^{\oplus 2}} \gX^{\tilde{e}_X}\gZ^{\tilde{e}_Z} \gCNOT^{\otimes C^{r_X}} \ket{\Enc_{\cC}^{r_X}(x_1),\Enc_{\cC}^{r_X}(x_2)} \\
    &\hspace{7em} \bra{\Enc_{\cC}^{r_X}(x'_1),\Enc_{\cC}^{r_X}(x'_2)} \gCNOT^{\otimes C^{r_X}} \gZ^{\tilde{e}_Z'}\gX^{\tilde{e}_X'} \otimes \rho_{x,x'} \\
    &= \sum_{x,x'\in(\cC^{r_X2})^{\oplus 2}} \gX^{\tilde{e}_X}\gZ^{\tilde{e}_Z} \ket{\Enc_{\cC}^{r_X}(x_1),\Enc_{\cC}^{r_X}(x_1+x_2)} \\
    &\hspace{7em} \bra{\Enc_{\cC}^{r_X}(x'_1),\Enc_{\cC}^{r_X}(x'_1+x'_2)} \gZ^{\tilde{e}_Z'}\gX^{\tilde{e}_X'} \otimes \rho_{x,x'} \\
    &= \gX^{\tilde{e}_X}\gZ^{\tilde{e}_Z} (\Enc_{\cC}^{r_X}\circ\gCNOT^{\otimes M^{r_X}}(\rho)) \gX^{\tilde{e}_X'}\gZ^{\tilde{e}_Z'}.
  \end{align*}
  Therefore because $\cR$ only has a single timestep, the fault $\cF$ simply applies a single Pauli error $F_1$ to the output of $\cR$, that is
  \begin{align}
    \label{eq:cnotout}
    \begin{split}
      \cR[\cF](\sigma)
      &= F_1(\gX^{\tilde{e}_X}\gZ^{\tilde{e}_Z} \Enc_{\cC}^{r_X}\circ\gCNOT^{\otimes M^{r_X}}(\rho) \gX^{\tilde{e}_X'}\gZ^{\tilde{e}_Z'}) \\
      &\propto \gX^{\tilde{e}_X}\gZ^{\tilde{e}_Z} F_1(\Enc_{\cC}^{r_X}\circ\gCNOT^{\otimes M^{r_X}}(\rho)) \gX^{\tilde{e}_X'}\gZ^{\tilde{e}_Z'}
    \end{split}
  \end{align}
  Now by definition $\tilde{e}_X,\tilde{e}_x'$ (resp.~$\tilde{e}_Z,\tilde{e}_Z'$) are supported inside $E_{\mathrm{in},Z,1}\cup E_{\mathrm{in},Z,2}$ (resp.~$E_{\mathrm{in},X,1}\cup E_{\mathrm{in},X,2}$). Furthermore, the Pauli error $\bar{F}_1$ given by the reductions of the $X$ and $Z$ errors comprising $F_1$ by definition satisfies $F_1(\Enc_{\cC}^{r_X}(\rho))=\bar{F}_1(\Enc_{\cC}^{r_X}(\rho))$, as the reduction of a Pauli error differs from the original Pauli error by a stabilizer. Furthemore, by \Cref{lem:redfault}, the $X$ and $Z$ errors comprising $\bar{F}_1$ are supported inside $\bar{E}_X$ and $\bar{E}_Z$ respectively. Hence the entire $\gX$ and $\gZ$ Pauli errors on the RHS of \Cref{eq:cnotout} are supported inside $E_{\mathrm{out},Z}$ and $E_{\mathrm{out},X}$ respectively, so $\cR[\cF](\sigma)$ is a $\comp{E_{\mathrm{out}}}$-deviation of $\Enc_{\cC}^{r_X}\circ\gCNOT^{\otimes M^{r_X}}(\rho)$, as desired.
\end{proof}

\subsubsection{Classically-Controlled Paulis}
In this section, we present our gadget for performing logical classically-controlled Pauli gates, by performing physical classically-controlled Pauli gates for the appropriate logical operators.

% \lnote{TODO: may want to discuss how we need reweightings in the logical channel here because we have to commute the classically-controlled Pauli through the input/fault error, even if the physical circuit did not have reweightings?}

\begin{lemma}
  \label{lem:ccp}
  Let $\cC^*=\cC^*(\epsilon,\bar{\ell},r_X,r_Z)$, $\cM^*=\cM^*(\epsilon,\bar{\ell},r_X,r_Z)$, and $\Enc_{\cC}:\cM^*\rightarrow\cC^*$ as in \Cref{def:prodcode}. Let $Q=(Q_X,Q_Z)$ be the $[[n,k]]$ code at level $r_X$ of $\cC^*$.

  For arbitrary $\eta,\gamma_{\mathrm{in}},\omega_{\mathrm{in}},\gamma_{\mathrm{run}},\omega_{\mathrm{run}}>0$, let
  \begin{equation}
    \label{eq:ccpparams}
    \gamma_{\mathrm{out}} = \gamma_{\mathrm{in}}+3\gamma_{\mathrm{run}}, \hspace{1em} \omega_{\mathrm{out}} = \omega_{\mathrm{in}}+3\omega_{\mathrm{run}}
  \end{equation}
  and
  \begin{align*}
    \cE_{\mathrm{run}} &= \cE_{\mathrm{run}}^{\cC}(\eta,\gamma_{\mathrm{run}},\omega_{\mathrm{run}};\beta).
  \end{align*}
  (see \Cref{def:badfault}). For $\alpha\in\{\mathrm{in},\mathrm{out}\}$, define the decorated code
  \begin{align*}
    D_\alpha = (Q,\; \Enc_{\cC}^{r_X},\; \cE_\alpha = \cE^{\cC}(\eta,\gamma_\alpha,\omega_\alpha;\beta)).
  \end{align*}
  (see \Cref{def:badfams}).

  For $P\in\{\gX,\gZ\}$, let $\bar{\cO}_P=\gCt{P}^{\otimes k}[*]$ denote the set of all reweighted versions of the channel that applies $k$ independent copies of the classically-controlled Pauli $\gCt{P}$. Therefore each $\bar{O}=\gCt{P}^{\otimes k}[\zeta]\in\bar{\cO}_P$ has $k$ input and output bits, along with $k$ input and output qubits. Then there exists a fault-tolerant gadget $\bfG_P=(\cR_P,\cE_{\mathrm{run}}^{\sqcup T},D_{\mathrm{in}},D_{\mathrm{out}},\emptyset)$ for $\bar{\cO}_P$, where $\cR_P$ is a quantum circuit using quantum space $\Qu{N}=C^{r_X}$ (so $|\Qu{N}|=n$), time $T=3$, and gate set $\cG_P=\{\gCt{P},\gCF_*\}$.
  % Note that third timestep is for terminating extra classical bits, which can be done with a classical function gate
\end{lemma}

Note that the set $\bar{\cO}_P=\gCt{P}^{\otimes k}[*]$ of logical superoperators implemented by the gadget in \Cref{lem:ccp} allows for arbitrary reweightings of the logical classically-controlled Pauli gates. As described in \Cref{sec:faulttol}, such logical reweightings may occur even in the absence of physical reweightings, under appropriate non-physical (i.e.~non-CPTP) faults.

\begin{proof}[Proof of \Cref{lem:ccp}]
  We will present the proof for the $P=\gX$ case; the $P=\gZ$ case is analogous, except the argument uses the dual complex  ${\cC^\vee}^*=\cC_*\cong\cC^*(\epsilon,\bar{\ell},r_Z,r_X)$.

  Our desired circuit $\cR_X$ uses $T=3$ timesteps. Recalling that $k=|M^{r_X}|=\dim(H^{r_X}(\cM)$ and $n=|C^{r_X}|$, consider an input state
  $\sigma = \sum_{x\in\bF_2^k}\ket{x}\bra{x}\otimes\sigma_x$,
  where $x$ denotes the classical input bits, and each $\sigma_x\in\bC^{2^n\times 2^n}$ represents the associated state of the $n$ input qubits. In the first timestep, we apply a classical function gate to compute some $z=z(x)\in\Enc_{\cC}^{r_X}(x)\subseteq C^{r_X}$, where we can assume the mapping $x\mapsto z(x)$ is linear (over $\bF_2$) because $\Enc_{\cC}^{r_X}$ is linear. In the second timestep, we apply the classically-controlled gates $\gCt{X}^{\otimes C^{r_X}}$ with classical control bits $x$ and target qubits $\sigma_x$, which therefore applies $\gX^x$ to the qubits. In the final timestep, we terminate the classical bits $z$ (which can be implemented with a classical function gate that maps $z\in\bF_2^n$ to the empty bit string $\emptyset\in\bF_2^0$).

  We now show that $\bfG_X$ is a fault-tolerant gadget for $\bar{\cO}_X=\gCt{X}^{\otimes k}[*]$. Fix a $\cE_{\mathrm{in}}$-avoiding pair $E_{\mathrm{in}}=(E_{\mathrm{in},X}\subseteq C^{r_X},\;E_{\mathrm{in},Z}\subseteq C^{r_X})$, and a $\cE_{\mathrm{run}}^{\sqcup T}$-avoiding set $E_{\mathrm{run}}$. By \Cref{lem:paulift}, it suffices to show that there exists a $\cE_{\mathrm{out}}$-avoiding pair $E_{\mathrm{out}}=(E_{\mathrm{out},X}\subseteq C^{r_X},\;E_{\mathrm{out},Z}\subseteq C^{r_X})$ such that for every finite set $\Rf{K}$, every operator $\rho\in\bC^{2^{M^{r_X}\sqcup M^{r_X}\sqcup\Rf{K}}\times 2^{M^{r_X}\sqcup M^{r_X}\sqcup\Rf{K}}}$ that is classical on bits in the first copy of $M^{r_X}$, every Pauli $\comp{E_{\mathrm{in}}}$-deviation $\sigma$ of $\Enc_{\cC}^{r_X}(\rho)$, every $\comp{E_{\mathrm{run}}}$-avoiding Pauli fault $\cF$, and every reweighting $\zeta$ for $\cR_X$, then there exists a reweighting $\bar{\zeta}$ for $\gCt{X}^{\otimes k}$ such that $\cR_X[\cF,\zeta](\sigma)$ is a $\comp{E_{\mathrm{out}}}$-deviation of $\Enc_{\cC}^{r_X}\circ\gCt{X}^{\otimes k}[\bar{\zeta}](\rho)$.
  % Note that in our specification of the input and output deviations, $M^{r_X}$ refers to the set of $k$ classical bits.

  To define $E_{\mathrm{out}}$, we first let $E\subseteq C^{r_X}$ be the set
  \begin{equation*}
    E = \bigcup_{t\in[T]}(E_{\mathrm{run}}\cap(N\times\{t\}))
  \end{equation*}
  of all qubits that may lie in the support of the fault across all timesteps, so that $E$ is $\cE_{\mathrm{run}}^{\cC}(\eta,T\gamma_{\mathrm{run}},T\omega_{\mathrm{run}};\beta)$-avoiding. Therefore by \Cref{lem:redfault}, the set
  \begin{equation*}
    \bar{E}_X = \Shad^{\times r_X}(E)\cup(C^{r_X}\setminus(V_Z\cup\cE_{Z,\Rd}^{\cC}))
  \end{equation*}
  is $\cE_Z^{\cC}(\eta,T\gamma_{\mathrm{run}},T\omega_{\mathrm{run}};\beta)$-avoiding. Similarly, applying \Cref{lem:redfault} to the dual complex ${\cC^{\vee}}^*=\cC_*\cong\cC^*(\epsilon,\bar{\ell},r_Z,r_X)$, we have that the set
  \begin{equation*}
    \bar{E}_Z = \Shad^{\times r_Z}(E)\cup(C^{r_X}\setminus(V_X\cup\cE_{X,\Rd}^{\cC}))
  \end{equation*}
  is $\cE_X^{\cC}(\eta,T\gamma_{\mathrm{run}},T\omega_{\mathrm{run}};\beta)$-avoiding. Note that here $\Shad^{\times r_X}$ applies $\Shad$ to the first $r_X$ components of the $(r_X+r_Z)$-dimensional product set $C$, and $\Shad^{\times r_Z}$ applies $\Shad$ to the latter $r_Z$ components. We then define
  \begin{align*}
    E_{\mathrm{out},X} &= E_{\mathrm{in},X} \cup \bar{E}_Z, \hspace{1em} E_{\mathrm{out},Z} = E_{\mathrm{in},Z} \cup \bar{E}_X,
  \end{align*}
  so that $E_{\mathrm{out}}$ is
  \begin{equation*}
    \cE^{\cC}(\eta,\; \gamma_{\mathrm{in}}+T\gamma_{\mathrm{run}},\; \omega_{\mathrm{in}}+T\omega_{\mathrm{run}};\; \beta) = \cE^{\cC}(\eta,\gamma_{\mathrm{out}},\omega_{\mathrm{out}};\beta)
  \end{equation*}
  avoiding, where the equality above holds by \Cref{eq:ccpparams} because $T=3$.

  Our goal is now to prove that there exists a reweighting $\bar{\zeta}$ for $\gCt{X}^{\otimes k}$ such that $\cR_X[\cF,\zeta](\sigma)$ is a $\comp{E_{\mathrm{out}}}$-deviation of $\Enc_{\cC}^{r_X}\circ\gCt{X}^{\otimes k}[\bar{\zeta}](\rho)$.
  To begin, we write
  \begin{align*}
    \rho
    &= \sum_{x,y,y'\in\bF_2^k} \ket{x}\bra{x} \otimes \ket{y}\bra{y'} \otimes \rho_{x,y,y'},
  \end{align*}
  where $\rho_{x,y,y'}\in\bC^{2^{\Rf{K}}\times 2^{\Rf{K}}}$. Then there exist $e_X,e_X',e_Z,e_Z'\in\cC^{r_X}$ such that $(\supp(e_X)\cup\supp(e_X'),\supp(e_Z)\cup\supp(e_Z'))$ is $\comp{E_{\mathrm{in}}}$-avoiding, and such that
  \begin{align*}
      \sigma
      &= \sum_{x,y,y'\in\bF_2^k} \ket{x}\bra{x} \otimes \gZ^{e_Z}\gX^{e_X} \ket{\Enc_{\cC}^{r_X}(y)}\bra{\Enc_{\cC}^{r_X}(y')} \gX^{e_X'}\gZ^{e_Z'} \otimes \rho_{x,y,y'}.
  \end{align*}
  We let $f_X,f_X',f_Z,f_Z'\in\cC^{r_X}$ specify the Pauli error $F_1$ from the fault following the first timestep, so that $\supp(f_X)\cup\supp(f_X')\cup\supp(f_Z)\cup\supp(f_Z')$ is $\comp{E_{\mathrm{run}}}$-avoiding. Then after the first timestep of $\cR_X[\cF,\zeta](\sigma)$, the state becomes
  \begin{align*}
    &\sum_{x,y,y'\in\bF_2^k} \ket{x}\bra{x} \otimes \ket{z(x)}\bra{z(x)} \otimes \gZ^{e_Z+f_Z}\gX^{e_X+f_X} \ket{\Enc_{\cC}^{r_X}(y)}\bra{\Enc_{\cC}^{r_X}(y')} \gX^{e_X'+f_X'}\gZ^{e_Z'+f_Z'} \otimes \rho_{x,y,y'}.
  \end{align*}
  Recalling that $z(x)\in\Enc_{\cC}^{r_X}(x)$, it follows that the subsequent application of $\gCt{X}^{\otimes r_X}[\zeta]$ in the second timestep maps the state above to
  \begin{align*}
    & \sum_{x,y,y'\in\bF_2^k} (-1)^{\zeta\cdot z(x)} \cdot \ket{x}\bra{x} \otimes \ket{z(x)}\bra{z(x)} \\
    &\hspace{4em} \otimes \gX^{z(x)}\gZ^{e_Z+f_Z}\gX^{e_X+f_X} \ket{\Enc_{\cC}^{r_X}(y)}\bra{\Enc_{\cC}^{r_X}(y')} \gX^{e_X'+f_X'}\gZ^{e_Z'+f_Z'}\gX^{z(x)} \otimes \rho_{x,y,y'} \\
    &= \sum_{x,y,y'\in\bF_2^k} (-1)^{(\zeta+e_Z+f_Z+e_Z'+f_Z')\cdot z(x)} \cdot \ket{x}\bra{x} \otimes \ket{z(x)}\bra{z(x)} \\
    &\hspace{4em} \otimes \gZ^{e_Z+f_Z}\gX^{e_X+f_X} \ket{\Enc_{\cC}^{r_X}(x+y)}\bra{\Enc_{\cC}^{r_X}(x+y')} \gX^{e_X'+f_X'}\gZ^{e_Z'+f_Z'} \otimes \rho_{x,y,y'}.
  \end{align*}
  Then $\cR_X[\cF,\zeta](\sigma)$ terminates the $z(x)$ register and returns the resulting state, which also experiences the $\comp{E_{\mathrm{run}}}$-avoiding Pauli errors $F_2$ and $F_3$. Let $\tilde{f}_X,\tilde{f}_X',\tilde{f}_Z,\tilde{f}_Z'\in\cC^{r_X}$ specify the accumulated Pauli fault error $F_1\circ F_2\circ F_3$, so that $\supp(\tilde{f}_X)\cup\supp(\tilde{f}_X')\cup\supp(\tilde{f}_Z)\cup\supp(\tilde{f}_Z')\subseteq E$. Then $\cR_X[\cF,\zeta](\sigma)$ returns the state
  \begin{align}
    \label{eq:ccpret1}
    \begin{split}
      & \sum_{x,y,y'\in\bF_2^k} (-1)^{(\zeta+e_Z+f_Z+e_Z'+f_Z')\cdot z(x)} \cdot \ket{x}\bra{x} \\
      &\hspace{4em} \otimes \gZ^{e_Z+\tilde{f}_Z}\gX^{e_X+\tilde{f}_X} \ket{\Enc_{\cC}^{r_X}(x+y)}\bra{\Enc_{\cC}^{r_X}(x+y')} \gX^{e_X'+\tilde{f}_X'}\gZ^{e_Z'+\tilde{f}_Z'} \otimes \rho_{x,y,y'}.
    \end{split}
  \end{align}

  We now simplify this returned state. Let $\Rdn^\vee$ denote the reduction map for the dual complex ${\cC^\vee}^*=\cC_*\cong\cC(\epsilon,\bar{\ell},r_Z,r_X)$. By definition, for every $b\in\cC^{r_X}$ the Paulis $\gX^{\delta^{\cC}(\Rd(b))}$ and $\gX^{\partial^{\cC}(\Rd^\vee(b))}$ are stabilizers of the codespace. Therefore we can replace each Pauli error in \Cref{eq:ccpret1} by its reduction, to obtain the equivalent expression
  \begin{align}
    \label{eq:ccpret2}
    \begin{split}
      & \sum_{x,y,y'\in\bF_2^k} (-1)^{(\zeta+e_Z+f_Z+e_Z'+f_Z')\cdot z(x)} \cdot \ket{x}\bra{x} \\
      &\hspace{4em} \otimes \gZ^{e_Z+\Rdn(\tilde{f}_Z)}\gX^{e_X+\Rdn(\tilde{f}_X)} \ket{\Enc_{\cC}^{r_X}(x+y)}\bra{\Enc_{\cC}^{r_X}(x+y')} \gX^{e_X'+\Rdn(\tilde{f}_X')}\gZ^{e_Z'+\Rdn(\tilde{f}_Z')} \otimes \rho_{x,y,y'}.
    \end{split}
  \end{align}
  Here recall that $e_X,e_X',e_Z,e_Z'$ are already reduced by the definition of $\cE_{\mathrm{in}}$ (see \Cref{def:badfams}).
  By definition the mapping $x\mapsto(\zeta+e_Z+f_Z+e_Z'+f_Z')\cdot z(x)$ is linear, so there exists $\bar{\zeta}\in\bF_2^k$ such that $(\zeta+e_Z+f_Z+e_Z'+f_Z')\cdot z(x)=\bar{\zeta}\cdot x$. Furthermore, \Cref{lem:redfault} implies that $\Rdn(\tilde{f}_X),\Rdn(\tilde{f}_X')$ are supported inside $\bar{E}_X$, and $\Rdn^\vee(\tilde{f}_Z),\Rdn^\vee(\tilde{f}_Z')$ are supported inside $\bar{E}_Z$. Let
  \begin{equation*}
    \tilde{e}_X = e_X+\Rdn(\tilde{f}_X), \hspace{1em} \tilde{e}_X' = e_X'+\Rdn(\tilde{f}_X'), \hspace{1em} \tilde{e}_Z = e_Z+\Rdn(\tilde{f}_Z), \hspace{1em} \tilde{e}_Z' = e_Z'+\Rdn(\tilde{f}_Z')
  \end{equation*}
  specify the Pauli errors on the output expression in \Cref{eq:ccpret2}.
  It follows that $\supp(\tilde{e}_X)\cup\supp(\tilde{e}_X')\subseteq E_{\mathrm{in},Z}\cup\bar{E}_X$ and $\supp(\tilde{e}_Z)\cup\supp(\tilde{e}_Z')\subseteq E_{\mathrm{in},X}\cup\bar{E}_Z$, so that $(\supp(\tilde{e}_Z)\cup\supp(\tilde{e}_Z'),\supp(\tilde{e}_X)\cup\supp(\tilde{e}_X'))$ is $\comp{E_{\mathrm{out}}}$-avoiding.
  % Recalling that $\cE_{\mathrm{in}} = \cE^{\cC}(\eta,\gamma_{\mathrm{in}},\omega_{\mathrm{in}};\beta))$ and that $\gamma_{\mathrm{out}}=\gamma_{\mathrm{in}}+6\gamma_{\mathrm{run}}$, $\omega_{\mathrm{out}}=\omega_{\mathrm{in}}+6\omega_{\mathrm{run}}$, it follows that $(\supp(\tilde{e}_X)\cup\supp(\tilde{e}_X'),\supp(\tilde{e}_Z)\cup\supp(\tilde{e}_Z'))$ is $\cE_{\mathrm{out}}=\cE^{\cC}(\eta,\gamma_{\mathrm{out}},\omega_{\mathrm{out}};\beta)$-avoiding.
  Therefore we can rewrite the output state of $\cR_X[\cF,\zeta](\sigma)$ in \Cref{eq:ccpret2} as
  \begin{align}
    \label{eq:ccpret3}
    \begin{split}
      \hspace{1em}&\hspace{-1em} \sum_{x,y,y'\in\bF_2^k} (-1)^{\bar{\zeta}\cdot x} \cdot \ket{x}\bra{x} \otimes \gZ^{\tilde{e}_Z}\gX^{\tilde{e}_X} \ket{\Enc_{\cC}^{r_X}(x+y)}\bra{\Enc_{\cC}^{r_X}(x+y')} \gX^{\tilde{e}_X'}\gZ^{\tilde{e}_Z'} \otimes \rho_{x,y,y'} \\
                  &= \gZ^{\tilde{e}_Z}\gX^{\tilde{e}_X}\Enc_{\cC}^{r_X}\left(\sum_{x,y,y'\in\bF_2^k} (-1)^{\bar{\zeta}\cdot x} \cdot \ket{x}\bra{x} \otimes \ket{x+y}\bra{x+y'} \otimes \rho_{x,y,y'}\right) \gX^{\tilde{e}_X'}\gZ^{\tilde{e}_Z'} \\
      &= \gZ^{\tilde{e}_Z}\gX^{\tilde{e}_X}\Enc_{\cC}^{r_X}\left(\gCt{X}^{\otimes k}[\bar{\zeta}](\rho)\right) \gX^{\tilde{e}_X'}\gZ^{\tilde{e}_Z'}.
    \end{split}
  \end{align}
  We have shown that the RHS of \Cref{eq:ccpret3} above is a $\comp{E_{\mathrm{out}}}$-deviation of $\Enc_{\cC}^{r_X}\left(\gCt{X}^{\otimes k}[\bar{\zeta}](\rho)\right)$, as desired. Thus $\bfG_X$ is a fault-tolerant gadget for $\bar{\cO}_X=\gCt{X}^{\otimes k}[*]$, as desired.
\end{proof}

\subsubsection{Termination}
In this section, we present our termination gadget, which simply terminates all qubits in a code block. Because we terminate all qubits in the first timestep of the gadget, there is no fault or output code to consider. The analysis follows from that of the logical measurement gadget in \Cref{lem:logmeas}.

\begin{lemma}
  \label{lem:term}
   For $r_X,r_Z,\bar{\ell}\in\bN$, let $r=r_X+r_Z$, define $\epsilon=\epsilon(r)$ as in \Cref{lem:ssflip}, and define $\mu=\mu(\epsilon)$ and $\Delta=\Delta(\epsilon)$ as in \Cref{lem:lossless}. Define $\cC^*=\cC^*(\epsilon,\bar{\ell},r_X,r_Z)$, $\cM^*=\cM^*(\epsilon,\bar{\ell},r_X,r_Z)$, and $\Enc_{\cC}:\cM^*\rightarrow\cC^*$ as in \Cref{def:prodcode}. Let $Q=(Q_X,Q_Z)$ be the $[[n,k]]$ code at level $r_X$ of $\cC^*$.

  Let $\beta\geq 2\Delta$, and for arbitrary positive
  \begin{equation*}
    \eta_{\mathrm{in}} \leq \mu 2^{\bar{\ell}/2}, \hspace{1em} \gamma_{\mathrm{in}} \leq \frac{1}{16r\Delta^2}, \hspace{1em} \omega_{\mathrm{in}} \leq \frac{\beta^{\bar{\ell}/2}}{8r\beta\Delta^2},
  \end{equation*}
  let
  \begin{equation*}
    \cE_{\mathrm{in}} = \cE^{\cC}(\eta_{\mathrm{in}},\gamma_{\mathrm{in}},\omega_{\mathrm{in}};\beta)
  \end{equation*}
  be the family of bad sets defined in \Cref{def:badfams}. Define the decorated code
  \begin{equation*}
    D_{\mathrm{in}} = (Q,\; \Enc_{\cC}^{r_X},\; \cE_{\mathrm{in}}).
  \end{equation*}

  Let $\bar{O}=\gTerm^{\otimes k}:\bC^{2^k\times 2^k}\rightarrow\bC$ be the termination channel, which has $k$ input qubits and zero output qubits.
  Then there exists a fault-tolerant gadget $\bfG=(\cR,\emptyset,D_{\mathrm{in}},\emptyset,\emptyset)$ for $\bar{O}$, where $\cR$ is the quantum circuit using quantum space $\Qu{N}=C^{r_X}$ (so $|\Qu{N}|=n$), time $T=1$, and gate set $\cG_Z=\{\gTerm\}$ that simply applies $\gTerm$ to all qubits in its only timestep, i.e.~$R_1=\gTerm^{\otimes n}$.
\end{lemma}
\begin{proof}
  By definition
  \begin{equation}
    \label{eq:termdecomp}
    \gTerm = \gTerm\circ\gMZ,
  \end{equation}
  that is, a qubit can be terminated by first measuring it, and then terminating the classical measurement outcome bit. Hence the circuit $\cR_0$ that applies the logical measurement gadget in \Cref{lem:logmeas} followed by termination gates on all classical bits provides a termination gadget satisfying all desired properties of the statement of \Cref{lem:term}, except that the time usage is $T_0=3$ instead of the desired $T=1$. However, by \Cref{eq:termdecomp}, the circuit $\cR$ that simply applies $\gTerm$ to all qubits in the first timestep has the exact same effect as $\cR_0$, as $\cR_0$ applies $\gMZ$ to all qubits in the first timestep, and then terminates the resulting classical bits (after applying a classical function gate). Hence $\cR$ also satisfies the desired properties in the statement of \Cref{lem:term}.
\end{proof}

\subsection{Composed Gadgets}
In this section, we present our error-correction and injection gadgets, which follow by simply composing our gadgets in \Cref{sec:coregad,sec:basicgad} above.

\subsubsection{Error Correction}
In \Cref{lem:errcorr} below, we present our error-correction gadget. We use ``Knill-type'' error correction, in which we repeatedly prepare fresh logical Bell pairs that we use to teleport the input state to a new code block. This procedures ensures that no errors can propagate from the input block to the output block, which immediately implies the refreshing property in \Cref{def:faulttol}.

\begin{lemma}
  \label{lem:errcorr}
  For integers $r_X,r_Z\geq 2$, $\bar{\ell}\geq 1$, let $r=r_X+r_Z$, define $\epsilon=\epsilon(r)$ as in \Cref{lem:ssflip}, and define $\mu=\mu(\epsilon)$ and $\Delta=\Delta(\epsilon)$ as in \Cref{lem:lossless}. Define $\cC^*=\cC^*(\epsilon,\bar{\ell},r_X,r_Z)$, $\cM^*=\cM^*(\epsilon,\bar{\ell},r_X,r_Z)$, and $\Enc_{\cC}:\cM^*\rightarrow\cC^*$ as in \Cref{def:prodcode}. Let $Q=(Q_X,Q_Z)$ be the $[[n,k]]$ code at level $r_X$ of $\cC^*$.

  Let $\beta\geq 2\Delta$, and for arbitrary positive
  \begin{align*}
    & \eta_{\mathrm{in}} \leq \mu 2^{\bar{\ell}/2}, \hspace{1em} \gamma_{\mathrm{in}} \leq \frac{1}{2^{10}r\Delta^2}, \hspace{1em} \omega_{\mathrm{in}} \leq \frac{\beta^{\bar{\ell}/2}}{2^{10}r\beta\Delta^2}, \\
    & \eta_{\mathrm{run}} \leq \mu 2^{\bar{\ell}/2}, \hspace{1em} \gamma_{\mathrm{run}} \leq \frac{1}{2^{30}r^5\Delta^8}, \hspace{1em} \omega_{\mathrm{run}} \leq \frac{\beta^{\bar{\ell}/2}}{2^{30}r^5\beta^3\Delta^7},
  \end{align*}
  let
  \begin{align*}
    & \eta_{\mathrm{out}} = \eta_{\mathrm{run}}, \hspace{1em} \gamma_{\mathrm{out}} = 2^{20}r^3\Delta^5\cdot\gamma_{\mathrm{run}}, \hspace{1em} \omega_{\mathrm{out}} = 2^{20}r^3\beta^2\Delta^4\cdot\omega_{\mathrm{run}}.
  \end{align*}
  Then let
  \begin{align*}
    \cE_{\mathrm{run}} &= \cE_{\mathrm{run}}^{\cC}(\eta_{\mathrm{run}},\gamma_{\mathrm{run}},\omega_{\mathrm{run}};\beta),
  \end{align*}
  be the family of bad sets from \Cref{def:badfault}. For $\alpha\in\{\mathrm{in},\mathrm{out}\}$ let
  \begin{align*}
    \cE_\alpha &= \cE^{\cC}(\eta_\alpha,\gamma_\alpha,\omega_\alpha;\beta)
  \end{align*}
  be the family of bad sets from \Cref{def:badfams}, and define the decorated code
  \begin{align*}
    D_\alpha &= (Q,\; \Enc_{\cC}^{r_X},\; \cE_\alpha).
  \end{align*}

  Then there exists a mending refreshing fault-tolerant gadget $\bfG=(\cR,(\cE_{\mathrm{run}}^{\sqcup 3})^{\sqcup T},D_{\mathrm{in}},D_{\mathrm{out}},\ps)$ for the $k$-qubit identity channel $\bar{O}=I_k:\bC^{2^k\times 2^k}\rightarrow\bC^{2^k\times 2^k}$, where $\cR$ is a quantum circuit using quantum space
  \begin{equation*}
    \Qu{N} = C^{r_X} \sqcup (C^{r_X}\sqcup C^{r_X+1}) \sqcup (C^{r_X}\sqcup C^{r_X-1}),
  \end{equation*}
  (so $|\Qu{N}|\leq 12r\Delta n$), time $T=r(2\Delta+2)+17\leq 16r\Delta$, and gate set
  \begin{equation*}
    \cG = \{\gInitX,\gInitZ,\gTerm,\gCNOT,\gMX,\gMZ,\gCt{\gX},\gCt{\gZ},\gCF_*\}.
  \end{equation*}
\end{lemma}

\begin{figure}
  \centering
  \begin{quantikz}[row sep={1.3cm,between origins}, column sep=0.6cm]
    \lstick{$\ket{\psi}$} & \qw      & \qw       & \ctrl{1} & \meter{X} & \cw & \cw & \cwbend{2} \\
    \lstick{$\ket{+}$}    & \qw      & \ctrl{1}  & \targ{}  & \meter{Z} & \cw & \cwbend{1} \\
    \lstick{$\ket{0}$}    & \qw      & \targ{}   & \qw      & \qw         & \qw & \gate{X}     & \gate{Z} & \qw \rstick{$\ket{\psi}$}
  \end{quantikz}
  \caption{\label{fig:teleport} Circuit $\cR_{\mathrm{tel}}$ to teleport a qubit $\ket{\psi}$.}
\end{figure}

To prove \Cref{lem:errcorr}, we will construct a fault-tolerant gadget to apply the logical teleportation circuit $\cR_{\mathrm{tel}}$ in \Cref{fig:teleport} on all $k$ logical qubits simultaneously. We will use \Cref{claim:telpost} below, which shows that $\cR_{\mathrm{tel}}^{\otimes k}[*]$ only contains the desired logical identity channel $\bar{O}=I_k$. That is, we show that $\cR_{\mathrm{tel}}$ still implements the teleportation (i.e.~identity) channel, up to a global phase/scalar, even if arbitrary reweightings $\zeta$ are applied to the classicall-controlled Paulis.

\begin{claim}
  \label{claim:telpost}
  For every reweighting $\zeta$ for the circuit $\cR_{\mathrm{tel}}$ shown in \Cref{fig:teleport}, we have $\cR_{\mathrm{tel}}[\zeta](\cdot)\propto\cR_{\mathrm{tel}}(\cdot)=I$.
\end{claim}
\begin{proof}
  First decompose
  \begin{equation*}
    \cR_{\mathrm{tel}} = \sum_{\text{postselections }\xi}\cR_{\mathrm{tel}}[\xi](\cdot).
  \end{equation*}
  By definition for every postselection $\xi$ for $\cR_{\mathrm{tel}}$, then $\cR_{\mathrm{tel}}[\xi](\cdot)\propto\cR_{\mathrm{tel}}(\cdot)$, as teleportation succeeds even under postselected measurements. Now for each fixed postselected circuit $\cR_{\mathrm{tel}}[\xi]$, the classical control bit values are fixed, so a reweighting $\zeta$ then only introduces an additional global phase. That is, $\cR_{\mathrm{tel}}[\zeta,\xi](\cdot)\propto\cR_{\mathrm{tel}}[\xi](\cdot)\propto\cR_{\mathrm{tel}}(\cdot)$. Hence
  \begin{equation*}
    \cR_{\mathrm{tel}}[\zeta](\cdot) = \sum_{\text{postselections }\xi}\cR_{\mathrm{tel}}[\zeta,\xi](\cdot) \propto \cR_{\mathrm{tel}}(\cdot),
  \end{equation*}
  as desired.
\end{proof}

We will also use \Cref{claim:telfault} below, which shows that Pauli errors right before the measurements in $\cR_{\mathrm{tel}}$ are equivalent to Pauli errors on the input/output state, even under arbitrary reweightings. Here by a (single-qubit) ``Pauli error'' we mean a Pauli superoperator. Similarly and by a (single-qubit) ``$\gX$ (resp.~$\gZ$) Pauli error'' we mean a Pauli superoperator of the form $\rho\mapsto\gX^a\rho\gX^{a'}$ (resp.~$\rho\mapsto\gZ^a\rho\gZ^{a'}$) for some $a,a'\in\{0,1\}$. Thus every Pauli error can be expressed as the composition of a $\gX$ error and a $\gZ$ error.

\begin{claim}
  \label{claim:telfault}
  Let $\cF$ be a Pauli fault for the circuit $\cR_{\mathrm{tel}}$ shown in \Cref{fig:teleport} that is supported in the space-time locations just prior to the Pauli measurements. Let $\zeta$ be a reweighting of $\cR_{\mathrm{tel}}$. Then there exists a Pauli superoperator $\bar{F}:\bC^{2\times 2}\rightarrow\bC^{2\times 2}$ depending only on $\cF$ for which $\cR_{\mathrm{tel}}[\cF,\zeta](\cdot)\propto\bar{F}$.
\end{claim}
\begin{proof}
  By definition $\cF$ only induces Pauli errors on the first two qubits, just prior to the measurements in $\cR_{\mathrm{tel}}$. A Pauli $\gX$ (resp.~$\gZ$) error just prior to the $\gMX$ (resp.~$\gMZ$) measurement has no effect on the circuit's execution except to change the reweighting on the subsequent classically-controlled Pauli $\gZ$ (resp.~$\gX$) correction. (Such a changed reweighting is for instance possible if a non-physical error such as $\rho\mapsto\rho\gX$ is applied just prior to the $\gMX$ measurement.) Thus we have
  \begin{equation*}
    \cR_{\mathrm{tel}}[\cF,\zeta] = \cR_{\mathrm{tel}}[\cF',\zeta'],
  \end{equation*}
  where $\zeta'$ is our new reweighting that absorbs the $\gX$ (resp.~$\gZ$) error just prior to the $\gMX$ (resp.~$\gMZ$) measurement, and $\cF'$ is the fault obtained by removing these errors from $\cF$. That is, $\cF'$ only applies a Pauli $\gZ$ (resp.~$\gX$) error just prior to the $\gMX$ (resp.~$\gMZ$) measurement on the first (resp.~second) qubit. Now conjugating the errors in $\cF'$ by the preceding $\gCNOT$ gates, we then have that
  \begin{equation*}
    \cR_{\mathrm{tel}}[\cF',\zeta'] = \cR_{\mathrm{tel}}[\cF'',\zeta'],
  \end{equation*}
  where $\cF''$ is a fault supported entirely in the first timestep of $\cR_{\mathrm{tel}}$, which applies a Pauli $\gZ$ error to the first qubit, and applies Pauli $\gX$ errors to the second and third qubits. However, because the second qubit is initialized to $\ket{+}\bra{+}$, which is invariant to $\gX$ errors, we in fact have
  \begin{equation*}
    \cR_{\mathrm{tel}}[\cF'',\zeta'] = \cR_{\mathrm{tel}}[\cF''',\zeta'],
  \end{equation*}
  where $\cF'''$ is again a fault supported entirely in the first timestep, which applies a $\gZ$ error to the first qubit, no error to the second qubit, and a $\gX$ error to the third qubit. We may now conjugate this $\gX$ error on the third qubit by the $\gCNOT$, $\gCt{\gX}$, and $\gCt{\gZ}$ gates applied to the third qubit, in order to push the $\gX$ error to the end of the circuit. Note that conjugating the $\gX$ error by $\gCt{\gZ}$ may introduce a $\gZ$ error on the classical control bit, or equivalently, may again change the reweighting. Thus
  \begin{equation*}
    \cR_{\mathrm{tel}}[\cF''',\zeta'] = \bar{F}_{\gX}\circ\cR_{\mathrm{tel}}[\zeta'']\circ\bar{F}_{\gZ},
  \end{equation*}
  where $\zeta''$ is a reweighting for $\cR_{\mathrm{tel}}$, and $\bar{F}_{\gX}$ (resp.~$\bar{F}_{\gZ}$) is a single-qubit $\gX$ (resp.~$\gZ$) error. But by \Cref{claim:telpost} we have $\cR_{\mathrm{tel}}[\zeta'']\propto I$, so combining the above equalities gives that
  \begin{equation*}
    \cR_{\mathrm{tel}}[\cF,\zeta] \propto \bar{F}_{\gX}\circ\bar{F}_{\gZ},
  \end{equation*}
  as desired.
\end{proof}

\begin{proof}[Proof of \Cref{lem:errcorr}]
  At a high level, we define $\cR$ to prepare two fresh code states containing logical Bell pairs, which we then use to teleport the input state into one of the two fresh code states. Specifically, we implement the teleportation circuit in \Cref{fig:teleport} on all logical qubits in the code blocks simultaneously. No errors can propagate from the first (input) block to the third (output) block, as after the initial Bell state preparation (which does not involve the input block), the only gates applied to the output block are classically-controlled Paulis (which do not propagate errors). Hence the gadget will be fault-tolerant and refreshing. The mending property follows because if the input has a bad corruption, we simply apply the wrong classically-controlled Pauli correction, resulting in a logical Pauli error on an otherwise valid code state.

  We now present the details. Let
  \begin{equation*}
    \Qu{N}_1 = C^{r_X}, \hspace{1em} \Qu{N}_2 = C^{r_X}\sqcup C^{r_X+1}, \hspace{1em} \Qu{N}_3 = C^{r_X}\sqcup C^{r_X-1},
  \end{equation*}
  so that $\Qu{N}=\Qu{N}_1\sqcup \Qu{N}_2\sqcup\Qu{N}_3$. We for instance write $\Qu{N}_2|_{C^{r_X}}$ to denote the copy of $C^{r_X}$ inside $\Qu{N}_2$. Our circuit $\cR$ takes as input qubits $N_{\mathrm{in}}=\Qu{N}_1$, and outputs qubits $N_{\mathrm{out}}=\Qu{N}_3|_{C^{r_X}}$. We define $\cR$ to perform the following operations in order to implement the logical circuit in \Cref{fig:teleport} on all $k$ logical qubits:
  \begin{enumerate}
  \item ($r(2\Delta+2)+6\leq 8r\Delta$ timesteps) Apply the $\ket{+}^{\otimes k}$-state preparation gadget in \Cref{lem:stateprep} on qubits $\Qu{N}_2$, and apply the $\ket{0}^{\otimes k}$-state preparation gadget in \Cref{lem:stateprep} on qubits $\Qu{N}_3$, while applying the idling gadget in \Cref{lem:idle} on qubits $\Qu{N}_1$.
  \item ($1$ timestep) Apply the $\gCNOT^{\otimes k}$ gadget in \Cref{lem:cnot} on qubits $\Qu{N}_2|_{C^{r_X}},\Qu{N}_3|_{C^{r_X}}$, while applying the idling gadget in \Cref{lem:idle} on qubits $\Qu{N}_1$.
  \item ($1$ timestep) Apply the $\gCNOT^{\otimes k}$ gadget in \Cref{lem:cnot} on qubits $\Qu{N}_1,\Qu{N}_2|_{C^{r_X}}$, while applying the idling gadget in \Cref{lem:idle} on qubits $\Qu{N}_3|_{C^{r_X}}$.
  \item ($2$ timesteps) Apply the $\gMX^{\otimes k}$ gadget in \Cref{lem:logmeas} on qubits $\Qu{N}_1$, and apply the $\gMZ^{\otimes k}$ gadget in \Cref{lem:logmeas} on qubits $\Qu{N}_2|_{C^{r_X}}$, while applying the idling gadget in \Cref{lem:idle} on qubits $\Qu{N}_3|_{C^{r_X}}$. Let $x_1,x_2\in\bF_2^k$ denote the classical output strings of the respective measurement gadgets.
  % \item Run a classical function gate to compute $z_1=z_1(x_1),\;z_2=z_2(x_2)\in\bF_2^n$ satisfying $z_1\in\Enc_{\cC}(x_1)$ and $z_2\in\Enc_{\cC^\vee}(x_2)$, where $\Enc_{\cC^\vee}$ is the dual encoding map to $\Enc_{\cC}$ (see \Cref{def:encdual,remark:encdual}). Meanwhile, apply the idling gadget in \Cref{lem:idle} on qubits $\Qu{N}_3|_{C^{r_X}}$.
  \item ($3$ timesteps) Apply the $\gCt{\gX}^{\otimes k}$ gadget in \Cref{lem:ccp} on control bits $x_2$ and target qubits~$\Qu{N}_3$.
  \item ($3$ timesteps) Apply the $\gCt{\gZ}^{\otimes k}$ gadget in \Cref{lem:ccp} on control bits $x_1$ and target qubits~$\Qu{N}_3$.
  \item ($1$ timestep) Apply $\gTerm^{\otimes 2k}$ to the classical bits $(x_1,x_2)$, while applying the idling gadget in \Cref{lem:idle} on qubits $\Qu{N}_3|_{C^{r_X}}$.
  \end{enumerate}
  By definition, $\cR$ uses quantum space $\Qu{N}$ and time $T=r(2\Delta+2)+17\leq 16r\Delta$.

  We may apply \Cref{lem:parcomp} across the gadgets applied in parallel in each of the steps defining $\cR$ above, and then apply \Cref{lem:seqcomp} across these parallel-compsed gadgets applied in sequence. Letting $\cR_{\mathrm{tel}}$ denote the logical teleportation circuit in \Cref{fig:teleport}, then it follows from the gadget lemmas listed above that the resulting composition $\bfG=(\cR,(\cE_{\mathrm{run}}^{\sqcup 3})^{\sqcup T},D_{\mathrm{in}},D_{\mathrm{out}},\ps)$ provides a fault-tolerant gadget for $\cR_{\mathrm{tel}}^{\otimes k}[*]$.

  It now only remains to show the mending and refreshing properties. For the refreshing property, we simply observe that when composing the gadgets above and applying \Cref{lem:seqcomp,lem:parcomp} to construct the $\cE_{\mathrm{out}}$-avoiding pair $E_{\mathrm{out}}$ given the $\cE_{\mathrm{in}}$-avoiding pair $E_{\mathrm{in}}$, the $(\cE_{\mathrm{run}}^{\sqcup 3})^{\sqcup T}$-avoiding set $E_{\mathrm{run}}$, and the postselection $\xi$, then $E_{\mathrm{out}}$ will depend on $E_{\mathrm{run}}$ and on $\xi$, but not on $E_{\mathrm{in}}$. Indeed, $E_{\mathrm{out}}$ is constructed by repeatedly applying \Cref{lem:seqcomp,lem:parcomp} to compose all gadgets described above into the circuit $\cR$. The state preparation and $\gCNOT$ gadgets acting on $\Qu{N}_2\sqcup\Qu{N}_3$ at the start of $\cR$ do not touch qubits $\Qu{N}_1$, so the resulting error sets on $\Qu{N}_2\sqcup\Qu{N}_3$ cannot depend on $E_{\mathrm{in}}\subseteq\Qu{N}_1$. The only subsequent gadgets applied on qubits $\Qu{N}_3$ are idle and classically-controlled Pauli gadgets, both of which do not touch any qubits outside of $\Qu{N}_3$, and hence cannot introduce a dependence of the output error set on $E_{\mathrm{in}}$. Thus the output error set $E_{\mathrm{out}}$ does not depend on $E_{\mathrm{in}}$, as desired.

  To show the mending property, consider a $(\cE_{\mathrm{run}}^{\sqcup 3})^{\sqcup T}$-avoiding set $E_{\mathrm{run}}$ and a postselection $\xi\in\bF_2^{\ps}$ for $\cR$. Let $E_{\mathrm{out}}$ be the associated $\cE_{\mathrm{out}}$-avoiding pair given by the definition of refreshing fault-tolerance. Here the refreshing property ensures that $E_{\mathrm{out}}$ is well-defined as a function of $E_{\mathrm{run}}$ and $\xi$, without needing a choice of $E_{\mathrm{in}}$. By \Cref{lem:paulift}, our goal is to show that for every finite set $\Rf{K}$, every operator $\rho\in\bC^{2^{M^{r_X}\sqcup\Rf{K}}\times 2^{M^{r_X}\sqcup\Rf{K}}}$, every Pauli $\Rf{K}$-deviation $\sigma$ of $\Enc_{\cC}^{r_X}(\rho)$, and every $\comp{E_{\mathrm{run}}}$-avoiding Pauli fault $\cF$ and reweighting $\zeta$ for $\cR$, then $\cR[\cF,\zeta,\xi](\sigma)$ is a linear combination of $\comp{E_{\mathrm{out}}}$-deviations of $\Enc_{\cC}^{r_X}\circ L(\rho)$ for superoperators $L$ acting on the $k$ output logical qubits.

  For this purpose, let $F_{\mathrm{in}}$ be the Pauli error superoperator acting on the circuit's input state, so that $\sigma=F_{\mathrm{in}}\circ\Enc_{\cC}^{r_X}(\rho)$. By definition, $F_{\mathrm{in}}$ simply propagates through the $\gCNOT$ gadget applied to qubits $\Qu{N}_1,\Qu{N}_2$ (along with any additional Pauli errors accumulated during idling timesteps), before being measured out in the logical measurement gadgets. Because these logical measurement gadgets from \Cref{lem:logmeas} are mending, it follows that $\cR[\cF,\zeta,\xi](\sigma)$ outputs a linear combination of $\comp{E_{\mathrm{out}}}$-deviations of states of the form $\Enc_{\cC}^{r_X}\circ\cR_{\mathrm{tel}}^{\otimes k}[\bar{\cF},\bar{\zeta}](\rho)$ for logical reweightings $\bar{\zeta}$, and for logical Pauli faults $\bar{\cF}$ supported only at the locations just prior to the measurements in $\cR_{\mathrm{tel}}^{\otimes k}$.
  Then by \Cref{claim:telfault}, for each such $\bar{\cF}$ there exists a $k$-qubit Pauli superoperator $L$ for which
  \begin{equation*}
    \Enc_{\cC}^{r_X}\circ\cR_{\mathrm{tel}}^{\otimes k}[\bar{\cF},\bar{\zeta}](\rho) \propto \Enc_{\cC}^{r_X}\circ L(\rho).
  \end{equation*}
  Thus $\cR[\cF,\zeta,\xi](\sigma)$ outputs a linear combination of $\comp{E_{\mathrm{out}}}$-deviations of states of the form $\Enc_{\cC}^{r_X}\circ L(\rho)$ for $k$-qubit superoperators $L$, as desired, so $\bfG$ is indeed mending.
\end{proof}

\subsubsection{Injection}
In \Cref{lem:inject} below, we present our injection gadget. Like our error-correction gadget in \Cref{lem:errcorr}, our injection gadget begins by preparing two fresh code blocks containing logical Bell pairs. However, we now apply our ejection gadget from \Cref{lem:eject} to one half of the Bell pairs, which we then use the teleport the input state into the other (encoded) half of the Bell pairs. We again perform teleportation using the circuit in \Cref{fig:teleport}.

\begin{lemma}
  \label{lem:inject}
  For integers $r_X,r_Z\geq 2$, $\bar{\ell}\geq 1$, let $r=r_X+r_Z$, define $\epsilon=\epsilon(r)$ as in \Cref{lem:ssflip}, and define $\mu=\mu(\epsilon)$ and $\Delta=\Delta(\epsilon)$ as in \Cref{lem:lossless}. Define $\cC^*=\cC^*(\epsilon,\bar{\ell},r_X,r_Z)$, $\cM^*=\cM^*(\epsilon,\bar{\ell},r_X,r_Z)$, and $\Enc_{\cC}:\cM^*\rightarrow\cC^*$ as in \Cref{def:prodcode}. Let $Q=(Q_X,Q_Z)$ be the $[[n,k]]$ code at level $r_X$ of $\cC^*$.

  Let $\beta\geq 2\Delta$, and for $\omega_{\mathrm{in}},\eta_{\mathrm{run}},\gamma_{\mathrm{run}},\omega_{\mathrm{run}}>0$ satisfying
  \begin{equation*}
    \eta_{\mathrm{run}} \leq \mu 2^{\bar{\ell}/2}, \hspace{1em} \gamma_{\mathrm{run}} \leq \frac{1}{2^{12}r^3\Delta^5}, \hspace{1em} \omega_{\mathrm{run}} \leq \frac{\beta^{\bar{\ell}/2}}{2^9r^3\beta^2\Delta^4}
  \end{equation*}
  (as in \Cref{eq:sprunbounds}),
  let
  \begin{align*}
    \bar{\omega} &= \omega_{\mathrm{in}}+2^{30}r^3\beta^2\Delta^4\cdot\omega_{\mathrm{run}}, \hspace{1em} \eta_{\mathrm{out}} = \eta_{\mathrm{run}}, \hspace{1em} \gamma_{\mathrm{out}} = 2^{30}r^3\Delta^5\cdot\gamma_{\mathrm{run}}, \hspace{1em} \omega_{\mathrm{out}} = 2^{30}r^3\beta^2\Delta^4\cdot\omega_{\mathrm{run}}.
  \end{align*}
  Then let
  \begin{align*}
    \cE_{\mathrm{run},1} &= 2^{M^{r_X}}|^{\geq\omega_{\mathrm{run},1}} \\
    \cE_{\mathrm{run},2} = \cE_{\mathrm{run},3} &= \cE_{\mathrm{run}}^{\cC}(\eta_{\mathrm{run}},\gamma_{\mathrm{run}},\omega_{\mathrm{run}};\beta) \\
    \cE_{\mathrm{run}} &= \cE_{\mathrm{run},1}\sqcup\cE_{\mathrm{run},2}\sqcup\cE_{\mathrm{run},3},
  \end{align*}
  where $\cE_{\mathrm{run}}^{\cC}(\cdot)$ is given by \Cref{def:badfault}. Define the decorated codes
  \begin{align*}
    D_{\mathrm{in}} &= ((\bF_2^{M^{r_X}},\bF_2^{M^{r_X}}),\; \id,\; \cE_{\mathrm{in}}=2^{M^{r_X}}|^{\geq\omega_{\mathrm{in}}}) \\
    D_{\mathrm{in}}(E_{\mathrm{in}}) &= ((\bF_2^{M^{r_X}},\bF_2^{M^{r_X}}),\; \id,\; \cE_{\mathrm{in}}(E_{\mathrm{in}})=\comp{E_{\mathrm{in}}}) \hspace{1em} \forall E_{\mathrm{in}}\subseteq M^{r_X} \\
    D_{\mathrm{out}} &= (Q,\; \Enc_{\cC}^{r_X},\; \cE_{\mathrm{out}} = \cE^{\cC}(\eta_{\mathrm{out}},\gamma_{\mathrm{out}},\omega_{\mathrm{out}};\beta)),
  \end{align*}
  where the encoding map for $D_{\mathrm{in}}$ is the $k=|M^{r_X}|$-qubit identity map $\id:\bF_2^{M^{r_X}}\rightarrow\bF_2^{M^{r_X}}$, and $\cE^{\cC}(\cdot)$ is given by \Cref{def:badfams}.

  Let $\bar{\cE}=2^{M^{r_X}}|^{\geq\bar{\omega}}$. Then there exists a quantum circuit $\cR$ using quantum space
  \begin{equation*}
    \Qu{N} = M^{r_X} \sqcup (C^{r_X}\sqcup C^{r_X+1}) \sqcup (C^{r_X}\sqcup C^{r_X-1})
  \end{equation*}
  (so $|\Qu{N}|\leq 12r\Delta n$), time $T=r(2\Delta+2)+21\leq 16r\Delta$, and gate set
  \begin{equation*}
    \cG = \{\gInitX,\gInitZ,\gTerm,\gCNOT,\gMX,\gMZ,\gCt{\gX},\gCt{\gZ},\gCF_*\}
  \end{equation*}
  with an associated set $\ps$ of posetselectable gates for which the following hold:
  \begin{enumerate}
  \item\label{it:ig} (General fault-tolerance) There exists a mending refreshing fault-tolerant gadget
    \begin{equation*}
      \bfG = (\cR,\; \cE_{\mathrm{run}}^{\sqcup T},\; D_{\mathrm{in}},\; D_{\mathrm{out}},\; \ps)
    \end{equation*}
    for the set
    \begin{equation*}
      \bar{\cO} = \left\{\bar{O}:\bC^{2^{M^{r_X}}\times 2^{M^{r_X}}}\rightarrow\bC^{2^{M^{r_X}}\times 2^{M^{r_X}}} : |\supp(\bar{O})| < \bar{\omega}\right\}
    \end{equation*}
    of superoperators acting on $k=|M^{r_X}|$ qubits whose support is $\bar{\cE}$-avoiding.
  \item\label{it:ifg} (Fine-grained logical error bound) For every $\cE_{\mathrm{in}}$-avoiding set $E_{\mathrm{in}}\subseteq M^{r_X}$, every $\cE_{\mathrm{run}}$-avoiding set $E_{\mathrm{run}}\subseteq\Qu{N}\times[T]$, and every postselection $\xi\in\bF_2^{\ps}$, there exists a $\bar{\cE}$-avoiding set $\bar{E}\subseteq M^{r_X}$ such that
    \begin{equation*}
      \bfG(E_{\mathrm{in}},E_{\mathrm{run}}) = (\cR,\; (\comp{E_{\mathrm{run}}})^{\sqcup T},\; D_{\mathrm{in}}(E_{\mathrm{in}}),\; D_{\mathrm{out}},\; \ps)
    \end{equation*}
    is a mending refreshing fault-tolerant gadget for the set
    \begin{equation*}
      \bar{\cO}(\bar{E}) = \left\{\bar{O}:\bC^{2^{M^{r_X}}\times 2^{M^{r_X}}}\rightarrow\bC^{2^{M^{r_X}}\times 2^{M^{r_X}}} : \supp(\bar{O}) \subseteq \bar{E}\right\}
    \end{equation*}
    of superoperators whose support is $\comp{\bar{E}}$-avoiding.
  \end{enumerate}
  Furthemore, no postselectable gate in $\ps$ lies in the forward lightcone of an input qubit of $\cR$, meaning that we only consider postselections on measurements that do not depend on the input state.
\end{lemma}

\begin{remark}
  \label{remark:istatements}
  We have included the general fault-tolerance statement (\Cref{it:ig}) in \Cref{lem:inject} as it may be easier for the reader to parse than the fine-grained logical error bound (\Cref{it:ifg}). However, the fine-grained bound is a strictly stronger result than the general statement. Indeed, the general statement allows the gadget's logical output to experience arbitrary linear combinations of errors $\bar{O}$ supported on different sets of size $<\bar{\omega}$. However, the fine-grained statement shows that once we fix the supports $E_{\mathrm{in}}$ and $E_{\mathrm{run}}$ of the input and fault errors respectively, along with the postselection $\xi$, then any resulting logical error $\bar{O}$ must have support inside a fixed set $\bar{E}$ of size $<\bar{\omega}$.

  Thus if we for instance assume that the input and fault errors are sampled from some locally stochastic distribution as described in \Cref{sec:bfls}, so that $E_{\mathrm{in}},E_{\mathrm{run}}$ are sampled from some probability distribution, and so that $\xi$ is also sampled from some well-defined distribution, then we have a well-defined probability distribution over the set $\bar{E}$ of size $<\bar{\omega}$ that is guaranteed to contain the logical error's support. In contrast, the general fault-tolerance statement in \Cref{it:ig} does not rule out more severe forms of logical errors, such as those arising from an adversary who measures some qubits, and then adaptively chooses which other qubits to corrupt.

  Note that by the final paragraph in the statement of \Cref{lem:inject}, as long as the fault does not perform entangling gates between the forward lightcone of the input qubits and the qubits that are measured in $\ps$, then the probability distribution of measurement outcomes $\xi\in\bF_2^{\ps}$ will be independent of the choice of input state $\rho$. In this case it follows that the probability distribution of the output error set $\bar{E}$ does not depend on the input state $\rho$. We emphasize here that $\rho$ denotes the uncorrupted input state, as $\bar{E}$ necessarily depends on the input error set $E_{\mathrm{in}}$.
\end{remark}

\Cref{lem:inject} is intended to be instantiated with $\omega_{\mathrm{in}},\omega_{\mathrm{run}},\omega_{\mathrm{out}},\bar{\omega}$ all equal to (arbitrarily) small constants times $k=M^{r_X}$. Similarly as described in the discussion around the ejection gadget in \Cref{sec:eject}, a locally stochastic fault of the form described in \Cref{sec:bfls} will only fail to support such choices of $\omega_{\mathrm{in}},\omega_{\mathrm{run}},\omega_{\mathrm{out}},\bar{\omega}$ with small constant probability by \Cref{lem:lsex}. Hence our gadget in \Cref{lem:inject} will only fail to inject the input qubits with small constant probability, and after injection only a small constant fraction of the input qubits will be corrupted. Similarly as described in \Cref{sec:eject}, once we fix the fault's probability distribution, we can therefore identify a large constant fraction of our input qubits with low marginal probability of corruption during injection.

\begin{proof}[Proof of \Cref{lem:inject}]
  As mentioned in \Cref{remark:istatements}, it suffices to show the fine-grained result in \Cref{it:ifg} in \Cref{lem:inject}. Indeed, to prove the general statement in \Cref{it:ig}, by the definition of fault-tolerance in \Cref{def:faulttol} we must first fix a $\cE_{\mathrm{in}}$-avoiding set $E_{\mathrm{in}}\subseteq M^{r_X}$, a $\cE_{\mathrm{run}}$-avoiding set $E_{\mathrm{run}}\subseteq\Qu{N}\times[T]$, and a postselection $\xi\in\bF_2^{\ps}$. But then the fine-grained bound in \Cref{it:ifg} implies the desired fault-tolerance statement in \Cref{it:ig}. Thus for the remainder of this proof we only consider \Cref{it:ifg}.

  At a high level, we define $\cR$ to prepare two fresh code states containing logical Bell pairs. We then apply \Cref{lem:eject} to eject out of one of these two code states, to create Bell pairs between bare physical qubits and the remaining code state. We then implement the teleportation circuit $\cR_{\mathrm{tel}}$ in \Cref{fig:teleport} on all qubits simultaneously, to teleport the input state (which also consists of bare physical qubits) into the remaining code state. This approach is similar to that of the error-correction gadget in \Cref{lem:errcorr}; the only difference is that our input now consists of bare physical qubits, so we need to eject out of one of the prepared code states in order to teleport.

  We now present the details. Let
  \begin{equation*}
    \Qu{N}_1 = M^{r_X}, \hspace{1em} \Qu{N}_2 = C^{r_X}\sqcup C^{r_X+1}, \hspace{1em} \Qu{N}_3 = C^{r_X}\sqcup C^{r_X-1},
  \end{equation*}
  so that $\Qu{N}=\Qu{N}_1\sqcup \Qu{N}_2\sqcup\Qu{N}_3$. We for instance write $\Qu{N}_2|_{C^{r_X}}$ to denote the copy of $C^{r_X}$ inside $\Qu{N}_2$. Also recall that we can view $M^{r_X}=V_{\bar{\ell}}^{\times r}$ as a subset of $C^{r_X}$, so that $\Qu{N}_2|_{M^{r_X}}$ is a well-defined set, which will contain the qubits ejected out of the code state we prepare in $\Qu{N}_2|_{C^{r_X}}$. Our circuit $\cR$ takes as input qubits $N_{\mathrm{in}}=\Qu{N}_1$, and outputs qubits $N_{\mathrm{out}}=\Qu{N}_3|_{C^{r_X}}$. We define $\cR$ to perform the following operations in order to implement the logical circuit in \Cref{fig:teleport} on all $k$ input qubits:
  \begin{enumerate}
  \item\label{it:isp} ($r(2\Delta+2)+6\leq 8r\Delta$ timesteps) Apply the $\ket{+}^{\otimes k}$-state preparation gadget in \Cref{lem:stateprep} on qubits $\Qu{N}_2$, and apply the $\ket{0}^{\otimes k}$-state preparation gadget in \Cref{lem:stateprep} on qubits $\Qu{N}_3$, while applying identity gates on qubits $\Qu{N}_1$.
  \item\label{it:ibp} ($1$ timestep) Apply the $\gCNOT^{\otimes k}$ gadget in \Cref{lem:cnot} on qubits $\Qu{N}_2|_{C^{r_X}},\Qu{N}_3|_{C^{r_X}}$, while applying identity gates on qubits $\Qu{N}_1$.
  \item\label{it:ie} ($5$ timesteps) Apply the ejection gadget in \Cref{lem:eject} on qubits $\Qu{N}_2|_{C^{r_X}}$, while applying identity gates on qubits $\Qu{N}_1$, and the idling gadget in \Cref{lem:idle} no qubits $\Qu{N}_3|_{C^{r_X}}$.
  \item\label{it:icnotin} ($1$ timestep) Apply $\gCNOT^{\otimes k}$ on qubits $\Qu{N}_1,\Qu{N}_2|_{M^{r_X}}$, while applying the idling gadget in \Cref{lem:idle} on qubits $\Qu{N}_3|_{C^{r_X}}$.
  \item\label{it:imeasure} ($1$ timestep) Apply $\gMX^{\otimes k}$ on qubits $\Qu{N}_1$, and apply $\gMZ^{\otimes k}$ on qubits $\Qu{N}_2|_{C^{r_X}}$, while applying the idling gadget in \Cref{lem:idle} on qubits $\Qu{N}_3|_{C^{r_X}}$. Let $x_1,x_2\in\bF_2^k$ denote the respective classical measurement outcomes.
  \item\label{it:ictx} ($3$ timesteps) Apply the $\gCt{\gX}^{\otimes k}$ gadget in \Cref{lem:ccp} on control bits $x_2$ and target qubits~$\Qu{N}_3$.
  \item\label{it:ictz} ($3$ timesteps) Apply the $\gCt{\gZ}^{\otimes k}$ gadget in \Cref{lem:ccp} on control bits $x_1$ and target qubits~$\Qu{N}_3$.
  \item\label{it:iterm} ($1$ timestep) Apply $\gTerm^{\otimes 2k}$ to the classical bits $(x_1,x_2)$, while applying the idling gadget in \Cref{lem:idle} on qubits $\Qu{N}_3|_{C^{r_X}}$.
  \end{enumerate}
  By definition, $\cR$ uses quantum space $\Qu{N}$ and time $T=r(2\Delta+2)+21\leq 16r\Delta$. Furthermore, we define $\ps$ to contain only those measurements in \Cref{it:isp} from the state preparation gadget in \Cref{lem:stateprep}, so that indeed no gates in $\ps$ lie in the forward lightcone of any of the input qubits of $\cR$.

  We will now prove the fault-tolerance statement in \Cref{it:ifg} in \Cref{lem:inject}. Fix a $\cE_{\mathrm{in}}$-avoiding set $E_{\mathrm{in}}\subseteq M^{r_X}$, a $\cE_{\mathrm{run}}$-avoiding set $E_{\mathrm{run}}\subseteq\Qu{N}\times[T]$, and a postselection $\xi\in\bF_2^{\ps}$. Let $E_{\mathrm{Bell},23}=E_{\mathrm{Bell},2}\sqcup E_{\mathrm{Bell},3}\subseteq\Qu{N}_2|_{M^{r_X}}\sqcup\Qu{N}_3|_{C^{r_X}}$ be the output error set of the gadget given by the subcircuit $\cR_{\mathrm{Bell}}$ of $\cR$ that prepares Bell pairs between the $k$ bare physical qubits $\Qu{N}_2|_{M^{r_X}}$ and the $k$ logical qubits encoded in the code block $\Qu{N}_3|_{C^{r_X}}$. Formally, this subcircuit $\cR_{\mathrm{Bell}}$ is given by the restriction of \Cref{it:isp,it:ibp,it:ie} in the definition of $\cR$ above to qubits $\Qu{N}_2\sqcup\Qu{N}_3$, where we apply \Cref{lem:seqcomp,lem:parcomp} to obtain the associated composed gadget from the primitive gadgets given by \Cref{lem:stateprep,lem:cnot,lem:eject,lem:idle}. Letting
  \begin{equation*}
    \gamma_{\mathrm{Bell},3} = 2^{20}r^3\Delta^5\cdot\gamma_{\mathrm{run}}, \hspace{1em} \omega_{\mathrm{Bell},2} = \omega_{\mathrm{Bell},3} = 2^{20}r^3\beta^2\Delta^4\cdot\omega_{\mathrm{run}},
  \end{equation*}
  it follows that $E_{\mathrm{Bell},2}$ is $2^{M^{r_X}}|^{\geq\omega_{\mathrm{Bell},2}}$-avoiding, and that $E_{\mathrm{Bell},3}$ is $\cE^{\cC}(\eta_{\mathrm{run}},\gamma_{\mathrm{Bell},3},\omega_{\mathrm{mid,3}};\beta)$-avoiding.
  Now let $E_{\mathrm{out}}\subseteq\Qu{N}_3|_{C^{r_X}}$ be the output error set of the gadget acting on physical qubits $\Qu{N}_3$ that is given by the subcircuit of $\cR$ that simply idles and applies classically-controlled Pauli gates to $\Qu{N}_3$, as specified by \Cref{it:icnotin,it:imeasure,it:ictx,it:ictz,it:iterm} in the definition of $\cR$ above. The input error set to this gadget is $E_{\mathrm{Bell},3}$ (which uniquely determined as a function of $E_{\mathrm{run}},\xi$), and the fault error set is given by $E_{\mathrm{run}}$ (restricted to qubits $\Qu{N}_3$ across the appropriate timesteps), so that output error set $E_{\mathrm{out}}$ is uniquely determined as a function of $E_{\mathrm{run}},\xi$, and is
  \begin{equation*}
    \cE^{\cC}(\eta_{\mathrm{run}},\; 2^{30}r^3\Delta^5\cdot\gamma_{\mathrm{run}},\; 2^{30}r^3\beta^2\Delta^4\cdot\omega_{\mathrm{run}};\; \beta) = \cE^{\cC}(\eta_{\mathrm{out}},\gamma_{\mathrm{out}},\omega_{\mathrm{out}};\beta)
  \end{equation*}
  avoiding. We also define $\bar{E}\subseteq M^{r_X}$ to be the set
  \begin{equation*}
    \bar{E} = E_{\mathrm{in}} \cup \left(\bigcup_{t=1}^{r(2\Delta+2)+13}(E_{\mathrm{run}}\cap(\Qu{N}_1\times\{t\}))\right) \cup E_{\mathrm{Bell},2} \cup (E_{\mathrm{run}}\cap(\Qu{N}_2|_{M^{r_X}}\times\{r(2\Delta+2)+13\})),
  \end{equation*}
  given by the union of $E_{\mathrm{in}}$, $E_{\mathrm{Bell},2}$, and all locations in $E_{\mathrm{run}}$ that act on bare physical qubits in $\Qu{N}_1$ or $\Qu{N}_2|_{M^{r_X}}$. By construction
  \begin{align*}
    |\bar{E}|
    &< \omega_{\mathrm{in}} + (r(2\Delta+2)+14)\cdot\omega_{\mathrm{run}} + \omega_{\mathrm{Bell},2} \leq \bar{\omega},
  \end{align*}
  so that $\bar{E}$ is $\bar{\cE}=2^{M^{r_X}}|^{\geq\bar{\omega}}$-avoiding.
  Therefore given input and fault errors supported inside $E_{\mathrm{in}},E_{\mathrm{run}}$ and the postselection $\xi$, then the resulting error on qubits $\Qu{N}_1\sqcup\Qu{N}_2|_{M^{r_X}}$ just prior to the measurements in \Cref{it:imeasure} above is supported inside $\bar{E}\sqcup\bar{E}\subseteq\Qu{N}_1\sqcup\Qu{N}_2|_{M^{r_X}}$. Here we use the fact that such errors just prior to the measurements can only arise from the input error supported in $E_{\mathrm{in}}$, from the post-ejection error supported in $E_{\mathrm{Bell},2}$, and from the fault error on qubits $\Qu{N}_1$ and $\Qu{N}_2|_{C^{r_X}}$ in the appropriate timesteps; these errors can also propagate between $\Qu{N}_1$ and $\Qu{N}_2|_{M^{r_X}}$ via the application of $\gCNOT^{\otimes k}$ in \Cref{it:icnotin} above.

  By \Cref{lem:paulift}, to prove fault-tolerance, our goal is to show that for every finite set $\Rf{K}$, every operator $\rho\in\bC^{2^{M^{r_X}}\times 2^{M^{r_X}}}$, every Pauli $\comp{E_{\mathrm{in}}}$-deviation $\sigma$ of $\rho$, and every $\comp{E_{\mathrm{run}}}$-avoiding Pauli fault $\cF$ and reweighting $\zeta$ for $\cR$, then $\cR[\cF,\zeta,\xi]$ is a $\comp{E_{\mathrm{out}}}$-deviation of $\Enc_{\cC}^{r_X}\circ(\bar{\cO}(\bar{E}))(\rho)$, meaning that $\cR[\cF,\zeta,\xi](\sigma)$ is a linear combination of $\comp{E_{\mathrm{out}}}$-deviations of states of the form $\Enc_{\cC}^{r_X}\circ\bar{O}(\rho)$ for $k$-qubit superoperators $\bar{O}$ supported inside $\bar{E}$.

  For this purpose, first observe that by applying \Cref{lem:seqcomp,lem:parcomp} to compose the gadgets comprising $\cR_{\mathrm{Bell}}$ given by \Cref{lem:stateprep,lem:cnot,lem:eject,lem:idle}, the output $\cR_{\mathrm{Bell}}[\cF,\zeta,\xi](1)$ is by definition a $\comp{E_{\mathrm{Bell},23}}$-deviation of $I_{M^{r_X}}\otimes\Enc_{\cC}^{r_X}(\ket{00}+\ket{11})^{\otimes M^{r_X}}$. Here we implicitly restrict $\cF,\zeta,\xi$ to space-time locations in the subcircuit $\cR_{\mathrm{Bell}}$ of $\cR$. Because $\cR_{\mathrm{Bell}}$ has no input (qu)bits, its associated superoperator simply takes the trivial input $1\in\bC^{1\times 1}$. We now apply \Cref{lem:seqcomp,lem:parcomp} to compose the remaining gadgets comprising $\cR$, which are simply idling (\Cref{lem:idle}) and classically-controlled Pauli (\Cref{lem:ccp}) gadgets on qubits $\Qu{N}_3|_{C^{r_X}}$, and bare physical identity, $\gCNOT$, and measurement gates on qubits $\Qu{N}_1\sqcup\Qu{N}_2|_{M^{r_X}}$. These bare physical gates can trivially can be viewed as gadgets whose decorated codes are specified by the identity encoding map.

  It then follows by \Cref{lem:seqcomp,lem:parcomp} that the output $\cR[\cF,\zeta,\xi](\sigma)$ is a linear combination of $\comp{E_{\mathrm{out}}}$-deviations of states of the form $\Enc_{\cC}^{r_X}\circ\cR_{\mathrm{tel}}^{\otimes k}[\bar{\cF},\bar{\zeta}](\rho)$, where each $\bar{\cF}$ in this linear combination is a logical Pauli fault for $\cR_{\mathrm{tel}}^{\otimes k}$ supported entirely inside qubits $\bar{E}\sqcup\bar{E}\subseteq\Qu{N}_1\sqcup\Qu{N}_2|_{M^{r_X}}$ in the timestep immediately proceeding the $\gMX$ and $\gMZ$ measurements, and each $\bar{\zeta}$ is some reweighting on the classically-controlled Pauli corrections. In more detail, by definition $\cR$ implements the circuit $\cR_{\mathrm{tel}}^{\otimes k}$, where the bare physical qubits $\Qu{N}_1$ and $\Qu{N}_2|_{M^{r_X}}$ provide the $k$ copies of each of the first two wires in $\cR_{\mathrm{tel}}$, and the $k$ logical qubits encoded in the code block $\Qu{N}_3|_{C^{r_X}}$ provides the $k$ copies of the third wire in $\cR_{\mathrm{tel}}$. We constructed the set $\bar{E}$ to contain the support of all Pauli errors on qubits $\Qu{N}_1\sqcup\Qu{N}_2|_{M^{r_X}}$ immediately prior to the measurements in \Cref{it:imeasure} above, which correspond to the Pauli measurements in $\cR_{\mathrm{tel}}^{\otimes k}$. Meanwhile, \Cref{lem:seqcomp} along with the definition of $E_{\mathrm{mid},3}$ and $E_{\mathrm{out}}$ then implies that the output $\cR[\cF,\zeta,\xi](\sigma)$ is a linear combination of $\comp{E_{\mathrm{out}}}$-deviations of the state $\Enc_{\cC}^{r_X}\circ\cR_{\mathrm{tel}}^{\otimes k}[\bar{\cF},\bar{\zeta}](\rho)$, for $\bar{\cF}$ as described above. Here we can assume that each $\bar{\cF}$ is a Pauli fault because an arbitrary fault can be decomposed into a linear combination of Pauli faults.

  Now by \Cref{claim:telfault}, there exists a $k$-qubit Pauli superoperator $\bar{O}$ depending only on $\bar{\cF}$, and supported inside the same set $\bar{E}\subseteq M^{r_X}$ as $\cF$, such that
  \begin{align*}
    \Enc_{\cC}^{r_X}\circ\cR_{\mathrm{tel}}^{\otimes k}[\bar{\cF},\bar{\zeta}](\rho)
    &\propto \Enc_{\cC}^{r_X}\circ\bar{O}(\rho).
  \end{align*}
  Thus we have shown that $\cR[\cF,\zeta,\xi](\sigma)$ is a $\comp{E_{\mathrm{out}}}$-deviation of $\Enc_{\cC}^{r_X}\circ(\bar{\cO}(\bar{E}))(\rho)$, as desired, so the fault-tolerance of the gadget $\bfG(E_{\mathrm{in}},E_{\mathrm{run}})$ in \Cref{it:ifg} in \Cref{lem:inject} holds.

  Because we defined the set $E_{\mathrm{out}}$ to be uniquely determined as a function of $E_{\mathrm{in}},\xi$, it immediately follows that $\bfG(E_{\mathrm{in}},E_{\mathrm{run}})$ is refreshing.

  To prove that this gadget is also mending, by \Cref{lem:paulift} it suffices to simply repeat the proof of fault-tolerance above, except that we now assume that $\sigma$ is a Pauli $\emptyset$-deviation of $\rho$. That is, whereas for fault-tolerance we let $\sigma$ differ from $\rho$ by a Pauli error supported inside $E_{\mathrm{in}}$, for mending we let $\sigma$ differ from $\rho$ by an arbitrary Pauli error. The same proof as used above to show fault-tolerance again goes through, except that the final logical error may now have arbitrary support in $M^{r_X}$, instead of being restricted to $\bar{E}\subseteq M^{r_X}$. Thus the output state $\cR[\cF,\zeta,\xi](\sigma)$ is now a linear combination of $\comp{E_{\mathrm{out}}}$-deviations of states of the form $\Enc_{\cC}^{r_X}\circ L(\rho)$ for arbitrary $k$-qubit Pauli superoperators $L$, so $\bfG(E_{\mathrm{in}},E_{\mathrm{run}})$ is indeed mending by \Cref{lem:paulift}.
\end{proof}

\section{Conclusion and Future Directions}
\label{sec:conclusion}
We introduce the first known family of quantum codes for which injection (i.e.~fault-tolerant encoding) and ejection (i.e.~fault-tolerant unencoding) can be performed with constant space-time overhead under locally stochastic noise. As stated in \Cref{thm:maininf} (and described in more detail throughout \Cref{sec:intro}), our codes also support various other constant-overhead fault-tolerant operations, including logical $\ket{0}$ and $\ket{+}$ preparation, error correction, and certain Clifford gates.

Our results open various possibilities for future work. In particular, while we provide a rigorous analysis of injection, ejection, and certain Clifford gates, we leave the development of a full universal fault-tolerance scheme based on these ideas for future work. Such a scheme may benefit heavily from efficient code switching and resource-state preparation procedures based on the ideas described in \Cref{sec:motivation,sec:apps}.

Beyond asymptotic performance, it would be interesting to see if the ideas we develop in this paper translate to improvements for finite-sized fault-tolerance schemes that may be implementable on hardware. Various schemes have been proposed for injecting states, and in particular magic states, into quantum LDPC codes with a focus on practical finite-sized paramters (e.g.~\cite{dennis_topological_2002,mazurek_long-distance_2014,lodyga_simple_2015,zhang_constant-overhead_2025,bhardwaj_high-rate_2026,liu_-situ_2026}). However, our construction uses various ingredients that were not previously used in such injection schemes, including inspiration from the linear-time-encodable codes of \cite{spielman_linear-time_1996}, and concatenation with inner codes of various sizes to realize non-uniform error rates. We highlight further development of such ideas, especially in practical regimes, as a focus for future work.

At a more technical level, recall that our injection and ejection gadgets guarantee that the data qubits have low marginal error probabilities (see \Cref{sec:maininf}). As we described in \Cref{sec:apps}, such low marginal error probabilities are sufficient in many applications of injection/ejection. However, we could still hope for a stronger guarantee such as locally stochastic data qubit errors. Our scheme may even provide such a guarantee, but our current analysis does not rule out correlated data errors.

% TODO: discuss transversal $\gCNOT$, code switching, uniform vs non-uniform locally stochastic noise, choice of ``good'' qubits as a few might have high error rates in injection/ejection, say error correction is single shot, lack of consideration of non-Clifford gates

\section{Acknowledgments}
We thank Christopher A.~Pattison, Katie Chang, Guanyu Zhu, and Zhiyang (Sunny) He for helpful discussions, some of which helped motivate this project.

\textbf{AI Use Statement:} All proofs and writing are the work of the human authors. ChatGPT 6 Astra was used to find references, make figures, and check over the final draft.

\bibliographystyle{alpha}
\bibliography{library}

\newcommand{\etalchar}[1]{$^{#1}$}
\begin{thebibliography}{LXJO{\etalchar{+}}26}

\bibitem[ABO97]{aharonov_fault-tolerant_1997}
D.~Aharonov and M.~Ben-Or.
\newblock Fault-tolerant quantum computation with constant error.
\newblock In {\em Proceedings of the twenty-ninth annual {ACM} symposium on
  {Theory} of computing}, {STOC} '97, pages 176--188, New York, NY, USA, May
  1997. Association for Computing Machinery.

\bibitem[ACQ22]{aharonov_quantum_2022}
Dorit Aharonov, Jordan Cotler, and Xiao-Liang Qi.
\newblock Quantum algorithmic measurement.
\newblock {\em Nature Communications}, 13(1):887, February 2022.

\bibitem[AMC{\etalchar{+}}25]{allen_quantum_2025}
Richard~R. Allen, Francisco Machado, Isaac~L. Chuang, Hsin-Yuan Huang, and
  Soonwon Choi.
\newblock Quantum {Computing} {Enhanced} {Sensing}, January 2025.
\newblock arXiv:2501.07625 [quant-ph].

\bibitem[BGV26]{breuckmann_fault-tolerant_2026}
Nikolas Breuckmann, Louis Golowich, and Umesh Vazirani.
\newblock Fault-{Tolerant} {Quantum} {Computation} with {Adversarial} {Errors}.
\newblock In {\em 2026 {IEEE} 67th {Annual} {Symposium} on {Foundations} of
  {Computer} {Science} ({FOCS})}, 2026.
\newblock To appear.

\bibitem[BH12]{bravyi_magic-state_2012}
Sergey Bravyi and Jeongwan Haah.
\newblock Magic-state distillation with low overhead.
\newblock {\em Physical Review A}, 86(5):052329, November 2012.

\bibitem[BK05]{bravyi_universal_2005}
Sergei Bravyi and Alexei Kitaev.
\newblock Universal {Quantum} {Computation} with ideal {Clifford} gates and
  noisy ancillas.
\newblock {\em Physical Review A}, 71(2):022316, February 2005.
\newblock arXiv:quant-ph/0403025.

\bibitem[BMM{\etalchar{+}}26]{bhardwaj_high-rate_2026}
Aditya Bhardwaj, Muzhou Ma, Nadine Meister, Robbie King, Dolev Bluvstein, John
  Preskill, Madelyn Cain, Qian Xu, and Hsin-Yuan Huang.
\newblock High-rate {qLDPC} processors, July 2026.
\newblock arXiv:2607.28795 [quant-ph].

\bibitem[Bom15]{bombin_single-shot_2015}
Héctor Bombín.
\newblock Single-{Shot} {Fault}-{Tolerant} {Quantum} {Error} {Correction}.
\newblock {\em Physical Review X}, 5(3):031043, September 2015.

\bibitem[Bom16]{bombin_dimensional_2016}
H.~Bombin.
\newblock Dimensional {Jump} in {Quantum} {Error} {Correction}, May 2016.
\newblock arXiv:1412.5079 [quant-ph].

\bibitem[CCHL22]{chen_exponential_2022}
Sitan Chen, Jordan Cotler, Hsin-Yuan Huang, and Jerry Li.
\newblock Exponential {Separations} {Between} {Learning} {With} and {Without}
  {Quantum} {Memory}.
\newblock In {\em 2021 {IEEE} 62nd {Annual} {Symposium} on {Foundations} of
  {Computer} {Science} ({FOCS})}, pages 574--585, February 2022.
\newblock ISSN: 2575-8454.

\bibitem[CFG26]{christandl_fault-tolerant_2026-1}
Matthias Christandl, Omar Fawzi, and Ashutosh Goswami.
\newblock Fault-{Tolerant} {Quantum} {Input}/{Output}.
\newblock {\em IEEE Transactions on Information Theory}, 72(5):3098--3128, May
  2026.

\bibitem[CRTS23]{cohen_hdx_2023}
Itay Cohen, Roy Roth, and Amnon Ta-Shma.
\newblock {HDX} {Condensers}.
\newblock In {\em 2023 {IEEE} 64th {Annual} {Symposium} on {Foundations} of
  {Computer} {Science} ({FOCS})}, pages 1649--1664, November 2023.

\bibitem[CRVW02]{capalbo_randomness_2002}
Michael Capalbo, Omer Reingold, Salil Vadhan, and Avi Wigderson.
\newblock Randomness conductors and constant-degree lossless expanders.
\newblock In {\em Proceedings of the thiry-fourth annual {ACM} symposium on
  {Theory} of computing}, {STOC} '02, pages 659--668, New York, NY, USA, May
  2002. Association for Computing Machinery.

\bibitem[DKLP02]{dennis_topological_2002}
Eric Dennis, Alexei Kitaev, Andrew Landahl, and John Preskill.
\newblock Topological quantum memory.
\newblock {\em Journal of Mathematical Physics}, 43(9):4452--4505, September
  2002.

\bibitem[DLV24]{dinur_expansion_2024}
Irit Dinur, Ting-Chun Lin, and Thomas Vidick.
\newblock Expansion of {High}-{Dimensional} {Cubical} {Complexes}: with
  {Application} to {Quantum} {Locally} {Testable} {Codes}.
\newblock In {\em 2024 {IEEE} 65th {Annual} {Symposium} on {Foundations} of
  {Computer} {Science} ({FOCS})}, pages 379--385, October 2024.
\newblock ISSN: 2575-8454.

\bibitem[FGL18]{fawzi_efficient_2018}
Omar Fawzi, Antoine Grospellier, and Anthony Leverrier.
\newblock Efficient decoding of random errors for quantum expander codes.
\newblock In {\em Proceedings of the 50th {Annual} {ACM} {SIGACT} {Symposium}
  on {Theory} of {Computing}}, {STOC} 2018, pages 521--534, New York, NY, USA,
  June 2018. Association for Computing Machinery.

\bibitem[FGL20]{fawzi_constant_2020}
Omar Fawzi, Antoine Grospellier, and Anthony Leverrier.
\newblock Constant overhead quantum fault tolerance with quantum expander
  codes.
\newblock {\em Communications of the ACM}, 64(1):106--114, December 2020.

\bibitem[GCZ25]{golowich_constant-overhead_2025}
Louis Golowich, Kathleen Chang, and Guanyu Zhu.
\newblock Constant-{Overhead} {Addressable} {Gates} via {Single}-{Shot} {Code}
  {Switching}, October 2025.
\newblock arXiv:2510.06760 [quant-ph].

\bibitem[Gol24]{golowich_new_2024}
Louis Golowich.
\newblock New {Explicit} {Constant}-{Degree} {Lossless} {Expanders}.
\newblock In {\em Proceedings of the 2024 {Annual} {ACM}-{SIAM} {Symposium} on
  {Discrete} {Algorithms} ({SODA})}, Proceedings, pages 4963--4971. Society for
  Industrial and Applied Mathematics, January 2024.

\bibitem[Got14]{gottesman_fault-tolerant_2014}
Daniel Gottesman.
\newblock Fault-tolerant quantum computation with constant overhead.
\newblock {\em Quantum Information \& Computation}, 14(15-16):1338--1372,
  November 2014.

\bibitem[HBC{\etalchar{+}}22]{huang_quantum_2022}
Hsin-Yuan Huang, Michael Broughton, Jordan Cotler, Sitan Chen, Jerry Li, Masoud
  Mohseni, Hartmut Neven, Ryan Babbush, Richard Kueng, John Preskill, and
  Jarrod~R. McClean.
\newblock Quantum advantage in learning from experiments.
\newblock {\em Science}, 376(6598):1182--1186, June 2022.

\bibitem[HLM{\etalchar{+}}25]{hsieh_explicit_2025}
Jun-Ting Hsieh, Alexander Lubotzky, Sidhanth Mohanty, Assaf Reiner, and
  Rachel~Yun Zhang.
\newblock Explicit {Lossless} {Vertex} {Expanders}.
\newblock In {\em 2025 {IEEE} 66th {Annual} {Symposium} on {Foundations} of
  {Computer} {Science} ({FOCS})}, pages 894--911, December 2025.
\newblock ISSN: 2575-8454.

\bibitem[HLW06]{hoory_expander_2006}
Shlomo Hoory, Nathan Linial, and Avi Wigderson.
\newblock Expander graphs and their applications.
\newblock {\em Bulletin of the American Mathematical Society}, 43(4):439--561,
  2006.

\bibitem[HNP25]{he_composable_2025}
Zhiyang He, Quynh~T. Nguyen, and Christopher~A. Pattison.
\newblock Composable {Quantum} {Fault}-{Tolerance}, August 2025.
\newblock arXiv:2508.08246 [quant-ph].

\bibitem[KP13]{kovalev_fault_2013}
Alexey~A. Kovalev and Leonid~P. Pryadko.
\newblock Fault tolerance of quantum low-density parity check codes with
  sublinear distance scaling.
\newblock {\em Physical Review A}, 87(2):020304, February 2013.

\bibitem[KP25]{kalachev_maximally_2025}
Gleb Kalachev and Pavel Panteleev.
\newblock Maximally {Extendable} {Product} {Codes} are {Good} {Coboundary}
  {Expanders}.
\newblock In {\em 2025 {IEEE} 66th {Annual} {Symposium} on {Foundations} of
  {Computer} {Science} ({FOCS})}, pages 1512--1524, December 2025.
\newblock ISSN: 2575-8454.

\bibitem[KPC26]{kannan_fault-tolerant_2026}
Ishaan Kannan, Harald Putterman, and Jordan Cotler.
\newblock Fault-tolerant quantum processing of physical experiments, September
  2026.
\newblock arXiv:2605.02057 [quant-ph] version: 2.

\bibitem[KT21]{kaufman_new_2021}
Tali Kaufman and Ran~J. Tessler.
\newblock New cosystolic expanders from tensors imply explicit {Quantum} {LDPC}
  codes with {Omega}(sqrt\{n\} log{\textasciicircum}k n) distance.
\newblock In {\em Proceedings of the 53rd {Annual} {ACM} {SIGACT} {Symposium}
  on {Theory} of {Computing}}, {STOC} 2021, pages 1317--1329, New York, NY,
  USA, June 2021. Association for Computing Machinery.

\bibitem[Li15]{li_magic_2015}
Ying Li.
\newblock A magic state’s fidelity can be superior to the operations that
  created it.
\newblock {\em New Journal of Physics}, 17(2):023037, February 2015.

\bibitem[LMGH15]{lodyga_simple_2015}
Justyna Lodyga, Pawel Mazurek, Andrzej Grudka, and Michal Horodecki.
\newblock Simple scheme for encoding and decoding a qubit in unknown state for
  various topological codes.
\newblock {\em Scientific Reports}, 5(1):8975, March 2015.
\newblock arXiv:1404.2495 [quant-ph].

\bibitem[LTZ15]{leverrier_quantum_2015}
Anthony Leverrier, Jean-Pierre Tillich, and Gilles Zémor.
\newblock Quantum {Expander} {Codes}.
\newblock In {\em 2015 {IEEE} 56th {Annual} {Symposium} on {Foundations} of
  {Computer} {Science}}, pages 810--824, October 2015.
\newblock arXiv:1504.00822 [quant-ph].

\bibitem[LXJO{\etalchar{+}}26]{liu_-situ_2026}
Kun Liu, Shifan Xu, Tomas Jochym-O'Connor, Zhiyang He, Shraddha Singh, and
  Yongshan Ding.
\newblock In-{Situ} {Simultaneous} {Magic} {State} {Injection} on {Arbitrary}
  {CSS} {qLDPC} {Codes}, April 2026.
\newblock arXiv:2604.05126 [quant-ph].

\bibitem[MG92]{misra_constructive_1992}
J.~Misra and David Gries.
\newblock A constructive proof of {Vizing}'s theorem.
\newblock {\em Information Processing Letters}, 41(3):131--133, March 1992.

\bibitem[MGH{\etalchar{+}}14]{mazurek_long-distance_2014}
Pawel Mazurek, Andrzej Grudka, Michal Horodecki, Pawel Horodecki, Justyna
  Lodyga, Lukasz Pankowski, and Anna Przysiezna.
\newblock Long-distance quantum communication over noisy networks without
  long-time quantum memory.
\newblock {\em Physical Review A}, 90(6):062311, December 2014.
\newblock arXiv:1202.1016 [quant-ph].

\bibitem[NP25]{nguyen_quantum_2025}
Quynh~T. Nguyen and Christopher~A. Pattison.
\newblock Quantum {Fault} {Tolerance} with {Constant}-{Space} and
  {Logarithmic}-{Time} {Overheads}.
\newblock In {\em Proceedings of the 57th {Annual} {ACM} {Symposium} on
  {Theory} of {Computing}}, {STOC} '25, pages 730--737, New York, NY, USA, June
  2025. Association for Computing Machinery.

\bibitem[Spi96]{spielman_linear-time_1996}
D.A. Spielman.
\newblock Linear-time encodable and decodable error-correcting codes.
\newblock {\em IEEE Transactions on Information Theory}, 42(6):1723--1731,
  November 1996.
\newblock Conference Name: IEEE Transactions on Information Theory.

\bibitem[SS96]{sipser_expander_1996}
M.~Sipser and D.A. Spielman.
\newblock Expander codes.
\newblock {\em IEEE Transactions on Information Theory}, 42(6):1710--1722,
  November 1996.

\bibitem[THL{\etalchar{+}}25]{tan_single-shot_2025}
Shi Jie~Samuel Tan, Yifan Hong, Ting-Chun Lin, Michael~J. Gullans, and Min-Hsiu
  Hsieh.
\newblock Single-{Shot} {Universality} in {Quantum} {LDPC} {Codes} via
  {Code}-{Switching}, October 2025.
\newblock arXiv:2510.08552 [quant-ph].

\bibitem[TZ14]{tillich_quantum_2014}
Jean-Pierre Tillich and Gilles Zemor.
\newblock Quantum {LDPC} codes with positive rate and minimum distance
  proportional to n{\textasciicircum}\{1/2\}.
\newblock {\em IEEE Transactions on Information Theory}, 60(2):1193--1202,
  February 2014.
\newblock arXiv: 0903.0566.

\bibitem[WHY24]{wills_constant-overhead_2024}
Adam Wills, Min-Hsiu Hsieh, and Hayata Yamasaki.
\newblock Constant-{Overhead} {Magic} {State} {Distillation}, August 2024.
\newblock arXiv:2408.07764 [quant-ph].

\bibitem[WLZH26]{wills_linear-time_2026}
Adam Wills, Ting-Chun Lin, Rachel~Yun Zhang, and Min-Hsiu Hsieh.
\newblock Linear-{Time} {Encodable} and {Decodable} {Quantum}
  {Error}-{Correcting} {Codes}, June 2026.
\newblock arXiv:2603.04543 [quant-ph].

\bibitem[XZB{\etalchar{+}}25]{xu_batched_2025}
Qian Xu, Hengyun Zhou, Dolev Bluvstein, Madelyn Cain, Marcin Kalinowski, John
  Preskill, Mikhail~D. Lukin, and Nishad Maskara.
\newblock Batched high-rate logical operations for quantum {LDPC} codes,
  October 2025.
\newblock arXiv:2510.06159 [quant-ph].

\bibitem[ZZYL25]{zhang_constant-overhead_2025}
Guo Zhang, Yuanye Zhu, Xiao Yuan, and Ying Li.
\newblock Constant-{Overhead} {Magic} {State} {Injection} into {qLDPC} {Codes}
  with {Error} {Independence} {Guarantees}, May 2025.
\newblock arXiv:2505.06981 [quant-ph].

\end{thebibliography}

\appendix

\section{Uniformizing the Error Model via Concatenation}
\label{sec:concat}
While the gadgets in \Cref{sec:ftnonu} perform our desired logical operations, they are designed to work under non-uniform locally stochastic noise, in which a small fraction of qubits must have subconstant error probabilities. In this section, we show how to concatenate these ``outer'' gadgets in \Cref{sec:ftnonu} with appropriate smaller ``inner'' fault-tolerance schemes of various sizes in order to obtain fault-tolerance against uniform locally stochastic errors. The main idea is to encode each physical qubit in our outer gadgets into an appropriately-sized code block from our inner scheme, and to replace each physical gate in our outer gadgets with a fault-tolerant gadget for that gate from our inner scheme.

In more detail, recall that every physical qubit $a$ in our gadgets from \Cref{sec:ftnonu} has a ``level'' $\Lev(a)\in\bZ_{\geq 0}$. We will encode each such qubit at level $\ell=\Lev(a)$ into a code of length $\bfn(\ell)=2^{\Theta(\ell)}$. Because the fraction of qubits at level $\ell$ is decays exponentially in $\ell$, we ensure this concatenation only increases the total number of physical qubits by a constant factor. The key point is that the length-$\bfn(\ell)$ encoding ensures an error rate under locally stochastic noise that decays doube-exponentially in $\ell$ (i.e.~as $2^{-2^{\Theta(\ell)}}$). Thus the concatenation allows us to simulate low error probabilities on high-level qubits of the un-concatenated ``outer'' scheme using uniform error probabilities on the concatenated scheme.

Our concatenation result formalizing this idea is given in \Cref{lem:lcft}. We remark that such concatenation of fault-tolerance schemes is notoriously delicate and often notation-intensive, see e.g.~\cite{aharonov_fault-tolerant_1997,he_composable_2025}. We attempt to simplify the presentation as much as possible by using slightly modified versions of our scheme in the main body, where the modifications are standard. As an example, below we allow the outer fault-tolerance scheme to have faults at timestep $0$. All of our results are easily adapted to this setting, even though in \Cref{def:fault} we previously defined faults to only first act after timestep $1$.

% Such concatenation is a well-established technique for constructing fault-tolerance schemes, dating back to the earliest proofs of the threshold theorem \cite{aharonov_fault-tolerant_1997}. A comprehensive treatment in a formalism similar to ours is given in \cite{he_composable_2025}.

% However, for completeness we still provide a rigorous proof of fault-tolerance for concatenated schemes in our setting. For simplicity in this section, we assume the outer gadget has no input and output qubits, but rather only has classical input and output bits. This assumption is not fundamental, but rather is made for expositional purposes to avoid further complicating the (already cumbersome) notation. Because many (but not all) applications of fault-tolerant quantum computers in the real world ultimately aim to solve problems with classical descriptions, the assumption that our final fault-tolerant circuit has classical input/output is not unreasonable. This circuit will typically include additional gadgets beyond the scope of the present paper, such as for implementing non-Clifford gates.

\subsection{Inner Fault-Tolerance Scheme}
In this section, we describe the inner fault-tolerance scheme with which we will concatenate our gadgets from \Cref{sec:ftnonu}. This scheme in particular follows from the work of \cite{golowich_constant-overhead_2025}, so we simply state its properties using our notation, which differs slightly from that of \cite{golowich_constant-overhead_2025}. Other related schemes, e.g.~\cite{bombin_dimensional_2016,xu_batched_2025,tan_single-shot_2025} provide similar properties that could also be used in place of \cite{golowich_constant-overhead_2025}, though with some caveats as discussed below.

The codes of \cite{golowich_constant-overhead_2025} have constant rate. However we want each code block of our inner scheme to only encode a single qubit, corresponding to a single physical qubit of our outer scheme. We therefore define the following notion of a decorated subsystem code, which allows us to ignore the states of some ``gauge'' qubits encoded in a code.

\begin{definition}
  Let $Q=(Q_X,Q_Z)$ be a $[[n,k+k',d]]$ quantum code with encoding isometry $\Enc:\bC^{2^k}\otimes\bC^{2^{k'}}\rightarrow\bC^{2^n}$. We call $Q$ a \emph{$[[n,k,d]]$ quantum CSS subsystem code} (or simply \emph{subsystem code}), where we call the first $k$ inputs to $\Enc$ \emph{logical qubits} and the latter $k'$ inputs \emph{gauge qubits}.

  Let $\cE_X,\cE_Z\subseteq 2^{[n]}$ be families of bad sets, and let $\cE=(\cE_X,\cE_Z)$. We say the data $D=(Q,\Enc,\cE)$ forms a \emph{$[[n,k,d]]$ decorated quantum CSS subsystem code} (which we shorten to \emph{decorated subsystem code}, or simply \emph{decorated code} when clear from context).

  As in \Cref{def:deccode}, if $\cE_X=\cE_Z$ we may write $D=(Q,\Enc,\cE_Z)$. Again we also let $\emptyset$ denote the trivial code, and let $D_1\sqcup D_2$ denote the disjoint union of decorated subsystem codes $D_1,D_2$.

  For an operator $\rho\in\bC^{2^k\times 2^k}$, we say an operator $\sigma\in\bC^{2^n\times 2^n}$ is a \emph{Pauli $\cE$-deviation of $\Enc(\rho)$} if $\sigma$ is a Pauli $\cE$-deviation of
  \begin{equation*}
    \Enc(\rho\otimes\bC^{2^{k'}\times 2^{k'}}) = \{\Enc(\rho\otimes\rho'):\rho'\in\bC^{2^{k'}\times 2^{k'}}\}
  \end{equation*}
  in the sense of \Cref{def:deccode}. Then we say $\sigma$ is a \emph{$\cE$-deviation of $\Enc(\rho)$} if $\sigma$ is a linear combination of Pauli $\cE$-deviations of $\Enc(\rho)$, and we extend this notation to sets $S\subseteq\bC^{2^k\times 2^k}$ of operators $\rho$ analogously as done in \Cref{def:deccode}.
\end{definition}

Our fault-tolerance formalism, including \Cref{def:faulttol,lem:paulift}, applies as before when we use decorated subsystem codes in place of ordinary decorated codes. Indeed, a similar formalism using subsystem codes was used in \cite{breuckmann_fault-tolerant_2026}. However, \cite{breuckmann_fault-tolerant_2026} required the gauge qubits to be in particular states for certain gadgets to function properly, which necessitated more data in their decorated subsystem codes. In contrast, we always allow the gauge qubits to be in arbitrary states, which we simply ignore for the purpose of implementing the desired logical superoperator.

We will also use the fact that the families of bad sets considered by \cite{golowich_constant-overhead_2025} satisfy the following key property.

\begin{definition}
  For $\underline{d},\overline{d}>0$, we say a family $\cE$ of sets is \emph{$(\underline{d},\overline{d})$-bounded} if $\cE$ contains $|\cE|\leq 2^{\overline{d}}$ sets, and each such set $E\in\cE_n$ has size $|E|\geq\underline{d}$.
\end{definition}

As we describe below, $(\underline{d},\overline{d})$-boundedness was implicit in \cite{golowich_constant-overhead_2025}, but was not explicitly named as such.

We are now ready to state the fault-tolerance scheme of \cite{golowich_constant-overhead_2025}, where we translate only the results relevant for our purposes to our notation in \Cref{lem:prior} below.

\begin{lemma}[Follows from \cite{golowich_constant-overhead_2025}]
  \label{lem:prior}
  There exists an infinite set $\bfN\subseteq\bN$ of positive integers that index a family of decorated subsystem codes
  \begin{equation*}
    D_n = (Q_n,\; \Enc_n,\; \cE_n) \hspace{1em} \text{ for } n\in\bfN
  \end{equation*}
  for which the following items hold for some constant $\mu>0$ and some constants $T_0,\nu\in\bN$. Below, we define the gate set
  \begin{equation*}
    \cG = \{\gI,\gX,\gZ,\gH,\gCNOT,\gInitX,\gInitZ,\gTerm,\gMX,\gMZ,\gCt{\gX},\gCt{\gZ},\gCF_*\},
  \end{equation*}
  and we let $k_{\mathrm{in},G}$ (resp.~$k_{\mathrm{out},G}$) denote the number of input (resp.~output) qubits of $G$, so that $k_{\mathrm{in},G},k_{\mathrm{out},G}\in\{0,1,2\}$.
  We take all asymptotic notation (e.g.~$\Omega(1)$) with respect to $n\rightarrow\infty$.
  \begin{enumerate}
  \item\label{it:prdensity} For every $n_0\in\bN$, there exists $n\in\bfN$ with\footnote{The constant $2$ here, along with the constant $8$ in \Cref{it:prcnot}, is somewhat arbitrary. These choices of constants are sufficient for our purposes, but \Cref{lem:prior} holds with these constants replaced by arbitrary constants $>1$.} $n_0\leq n<2n_0$.
  \item For every $n\in\bfN$, the subsystem code $Q_n$ has parameters $[[n,\; 1,\; n^{\Omega(1)}]]$.
  \item For every $n\in\bfN$, the family $\cE_n\subseteq 2^{[n]}$ is $(\Omega(n^\mu),O(n^\mu))$-bounded.
  \item\label{it:prgads} For every $n\in\bfN$ and every gate $G\in\cG\setminus\gCF_*$, there exists a mending refreshing fault-tolerant gadget
    \begin{equation*}
      \bfG_{n,G} := (\cR_{n,G},\; \cE_{\mathrm{run},n,G},\; D_n^{\sqcup k_{\mathrm{in},G}},\; D_n^{\sqcup k_{\mathrm{out},G}},\; \ps_{n,G})
    \end{equation*}
    for\footnote{Recall from \Cref{def:gates} that $G[*]$ denotes the set of all reweightings $G[\zeta]$ of $G$, which is only larger than the singleton set $\{G\}$ when $G$ is one of the classically controlled gates $\gCt{\gX},\gCt{\gZ}\in\cG$.} $G[*]$, where $\cR_{n,G}$ is a quantum circuit using quantum space $\nu n$, time $T_0$, and gate set $\cG$, and the family $\cE_{\mathrm{run},n,G}\subseteq 2^{[\nu n]\times[T_0]}$ is $n^\mu$-bounded.

    For $G\in\{\gMX,\gMZ\}$, the circuit $\cR_{n,G}$ simply measures all qubits in its first timestep, and then applies a classical function gate to the measurement outcome bits. For such $G\in\{\gMX,\gMZ\}$, we have $\ps_{n,G}=\emptyset$, but we let $\ps_{n,G}'$ be the set of all measurement gates in $\cR_{n,G}$. Then for every $\xi'\in\bF_2^{\ps_{n,G}'}$, there exists $\bar{\xi}\in\bF_2$ such that\footnote{Because we impose the postselection $\xi'$ directly on the circuit $\cR_{n,G}[\xi']$ in $\bfG_{n,G}[\xi']$, the remaining set of postselectable gates is empty.}
    \begin{equation*}
      \bfG_{n,G}[\xi'] := (\cR_{n,G}[\xi'],\; \cE_{\mathrm{run},n,G},\; D_n^{\sqcup k_{\mathrm{in},G}},\; D_n^{\sqcup k_{\mathrm{out},G}},\; \emptyset)
    \end{equation*}
    provides a mending refreshing fault-tolerant gadget for $G[\bar{\xi}]$.
  \item\label{it:prcnot} For every $n,n'\in\bfN$ with $n/8\leq n'\leq 8n$, there exists a mending refreshing fault-tolerant gadget
    \begin{equation*}
      \bfG_{n,n',\gCNOT} = (\cR_{n,n',\gCNOT},\; \cE_{\mathrm{run},n,n',\gCNOT},\; D_n\sqcup D_{n'},\; D_n\sqcup D_{n'},\; \ps_{n,n',\gCNOT})
    \end{equation*}
    for $\gCNOT$, where $\cR_{n,n',\gCNOT}$ is a quantum circuit using quantum space $\nu n$, time $T_0$, and gate set $\cG$, and the family $\cE_{\mathrm{run},n,n',\gCNOT}\subseteq 2^{[\nu n]\times[T_0]}$ is $n^\mu$-bounded.
  \end{enumerate}
  Furthermore, all classical function gates used in the quantum circuits above are computable by $\poly(n)$-size classical circuits.
\end{lemma}

We now describe how \Cref{lem:prior} follows from \cite{golowich_constant-overhead_2025}. The codes and gadgets of \cite{golowich_constant-overhead_2025} are similar to our codes and gadgets considered in \Cref{sec:ftnonu}. The principal difference is that in \Cref{sec:ftnonu} we consider tensor products of the 1-dimensional chain complexes in \Cref{sec:classcodes} constructed from many lossless expanders of different sizes strung together, whereas \cite{golowich_constant-overhead_2025} takes tensor products of 1-dimensional chain complexes constructed from a single lossless expander. Our more involved classical code construction in \Cref{sec:classcodes} arises from our need to support fault-tolerant injection and ejection. However, many of the techniques we use to perform fault-tolerant gates in the set $\cG$ in \Cref{lem:prior} are extensions of those used in \cite{golowich_constant-overhead_2025}. Below, we therefore briefly overview how the results of \cite{golowich_constant-overhead_2025} imply \Cref{lem:prior}, while drawing comparisons to some of our gadgets in \Cref{sec:ftnonu} to provide more intuition for the relevant techniques.

\begin{proof}[Proof sketch of \Cref{lem:prior}]
  We take the code instantiation described in \cite[Section~10]{golowich_constant-overhead_2025} with dimension parameter $r=4$. That is, we take each $Q_n$ to be the qLDPC code associated to level $2$ of the tensor product of four 1-dimensional (co)chain complexes, each of which is given by a classical LDPC code specified by a constant-degree lossless expander. As described in \cite{golowich_constant-overhead_2025}, random constant-degree bipartite graphs provide such expanders with any desired number of vertices above some constant lower bound (see e.g.~\cite[Theorem 4.16]{hoory_expander_2006}). An explicit construction with slightly less flexibility in parameters is also given by \cite{hsieh_explicit_2025}. For all constant-sized block lengths that are too small to use the lossless expander construction, we simply use the trivial identity encoding, so that a bare physical qubit encodes the logical qubit. Such a trivial code can still be obtained from the tensor product of four trivial 1-dimensional (co)chain complexes, and can also be padded to have an arbitrary constant block length. Hence we obtain a family of codes whose block lengths satisfy \Cref{it:prdensity} in \Cref{lem:prior}.

  \cite{golowich_constant-overhead_2025} also give a natural encoding map $\Enc_n$ for $Q_n$. Because the resulting 4-dimensional product code as constructed in \cite{golowich_constant-overhead_2025} has parameters $[[n,\Theta(n),\Omega(n^{(1/4)})]]$, we designate a single message qubit in $(Q_n,\Enc_n)$ as a logical qubit, and all other message qubits as gauge qubits. That is, if $Q_n$ has $k'+1$ message qubits, we view $Q_n$ as a $[[n,1,n^{\Omega(1)}]]$ subsystem code with encoding isometry $\Enc_n:\bC^2\otimes\bC^{2^{k'}}\rightarrow\bC^{2^n}$.

  We let $\cE_n$ be a family of bad sets of the following form. We fix an appropriate graph $G_n$ on $\Theta(n)$ vertices of maximum degree $O(1)$ and an appropriate constant $\gamma=\Omega(1)$. Then for every connected subgraph $G'$ of $G_n$ with $|V(G')|=\Theta(n^{1/4})$, we let $\cE_n$ contain all subsets $E\subseteq V(G')$ for which $|E|/|V(G')|\geq\gamma$. By definition $|\cE_n|\leq O(n)\cdot 2^{O(n^{1/4})}\leq 2^{O(n^{1/4})}$, so $\cE_n$ is $(\Omega(n^{1/4}),O(n^{1/4}))$-bounded. Note that \cite{golowich_constant-overhead_2025} in fact use the family of bad sets $\cE(G_n,\Theta(n^{1/4}),\gamma)$ described in \Cref{def:clusterfam}, which is similar to our family $\cE_n$, except that $|V(G')|\geq\Omega(n^{1/4})$ is allowed to be arbitrarily large. However, every set in $\cE(G_n,\Theta(n^{1/4}),\gamma)$ contains a set in $\cE_n$, so all results in \cite{golowich_constant-overhead_2025} hold with $\cE_n$ in place of $\cE(G_n,\Theta(n^{1/4}),\gamma)$.
  
  % We also let $\cE_n$ be the family of bad sets constructed in \cite{golowich_constant-overhead_2025}. Specifically, this family is of the form $\cE(G_n,\eta_n,\gamma)$ described in \Cref{def:clusterfam}, where $G_n$ is a graph on $\Theta(n)$ vertices of maximum degree $O(1)$, and where $\eta_n=\Theta(n^{1/4})$ and $\gamma=\Omega(1)$. By definition every $E\in\cE_n$ then has $|E|\geq\Omega(n^{1/4})$, and it is shown in the proof of \cite[Lemma~3.31]{golowich_constant-overhead_2025} that $|\cE_n|\leq O(n)\cdot 2^{O(n^{1/4})}=2^{O(n^{1/4})}$. Hence $\cE_n$ is $(\Omega(n^{1/4}),O(n^{1/4}))$-bounded.

\cite{golowich_constant-overhead_2025} then provide fault-tolerant gadgets for performing Pauli, Hadamard, $\gCNOT$, $\ket{0}$ and $\ket{+}$ state initialization, and measurement gates on these decorated codes. Note that our choice of $Q_n$ as the code at level $2$ of $r=4$-dimensional tensor product complexes ensures that \cite[Section~7]{golowich_constant-overhead_2025} provides fault-tolerant initialization gadgets for both $\ket{0}$ and $\ket{+}$ states. While the gadgets stated in \cite{golowich_constant-overhead_2025} act on a linear number of logical qubits in a code block, ignoring the logical action on all but one logical qubit immediately provides analogous gadgets for our $[[n,1,\Omega(n^{1/4})]]$ decorated subsystem codes $D_n$. Note also that \cite[Section~10.4]{golowich_constant-overhead_2025} shows how to swap chosen pairs of encoded qubits, within or across code blocks, to perform targeted gates on individual (or pairs of) such qubits. Hence regardless of which encoded qubit we specify as the logical qubit vs the gauge qubits, we still obtain gadgets from \cite{golowich_constant-overhead_2025} that perform the desired logical gates. For sufficiently small constant block sizes $n$ such that $\Enc_n$ is the trivial identity encoding (see above), all gadgets are instead given by applying the appropriate bare physical gates. In this case, fault-tolerance holds trivially, assuming appropriate choices of the constants hidden by the big-$O$'s in \Cref{lem:prior}.

To ensure that these gadgets are mending and refreshing, by \Cref{lem:seqcomp} we can pre-compose them with a mending error-correction gadget. In more detail, while \cite{golowich_constant-overhead_2025} does not explicitly consider the mending property, they provide an analogous state-preparation gadget as our gadget in \Cref{lem:stateprep}. Hence a mending refreshing fault-tolerant error-correction gadget can be constructed for the codes of \cite{golowich_constant-overhead_2025} by fault-tolerantly implementing a teleportation circuit, exactly analogously to our error-correction gadget in \Cref{lem:errcorr}. Note that for $G\in\{\gMX,\gMZ\}$, we use the measurement gadget from \cite{golowich_constant-overhead_2025} directly without precomposing with an error-correction gadget, in order to ensure that all qubits are measured in the first timestep as stated in \Cref{it:prgads} of \Cref{lem:prior}. However, the measurement gadget of \cite{golowich_constant-overhead_2025} is naturally mending, by exactly analogous reasoning as we used to show that our measurement gadget in \Cref{lem:logmeas} is mending. Measurement gadgets are trivially refreshing, as they have no output qubits.

As another technical point, the fault-tolerance framework used in \cite{golowich_constant-overhead_2025} is similar to our framework in this paper, with the principal difference being that in \Cref{def:faulttol} in this paper we require the existence of an appropriate output error set $E_{\mathrm{out}}$ depending only on the input error set $E_{\mathrm{in}}$, the fault error set $E_{\mathrm{run}}$, and the postselection $\xi$. In contrast, \cite{golowich_constant-overhead_2025} considers the slightly weaker setting where $E_{\mathrm{out}}$ is allowed to depend on the precise input and fault errors, rather than simply the sets $E_{\mathrm{in}},E_{\mathrm{run}}$ bounding their supports. Nevertheless, their techniques extend naturally to our stronger notion of fault-tolerance. In particular, applying their decoder to the direct sum complex of their (co)chain complexes, exactly as done in our \Cref{sec:decoder}, we obtain a state preparation gadget for $D_n$ analogous to our state preparation gadget in \Cref{lem:stateprep}, in which $E_{\mathrm{out}}$ is determined from $E_{\mathrm{in}},E_{\mathrm{run}},\xi$. As in \Cref{lem:stateprep}, this state preparation gadget for $D_n$ will have the set $\ps$ of postselectable gates containing all of the gadget's measurement gates.

Hence we obtain a mending refreshing fault-tolerant gadget $\bfG_{n,G}$ for every $G\in\cG\setminus\gCF_*$. All these gadgets from \cite{golowich_constant-overhead_2025} use space $O(n)$ and time $O(1)$. Furthermore, the family $\cE_{\mathrm{run},n,G}$ of bad fault sets used in \cite{golowich_constant-overhead_2025} consists of the disjoint union across all timesteps $t$ of a family of the form $\cE(G_{\mathrm{run},n,G},\;\eta_n,\;\gamma_{\mathrm{run}})$ described in \Cref{def:clusterfam}, where $G_{\mathrm{run},n,G}$ is again a graph on $\Theta(n)$ vertices of maximum degree $O(1)$, and where $\eta_n=\Theta(n^{1/4})$ and $\gamma_{\mathrm{run}}=\Omega(1)$. Hence $\cE_{\mathrm{run},n,G}$, like our decorated code family $\cE_n$, is $(\Omega(n^{1/4},O(n^{1/4})))$-bounded.

Thus we obtain all of the desired gadgets $\bfG_{n,G}$ in \Cref{it:prgads} in \Cref{lem:prior}. The statement regarding the postselected gadget $\bfG_{n,G}[\xi']$ for $G\in\{\gMX,\gMZ\}$ holds because the measurement gadgets in \cite[Section~9.3]{golowich_constant-overhead_2025}, like our measurement gadgets in \Cref{lem:logmeas}, simply measure all physical qubits in the input code block, and then decode the output. Postselections on these physical measurements simply induce postselection on the encoded logical state, as stated in \Cref{it:prgads} in \Cref{lem:prior}.

It only remains to show \Cref{it:prcnot}. For this purpose, we simply apply the code switching gadgets of \cite{golowich_constant-overhead_2025} to switch our logical qubit encoded in $D_{n'}$ into the code $D_n$. We then apply the $\gCNOT$ gadget from \Cref{it:prgads} on our two copies of $D_n$, before switching the appropriate code block back to $D_{n'}$. In more detail, to perform this fault-tolerant code switching, we apply the code switching gadgets from \cite[Sections~5 and~6]{golowich_constant-overhead_2025}, which allow us to fault-tolerantly switch out the four $1$-dimensional factor chain complexes comprising the 4-dimensional tensor product chain complex associated to $Q_n$. That is, we loop through each of the four factor complexes comprising $Q_n$, and for each one we perform the ``downwards switching'' in \cite[Section~5]{golowich_constant-overhead_2025} followed by the ``upwards switching'' in \cite[Section~6]{golowich_constant-overhead_2025} to swap out the factor complex of $Q_n$ for that of $Q_{n'}$. Indeed, a similar technique of swapping out the factor complexes one at a time is used in \cite[Sections~10.2 and~10.3]{golowich_constant-overhead_2025} to perform logical qubit permutations and Hadamard gates. As shown in \cite{golowich_constant-overhead_2025}, this entire procedure uses space $O(n)$ and time $O(1)$. This procedure is again fault-tolerant with respect to a $(\Omega(n^{1/4},O(n^{1/4})))$-bounded family of bad fault sets, which as before consists of the disjoint union of a constant number of families of the form described in \Cref{def:clusterfam}.

Furthermore, all classical function gates used in \cite{golowich_constant-overhead_2025} run $\poly(n)$-sized classical circuits. This classical efficiency arises from the fact that \cite{golowich_constant-overhead_2025} construct their gadgets using an efficient small-set flip decoder that they develop for their codes, which is analogous to our efficient decoder in \Cref{sec:decoder}.
\end{proof}

Other constructions in the literature that are also based on code switching, e.g.~\cite{bombin_dimensional_2016,xu_batched_2025,tan_single-shot_2025}, also imply results that are similar to \Cref{lem:prior}. However, these prior works rely on minimum-weight decoders for their quantum codes, which generically do not have known implementations by polynomial-time classical circuits. Nevertheless, if we drop the requirement that all classical function gates have such efficient implementations, then these alternative constructions may also suffice in place of \cite{golowich_constant-overhead_2025}.

\subsection{Concatenated Fault-Tolerance Scheme}
In this section, we show how to concatenate an ``outer'' gadget with the ``inner'' scheme of \Cref{lem:prior}. Here we think of the outer gadget as being constructed from some parallel and/or sequential compositions of gadgets from \Cref{sec:ftnonu}, so that we have mapping $\Lev:\Qu{N}\rightarrow\bZ_{\geq 0}$ that assigns a nonnegative integer ``level'' to every physical qubit. Our goal is to simulate execution of the outer scheme under non-uniform locally stochastic noise, where higher-level qubits have lower error-probabilities, using an execution of the concatenated scheme under uniform locally stochastic noise. Our concatenation will therefore encode higher-level qubits of the outer gadget into larger code blocks of the inner scheme.

We encode each level-$\ell$ qubit of the outer gadget using the code $D_{\bfn(\ell)}$ from \Cref{lem:prior}, for $\bfn(\ell)\in\bfN$ given by \Cref{def:tn} below.

\begin{definition}
  \label{def:tn}
  Let $\bfN\subseteq\bN$ be the set defined in \Cref{lem:prior}. For every $\ell\in\bZ_{\geq 0}$, we define $\bfn(\ell)\in\bfN$ to be the least element of $\bfN$ that is at least $2^{\ell/8}$. It in particular follows from \Cref{it:prdensity} in \Cref{lem:prior} that $\bfn(0)=1$, and that for every $\ell\in\bZ_{\geq 0}$ we have $2^{\ell/8}\leq\bfn(\ell)<2^{\ell/8+2}$ and $1\leq\bfn(\ell+1)/\bfn(\ell)\leq 8$.
\end{definition}

The specific constant $8$ in \Cref{def:tn} is not important. Rather, we simply need $\bfn(\ell)$ to:
\begin{enumerate}
\item Grow exponentially in $\ell$, to ensure sufficient error suppression, and
\item Grow more slowly than $2^\ell$, to ensure that encoding each level-$\ell$ physical qubit of our codes from \Cref{sec:quantumprod} into a length-$\bfn(\ell)$ code results in a concatenated code that still has constant rate. Recall here that by construction, the fraction of qubits at level $\geq\ell$ grows as $O(1/2^\ell)$.
\end{enumerate}

We now define decorated codes obtained by concatenation with those in \Cref{lem:prior}.

\begin{definition}
  \label{def:lccode}
  Let $D=(Q,\Enc,\cE)$ be a $[[n,k,d]]$ decorated code for which every physical qubit $a\in[n]$ has an assigned level $\Lev(a)\in\bZ_{\geq 0}$. Define $(D_{\bfn(\ell)})_{\ell\in\bN}$ as given by \Cref{lem:prior,def:tn}. We then define the \emph{level-concatenated code} $\tilde{D}=D_{\bfn(\cdot)}\diamond D$ to be the $[[\bfn,k,\tilde{d}]]$ concatenated subsystem code
  \begin{equation*}
    \tilde{D}=(\tilde{Q},\widetilde{\Enc},\tilde{\cE})
  \end{equation*}
  defined as follows. We let
  \begin{align*}
    \widetilde{\Enc} &= \left(\bigotimes_{a\in[n]}\Enc_{\bfn(\Lev(a))}\right)\circ\Enc
  \end{align*}
  simply first apply $\Enc$, and then apply $\Enc_{\bfn(\ell)}$ to each resulting level-$\ell$ qubit. We define $\tilde{Q}$ to be the resulting concatenated CSS (subsystem) code defined by the encoding map $\widetilde{\Enc}$. Therefore $\tilde{Q}$ has length $\bfn=\sum_{a\in[n]}\bfn(\Lev(a))$, dimension $k$, and distance $\tilde{d}\geq d$. We define the pair of families of bad sets $\tilde{\cE}=(\tilde{\cE}_X,\tilde{\cE}_Z)$ such that for $\alpha\in\{X,Z\}$, then $\tilde{\cE}_\alpha\subseteq\bigsqcup_{a\in[n]}[\bfn(\Lev(a))]$ is given by
  \begin{align*}
    \tilde{\cE}_\alpha &= \bigcup_{E\in\cE_\alpha}\left(\prod_{a\in E}\cE_{\bfn(\Lev(a))}\right).
  \end{align*}
  That is, for every $E\in\cE_\alpha$, $\tilde{\cE}_\alpha$ contains all sets that consist of the union over all $a\in E$ of a bad set in $\cE_{\bfn(\Lev(a))}$.
\end{definition}

We next define circuits, along with associated sets of postselectable gates and families of bad fault sets, that arise from concatenation with the gadgets in \Cref{lem:prior}.

\begin{definition}
  \label{def:lccirc}
  Let $\cR$ be a quantum circuit using quantum space $\Qu{N}$ and time $T$, in which every qubit $a\in\Qu{N}$ has an assigned level $\Lev(a)\in\bZ_{\geq 0}$. We assume every $\gCNOT$ gate in $\cR$ acts on a pair of qubits whose levels differ by at most $1$. Let $\ps$ be an associated set of postselectable gates, and let $\cE_{\mathrm{run}}\subseteq\Qu{N}\times[T]$ be a family of bad fault sets.

  Let the gate set $\cG$, all gadgets $\bfG_{n,G}$, all decorated subsystem codes $D_n$, and the constants $\mu,\nu,T_0$ be defined as in \Cref{lem:prior}, and let $\bfn(\cdot)$ be defined as in \Cref{def:tn}.

  We then define the \emph{level-concatenated circuit} $\tilde{\cR}$ as follows. $\tilde{\cR}$ uses the exact set of classical bits $\Cl{\tilde{\cR}}=\Cl{\cR}$ as $\cR$. However, $\tilde{\cR}$ uses quantum space $\Qu{\tilde{N}}=\bigsqcup_{a\in\Qu{N}}[\nu\cdot\bfn(\Lev(a))]$ and time $\tilde{T}=T_0\cdot T$. Therefore each qubit $a\in\Qu{N}$ has a set $\tilde{N}_a\cong[\nu\cdot\bfn(\Lev(a))]$ of associated qubits in $\Qu{\tilde{N}}$, and each timestep $t\in[\tilde{T}]$ has a set $\tilde{T}_a=\{T_0(t-1)+1,\dots,T_0t\}$ of associated timesteps in $[\tilde{T}]$.

  Let $H$ denote the set of all instances of all gates applied in $\cR$. We let $N_h\subseteq\Qu{N}$ and $t_h\in[T]$ denote the set of qubits and timestep on which $h$ is applied in $\cR$. We let $G_h\in\cG$ denote the type of gate that $h$ applies. In particular, it follows that $|N_h|\in\{0,1,2\}$, and $|N_h|=2$ if and only if $G_h=\gCNOT$ gate.

  For every $t\in[T]$ and every gate $h\in H$ applied in $\cR$, we define $\tilde{\cR}$ to apply an associated sub-circuit $\tilde{\cR}_h$ to qubits $\tilde{N}_h=\bigsqcup_{a\in N_h}\tilde{N}_a$ in timesteps $\tilde{T}_h=\{T_0(t-1)+1,\dots,T_0t\}$ as follows:
  \begin{itemize}
  \item If $G_h\in\cG\setminus\{\gCNOT,\gCF_*\}$, then $h$ acts on a single qubit $N_h=\{a\}$. We therefore define $\tilde{\cR}_h:=\cR_{\bfn(\Lev(a)),G_h}$ to be the circuit from the gadget $\tilde{\bfG}_h:=\bfG_{\bfn(\Lev(a)),G_h}$ from \Cref{it:prgads} in \Cref{lem:prior}.
  \item If $G_h=\gCNOT$, then $G_h$ acts on a pair of qubits $N_h=\{a,a'\}$. We therefore define $\tilde{\cR}_h:=\cR_{\bfn(\Lev(a)),\bfn(\Lev(a')),\gCNOT}$ to be the circuit from the gadget $\tilde{\bfG}_h:=\bfG_{\bfn(\Lev(a)),\bfn(\Lev(a')),\gCNOT}$ from \Cref{it:prcnot} in \Cref{lem:prior}.
  \item If $G_h\in\gCF_*$, then $G_h$ acts only on classical bits, and not on any qubits. Recalling that $\Cl{\tilde{\cR}}=\Cl{\cR}$, we define $\tilde{\cR}_h$ to idle (i.e.~apply identity gates) on the input bits to $G_h$ in timesteps $\{T_0(t-1)+1,\dots,T_0t-1\}$, and then to apply $G_h$ in timestep $T_0t$.
  \end{itemize}

  We similarly define the level-concatenated set of postselectable gates $\tilde{\ps}$ to contain the sets of measurement gates $\ps_{n,G}'$ defined in \Cref{it:prgads} in \Cref{lem:prior} for all instances of $\cR_{n,G}$ in $\tilde{\cR}$ associated to a measurement gate $G$ in $\ps$, along with the sets $\ps_{n,G}$ defined in \Cref{it:prgads} in \Cref{lem:prior} for all instances of $\cR_{n,G}$ in $\tilde{\cR}$ associated to a gate $G$ not in $\ps$.

  Furthermore, we define the level-concatenated family of bad fault sets $\tilde{\cE}_{\mathrm{run}}\subseteq 2^{\Qu{\tilde{N}}\times[\tilde{T}]}$ as follows. Intuitively, each family of bad sets in $\tilde{\cE}_{\mathrm{run}}$ consists of the union of families of bad sets for gadgets $\bfG_{n,G}$ at locations within $\tilde{\cR}$ that correspond to elements of some chosen bad set $E\in\cE_{\mathrm{run}}$. Formally, first let
  \begin{equation*}
    \tilde{\cE}_{h,\mathrm{run}} \subseteq \tilde{N}_h\times\tilde{T}_h
  \end{equation*}
  be the family of bad fault sets for the gadget $\tilde{\bfG}_h$, as given by \Cref{lem:prior} as described above. Thus $\tilde{\cE}_{h,\mathrm{run}}$ is $(\Omega(n^\mu),O(n^\mu))$-bounded, where $n=\Theta(\bfn(\Lev(a)))$ for $a\in N_h$. For a set $E\in\cE_{\mathrm{run}}$, we define $\cH_E\subseteq 2^H$ to be the set of all $H'\subseteq H$ satisfying both of the following properties:
  \begin{enumerate}
  \item\label{it:EtoH} For every $(a,t)\in E$, then either $(a,t)$ or $(a,t+1)$ lies in the support $N_h\times\{t_h\}$ of some $h\in H'$.
  \item\label{it:HtoE} For every $h\in H'$, there exists $(a,t)\in N_h\times\{t_h\}$ such that either $(a,t)\in E$ or $(a,t-1)\in E$.
  \end{enumerate}
  For $H'\in\cH_E$, we define $\tilde{\cE}_{\mathrm{run},H'}\subseteq 2^{\Qu{\tilde{N}}\times[\tilde{T}]}$ by
  \begin{equation*}
    \tilde{\cE}_{\mathrm{run},H'} = \prod_{h\in H'}\tilde{\cE}_{h,\mathrm{run}}.
  \end{equation*}
  We then define
  \begin{equation}
    \label{eq:tErE}
    \tilde{\cE}_{\mathrm{run}} = \bigcup_{E\in\cE_{\mathrm{run}}} \tilde{\cE}_{\mathrm{run},E}, \hspace{1em} \text{where each} \hspace{1em} \tilde{\cE}_{\mathrm{run},E} = \bigcup_{H'\in\cH_E}\tilde{\cE}_{\mathrm{run},H'}.
  \end{equation}
\end{definition}

\Cref{claim:lsconvert} below shows that we can convert locally stochastic noise with non-uniform error probabilities on a family of bad fault sets $\cE_{\mathrm{run}}$ to locally stochastic noise with uniform error probabilities on the level-concatenated family $\tilde{\cE}_{\mathrm{run}}$. More precisely, we assume that each qubit at level $\ell\in\bZ_{\geq 0}$ of the non-concatenated circuit $\cR$ experiences an error with probability $p^{2^{(\mu/8)\cdot\ell}}=p^{2^{\Theta(\ell)}}$. Then we show that there exists $\tilde{p}>0$ such that if each qubit of the level-concatenated circuit $\tilde{\cR}$ experiences an error with probability $\tilde{p}$, a union bound on the chance of a bad error in $\tilde{\cE}_{\mathrm{run}}$ (under uniform error probabilities) is less than a union bound on the chance of a bad error in $\cE_{\mathrm{run}}$ (under non-uniform error probabilities). As we typically bound the chance of a bad error from non-uniform locally stochastic noise using union bounds, it follows that our schemes with low error probability under non-uniform locally stochastic noise can be converted to level-concatenated schemes with low error probability under uniform locally stochastic noise.

\begin{claim}
  \label{claim:lsconvert}
  Define the family of bad fault sets $\cE_{\mathrm{run}}$ and its level-concatenated lift $\tilde{\cE}_{\mathrm{run}}$ as in \Cref{def:lccirc}. Let $\mu>0$ be the constant defined in \Cref{lem:prior}. Then for every $p>0$, there exists $\tilde{p}=\tilde{p}(p)>0$ such that
  \begin{equation*}
    \sum_{\tilde{E}\in\tilde{\cE}_{\mathrm{run}}}\tilde{p}^{|\tilde{E}|} \leq \sum_{E\in\cE_{\mathrm{run}}}\prod_{(a,t)\in E}p^{2^{(\mu/8)\Lev(a)}} = \sum_{E\in\cE_{\mathrm{run}}}p^{\sum_{(a,t)\in E}2^{(\mu/8)\Lev(a)}}.
  \end{equation*}
\end{claim}
\begin{proof}
  By definition, for each $E\in\cE_{\mathrm{run}}$, then $\tilde{\cE}_{\mathrm{run}}$ contains the associated family $\tilde{\cE}_{\mathrm{run},E}$ from \Cref{eq:tErE}. Now we must have
  \begin{equation*}
    |\cH_E| \leq 4^{|E|},
  \end{equation*}
  as by \Cref{it:HtoE} in \Cref{def:lccirc}, each $H'\in\cH_E$ is entirely supported on gates adjacent to space-time locations in $E$. Then
  \begin{equation*}
    |\tilde{\cE}_{\mathrm{run},E}| \leq \sum_{H'\in\cH_E}\prod_{h\in H'}|\tilde{\cE}_{h,\mathrm{run}}| \leq \prod_{(a,t)\in E}2^{O(\bfn(\Lev(a))^\mu)} = 2^{\sum_{(a,t)\in E}O(\bfn(\Lev(a))^\mu)},
  \end{equation*}
  where the second inequality above uses the fact that $\tilde{\cE}_{h,\mathrm{run}}$ is $(\Omega(n^\mu),O(n^\mu))$-bounded with $n=\Theta(\bfn(\Lev(a)))$ for $a\in N_h$.

  Furthermore, because by \Cref{it:EtoH} in \Cref{def:lccirc}, each $H'\in\cH_E$ must contain some gate adjacent to each space-time location in $E$, it follows that every $\tilde{E}\in\tilde{\cE}_{\mathrm{run},E}$ must satisfy
  \begin{equation*}
    |\tilde{E}| \geq \sum_{(a,t)\in E}\Omega(\bfn(\Lev(a))^\mu),
  \end{equation*}
  where we again apply $(\Omega(n^\mu),O(n^\mu))$-boundedness.

  Thus for $0<\tilde{p}<1$ we have
  \begin{align*}
    \sum_{\tilde{E}\in\tilde{\cE}_{\mathrm{run}}}\tilde{p}^{|\tilde{E}|}
    &\leq 2^{\sum_{(a,t)\in E}O(\bfn(\Lev(a))^\mu)} \cdot \tilde{p}^{\sum_{(a,t)\in E}\Omega(\bfn(\Lev(a))^\mu)}.
  \end{align*}
  Therefore for sufficiently small $\tilde{p}=\tilde{p}(p)>0$ relative to the constants hidden by the big-$O$ and $\Omega$ above, we have
  \begin{align*}
    \sum_{\tilde{E}\in\tilde{\cE}_{\mathrm{run}}}\tilde{p}^{|\tilde{E}|}
    &\leq p^{\sum_{(a,t)\in E}\bfn(\Lev(a))^\mu} \leq p^{\sum_{(a,t)\in E}2^{(\mu/8)\Lev(a)}},
  \end{align*}
  as desired.
\end{proof}

Our goal is to use the above definitions of level-concatenated objects to concatenate outer gadgets such as those given in \Cref{sec:ftnonu} with inner gadgets from \Cref{lem:prior}. Specifically, we will show how to perform such concatenation for outer gadgets that are \emph{concatenation-compatible} in the sense of \Cref{def:coco} below.

\begin{definition}
  \label{def:coco}
  A fault-tolerant gadget
  \begin{equation*}
    \bfG = (\cR,\; \cE_{\mathrm{run}},\; D_{\mathrm{in}}=(Q_{\mathrm{in}},\Enc_{\mathrm{in}},\cE_{\mathrm{in}}),\; D_{\mathrm{out}}=(Q_{\mathrm{out}},\Enc_{\mathrm{out}},\cE_{\mathrm{out}}),\; \ps)
  \end{equation*}
  for $\bar{\cO}$ is \emph{concatenation-compatible} if the following hold:
  \begin{enumerate}
  \item Assume $\cR$ be a quantum circuit using quantum space $\Qu{N}$ and time $T$, and assume we have a fixed mapping $\Lev:\Qu{N}\rightarrow\bZ_{\geq 0}$ assigning a nonnegative integer \emph{level} to each qubit in $\Qu{N}$. We require each $\gCNOT$ gate in $\cR$ to act on a pair of qubits whose levels differ by at most $1$.
  \item We require fault-tolerance to hold even under faults $\cF=(F_0,\dots,F_T)$ with a layer of errors $F_0$ at timestep $0$, so that $\cE_{\mathrm{run}}\subseteq 2^{\Qu{N}\times\{0,\dots,T\}}$ and $\cR[\cF]=(F_0,R_1,F_1,\dots,R_T,F_T)$. (Recall that previously in \Cref{def:fault} we instead considered faults of the form $\cF=(F_1,\dots,F_T)$.)
  \item\label{it:outbad} We assume that for every $\cE_{\mathrm{in}}$-avoiding pair $E_{\mathrm{in}}=(E_{\mathrm{in},X},E_{\mathrm{in},Z})$, every $\cE_{\mathrm{run}}$ avoiding set $E_{\mathrm{run}}\subseteq 2^{\Qu{N}\times\{0,\dots,T\}}$, and every postselection $\xi\in\bF_2^{\ps}$, then the resulting $\cE_{\mathrm{out}}$-avoiding pair $E_{\mathrm{out}}=(E_{\mathrm{out},X},E_{\mathrm{out},Z})$ guaranteed by fault-tolerance satisfies the following: both $E_{\mathrm{out},X},E_{\mathrm{out},Z}$ contain every output qubit $a\in N_{\mathrm{out}}$ for which $(a,T)\in E_{\mathrm{run}}$.
  % \item\label{it:outbad} We assume that for every $\cE_{\mathrm{in}}$-avoiding pair $E_{\mathrm{in}}=(E_{\mathrm{in},X},E_{\mathrm{in},Z})$, every $\cE_{\mathrm{run}}$ avoiding set $E_{\mathrm{run}}\subseteq 2^{\Qu{N}\times\{0,\dots,T\}}$, and every postselection $\xi\in\bF_2^{\ps}$, then the resulting $\cE_{\mathrm{out}}$-avoiding pair $E_{\mathrm{out}}=(E_{\mathrm{out},X},E_{\mathrm{out},Z})$ guaranteed by fault-tolerance satisfies the following: both $E_{\mathrm{out},X},E_{\mathrm{out},Z}$ contain the set $N_h$ for every $h\in H$ for which $t_h=T$ and $(N_h\times\{t_h\})\cap E_{\mathrm{run}}\neq\emptyset$.
  \end{enumerate}
\end{definition}

We briefly explain why the three requirements in \Cref{def:coco} are reasonable, and in particular, are satisfied by slightly modified versions of our gadgets from \Cref{sec:ftnonu}:

\begin{enumerate}
\item All physical circuits $\cR$ we consider in \Cref{sec:ftnonu} only apply $\gCNOT$ gates to pairs of qubits whose levels differ by at most $1$. Furthermore, all circuits $\cR$ we consider also have at most $O(1/2^\ell)$-fraction of all qubits at level $\geq\ell$, so because $\bfn(\ell)<2^{\ell/8+2}$, there is just a constant factor blowup in space usage from $\cR$ to the level-concatenated circuit $\tilde{\cR}$.
\item\label{it:fault0explain} While in \Cref{sec:ftnonu} we considered faults starting at timestep $1$, all of our gadgets extend naturally to allow for faults starting at timestep $0$, up to constant-factor changes in the parameters. Indeed, note that when we sequentially compose two of our gadgets, errors from the fault in the final timestep of the former can be viewed as errors at timestep $0$ of the latter.
\item Following the final timestep $T$ of our circuit $\cR$, the fault applies the error $F_T$ on the output state, which will be a corrupted code state of $D_{\mathrm{out}}$. Because our gadgets are composable (similarly as argued in \Cref{it:fault0explain} above), they can correct this error from $F_T$, which is supported inside the time-$T$ layer of $E_{\mathrm{run}}$. Thus up to constant-factor changes in parameters, we may define $\cE_{\mathrm{out}}$ and the $\cE_{\mathrm{out}}$-avoiding set $E_{\mathrm{out}}$ so that $E_{\mathrm{out}}$ always contains the time-$T$ layer of $E_{\mathrm{run}}$.

  In more detail, in \Cref{sec:ftnonu} we previously defined $E_{\mathrm{out}}$ to simply contain the support of the \emph{reduction} of errors supported in the time-$T$ layer of $E_{\mathrm{out}}$. However, because the reduction of an error has the same effect as the error itself on the code space, all our gadgets still work with $E_{\mathrm{out}}$ containing the non-reduced error support as well. We may similarly modify $\cE_{\mathrm{in}}$ and $E_{\mathrm{in}}$ to account for the support of the non-reduced error, to retain symmetry between the input and output bad error families, and hence retain composability.

  % Note that we may also assume without loss of generality (up to constant factor changes in parameters) that the final timestep of $\cR$ applies identity gates to all qubits, so that the union of all $N_h$ for $h\in H$ for which $t_h=T$ and $(N_h\times\{t_h\})\cap E_{\mathrm{run}}\neq\emptyset$, as defined in \Cref{it:outbad} of \Cref{def:coco}, is precisely the time-$T$ layer of $E_{\mathrm{out}}$.
\end{enumerate}

\Cref{lem:lcft} below provides our main lemma showing fault-tolerance of a level-concatenated gadget. Related concatenation results, though in slightly different settings, have been shown before (see e.g.~\cite{he_composable_2025}). In particular, our proof of fault-tolerance for the concatenated scheme is similar to that of \cite{breuckmann_fault-tolerant_2026}.

\begin{lemma}
  \label{lem:lcft}
  Let
  \begin{align*}
    \bfG &= (\cR,\cE_{\mathrm{run}},D_{\mathrm{in}},D_{\mathrm{out}},\ps)
  \end{align*}
  be a concatenation-compatible fault-tolerant gadget for a set $\bar{\cO}$ of superoperators, where $\cR$ is a quantum circuit using quantum space $\Qu{N}$, time $T$, and gate set
  \begin{equation*}
    \cG = \{\gI,\gX,\gZ,\gH,\gCNOT,\gInitX,\gInitZ,\gTerm,\gMX,\gMZ,\gCt{\gX},\gCt{\gZ},\gCF_*\}.
  \end{equation*}
  % We assume that we have a mapping $\Lev:\Qu{N}\rightarrow\bZ_{\geq 0}$ assigning a nonnegative integer \emph{level} to each qubit in $\Qu{N}$. Also assume that every $\gCNOT$ gate in $\cR$ acts on qubits whose levels differ by at most $1$.

  Then we have a level-concatenated fault-tolerant gadget
  \begin{align*}
    \tilde{\bfG} &= (\tilde{\cR},\tilde{\cE}_{\mathrm{run}},\tilde{D}_{\mathrm{in}},\tilde{D}_{\mathrm{out}},\tilde{\ps})
  \end{align*}
  for $\bar{\cO}$, where $\tilde{\cR}$ is the level-concatenated quantum circuit using quantum space $\Qu{\tilde{N}}$ and time $\tilde{T}$ as defined in \Cref{def:lccirc}, so that $|\Qu{\tilde{N}}|=O(\sum_{a\in\Qu{N}}2^{\Lev(a)/8})$ and $\tilde{T}=O(T)$. Here $\tilde{\cE}_{\mathrm{run}},\tilde{\ps}$ and $\tilde{D}_{\mathrm{in}},\tilde{D}_{\mathrm{out}}$ are the level-concatenated objects defined in \Cref{def:lccirc} and \Cref{def:lccode}, respectively.
\end{lemma}

To prove \Cref{lem:lcft}, we will use the following notion of an \emph{extended fault}, which gives the adversary access to a private register of qubits $\Ad{N}$ on which they can perform arbitrary gates, including entangling gates with the original circuit's qubits $\Qu{N}$.

\begin{definition}
  \label{def:extfault}
  Let $\cR=(R_1,\dots,R_T)$ be a quantum circuit using space $N=\Cl{N}\sqcup\Qu{N}$ and time $T$, with each $R_t:\bC^{2^{N_{t-1}}\times 2^{N_{t-1}}}\rightarrow\bC^{2^{N_t}\times 2^{N_t}}$. For a set of qubits $\Ad{N}$ that is disjoint from $N=\Cl{N}\sqcup\Qu{N}$, we define an \emph{extended fault} for $\cR$ to be a sequence $\cF=(F_0,\dots,F_T)$ of superoperators
  \begin{equation*}
    F_t:\bC^{2^{\Cl{N}_t\sqcup\Qu{N}_t\sqcup\Ad{N}}\times 2^{\Cl{N}_t\sqcup\Qu{N}_t\sqcup\Ad{N}}}\rightarrow\bC^{2^{\Cl{N}_t\sqcup\Qu{N}_t\sqcup\Ad{N}}\times 2^{\Cl{N}_t\sqcup\Qu{N}_t\sqcup\Ad{N}}},
  \end{equation*}
  with the restriction that the only gates that $F_t$ applies to the classical bits $\Cl{N}_t$ are classically-controlled gates $\gCt{G}$, for gates $G$ acting on qubits $\Qu{N}_t\sqcup\Ad{N}$. In particular, it follows that the Pauli decomposition of each $F_t$ only applies the $\gI$ and $\gZ$ Paulis to bits in $\Cl{N}$.
  We define $\supp(\cF)\subseteq \Qu{N}\times\{0,\dots,T\}$ to be the support
  \begin{equation*}
    \supp(\cF) = (\supp(F_0)\cap\Qu{N})\sqcup\cdots\sqcup(\supp(F_T)\cap\Qu{N})
  \end{equation*}
  of $\cF$ on the original qubits $\Qu{N}$ of $\cR$. An \emph{extended Pauli fault} $\cF$ is an extended fault in which each $F_t$ is a Pauli superoperator.

  The \emph{$\cF$-corrupted circuit} $\cR[\cF]$ is defined by
  \begin{equation*}
    \cR[\cF] = (\gInitZ^{\otimes\Ad{N}}\otimes I_{N_0},\; F_0,\; R_1,\; F_1,\; R_2,\; F_2,\dots,R_T,\; F_T,\; \gTerm^{\otimes\Ad{N}}\otimes I_{N_t}).
  \end{equation*}
  That is, $\cR[\cF]$ initializes all qubits in $\Ad{N}$, then alternatatively applies the fault and circuit superoperators $F_t$ and $R_t$ respectively, and finally terminates all qubits in $\Ad{N}$.
\end{definition}

\begin{proof}[Proof of \Cref{lem:lcft}]
  Fix a $\tilde{\cE}_{\mathrm{in}}$-avoiding pair of sets $\tilde{E}_{\mathrm{in}}=(\tilde{E}_{\mathrm{in},X}\subseteq\Qu{\tilde{N}}_{\mathrm{in}},\tilde{E}_{\mathrm{in},Z}\subseteq\Qu{\tilde{N}}_{\mathrm{in}})$, a $\tilde{\cE}_{\mathrm{run}}$-avoiding set $\tilde{E}_{\mathrm{run}}\subseteq\Qu{\tilde{N}}\times[\tilde{T}]$, and a postselection $\tilde{\xi}\in\bF_2^{\tilde{\ps}}$ for $\tilde{\bfG}$.

  We define a pair of sets $E_{\mathrm{in}}=(E_{\mathrm{in},X},E_{\mathrm{in},Z})$ as follows: for $\alpha\in\{X,Z\}$, we let $E_{\mathrm{in},\alpha}\subseteq\Qu{N}_{\mathrm{in}}$ contain every $a\in\Qu{N}_{\mathrm{in}}$ for which $\tilde{E}_{\mathrm{in},\alpha}\cap[\bfn(\Lev(a))]$ is not $\cE_{\bfn(\Lev(a))}$-avoiding. Here $\tilde{E}_{\mathrm{in},\alpha}\cap[\bfn(\Lev(a))]$ refers to the restriction of $\tilde{E}_{\mathrm{in},\alpha}$ to qubits in the length-$\bfn(\Lev(a))$ encoding of $a\in N_{\mathrm{in}}$, and $\cE_{\bfn(\Lev(a))}$ is defined as in \Cref{lem:prior}. Then by \Cref{def:lccode}, the pair $E_{\mathrm{in}}$ must be $\cE_{\mathrm{in}}$-avoiding.

  We next define a set $E_{\mathrm{run}}\subseteq\Qu{N}\times\{0,\dots,T\}$ as follows, where below we use the notation $H$, $\tilde{\cE}_{h,\mathrm{run}}$, etc.~from \Cref{def:lccirc}. First, we define $H_{\mathrm{run}}\subseteq H$ to contain every gate $h\in H$ for which the restriction $\tilde{E}_{\mathrm{run}} \cap (\tilde{N}_h\times \tilde{T}_h)$ of $\tilde{E}_{\mathrm{run}}$ to the subcircuit $\tilde{\cR}_h$ is not $\tilde{\cE}_{h,\mathrm{run}}$-avoiding.
  We then let $E_{\mathrm{run}}$ contain every $(a,t)\in\Qu{N}\times\{0,\dots,T\}$ such that either $(a,t)$ or $(a,t+1)$ lies in the support $N_h\times\{t_h\}$ of some gate $h\in H_{\mathrm{run}}$. Then by \Cref{def:lccirc}, the set $E_{\mathrm{run}}$ must be $\cE_{\mathrm{run}}$-avoiding.

  We furthermore define a postselection $\xi\in\bF_2^{\ps}$ as follows. For every $h\in\ps$, to define $\xi_h$, we first let $\tilde{\xi}_h$ denote the restriction of $\tilde{\xi}$ to components corresponding to postselectable gates in $\tilde{N}_h\times\tilde{T}_h$. For every $h\in\ps\subseteq H$, so that $G_h\in\{\gMX,\gMZ\}$, recall that the gadget $\tilde{\bfG}_h$ occupies space-time locations $\tilde{N}_h\times\tilde{T}_h$, so that $\tilde{\xi}_h$ is a postselection for $\tilde{\bfG}_h$. By \Cref{lem:prior}, there exists $\xi_h\in\bF_2$ such that $\tilde{\bfG}_h[\tilde{\xi}_h]$ provides a fault-tolerant gadget for $G_h[\xi_h]$.

  Given $E_{\mathrm{in}},E_{\mathrm{run}},\xi$ defined above, the fault-tolerance of $\bfG$ provides a $\cE_{\mathrm{out}}$-avoiding output error set pair $E_{\mathrm{out}}=(E_{\mathrm{out},X}\subseteq\Qu{N}_{\mathrm{out}},E_{\mathrm{out},Z}\subseteq\Qu{N}_{\mathrm{out}})$. We define the error set pair $\tilde{E}_{\mathrm{out}}=(\tilde{E}_{\mathrm{out},X}\subseteq\Qu{\tilde{N}}_{\mathrm{out}},\tilde{E}_{\mathrm{out},Z}\subseteq\Qu{\tilde{N}}_{\mathrm{out}})$ as follows, for each $\alpha\in\{X,Z\}$:
  \begin{itemize}
  \item For every $a\in E_{\mathrm{out},\alpha}$, we let $\tilde{E}_{\mathrm{out},\alpha}$ contain all qubits in the set $\tilde{N}_a$.
  \item For every $a\in N_{\mathrm{out}}\setminus E_{\mathrm{out},\alpha}$, by \Cref{it:outbad} in \Cref{def:coco} there exists $h\in H$ such that $h_t=T$, $a\in N_t$, and $(N_h\times\{t_h\})\cap E_{\mathrm{run}}=\emptyset$. Then by the definition of $H_{\mathrm{run}},E_{\mathrm{run}}$ above, we must have $h\notin H_{\mathrm{run}}$. Therefore given the restriction of $\tilde{E}_{\mathrm{run}}$ and $\tilde{\xi}$ to the space-time locations $\tilde{N}_h\times\tilde{T}_h$, the mending refreshing fault-tolerance of $\tilde{\bfG}_h$ from \Cref{lem:prior} provides an associated output error set $\tilde{E}_{h,\mathrm{out}}\subseteq\bigsqcup_{a\in N_h}[\bfn(\Lev(a))]$. For each $a\in N_h$, we let $\tilde{E}_{\mathrm{out},\alpha}$ contain all qubits in the set $\tilde{E}_{h,\mathrm{out}}\cap[\bfn(\Lev(a))]$, which is $\cE_{\bfn(\Lev(a))}$-avoiding by \Cref{lem:prior}.
  \end{itemize}
  
  Now to show that $\tilde{\bfG}$ provides a fault-tolerant gadget for $\bar{\cO}$, fix a finite set $\Rf{K}$, an operator $\rho\in\bC^{2^{K_{\mathrm{in}}\sqcup\Rf{K}}\times 2^{K_{\mathrm{in}}\sqcup\Rf{K}}}$ that is classical on bits $\Cl{K}_{\mathrm{in}}$, a Pauli $\comp{\tilde{E}}_{\mathrm{in}}$-deviation $\tilde{\sigma}$ of $\widetilde{\Enc}_{\mathrm{in}}(\rho)$, and a $\comp{\tilde{E}}_{\mathrm{run}}$-avoiding Pauli fault $\tilde{\cF}$ and reweighting $\tilde{\zeta}$ for $\tilde{\cR}$. By \Cref{lem:paulift}, our goal is to show that $\tilde{\cR}[\tilde{\cF},\tilde{\zeta},\tilde{\xi}](\tilde{\sigma})$ is a $\comp{\tilde{E}}_{\mathrm{out}}$-deviation of $\widetilde{\Enc}_{\mathrm{out}}\circ\bar{\cO}(\rho)$.

  For this purpose, we begin by constructing an extended fault $\tilde{\cF}'$ under which $\tilde{\cR}$ has the same output as under $\tilde{\cF}$. We define $\tilde{\cF}'$ to use a private side register $\Ad{\tilde{N}}\cong\Qu{\tilde{N}}$ of qubits that is isomorphic to the main set of qubits used by $\tilde{\cR}$. For each gate $h\in H\setminus H_{\mathrm{run}}$, we let $\tilde{\cF}'$ simply apply the same Pauli errors as $\tilde{\cF}$ within space-time locations $\tilde{N}_h\times\tilde{T}_h\subseteq\Qu{\tilde{N}}$. However, for each $h\in H_{\mathrm{run}}$, within space-time locations $\tilde{N}_h\times\tilde{T}_h$ inside both $\Qu{\tilde{N}}$ and $\Ad{\tilde{N}}$, we let $\tilde{\cF}'$ perform the following:
  \begin{enumerate}
  \item In timestep $T_0(t_h-1)$, swap all active qubits in $\tilde{N}_h\subseteq\Qu{\tilde{N}}$ for their counterparts in $\tilde{N}_h\subseteq\Ad{\tilde{N}}$.
  \item In timesteps $\tilde{T}_h$, let $\tilde{\cF}'$ apply the corrupted sub-circuit $\tilde{\cR}_h[\tilde{\cF},\tilde{\zeta},\tilde{\xi}]$ to the qubits $\tilde{N}_h\subseteq\Ad{\tilde{N}}$, where here we implicitly restrict $\tilde{\cF},\tilde{\zeta},\tilde{\xi}$ to the appropriate space-time locations $\tilde{N}_h\times\tilde{T}_h$. Meanwhile, let $\tilde{\cF}'$ apply no errors (i.e.~only identity gates) to qubits $\tilde{N}_h\subseteq\Qu{\tilde{N}}$ in timesteps $T_0(t_h-1)+1,\dots,T_0t_h-1$.
  \item In timestep $T_0t_h$, again swap all active qubits in $\tilde{N}_h\subseteq\Qu{\tilde{N}}$ for their counterparts in $\tilde{N}_h\subseteq\Ad{\tilde{N}}$.
  \end{enumerate}
  We define an associated reweighting $\tilde{\zeta}'$ and postselection $\tilde{\xi}'$ that simply consist of the restrictions of $\tilde{\zeta}$ and $\tilde{\xi}$ respectively to space-time locations $\bigsqcup_{h\in H\setminus H_{\mathrm{run}}}\tilde{N}_h\times\tilde{T}_h\subseteq\Qu{\tilde{N}}$, as we may apply trivial reweightings and postselections to sub-circuits where we have swapped out the execution to $\Ad{\tilde{N}}$.

  The above definition of $\tilde{\cF}'$ is valid because the only gates in $\cG$ that access classical bits are measurement gates and classically-controlled Pauli gautes. But by \Cref{lem:prior}, for $G_h\in\{\gMX,\gMZ\}$, the circuit $\tilde{\cR}_h$ will immediately measure all input qubits, so there are no active qubits in $\tilde{N}_h\times\tilde{T}_h$, and hence $h\notin H_{\mathrm{run}}$. In particular, it follows that $\tilde{\xi}'$ agrees with $\tilde{\xi}$ on the restriction to subcircuits $\tilde{\cR}_h$ associated to measurement gates $h$. Meanwhile, by \Cref{def:extfault} we allow the extended fault $\tilde{\cF}'$ to apply classically-controlled Pauli gates using classical-control bits in $\Cl{\tilde{N}}$.

  \Cref{claim:swapsupp} below follows immediately from our definition of $\tilde{\cF}'$. The key point is that the restriction of $\supp(\tilde{\cF}')$ to each gadget $\tilde{\bfG}_h$ does not contain any of the gadget's bad fault error sets, so we can apply the gadget's fault-tolerance provided by \Cref{lem:prior}.
  
  \begin{claim}
    \label{claim:swapsupp}
    For every $h\in H$, the restriction of $\supp(\tilde{\cF}')$ to space-time locations $\tilde{N}_h\times\tilde{T}_h$ is of the form $E_{h,\mathrm{run}}\cup E_{h,\mathrm{swap}}$, where $E_{h,\mathrm{run}}\subseteq\tilde{N}_h\times\tilde{T}_h$ is an $\tilde{\cE}_{h,\mathrm{run}}$-avoiding set, and $E_{h,\mathrm{swap}}\subseteq\tilde{N}_h\times\{T_0t_h\}$ is the set $E_{h,\mathrm{swap}}=\bigsqcup_{a\in N_h:(a,t_h)\in E_{\mathrm{run}}}\tilde{N}_a$.
  \end{claim}

  Now by construction we have
  \begin{equation*}
    \tilde{\cR}[\tilde{\cF},\tilde{\zeta},\tilde{\xi}](\tilde{\sigma}) = \tilde{\cR}[\tilde{\cF}',\tilde{\zeta}',\tilde{\xi}'](\tilde{\sigma}),
  \end{equation*}
  as the RHS above applies the exact same circuit as the LHS, except some parts of the circuit are performed by swapping some qubits from $\Qu{\tilde{N}}$ into $\Ad{\tilde{N}}$, applying some gates on $\Ad{\tilde{N}}$, and then swapping the resulting qubits back into $\Qu{\tilde{N}}$.

  Now we may decompose each $\tilde{F}_t'$ comprising $\tilde{\cF}'=(\tilde{F}'_0,\dots,\tilde{F}'_{\tilde{T}})$ into a linear combination of Pauli superoperators $\tilde{F}_t''$, and then analyze each Pauli term separately. That is, let $\tilde{\cF}''=(\tilde{F}''_0,\dots,\tilde{F}''_T)$ be an extended Pauli fault term in the Pauli decomposition of $\tilde{\cF}'$, so that $\supp(\tilde{\cF}'')\subseteq\supp(\tilde{\cF}')$. Our goal is now to show that $\tilde{\cR}[\tilde{\cF}'',\tilde{\zeta}',\tilde{\xi}'](\tilde{\sigma})$ is a $\comp{\tilde{E}}_{\mathrm{out}}$-deviation of $\widetilde{\Enc}_{\mathrm{out}}\circ\bar{\cO}(\rho)$.

  Because a Pauli superoperator is a tensor product of Paulis operating on individual qubits, $\tilde{\cF}''$ introduces no entanglement between $\Ad{\tilde{N}}$ and $\Cl{\tilde{N}}\sqcup\Qu{\tilde{N}}$, so we may ignore qubits $\Ad{\tilde{N}}$ when considering the execution of $\tilde{\cR}[\tilde{\cF}'',\tilde{\zeta}',\tilde{\xi}'](\tilde{\sigma})$. Furthermore, as described in \Cref{def:extfault}, $\tilde{\cF}''$ will only apply Pauli $\gZ$ (and trivially, $\gI$) errors on bits in $\Cl{\tilde{N}}$ when $\tilde{\cR}$ applies classically-controlled Pauli gates. Such $\gZ$ errors on classical-control bits in $\Cl{\tilde{N}}$ simply have the effect of modifying the reweighting on the associated classically-controlled Pauli gates applied by $\tilde{\cR}$. Hence let $\tilde{\cF}'''$ denote the restriction of $\tilde{\cF}''$ to qubits in $\Qu{\tilde{N}}$, so that $\tilde{\cF}'''$ is an ordinary (i.e.~non-extended) Pauli fault for $\tilde{\cR}$. Then there exists a reweighting $\zeta''$ such that
  \begin{equation*}
    \tilde{\cR}[\tilde{\cF}''',\tilde{\zeta}'',\tilde{\xi}'](\tilde{\sigma}) = \tilde{\cR}[\tilde{\cF}'',\tilde{\zeta}',\tilde{\xi}'](\tilde{\sigma}).
  \end{equation*}
  It therefore now suffices to show that $\tilde{\cR}[\tilde{\cF}''',\tilde{\zeta}'',\tilde{\xi}'](\tilde{\sigma})$ is a $\comp{\tilde{E}}_{\mathrm{out}}$-deviation of $\widetilde{\Enc}_{\mathrm{out}}\circ\bar{\cO}(\rho)$. For this purpose, we will inductively analyze the execution of $\tilde{\cR}[\tilde{\cF}''',\tilde{\zeta}'',\tilde{\xi}'](\tilde{\sigma})$ in \Cref{claim:concatind} below.

  To state and prove \Cref{claim:concatind}, we need the following notation.
  For $t\in\{0,\dots,T\}$, let $\tilde{\cR}_{\leq T_0t}=(\tilde{R}_1,\dots,\tilde{R}_{T_0t})$ denote the restriction of $\tilde{\cR}$ to the first $T_0t$ timesteps. We similarly let $\tilde{\cR}_{\leq T_0t}[\tilde{\cF}''',\tilde{\zeta}'',\tilde{\xi}'](\cdot)$ denote the superoperator associated to $\tilde{\cR}_{\leq T_0t}$ under fault $\tilde{\cF}''$, reweighting $\tilde{\zeta}''$, and postselection $\tilde{\xi}'$, where we implicitly restrict $\tilde{\cF}'',\tilde{\zeta}'',\tilde{\xi}'$ to timesteps $\leq T_0t$. We also define $\cR_{\leq t}[\cdots]$ analogously.

  By definition, $\tilde{\cR}_{\leq T_0t}[\tilde{\cF}''',\tilde{\zeta}'',\tilde{\xi}'](\tilde{\sigma})$ is a state across the active classical bits $\Cl{\tilde{N}}_{T_0t}=\Cl{N}_t$ and qubits $\tilde{N}_{T_0t}=\bigsqcup_{a\in N_t}\tilde{N}_a'$ at timestep $T_0t$ of $\tilde{\cR}$, where here we let $\tilde{N}_a'=[\bfn(\Lev(a))]\subseteq[\nu\cdot\bfn(\Lev(a))]=\tilde{N}_a$ denote the set of qubits containing the encoding of $a\in N_t$ in $D_{\bfn(\Lev(a))}$.

  \begin{claim}
    \label{claim:concatind}
    For every $t\in[T]$, the state $\tilde{\cR}_{\leq T_0t}[\tilde{\cF}''',\tilde{\zeta}'',\tilde{\xi}'](\tilde{\sigma})$ is a linear combination of states of the form
    \begin{equation*}
      P \circ \left(\bigotimes_{a\in N_t}\Enc_{\bfn(\Lev(a))}\right) \circ \cR_{\leq t}[\cF,\zeta,\xi](\sigma),
    \end{equation*}
    where:
    \begin{enumerate}
    \item $\sigma$ is a $\comp{E}_{\mathrm{in}}$-deviation of $\Enc_{\mathrm{in}}(\rho)$.
    \item $\cF=(F_0,\dots,F_T)$ is a $\comp{E}_{\mathrm{run}}$-avoiding fault for $\cR$ supported in timesteps $<t$.
    \item $\zeta$ is the restriction of a reweighting for $\cR$ to timesteps $\leq t$.
    \item $\xi\in\bF_2^{\ps}$ is the postselection defined above.
    \item\label{it:ciP} $P$ is a Pauli superoperator acting on the active qubits $\tilde{N}_{T_0t}$ at timestep $T_0t$ of $\tilde{\cR}$, whose support satisfies the following: for every $a\in N_t$ such that $(a,t)\notin E_{\mathrm{run}}$, the set $\supp(P)\cap\tilde{N}_a'$ is $\cE_{\bfn(\Lev(a))}$-avoiding.

      Furthermore, if $t=T$, then letting $h\in H$ be the unique gate with $t_h=T$ and $a\in N_h$, we additionally have $\supp(P)\cap\tilde{N}_a'\subseteq\tilde{E}_{h,\mathrm{out}}\cap\tilde{N}_a'$. Here $\tilde{E}_{h,\mathrm{out}}$ is the output error set of the gadget $\tilde{G}_h$ from \Cref{lem:prior} under fault error set $\tilde{E}_{\mathrm{run}}$ and postselection $\tilde{\xi}$ restricted to space-time locations $\tilde{N}_h\times\tilde{T}_h$, as defined above.
    \end{enumerate}
  \end{claim}

  In \Cref{it:ciP} of \Cref{claim:concatind}, we emphasize that the restriction of $\tilde{\xi}$ to space-time locations $\tilde{N}_h\times\tilde{T}_h$ equals the restriction of $\tilde{\xi}'$ to $\tilde{N}_h\times\tilde{T}_h$. Indeed, as here we restrict attention to $(a,T)\notin E_{\mathrm{run}}$, it follows that $h\notin H_{\mathrm{run}}$, and hence $\tilde{\cF}'$ did not swap out the execution of $\tilde{\cR}_h$ to qubits in $\Ad{\tilde{N}}$. Thus we must have $\tilde{\xi}_h=\tilde{\xi}'_h$.

  \begin{proof}[Proof of \Cref{claim:concatind}]
    We show the claim by induction on the timestep $t$. Recall that $\supp(\tilde{\cF}''')\subseteq\supp(\tilde{\cF}')$, so that \Cref{claim:swapsupp} also applies with $\supp(\tilde{\cF}''')$ in place of $\supp(\tilde{\cF}')$. Then by the definition of $E_{\mathrm{in}}$ and $E_{\mathrm{run}}$ along with the mending refreshing fault-tolerance from \Cref{lem:prior} of the gadgets $\tilde{\bfG}_h$ for $h\in H$ with $t_h=1$, the $t=1$ case of \Cref{claim:concatind} must hold. In more detail, the mending refreshing fault-tolerance of the gadgets from \Cref{lem:prior} ensures that $\tilde{\cR}_{\leq T_0}[\tilde{\cF}''',\tilde{\zeta}'',\tilde{\xi}'](\tilde{\sigma})$ is a linear combination of states of the form
    \begin{equation*}
      P \circ \left(\bigotimes_{a\in N_1}\Enc_{\bfn(\Lev(a))}\right) \circ \cR_{\leq 1}[\zeta,\xi] \circ F_0 \circ L_{\mathrm{in}} \circ \Enc_{\mathrm{in}}(\rho),
    \end{equation*}
    for Pauli errors $P$, reweightings $\zeta$, and error superoperators $L_{\mathrm{in}}$ and $F_0$ supported inside $E_{\mathrm{in}}$ and inside the restriction of $E_{\mathrm{run}}$ to timestep $0$, respectively. Here $P$ is specifically guaranteed to be $\cE_{\bfn(\Lev(a))}$-avoiding on the restriction to $\tilde{N}_a$ for every $a\in N_1$ that is not in the restriction of $E_{\mathrm{run}}$ to timestep $1$. We thus let $\sigma=L\circ\Enc_{\mathrm{in}}(\rho)$, and let $F_0$ be timestep-$0$ error of our fault $\cF$. Thus the $t=1$ base case of \Cref{claim:concatind} holds, as desired.

    The inductive step follows by similar reasoning as the base case. For some $2\leq t\leq T$, assume that the timestep $t-1$ case of \Cref{claim:concatind} holds. Then the mending refreshing fault-tolerance from \Cref{lem:prior} of the gadgets $\tilde{\bfG}_h$ for $h\in H$ with $t_h=t$ ensures that $\tilde{\cR}_{\leq T_0}[\tilde{\cF}''',\tilde{\zeta}'',\tilde{\xi}'](\tilde{\sigma})$ is a linear combination of states of the form
    \begin{equation*}
      P \circ \left(\bigotimes_{a\in N_t}\Enc_{\bfn(\Lev(a))}\right) \circ \cR_{\leq t}[\cF_{<t-1},\zeta,\xi] \circ F_{t-1} \circ \Enc_{\mathrm{in}}(\rho),
    \end{equation*}
    for Pauli errors $P$, reweightings $\zeta$, $\comp{E}_{\mathrm{run}}$-avoiding faults $\cF_{<t-1}=(F_0,\dots,F_{t-1})$ as defined by previous steps of induction, and error superoperators $F_{t-1}$ supported inside the restriction of $E_{\mathrm{run}}$ to timestep $t-1$. Here $P$ is specifically guaranteed to be $\cE_{\bfn(\Lev(a))}$-avoiding on the restriction to $\tilde{N}_a$ for every $a\in N_t$ that is not in the restriction of $E_{\mathrm{run}}$ to timestep $t$. Furthermore, when $t=T$, then for every such $a\in N_T$ not in the restriction of $E_{\mathrm{run}}$ to timestep $T$, letting $h\in H$ be such that $t_h=T$ and $a\in N_h$, then we defined $\tilde{E}_{h,\mathrm{out}}$ precisely to contain the restriction of $\supp(P)$ to qubits in $\tilde{N}_h$. Thus letting $F_{t-1}$ be the timestep-$(t-1)$ error of our fault $\cF$, we have shown that the timestep $t$ case of \Cref{claim:concatind} holds, completing the inductive step.
  \end{proof}

  We now complete the proof of \Cref{lem:lcft}. By the $t=T$ case of \Cref{claim:concatind}, the state $\tilde{\cR}[\tilde{\cF}''',\tilde{\zeta}'',\tilde{\xi}'](\tilde{\sigma})$ is a linear combination of states of the form
    \begin{equation*}
      P \circ \left(\bigotimes_{a\in N_t}\Enc_{\bfn(\Lev(a))}\right) \circ \cR[\cF,\zeta,\xi](\sigma),
    \end{equation*}
    for $\sigma,\cF,\zeta,\xi,P$ defined as in \Cref{claim:concatind}. Because $\bfG$ is a concatenation-compatible fault-tolerant gadget for $\bar{\cO}$, it follows that each such state is in turn a linear combination of states of the form
    \begin{equation*}
      P \circ \left(\bigotimes_{a\in N_t}\Enc_{\bfn(\Lev(a))}\right) \circ L_{\mathrm{out}} \circ \Enc_{\mathrm{out}} \circ \bar{O}(\rho),
    \end{equation*}
    for $\comp{E}_{\mathrm{out}}$-avoiding Pauli error superoperators $L_{\mathrm{out}}$, and for $\bar{O}\in\bar{\cO}$. We can absorb the error $L_{\mathrm{out}}$ into the Pauli error $P$, to rewrite the above state as
    \begin{equation*}
      P'_X \circ P'_Z \circ \left(\bigotimes_{a\in N_t}\Enc_{\bfn(\Lev(a))}\right) \circ \Enc_{\mathrm{out}} \circ \bar{O}(\rho) = P' \circ \widetilde{\Enc} \circ \bar{O}(\rho),
    \end{equation*}
    where for $\alpha\in\{X,Z\}$, $P'_\alpha$ is a Pauli-$\alpha$ error with following property: for every $a\in N_{\mathrm{out}}$ such that $(a,T)\notin E_{\mathrm{run}}$ and $a\notin E_{\mathrm{out},\alpha}$, letting $h\in H$ be the unique gate with $t_h=T$ and $a\in N_h$, then $\supp(P_\alpha)\cap\tilde{N}'_a\subseteq\tilde{E}_{h,\mathrm{out}}\cap\tilde{N}'_a$. Now by \Cref{it:outbad} in \Cref{def:coco}, if $a\notin E_{\mathrm{out},\alpha}$, then we must have $(a,T)\notin E_{\mathrm{run}}$. Hence by the definition of $\tilde{E}_{\mathrm{out}}$ above, $P'_\alpha$ is supported inside $\tilde{E}_{\mathrm{out},\alpha}$-avoiding. Hence $P'=P'_X\circ P'_Z$ is $\comp{\tilde{E}}_{\mathrm{out}}$-avoiding. Thus we have shown that $\tilde{\cR}[\tilde{\cF}''',\tilde{\zeta}'',\tilde{\xi}'](\tilde{\sigma})$ is a $\comp{\tilde{E}}_{\mathrm{out}}$-deviation of $\widetilde{\Enc}_{\mathrm{out}}\circ\bar{\cO}(\rho)$, as desired.
\end{proof}

\end{document}